\documentclass[12pt]{iiscthes}
\usepackage{iisclogo}

\usepackage{basicreq}
\usepackage{subfiles}
\usepackage[toc,page]{appendix}

\graphicspath{{./}{../figures/}{./figures/}}

\usepackage{amsmath}
\usepackage{amssymb}
\usepackage{amsfonts}
\usepackage{amsthm}
\usepackage{epsfig}
\usepackage{calc}
\usepackage{color}
\usepackage[all]{xy}
\usepackage{bm}
\usepackage{enumerate}
\usepackage{lipsum} 
\usepackage{tikz}
\usetikzlibrary{shapes,arrows}
\usepackage{subfigure}
\usepackage{longtable}
\usepackage{graphics}
\usepackage{array}

\newcolumntype{P}[1]{>{\centering\arraybackslash}p{#1}}

\usepackage{graphicx,wrapfig,lipsum}

\newtheorem{defn}{Definition}
\newtheorem{thm}{Theorem}[section]
\newtheorem{cor}[thm]{Corollary}
\newtheorem{prop}{Proposition}
\newtheorem{lem}[thm]{Lemma}
\newtheorem{conj}[thm]{Conjecture}
\newtheorem{const}[thm]{Construction}
\newtheorem{note}{Remark}
\newtheorem{eg}{Example}

\newcommand{\bit}{\begin{itemize}}
	\newcommand{\eit}{\end{itemize}}
\newcommand{\bcor}{\begin{cor}}
	\newcommand{\ecor}{\end{cor}}
\newcommand{\beq}{\begin{equation}}
\newcommand{\eeq}{\end{equation}}
\newcommand{\beqn}{\begin{equation*}}
\newcommand{\eeqn}{\end{equation*}}
\newcommand{\bea}{\begin{eqnarray}}
\newcommand{\eea}{\end{eqnarray}}
\newcommand{\bean}{\begin{eqnarray*}}
	\newcommand{\eean}{\end{eqnarray*}}
\newcommand{\ben}{\begin{enumerate}}
	\newcommand{\een}{\end{enumerate}}
\newcommand{\bc}{\begin{center}}
	\newcommand{\ec}{\end{center}}
\newcommand{\bdefn}{\begin{defn}}
	\newcommand{\edefn}{\end{defn}}
\newcommand{\bnote}{\begin{note}}
	\newcommand{\enote}{\end{note}}
\newcommand{\bprop}{\begin{prop}}
	\newcommand{\eprop}{\end{prop}}
\newcommand{\blem}{\begin{lem}}
	\newcommand{\elem}{\end{lem}}
\newcommand{\bthm}{\begin{thm}}
	\newcommand{\ethm}{\end{thm}}
\newcommand{\bconj}{\begin{conj}}
	\newcommand{\econj}{\end{conj}}
\newcommand{\bconstr}{\begin{constr}}
	\newcommand{\econstr}{\end{constr}}
\newcommand{\bpf}{\begin{proof}}
	\newcommand{\epf}{\end{proof}}

\newcommand{\gmds}{\mbox{$G_{\text{\tiny MDS}}$}}

\newcommand{\fq}{\ensuremath{\mathbb{F}_q}}

\newcommand{\gpyr}{\mbox{$G_{\text{\tiny PYR}}$}}

\newcommand{\uc}[1]{\ensuremath{\underline{c}_{#1}}}
\newcommand{\ug}{\mbox{$\underline{g}$}}
\newcommand{\fqstar}{\mbox{$\mathbb{F}_q^{*}$}}

\newcommand{\vzero}{\mbox{$V_0$}}
\newcommand{\vone}{\mbox{$V_1$}}
\newcommand{\vtwo}{\mbox{$V_2$}}
\newcommand{\vi}{\mbox{$V_i$}}
\newcommand{\vs}{\mbox{$V_{s}$}}
\newcommand{\vsmone}{\mbox{$V_{s-1}$}}
\newcommand{\vinfty}{\mbox{$V_{\infty}$}}

\newcommand{\eone}{\mbox{$E_1$}}
\newcommand{\etwo}{\mbox{$E_2$}}

\newcommand{\ej}{\mbox{$E_j$}}
\newcommand{\es}{\mbox{$E_s$}}
\newcommand{\espone}{\mbox{$E_{s+1}$}}
\newcommand{\esmone}{\mbox{$E_{s-1}$}}

\newcommand{\gzero}{\mbox{${\cal G}_0$}}
\newcommand{\gbase}{\mbox{${\cal G}_{\text{base}}$}}
\newcommand{\gexp}{\mbox{${\cal G}_{\text{exp}}$}}
\newcommand{\gone}{\mbox{${\cal G}_1$}}
\newcommand{\gtwo}{\mbox{${\cal G}_2$}}
\newcommand{\gi}{\mbox{${\cal G}_i$}}
\newcommand{\gs}{\mbox{${\cal G}_s$}}
\newcommand{\nbase}{\mbox{$N_{\text{base}}$}}
\newcommand{\naux}{\mbox{$N_{\text{aux}}$}}

\newcommand{\gsmone}{\mbox{${\cal G}_{s-1}$}}
\newcommand{\gsub}{\mbox{${\cal G}_{\text{sub}}$}}
\newcommand{\aux}{\mbox{${\cal A}$}}
\newcommand{\auxi}{\mbox{${\cal A}_i$}}

\newcommand{\heven}{\mbox{$H_{2s+2}$}}
\newcommand{\hevena}{\mbox{$H^{(a)}_{2s+2}$}}
\newcommand{\hodd}{\mbox{$H_{2s+1}$}}
\newcommand{\hodda}{\mbox{$H^{(a)}_{2s+1}$}}

\newcommand{\fqalpha}{\mbox{$\mathbb{F}_q^{\alpha}$}}
\newcommand{\fqbeta}{\mbox{$\mathbb{F}_q^{\beta}$}}
\newcommand{\fqB}{\mbox{$\mathbb{F}_q^B$}}
\newcommand{\uu}{\mbox{$\underline{u}$}}
\newcommand\MyLBrace[2]{%
	\left.\rule{0pt}{#1}\right\}\text{#2}}
\newcommand{\ux}[1]{\ensuremath{\underline{x}_{{#1}}}}
\newcommand{\uxx}{\ensuremath{\underline{x} }}

\newcommand{\hmds}{\ensuremath{H_{\text{MDS}} \ }}

\newcommand{\code}{\ensuremath{{\cal C} \ }}

\newcommand{\calc}{\ensuremath{{\cal C}\ }}
\newcommand{\dmin}{\ensuremath{d_{\min}\ }}

\newcommand\scalemath[2]{\scalebox{#1}{\mbox{\ensuremath{\displaystyle #2}}}}

\begin{document}

\begin{frontmatter}
%
%

\title{Erasure Codes for Distributed Storage: Tight Bounds and Matching Constructions}
\author{Balaji S.B.}
\advisers{P. Vijay Kumar}
\submitdate{June 2018}
\dept{Electrical Communication Engineering}
\enggfaculty
\iisclogotrue 
\tablespagetrue 
 \maketitle

\begin{dedication}
\begin{center}
{\em Dedicated to} \\[2em]
{\em My Mother, Father, Brother,}\\
	 {\em My Brother's wife and Little Samyukhtha.}
\end{center}
\end{dedication}

%
\clearpage
\vspace*{\fill}
\begin{center}
\begin{minipage}{.6\textwidth}
\medskip

\emph{If numbers aren't beautiful, I don't know what is.}

\flushright{--- Paul Erd\H{o}s}
\end{minipage}
\end{center}
\vfill 
\clearpage

\acknowledgements
I would like to thank my mother, father, brother for being kind to me during tough times. I would like to thank my advisor Prof. P. Vijay Kumar for helping me to continue my PhD during a rough patch of time.
I would like to thank all my well wishers who helped me technically or non-technically duing my tenure as a PhD student at IISc.
I would like to thank my collabotators Prashanth, Ganesh and Myna with whom i had the pleasure of working with. I enjoyed the technical discussions with them. I would like to thank all my course instructors (in alphabetical order): Prof. Chiranjib Bhattacharyya, Prof. Navin Kashyap, Prof. Pooja Singla, Prof. Pranesachar, Prof. Shivani Agarwal, Prof. Sundar Rajan, Prof. Sunil Chandran, Prof. Thanngavelu, Prof. Venkatachala. I would like to thank all the professors and students and all people who helped me directly or indirectly.
The thought process i am going through now while doing research is a combination of my efforts and the thought process of my advisor. I picked up some of the thought process from my advisor P. Vijay Kumar like trying to break down any idea into its simplest form possible for which i am grateful. I also would like to thank my advisor in heping me writing this thesis partly. My writing skills and presentation skills greatly improved (although still not great) because of the teachings of my advisor. 
I have no friends in IISc. So my tenure as a PhD student was an extremely tough one. I would like to thank my family for taking the trouble to shift to Bangalore and stay by my side.
Whatever little intelligence i have is attributed to all the above. I also would like to thank my labmates Bhagyashree and Vinayak. I shared many heated conversations with Vinayak which in hindsight was an enjoyable one. I also would like to thank Shashank, Manuj, Mahesh, Anoop Thomas, Nikhil, Gautham Shenoy,  Birenjith, Myna for sharing a trip to ISIT with me. I would like to thank Anoop Thomas for always talking in an encouraging tone. I would like to thank Shashank, Vinayak and Avinash for sharing little tehnical conversations with me after attending talks at IISc. I would like to thank Lakshmi Narasimhan with whom i shared many conversations during my PhD. I also would like to thank Samrat for sharing some tehnical conversations with me after the math classes. I also would like to thank students who attended graph theory course with me as it was a sparse class and i shared many technical conversations with them. I wish to thank all the people at IISc who treated me kindly. Finally i would like to thank Anantha, Mahesh,  Aswin who are going to share a trip to ISIT with me this month. If someone's name is left out in the above, it is not intentional.

%
%
%
%
%
\begin{abstract}

\paragraph{The Node Repair Problem}

In the distributed-storage setting, data pertaining to a file is stored across spatially-distributed nodes (or storage units) that are assumed to fail independently. The explosion in amount of data generated and stored has caused renewed interest in erasure codes as these offer the same level of reliability (recovery from data loss) as replication, with significantly smaller storage overhead. For example, a provision for the use of erasure codes now exists in the latest version of the popular Hadoop Distributed File System, Hadoop 3.0. It is typically the case that code symbols of a codeword in an erasure code is distributed across nodes. This `Big-Data' setting also places a new and additional requirement on the erasure code, namely that the code must enable the efficient recovery of a single erased code symbol. An erased code symbol here corresponds to a failed node and recovery from single-symbol erasure is termed as node repair. Node failure is a common occurrence in a large data center and the ability of an erasure code to efficiently handle node repair is a third important consideration in the selection of an erasure code. Node failure is a generic term used to describe not just the physical failure of a node, but also its non-availability for reasons such as being down for maintenance or simply being busy serving other, simultaneous demands on its contents.  Parameters relevant to node repair are the amount of data that needs to be downloaded from other surviving (helper) nodes to the replacement of the failed node, termed the {\em repair bandwidth} and the number of helper nodes contacted, termed the {\em repair degree}. Node repair is said to be efficient if either repair bandwidth or repair degree is less.

\paragraph{Different Approaches to Node Repair} \ The conventional node repair of the ubiquitous Reed-Solomon (RS) code is inefficient in that in an $[n,k]$ RS code having block length $n$ and dimension $k$, the repair bandwidth equals $k$ times the amount of data stored in the replacement node and the repair degree equals $k$, both of which are excessively large.  In response, coding theorists have come up with different approaches to handle the problem of node repair.  Two new classes of codes have sprung up, termed as {\em regenerating (RG) codes} and {\em locally recoverable (LR) codes} respectively that provide erasure codes which minimize respectively the repair bandwidth and repair degree.  A third class termed as {\em locally regenerating (LRG) codes}, combines the desirable features of both RG and LR codes and offers both small repair bandwidth as well as a small value of repair degree. 

In a different direction, coding theorists have taken a second, closer look at the RS code and have devised efficient approaches to node repair in RS code.  Yet another direction is that adopted by {\em liquid storage} codes which employ a {\em lazy} strategy  approach to node repair to achieve a fundamental bound on information capacity.

\begin{figure}[h!]
	\centering
	\includegraphics[width=7.25in]{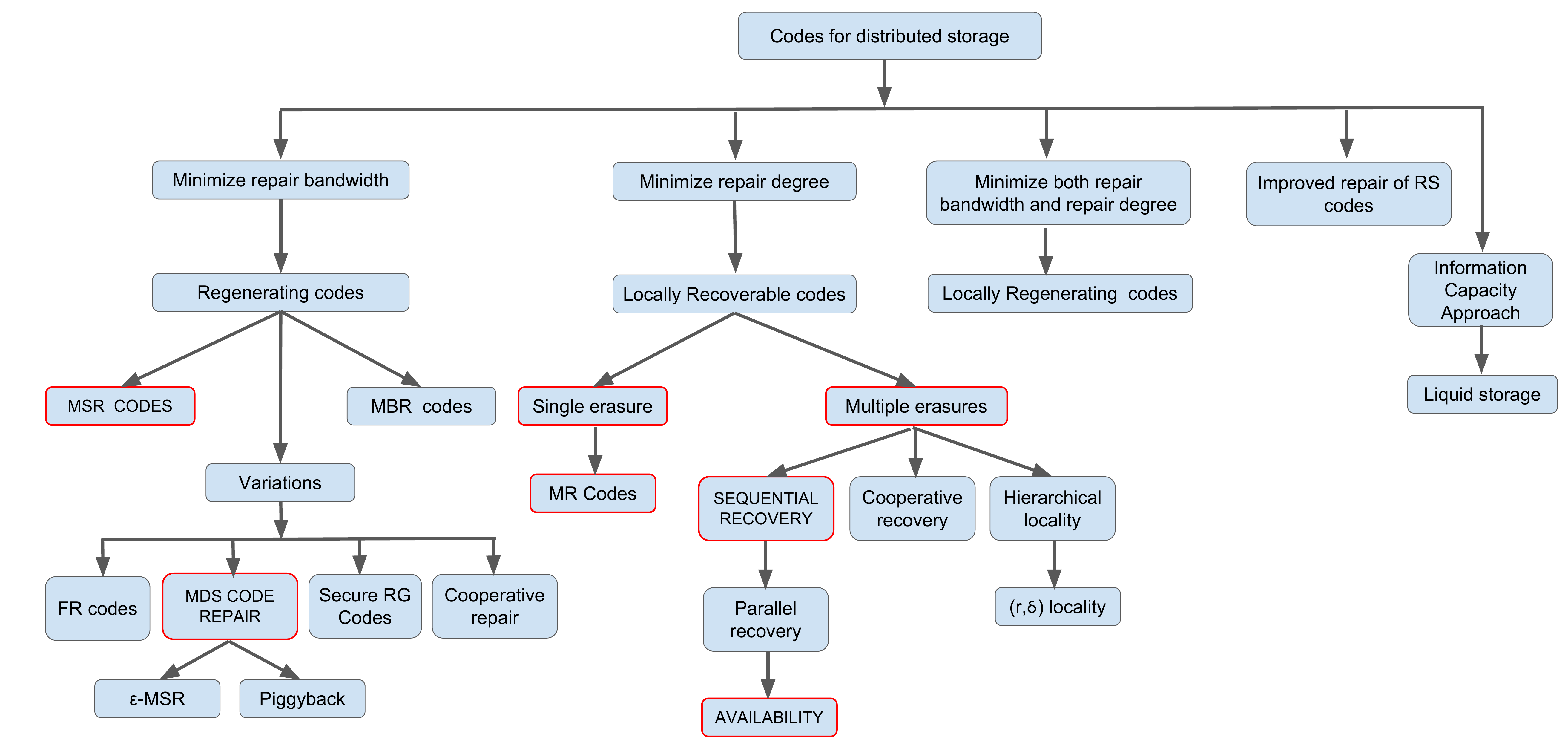}
	\caption{Flowchart depicts various techniques used to handle node repair. The boxes outlined in red indicate topics to which this thesis has contributed. The primary contributions of this thesis correspond to boxes outlined in red that is fully written in upper-case lettering. In the flowchart above, FR codes stands for fractional repetition codes and MBR codes stands for Minimum Bandwidth RG codes. Rest of the abbreviations are either common or defined in the text. }\label{fig:overview}
\end{figure}

Fig.~\ref{fig:overview} provides a classification of the various classes of codes that have been developed by coding theorists to address the 
problem of node repair.  The boxes outlined in red in Fig.~\ref{fig:overview}, correspond to the classes of codes towards which this thesis has made significant contributions.  The primary contributions are identified by a box that is outlined in red which is fully written in upper-case lettering. 

\paragraph{Contributions of the Thesis}

\begin{figure}[ht!]
	\centering
	\includegraphics[width=7in]{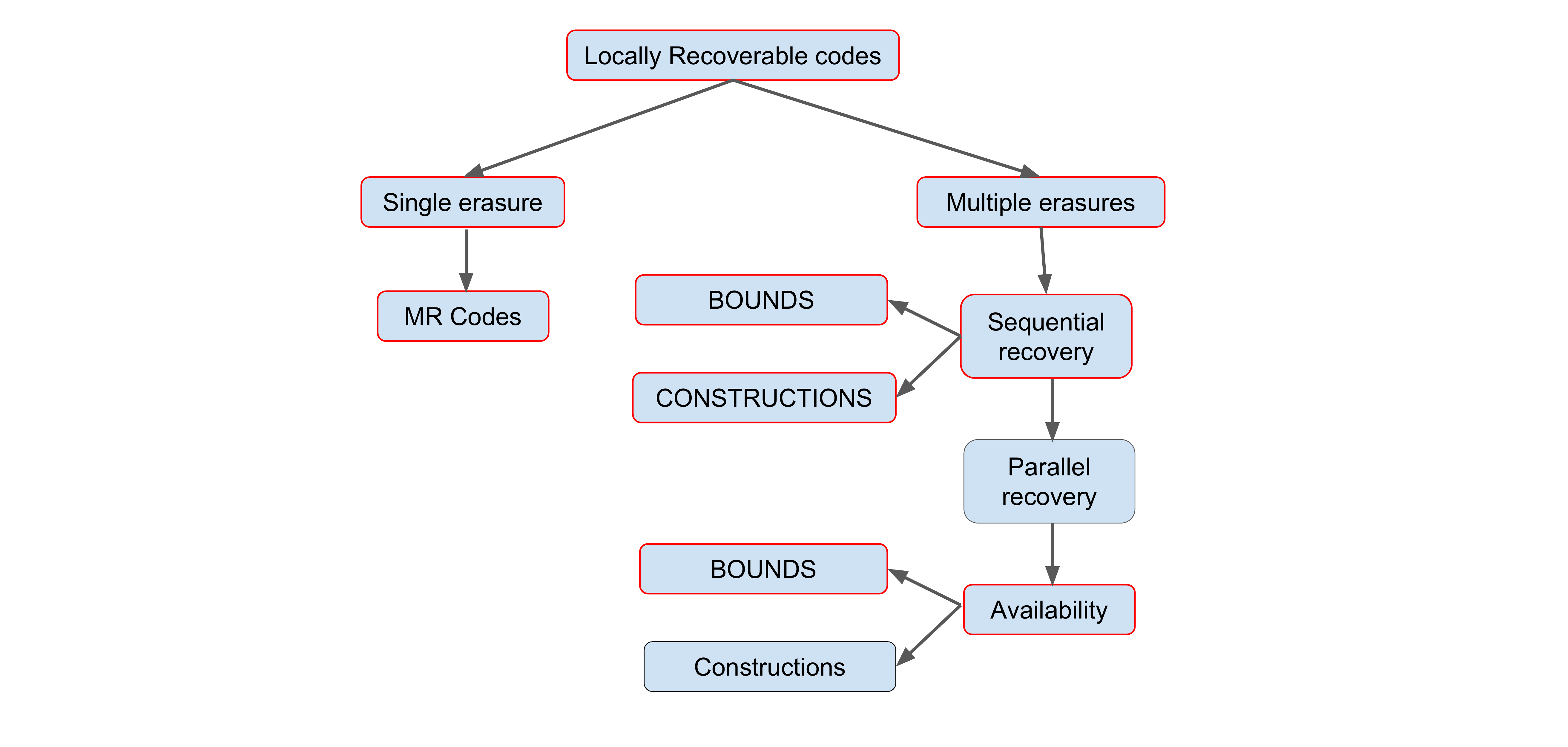}
	\caption{The flowchart shown here is an extensions of the portion of the flowchart appearing in Fig.~\ref{fig:overview} correspoding to LR codes. The boxes outlined in red indicate topics to which this thesis has contributed. The primary contributions of this thesis correspond to boxes outlined in red that is fully written in upper-case lettering.}\label{fig:flowchart_lr}
\end{figure}

As noted above, this thesis makes contributions that advance the theory of both RG and LR codes.  We begin with LR codes, since the bulk of the contributions of the thesis relate to this class of codes (see Fig.~\ref{fig:flowchart_lr}). In the following, a local code refers to a punctured code of an $[n,k]$ code of block length $\leq r+1$ and dimension strictly less than block length where $r < n$.

\paragraph{Contributions to LR Codes} \  LR codes are designed with the objective of reducing the repair degree and accomplish this by making sure that the overall erasure code has several local codes 
in such a way that any single erased code symbol can be repaired by containing at most $r$ other code symbols. The parameter $r$ is termed the {\em locality parameter} of the LR code.  Our contributions in the direction of an LR code are aimed at LR codes that are capable of handling multiple erasures efficiently.  Improved bounds on both the rate $R=k/n$ of an LR code as well as its minimum Hamming distance $d_{\min}$ are provided. We provide improved bounds under a constraint on the size $q$ of the code-symbol alphabet and also provide improved bounds without any constraint on the size $q$ of the code-symol alphabet. 

\vspace*{0.1in}

\paragraph{LR Codes for Multiple Erasures} \ The initial focus in the theory of LR codes was the design of codes that can recover from the erasure of a single erased symbol efficiently.  
Given that constructions that match the bounds on performance metrics on LR codes are now available in the literature, the attention of the academic community has since shifted in the direction of the design of LR codes to handle {\em multiple erasures}. An LR code is said to recover multiple erasures if it can recover from multiple simulataneus erasures by accessing a small number $<k$ of unerased code symbols. A strong motivation for developing the theory of LR codes which can handle multiple-erasure comes from the notion of {\em availability} because a sub class of LR codes called $t$-availability codes has the ability to recover from $t$ simultaneous erasures and also has the interesting property called availability which is explained in the following. In a data center, there could be storage units that hold popular data for which there could be several simultaneous competing demands.  In such situations, termed in the industry as a {\em degraded read}, the single-node repair capability of an erasure code is called upon to recreate data that is unavailable on account of multiple, competing demands for its data.  This calls for an ability to recreate multiple copies, say $t$, of the data belonging to the unavailable node.  To reduce latency, these multiple recreations must be drawn from disjoint sets of code symbols. This property called availability is achieved by $t$-availability codes. An LR code constructed in such a way that for each code symbol $c_i$ there is a set of $t$ local codes such that any two out of the $t$ local codes have only this code symbol $c_i$ in common, is termed as a {\em $t$-availability code}. 

\vspace*{0.1in}

\paragraph{Contributions to Availability Codes} \ The contributions of the thesis in the direction of $t$-availability codes include improved upper bounds on the minimum distance $d_{\min}$ of this class of codes, both with and without a constraint on the size $q$ of the code-symbol alphabet. An improved upper bound on code rate $R$ is also provided for a subclass of $t$-availability codes, termed as codes with {\em strict availability}. Among the class of $t$-availability codes, codes with strict availability typically have high rate. A complete characterization of optimal tradeoff between rate and fractional minimum distance for a special class of $t$-availability codes is also provided.

\vspace*{0.1in}

\paragraph{Contributions to LR Codes with Sequential Recovery} \ Since a $t$-availability code also has the ability to recover from $t$ simultaneous erasures. This leads naturally to the study of other LR codes that can recover from multiple, simultaneous erasures. There are several approaches to handling multiple erasures. We restrict ourselves to a subclass of LR codes which can recover from multiple erasures where we use atmost $r$ symbols for recovering an erased symbol. Naturally the most general approach in this subclass of LR codes is one in which the LR code recovers from a set of $t$ erasures by repairing them one by one in a sequential fashion, drawing at each stage from at most $r$ other code symbols. Such codes are termed as LR codes with {\em sequential recovery} and quite naturally, have the largest possible rate of any LR code that can recover from multiple erasures in the subclass of LR codes we are considering. A major contribution of the thesis is the derivation of a new tight upper bound on the rate of an LR code with sequential recovery.   While the upper bound on rate for the cases of $t=2,3$ was previously known, the upper bound on rate for $t \geq 4$ is a contribution of this thesis. This upper bound on rate proves a conjecture on the maximum possible rate of LR codes with sequential recovery that had previously appeared in the literature, and is shown to be tight by providing construction of codes with rate equal to the upper bound for every $t$ and every $r \geq 3$.    

Other contributions in the direction of codes with sequential recovery, include identifying instances of codes arising from a special sub-class of $(r+1)$-regular graphs known as Moore graphs, that are optimal not only in terms of code rate, but also in terms of having the smallest block length possible. Unfortunately, Moore graph exists only for a restricted set of values of $r$ and girth (length of cycle with least number of edges in the graph). This thesis also provides a characterization of codes with sequential recovery with rate equal to our upper bound for the case $t=2$ as well as an improved lower bound on block length for the case $t=3$.  

\paragraph{Contributions to RG Codes} 

\begin{figure}[ht!]
	\centering
	\includegraphics[width=6in]{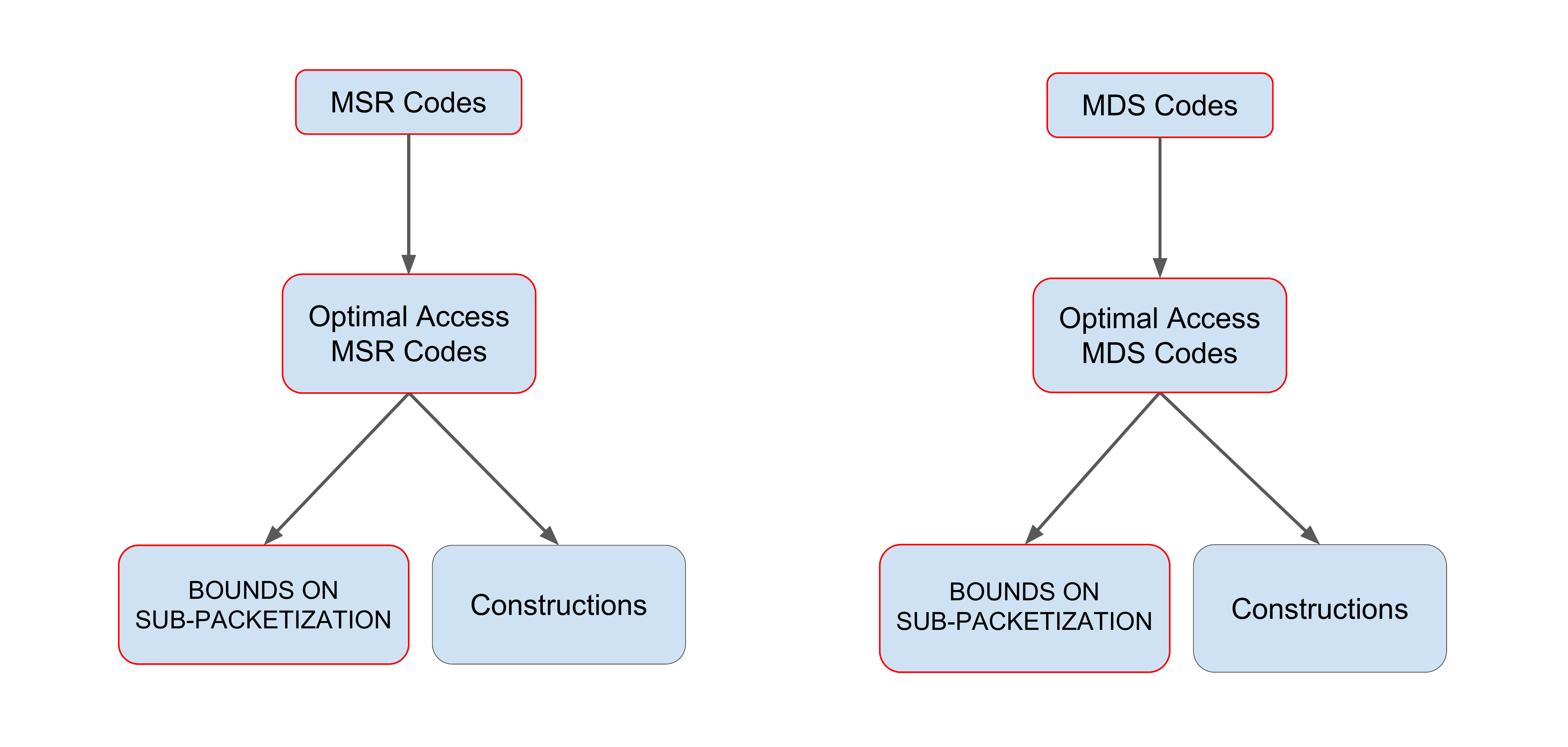}
	\caption{The two flowcharts shown here are extensions of the flowchart appearing in Fig.~\ref{fig:overview} corresponding to MSR and MDS codes. The boxes outlined in red indicate topics to which this thesis has contributed. The primary contributions of this thesis correspond to boxes outlined in red that is fully written in upper-case lettering.}\label{fig:detail}
\end{figure}

An RG code derives its ability to minimize the repair bandwidth while handling erasures from the fact that these codes are built over a vector symbol alphabet for example over $\mathbb{F}^{\alpha}_q$ for a finite field $\mathbb{F}_q$ and some $\alpha \geq 1$. The necessary value of the size $\alpha$ of this vector alphabet, also termed as the sub-packetization level of an RG code, tends to grow very large as the rate of the RG code approaches $1$, corresponding to a storage overhead which also approaches $1$.  In practice, there is greatest interest in high-rate RG codes and hence there is interest in knowing the minimum possible value of $\alpha$. An {\em optimal-access} RG code is an RG code in which during node repair, the number of scalar code symbols accessed at a helper node (i.e., for example the number of symbols over $\mathbb{F}_q$ accessed in a vector code symbol from $\mathbb{F}^{\alpha}_q$) equals the number of symbols passed on by the helper node to the replacement node. A node repair satisfying this property is called repair by help-by-transfer. This has the practical importance that no computation is needed at a helper node. The number of helper nodes contacted during a node repair is usually denoted by $d$ in the context of RG codes. A sub-class of optimal access RG codes called optimal-access Minimum Storage RG (MSR) codes refers to optimal access RG codes which are also vector MDS codes. In this thesis, we provide a tight lower bound on the sub-packetization level of optimal-access MSR codes. We do the same for Maximum Distance Separable (MDS) codes over a vector alphabet, which are designed to do repair by help-by-transfer for repairing any node belonging to a restricted subset of nodes, with minimum possible repair bandwidth. We refer to these codes as optimal access MDS codes. In both cases, we point to the literature on sub-packetization level of existing RG codes to establish that the bounds on sub-packetization level derived in this thesis are tight. See Fig. \ref{fig:detail} for a summary of contributions of the thesis to RG codes. The equations used in the derivation of our lower bound on sub-packetization level $\alpha$ in the case of optimal access MDS codes, also provides information on the structure of such codes with sub-packetization level equal to our lower bound.   The suggested structure is present in a known construction of a high-rate, optimal-access MSR code having least possible sub-packetization level. 

\paragraph{Contributions to Maximal Recoverable (MR) Codes} Returning to the topic of LR codes, we note that an LR code is constrained by a set of parity checks that give the code the ability to recover from the erasure of any given code symbol by connecting to at most $r$ other code symbols. We call this set of parity checks as local parity checks. The code with only local parity checks imposed on it typically result in a code of dimension $k_0$ that is larger than the dimension of the desired code. Thus one has the option of adding additional parity checks to bring the dimension down to $k$. Naturally, these additional parity checks are added so as to give the code the ability to recover from additional erasure patterns as well, for example recover from any pattern of $s \geq 2$ erasures, without any constraint on the number of helper nodes contacted during the recovery from this larger number of erasures. MR codes are the subclass of LR codes which given local parity checks and the desired overall dimension $k$, have the ability to recover from all possible erasure patterns which are not precluded by the local parity checks. It is an interesting and challenging open problem to construct MR codes having code symbol alphabet of small size. Our contributions in this area are constructions of MR codes over finite fields of small size.

There are several other contributions in the thesis, that are not described here for lack of space. 
	
\end{abstract}

\publications
%
{\bf Conference}
\begin{enumerate}
	\item S. B. Balaji and P. V. Kumar, "A tight lower bound on the sub-packetization level of optimal-access MSR and MDS codes,'' {\em CoRR, (Accepted at ISIT 2018), vol. abs/1710.05876}, 2017.
	\item M. Vajha, S. B. Balaji, and P. V. Kumar, "Explicit MSR Codes with Optimal Access, Optimal Sub-Packetization and Small Field Size
	for $d = k + 1; k + 2; k + 3$,'' { \em CoRR (Accepted at ISIT 2018), vol. abs/1804.00598 }, 2018.
	\item S. B. Balaji, G. R. Kini, and P. V. Kumar, "A Rate-Optimal Construction of Codes with Sequential Recovery with Low Block Length,'' {\em CoRR, (Accepted at NCC 2018), vol. abs/1801.06794}, 2018.
	\item S. B. Balaji and P. V. Kumar, "Bounds on the rate and minimum distance of codes with availability,'' { \em in IEEE International
		Symposium on Information Theory (ISIT), June 2017}, pp. 3155-3159.
	\item S. B. Balaji, G. R. Kini, and P. V. Kumar, "A tight rate bound and a matching construction for locally recoverable codes with sequential recovery from any number of multiple erasures,'' {\em in IEEE International Symposium on Information Theory (ISIT), June 2017}, pp.
	1778-1782.
	\item S. B. Balaji, K. P. Prasanth, and P. V. Kumar, "Binary codes with locality for multiple erasures having short block length,'' {\em in IEEE International Symposium on Information Theory (ISIT), July 2016}, pp. 655-659.
	\item S. B. Balaji and P. V. Kumar, "On partial maximally-recoverable and maximally-recoverable codes,'' {\em in IEEE International Symposium on Information Theory (ISIT), Hong Kong, 2015}, pp. 1881-1885.
\end{enumerate}




%
\notations
 \begin{tabular}{|>{\centering}p{3cm}|p{11cm}|}
 	\multicolumn{2}{c}{}\\ \hline
 	$[\ell]$ & The set $\{1,2,...,\ell\}$\\
 	$[\ell]^t$ & The cartesian product of $[\ell]$, $t$ times i.e., $[\ell] \times ....\times [\ell]$\\
 	$\Z$ & The set of integers\\
 	$\mathbb{N}$ & The set of natural numbers $\{1,2,...\}$\\
 	$S^c$ & Complement of the set $S$ \\
 	$\bF_q$ & The Finite field with $q$ elements\\
 	$\mathcal{C}$ & A linear block code \\
 	$\dim(\mathcal{C})$ & The dimension of the code $\mathcal{C}$\\
 	$n$ & Block length of a code \\
 	$k$ & Dimension of a code \\
 	$d _{\min}$ & Minimum distance of a code \\
 	$d$ & Minimum distance of a code or parameter of a regenerating code \\
 	$r$ & Locality parameter of a Locally Recoverable code \\
 	$[n,k]$ & Parameters of a linear block code \\
 	$[n,k,d]$ & Parameters of a linear block code \\
 	$(n,k,d)$ & Parameters of a nonlinear code \\
 	$(n,k,r,t)$ & Parameters of a $[n,k]$ Sequential-recovery LR code or an availability code with locality parameter $r$ for $t$ erasure recovery\\
 	$(n,k,r,d)$ & Parameters of an $[n,k]$ LR code with locality parameter $r$ and minimum distance $d$.\\
 	\hline
 \end{tabular}
 
 \begin{tabular}{|>{\centering}p{3cm}|p{11cm}|}
 	\multicolumn{2}{c}{}\\ \hline
 	$supp(\underline{c})$ & Support of the vector $\underline{c}$ \\
 	$supp(\mathcal{D})$ & Support of the subcode $\mathcal{D}$ i.e., the set $\cup_{\{\underline{c} \in \mathcal{D}\}} supp(\underline{c})$ \\
 	$\mathcal{C}^{\perp}$ & The dual code of the code $\mathcal{C}$ \\
 	$\mathcal{C}|_{S}$ & A punctured code obtained by puncturing the code $\mathcal{C}$ on $S^c$ \\
 	$H$ & A parity check matrix of a code \\
 	$\underline{c}$ & A code word of a code\\
 	$c_i$ & $i$th code symbol of a code word\\
 	$d_i$ & $i$th Generalized Hamming Weight of a code\\
 	$V(G)$ & Vertex set of a graph $G$ \\
 	$\{ (n,k,d),$ & Parameters of a regenerating code \\
 	$(\alpha,\beta),\ B) \}$ & \\
 	$(r, \delta, s)$ & Parameters of a Maximal Recoverable code \\
 	$<A>$ & Row space of the matrix $A$\\
 	$VA$ & The set $\{\underline{v}A: \underline{v} \in V\}$ for a vector space $V$ and matrix $A$ \\
 	$A^T$ & The Transpose of the matrix $A$ \\
 	\hline
 \end{tabular}

\abbreviations
       \begin{tabular}{|>{\centering}p{3cm}|p{11cm}|}
       	\multicolumn{2}{c}{}\\ \hline
       	LR & Locally Recoverable \\
       	IS & Information Symbol \\
       	AS & All Symbol \\
       	S-LR & Sequential-recovery LR \\
       	SA & Strict Availability \\
       	MDS & Maximum Distance Separable \\
       	RG & Regenerating \\
       	MSR & Minimum Storage Regenerating \\
       	PMR & Partial Maximal Recoverable \\
       	MR & Maximal Recoverable \\
       	$gcd(a,b)$ & Greatest Common Divisor of $a$ and $b$\\  
       	\hline
       \end{tabular}

\newpage
\pdfbookmark[section]{\contentsname}{toc}
\makecontents

\end{frontmatter}

\chapter{Introduction}

\section{The Distributed Storage Setting}

In a distributed storage system, data pertaining to a single file is spatially distributed across nodes or storage units (see Fig.~\ref{fig:dss}). Each node stores a large amounts of data running into the terabytes or more.   A node could be in need of repair for several reasons including (i) failure of the node, (ii) the node is undergoing maintenance or (iii) the node is busy serving other demands on its data.  For simplicity, we will refer to any one of these events causing non-availability of a node, as {\em node failure}.  It is assumed throughout, that node failures take place independently.

\begin{figure}[ht!]
	\centering
	\begin{minipage}[c]{0.3\textwidth}
		\centering
		\includegraphics[width=2.in]{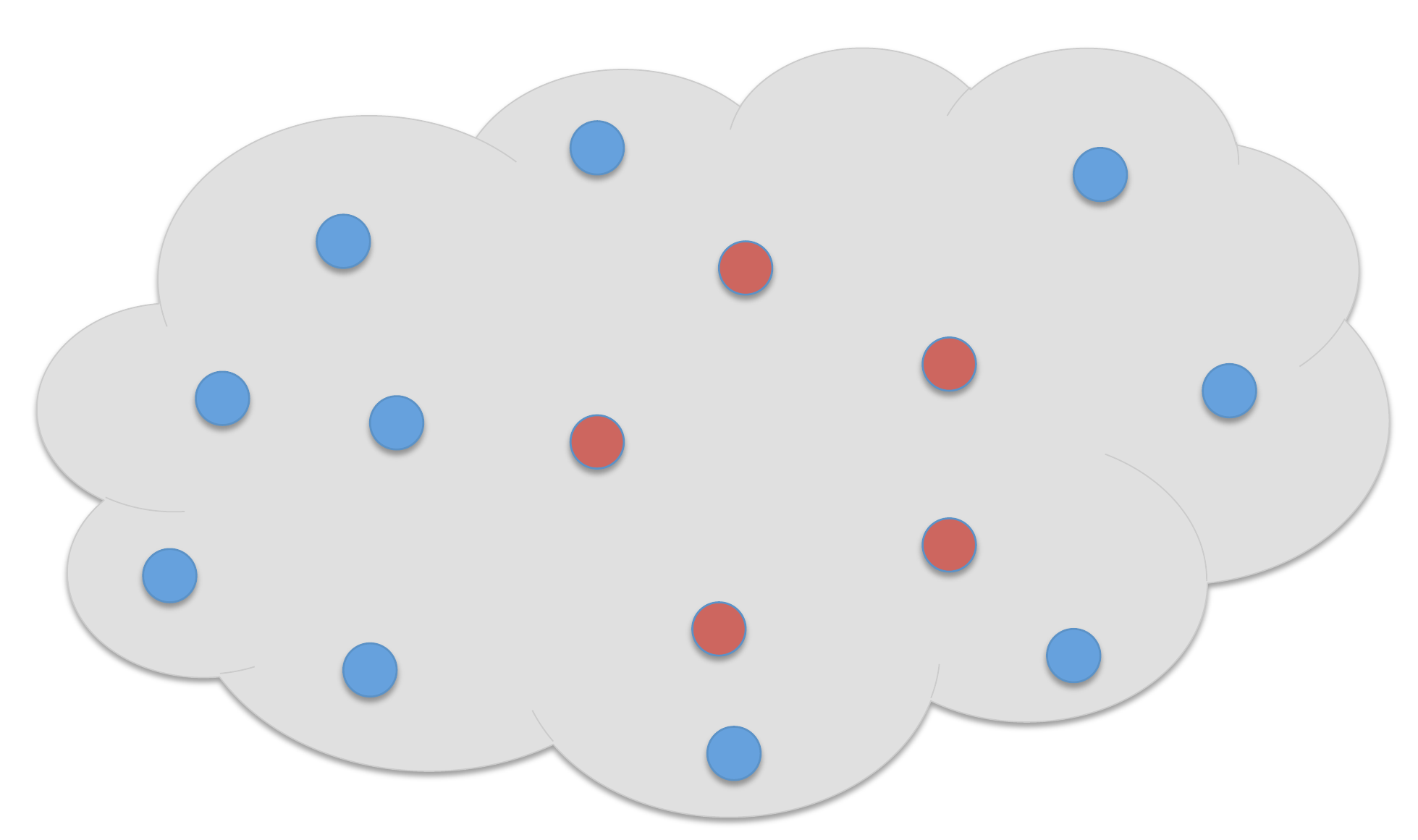}
	\end{minipage}
	\hspace{0.15\textwidth}
	\begin{minipage}[c]{0.3\textwidth}
		\centering
		\includegraphics[width=2.in]{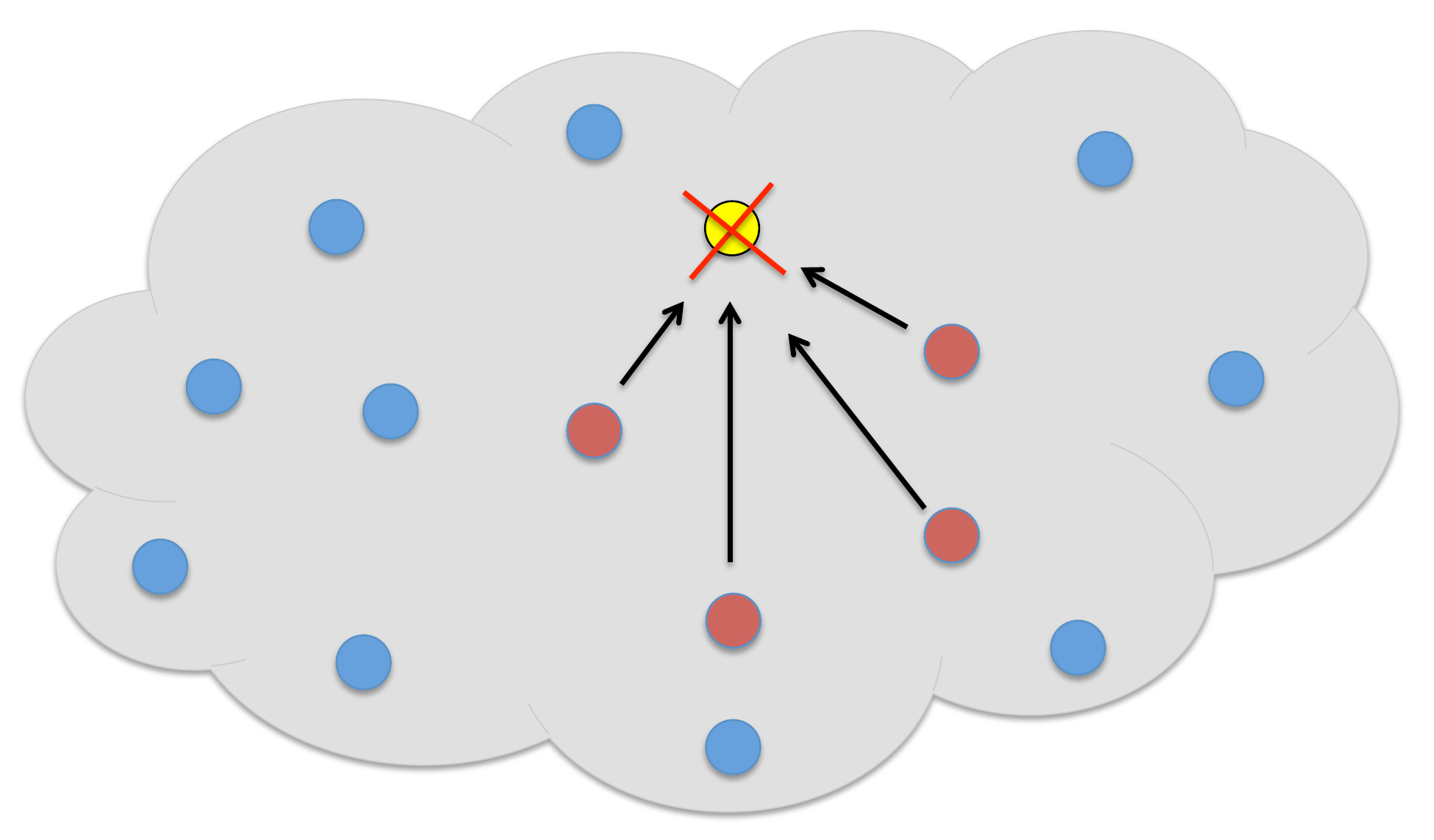}
	\end{minipage}
	\caption{In the figure on the left, the $5$ nodes in red are the nodes in a distributed storage system, that store data pertaining to a given data file.  The figure on the right shows an instance of node failure (the node in yellow), with repair being accomplished by having the remaining $4$ red (helper) nodes pass on data to the replacement node to enable node repair.} \label{fig:dss}
\end{figure}

In \cite{RashmiShahGuKuang} and \cite{SathiaAstPap_Xorbas}, the authors study the Facebook warehouse cluster and analyze the frequency of node failures as well as the resultant network traffic relating to node repair. It was observed in \cite{RashmiShahGuKuang} that a median of $50$ nodes are unavailable per day and that a median of $180$TB of cross-rack traffic is generated as a result of node unavailability (see Fig.~\ref{fig:fbdata}). 

\begin{figure}[h!]
	\centering
	\begin{minipage}[c]{0.4\textwidth}
		\centering
		\includegraphics[width=2.5in]{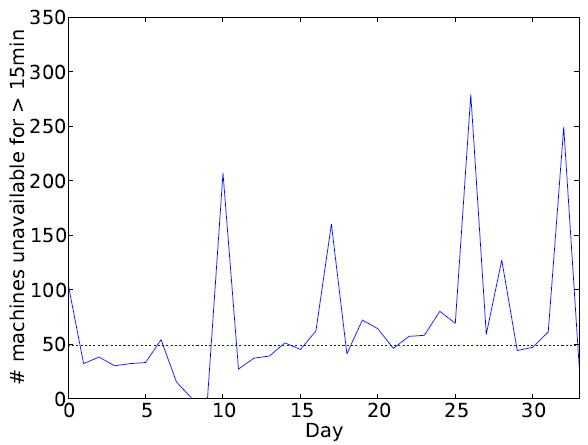}
	\end{minipage}
	\hspace{0.02\textwidth}
	\begin{minipage}[c]{0.5\textwidth}
		\centering
		\includegraphics[width=3in]{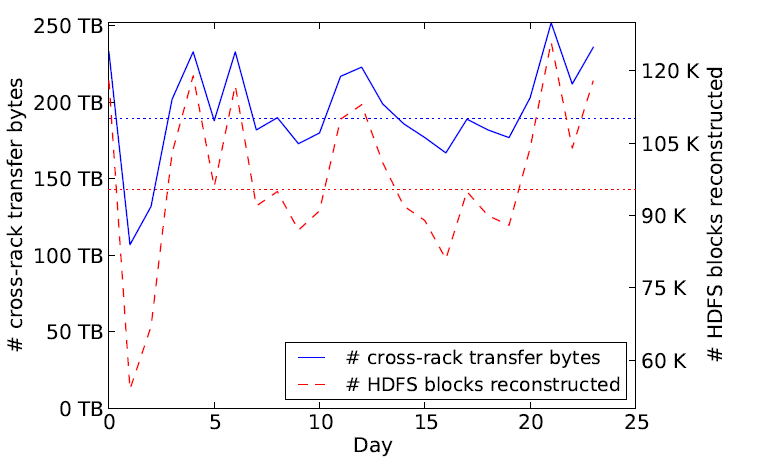}
	\end{minipage}
	\caption{The plot on the left shows the number of machines unavailable for more than
		$15$ minutes in a day, over a period of $34$ days. Thus, a median of in excess of $50$ machines become unavailable
		per day~\cite{RashmiShahGuKuang}.   The plot on the right is of the cross rack traffic generated as well as the number of Hadoop Distributed File System (HDFS) blocks reconstructed as a result of unavailable nodes.  The plot shows a median of $180$TB of cross-rack traffic generated as a result of node unavailability~\cite{RashmiShahGuKuang}. \label{fig:fbdata}}
\end{figure}

Thus there is significant practical interest in the design of erasure-coding techniques that offer both low overhead and which can also be repaired efficiently.  This is particularly the case, given the large amounts of data running into the tens or $100$s of petabytes, that are stored in modern-day data centers (see Fig.~\ref{fig:google_dc}). 

\begin{figure}[ht!]
	\centering
	\includegraphics[width=2in, angle=90]{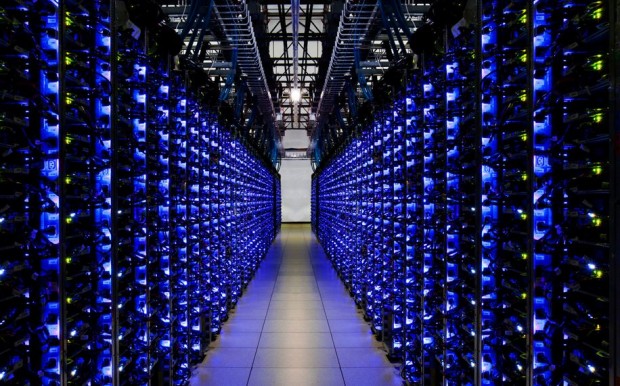}
	\caption{An image of a Google data center. }\label{fig:google_dc}
\end{figure}

\section{Different Approaches to Node Repair} 

The flowchart in Fig.~\ref{fig:flowchart1}, provides a detailed overview of the different approaches by coding theorists to efficiently handle the problem of node repair and the numerous subclasses of codes that they have given rise to.  In the description below, we explain the organization presented in the flowchart.  The boxes outline in red in the flowchart are the topics to which this thesis has made contributions.  These topics are revisited in detail in subsequent section of the chapter. 

\begin{figure}[h!]
	\centering
	\includegraphics[width=6in]{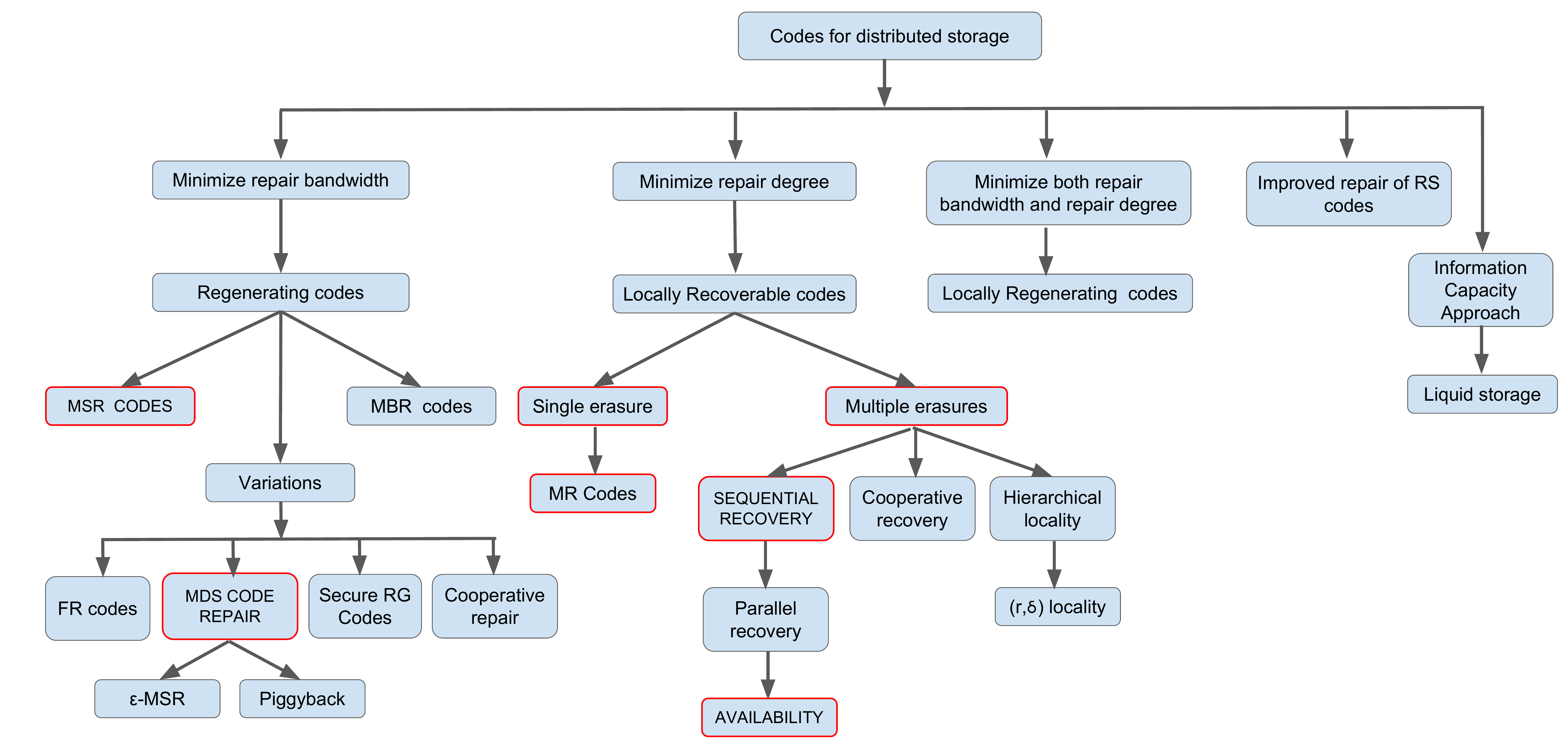}
	\caption{Flowchart depicts various techniques used to handle node repair. Current thesis has contributions in the topics corresponding to boxes highlighted in red. }\label{fig:flowchart1}
\end{figure}

\paragraph{Drawbacks of Conventional Repair} \ The conventional repair of an $[n,k]$ Reed-Solomon (RS) code where $n$ denotes the block length of the code and $k$ the dimension is inefficient in that the repair of a single node, calls for contacting $k$ other (helper) nodes and downloading $k$ times the amount of data stored in the failed node.  This is inefficient in $2$ respects.  Firstly the amount of data download needed to repair a failed node, termed the {\em repair bandwidth}, is $k$ times the amount stored in the replacement node.  Secondly, to repair a failed node, one needs to contact $k$ helper nodes.  The number of helper nodes contacted is termed the {\em repair degree}. Thus in the case of the $[14,10]$ RS code employed in Facebook, the repair degree is $10$ and the repair bandwidth is $10$ times the amount of data that is stored in the replacement node which is clearly inefficient. 

\paragraph{The Different Approaches to Efficient Node Repair} \ Coding theorists have responded to this need by coming up with two new classes of codes, namely {\em ReGenerating (RG) \cite{DimGodWuWaiRam,Wu}.and Locally Recoverable (LR) codes~\cite{GopHuaSimYek}}.  The focus in an RG code is on minimizing the repair bandwidth while LR codes seek to minimize the repair degree. In a different direction, coding theorists have also re-examined the problem of node repair in RS codes and have come up ~\cite{GuruWootRSRepair17} with new and more efficient repair techniques.  An alternative information-theoretic approach which permits {\em lazy repair}, i.e., which does not require a failed node to be immediately restored, can be found on \cite{LubyPadRic}. 

\paragraph{Different Classes of RG Codes} \ Regenerating codes are subject to a tradeoff termed as the {\em storage-repair bandwidth (S-RB) tradeoff}, between the storage overhead $\frac{n\alpha}{B}$ of the code and the normalized repair bandwidth (repair bandwidth normalized by the file size).  This tradeoff is derived by using principles of network coding.  Any code operating on the tradeoff is optimal with respect to file size.  At the two extreme ends of the tradeoff are codes termed as minimum storage regenerating codes (MSR) and minimum bandwidth regenerating (MBR) codes.  MSR codes are of particular interest as these codes are Maximum Distance Separable (MDS), meaning that they offer the least amount of storage overhead for a given level of reliability and also offer the potential of low storage overhead.    We will refer to codes corresponding to interior points of the S-RB tradeoff as interior-point RG codes.   It turns out the precise tradeoff in the interior is unknown, thus it is an open problem to determine the true tradeoff as well as provide constructions that are optimal with respect to this tradeoff.  Details pertaining to the S-RB tradeoff can be found in \cite{ShaRasKumRam_rbt,Tia,SasSenKum_isit,SasPraKriVajSenKum}.

\paragraph{Variations on the Theme of RG Codes} \ The theory of regenerating codes has been extended in several other directions.  {\em Secure RG codes} (see \cite{PawarRouayRam}) are RG codes which offer some degree of protection against a passive or active eavesdropper.  {\em Fractional Repair (FR) codes} (see \cite{RouRamFR10}) are codes which give up on some requirements of an RG code and in exchange provide the convenience of being able to repair a failed node simply by transferring data (without need for computation at either end) between helper and replacement node.  {\em Cooperative RG codes} (see \cite{HuXuWanZhaLi,KerScoStr,ShuHu}) are RG codes which consider the simultaneous repair of several failed nodes and show that there is an advantage to be gained by repairing the failed nodes collectively as opposed to in a one-by-one fashion. 

\paragraph{MDS codes with Efficient Repair} \ There has also been interest in designing other classes of Maximum Distance Separable (MDS) codes that can be repaired efficiently.  Under the {\em Piggyback Framework} (see \cite{RasShaRam_Piggyback}), it is shown how one can take a collection of MDS codewords and couple the contents of the different layers so as to reduce the repair bandwidth per codeword.  RG codes are codes over a vector alphabet $\mathbb{F}_q^{\alpha}$ and the parameter $\alpha$ is referred to as the {\em sub-packetization level} of the code.  It turns out that in an RG code, as the storage overhead gets closer to $1$, the sub-packetization level $\alpha$, rises very quickly.  $\epsilon$-MSR codes (see \cite{RawTamGur_epsilonMSR}) are codes which for a multiplicative factor ($1+ \epsilon$) increase in repair bandwidth over that required by an MSR code, are able to keep the sub-packetizatin to a very small level. 

\paragraph{Locally Recoverable Codes} \ Locally recoverable codes (see \cite{HanMon,HuaChenLi,OggDat,GopHuaSimYek,PapDim}) are codes that seek to lower the repair degree.  This is accomplished by constructing the erasure codes in such a manner that each code symbol is protected by a single-parity-check (spc) code of smaller blocklength, embedded within the code.  Each such spc code is termed as a {\em local code}. Node repair is accomplished by calling upon the short blocklength code, thereby reducing the repair degree.  The coding scheme used in the Windows Azure is an example of an LR code.   The early focus on the topic of LR codes was on the single-erasure case.  Within the class of single-erasure LR codes, is the subclass of {\em Maximum Recoverable (MR) codes}.    An MR code is capable of recovering from any erasure pattern that is not precluded by the locality constraints imposed on the code. 

\paragraph{LR Codes for Multiple-Erasures} \ More recent work in the literature has been directed towards the repair of multiple erasures. Several approaches have been put forward for multiple-erasure recovery.  The approach via {\em $(r,\delta)$ codes} (see \cite{PraKamLalKum,SonDauYueLi}), is simply to replace the spc local codes with codes that have larger minimum distance. {\em Hierarchical codes} are codes which offer different tiers of locality.   The local codes of smallest block length offer protection against single erasures.  Those with the next higher level of blocklength, offer protection against a larger number of erasures and so on. 

\paragraph{Codes with Sequential and Parallel Recovery} \ The class of codes for handling multiple erasures using local codes, that are most efficient in terms of storage overhead, are the class of codes with sequential recovery (for details on sequential recovery, please see \cite{PraLalKum,BalPraKum,SongCaiYue_L,BalKinKum_ISIT,BalKinKum_NCC}).  As the name suggests, in this class of codes, for any given pattern of $t$ erasures, there is an order under which recovery from these $t$ erasures is possible by contacting atmost $r$ code symbols for the recovery of each erasure. {\em Parallel Recovery} places a more stringent constraint, namely that one should be able to recover from any pattern of $t$ erasures in parallel.  

\paragraph{Availability Codes} \ Availability codes (see \cite{WanZha,WanZhaLiu,TamBarFro,BalKum}) require the presence of $t$ disjoint repair groups with each repair group contains atmost $r$ code symbols that are capable of repairing a single erased symbol.  The name availability stems from the fact that this property allows the recreation of a single erased symbol in $t$ different ways, each calling upon a disjoint set of helper nodes.  This allows the $t$ simultaneous demands for the content of a single node to be met, hence the name availability code. 
In the class of codes with cooperative recovery (see \cite{RawMazVis}), the focus is on the recovery of multiple erasures at the same time, while keeping the average number of helper nodes contacted per erased symbol, to a small value. 

\paragraph{Locally Regenerating (LRG) Codes} Locally regenerating codes (see \cite{KamPraLalKum}) are codes in which each local code is itself an RG code. Thus this class of codes incorporates into a single code, the desirable features of both RG and LR codes, namely {\em both} low repair bandwidth and low repair degree. 

\paragraph{Efficient Repair of RS Codes} \ In a different direction, researchers have come up with alternative means of repairing RS codes (\cite{ShanPapDimCaiMDSRepair14,GuruWootRSRepair17}).  These approaches view an RS code over an alphabet $\mathbb{F}_q$, $q=p^t$ as a vector code over the subfield $\mathbb{F}_p$ having sub-packetization level $t$ and use this perspective, to provide alternative, improved approaches to the repair of an RS code. 

\paragraph{Liquid Storage Codes} \ These codes are constructed in line with an information-theoretic approach which permits {\em lazy repair}, i.e., which does not require a failed node to be immediately restored can be found on \cite{LubyPadRic}. 

\section{Literature Survey}

\subsection{Locally Recoverable (LR) codes for Single Erasure}
In \cite{HanMon}, the authors consider designing codes such that the code designed and codes of short block length derived from the code designed through puncturing operations all have good minimum distance. The requirement of such codes comes from the problem of coding for memory where sometimes you want to read or write only parts of memory.  These punctured codes are what would today be regarded as local codes. The authors derive an upper bound on minimum distance of such codes under the constraint that the code symbols in a local code and code symbols in another local code form disjoint sets and provide a simple parity-splitting construction that achieves the upper bound. Note that this upper bound on minimum distance is without any constraint on field size and achieved for some restricted set of parameters by parity splitting construction which has field size of $O(n)$. In \cite{HuaChenLi}, the authors note that when a single code symbol is erased in an MDS code, $k$ code symbols need to be contacted to recover the erased code symbol where $k$ is the dimension of the MDS code. This led them to design codes called Pyramid Codes which are very simply derived from the systematic generator matrix of an MDS code and which reduce the number of code symbols that is needed to be contacted to recover an erased code symbol.  In \cite{OggDat}, the authors recognize the requirement of recovering a set of erased code symbols by contacting a small set of remaining code symbols and provide a code construction for the requirement based on the use of linearized polynomials.  

In \cite{GopHuaSimYek}, the authors introduce the class of LR codes in full generality, and present an upper bound on minimum distance $d_{\min}$ without any constraint on field size.   This paper along with the paper \cite{HuaSimXu_Azure} (sharing a common subset of authors) which presented the practical application of LR codes in Windows Azure storage, are to a large extent, responsible for drawing the attention of coding theorists to this class of codes.  


The extension to the non-linear case 
appears in \cite{PapDim},\cite{ForYek} respectively.  All of these papers were primarily concerned with local recoverability in the case of a  single erasure i.e., recovering an erased code symbol by contacting a small set of code symbols. More recent research has focused on the multiple-erasure case and multiple erasures are treated in subsequent chapters of this thesis.

For a detailed survery on alphabet size dependent bounds for LR codes and constructions of LR codes with small alphabet size, please refer to Chapter \ref{ch:LRSingleErasure}.
A tabular listing of some constructions of Maximal Recoverbale or partial-MDS codes appears in Table \ref{tab:pmds} in Chapter \ref{ch:MR}. 

\subsection{Codes with Sequential Recovery}
The sequential approach to recovery from erasures, introduced by Prakash et al. \cite{PraLalKum} is one of several approaches to local recovery from multiple erasures as discussed in Chapter \ref{ch:SeqBound}, Section \ref{sec:seq_classification}.   As indicated in Fig.~\ref{fig:me_classes}, Codes with Parallel Recovery and Availability Codes can be regarded as sub-classes of Codes with Sequential Recovery (S-LR codes).    Among the class of codes which contact at most $r$ other code symbols for recovery from each of the $t$ erasures, codes employing this approach (see \cite{PraLalKum,RawMazVis,SonYue,SonYue_Prdr,BalPraKum,BalPraKum4Era,SongCaiYue_L,BalKinKum_ISIT,BalKinKum_NCC}) have improved rate simply because sequential recovery imposes the least stringent constraint on the LR code.

\paragraph{Two Erasures} Codes with sequential recovery (S-LR code) from two erasures ($t=2$) are considered in \cite{PraLalKum} (see also \cite{SonYue}) where a tight upper bound on the rate and a matching construction achieving the upper bound on rate is provided.  A lower bound on block length and a construction achieving the lower bound on block length is provided in \cite{SonYue}. 

\paragraph{Three Erasures} Codes with sequential recovery from three erasures ($t=3$) can be found discussed in \cite{SonYue,SongCaiYue_L}.  A lower bound on block length as well as a construction achieving the lower bound on block length appears in \cite{SonYue}. 

\paragraph{More Than $3$ Erasures}  A general construction of S-LR codes for any $r,t$ appears in \cite{SonYue_Prdr,SongCaiYue_L}.  Based on the tight upper bound on code rate presented in Chapter \ref{ch:SeqBound}, it can be seen that the constructions provided in \cite{SonYue_Prdr,SongCaiYue_L} do not achieve the maximum possible rate of an S-LR code. In \cite{RawMazVis}, the authors provide a construction of S-LR codes for any $r,t$ with rate $\geq \frac{r-1}{r+1}$.  Again, the upper bound on rate presented in Chapter \ref{ch:SeqBound} shows that $\frac{r-1}{r+1}$ is not the maximum possible rate of an S-LR code.   In Chapter \ref{ch:SeqConstr}, we observe that the rate of the construction given in \cite{RawMazVis} is actually $\frac{r-1}{r+1}+\frac{1}{n}$ which equals the upper bound on rate derived here only for two cases: case (i) for $r=1$ and case (ii) for $r \geq 2$ and $t \in \{2,3,4,5,7,11\}$ exactly corresponding to those cases where a Moore graph of degree $r+1$ and girth $t+1$ exist.  In all other cases, the construction given in \cite{RawMazVis} does not achieve the maximum possible rate of an S-LR code.    

\subsection{Codes with Availability}
The problem of designing codes with availability in the context of LR codes was introduced in \cite{WanZha}. High rate constructions for availability codes appeared in \cite{WanZhaLiu},\cite{HuaYaaUchSie},\cite{WanZhaLin1},\cite{SilZeh}. Constructions of availability codes with large minimum distance appeared in \cite{TamBar_LRC},\cite{TamBarFro,BarTamVla,KruFro}, \cite{BhaTha}. For more details on constructions of availability codes please see Chapter \ref{ch:avail}. Upper bounds on minimum distance and rate of an availabiltiy code appeared in \cite{WanZha}, \cite{TamBarFro}, \cite{HuaYaaUchSie}, \cite{KruFro}, \cite{KadCal}. For exact expressions for upper bounds on minimum distance and rate which appeared in literature please refer to Chapter \ref{ch:avail}.

\subsection{Regenerating codes}
In the following, we focus only on sub-packetization level $\alpha$ of regenerating codes as this thesis is focussed only on this aspect.
An open problem in the literature on regenerating codes is that of determining the smallest value of sub-packetization level $\alpha$ of an optimal-access (equivalently, help-by-transfer) MSR code, given the parameters $\{(n,k,d=(n-1)\}$. 
This question is addressed in \cite{TamWanBru_access_tit}, where a lower bound on $\alpha$ is given for the case of a regenerating code that is MDS and where only the systematic nodes are repaired in a help-by-transfer fashion with minimum repair bandwidth.   In the literature these codes are often referred to as optimal access MSR codes with systematic node repair. The authors of  \cite{TamWanBru_access_tit} establish that:
\bean
\alpha \geq r^{\frac{k-1}{r}},
\eean
in the case of an optimal access MSR code with systematic node repair. 

In a slightly different direction, lower bounds are established in \cite{GopTamCal} on the value of $\alpha$ in a general MSR code that does not necessarily possess the help-by-transfer repair property. In \cite{GopTamCal} it is established that:
\bean
k \leq 2 \log_2(\alpha) (\lfloor \log_{\frac{r}{r-1}} (\alpha) \rfloor + 1),
\eean
while more recently, in \cite{HuaparXia} the authors prove that:
\bean
k \leq 2 \log_r(\alpha) (\lfloor \log_{\frac{r}{r-1}} (\alpha) \rfloor + 1). 
\eean
A brief survey of regenerating codes and in particular MSR codes appear in Chapter \ref{ch:Subpkt}.

\section{Codes in Practice}

The explosion in amount of storage required and the high cost of building and maintaining a data center, has led the storage industry to replace the widely-prevalent replication of data with erasure codes, primarily the RS code (see Fig.~\ref{fig:RS_applications}).   For example, the new release Hadoop 3.0 of the Hadoop Distributed File System (HDFS), incorporates HDFS-EC (for HDFS- Erasure Coding) makes provision for employing RS codes in an HDFS system.  

\begin{figure}[ht!]
	\centering
	\includegraphics[width=4.75in]{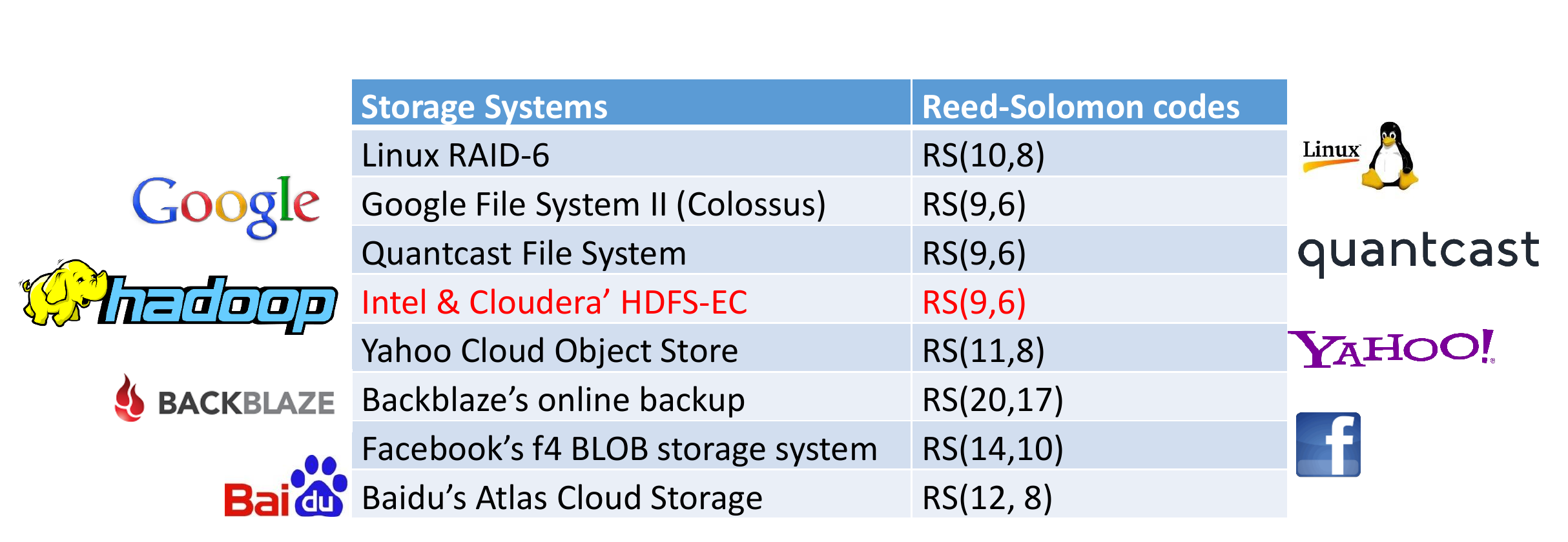}            
	\caption{Some examples of the RS code employed in industry. {\scriptsize (taken from Hoang Dau, Iwan Duursma, Mao Kiah and Olgica Milenkovic, ``Optimal repair schemes for Reed-Solomon codes with single and multiple erasures,'' {\em 2017 Information Theory and Applications Workshop}, San Diego, Feb 12-17.) }}
	\label{fig:RS_applications}
\end{figure}

However, the use of traditional erasure codes results in a repair overhead, measured in terms of additional repair traffic resulting in larger repair times and the tying up of nodes in non productive, node-repair-related activities.  This motivated the academic and industrial-research community to explore approaches to erasure code construction which were more efficient in terms of node repair and many of these approaches were discussed in the preceding section.  

An excellent example of research in this direction is the development of the theory of LR codes and their immediate deployment in data storage in the form of the Windows Azure system. 

\paragraph{LR Codes in Windows Azure:}  In \cite{HuaSimXu_Azure}, the authors compare performance-evaluation results of an $(n=16,k=12,r=6)$ LR code with that of $[n=16,k=12]$  RS code in Azure production cluster and demonstrates the repair savings of LR code. Subsequently the authors implemented an $(n=18,k=14,r=7)$ LR code in Windows Azure Storage and showed that this code has repair degree comparable to that of an $[9,6]$ RS code, but has storage overhead $1.29$ versus $1.5$ in the case of the RS code (see Fig.~\ref{fig:RS}, and Fig.~\ref{fig:Azure}). This $(n=18,k=14,r=7)$ LR code is currently is use now and has reportedly resulted in the savings of millions of dollars for Microsoft \cite{microsoft}. 

\begin{figure}[ht!]
	\centering
	\includegraphics[width=4in]{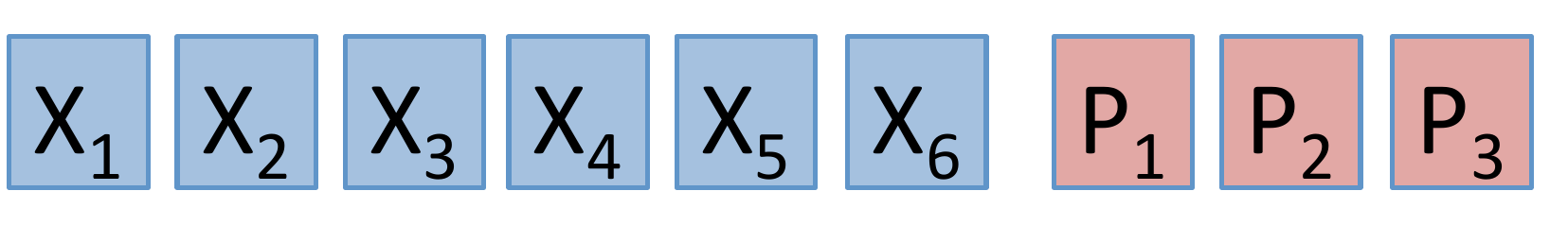}
	\caption{A $[9,6]$ RS code having a repair degree of $6$ and a storage overhead of $1.5$. }\label{fig:RS}
\end{figure}
\begin{figure}[ht!]
	\centering
	\includegraphics[width=6in]{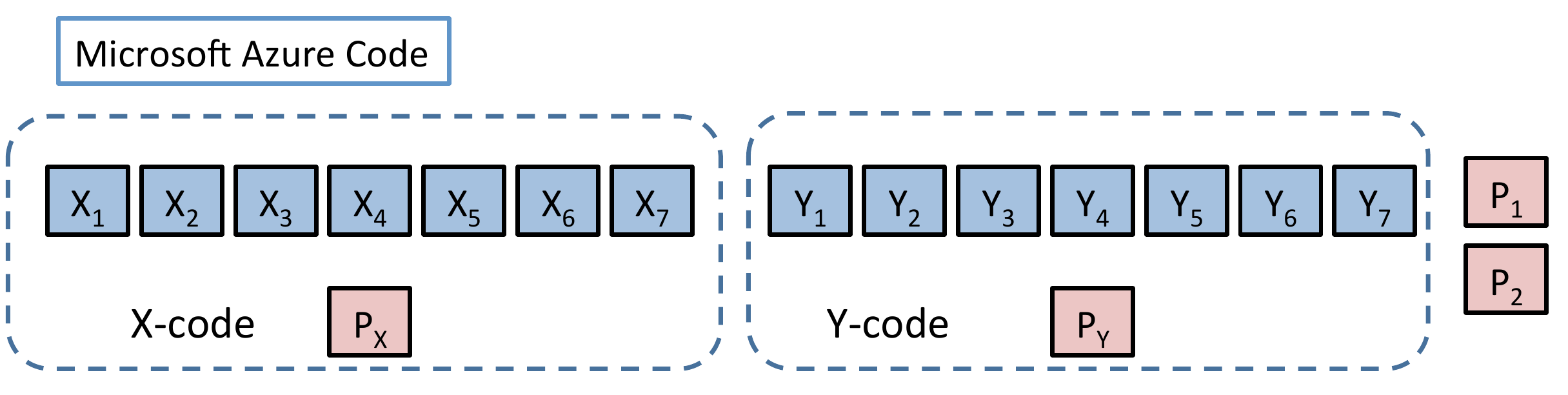}
	\caption{The LR code employed in Windows Azure.  This code has repair degree $7$, which is only slightly larger than the repair degree $6$ of the $[9,6]$ RS code in Fig.~\ref{fig:RS}.  However, the storage overhead of this code at $1.29$, is much smaller than the comparable value in the case of the $(1.5$) of the RS code. }\label{fig:Azure}
\end{figure}

A second poular distributed storage system is Ceph  and Ceph currently has an LR code plug-in \cite{lrc_plugin}.

Some other examples of work directed towards practical applications are described below.  Most of this work is work carried out by an academic group and presented at a major storage industry conference and involves performance evaluation through emulation of the codes in a real-world setting. 

\ben
\item In \cite{SathiaAstPap_Xorbas}, the authors implement HDFS-Xorbas.  This system employs LR codes in place of RS codes in HDFS-RAID.  The experimental evaluation of Xorbas was carried out in Amazon EC2 and a cluster in Facebook and the repair performance of $(n=16,k=10,r=5)$ LR code  was compared against a $[14,10]$ RS code. 

\item A method, termed as piggybacking, of layering several RS codewords and then coupling code symbols across layers is shown in \cite{RasShaRam_Piggyback}, to yield a code over a vector alphabet, that has reduced repair bandwidth, without giving up on the MDS property of an RS code.  A practical implementation of this is implemented in the Hitchhiker erasure-coded system ~\cite{RasShaGu_Hitchhiker}.   Hitchhiker was implemented in HDFS and its performance was evaluated on a data-warehouse cluster at Facebook. 

\item The HDFS implementation of a class of codes known as HashTag codes is discussed in \cite{KraGliJen_HashTag} (see also \cite{KraGliOVe}).   These are codes designed to efficiently repair systematic nodes and have a lower sub-packetization level in comparison to an RG code at the expense of a larger repair bandwidth. 

\item The NCCloud \cite{HuChenLee_NCcloud} is an early work that dealt with the practical performance evaluation of regenerating codes and employs a class of MSR code known as functional-MSR code having $2$ parities. 

\item In \cite{KriPraLal_Pentagon}, the performance of an MBR code known as the pentagon code as well as  an LRG code known as the heptagon local code are studied and their performance compared against double and triple replication.  These code possess inherent double replication of symbols as part of the construction. 

\item The product matrix (PM) code construction technique yields a general construction of MSR and MBR codes.  The PM  MSR codes have storage overhead that is approximately lower bounded by a factor of $2$.   The performance evaluation of an optimal-access version of a rate $\frac{1}{2}$ PM code, built on top of Amazon EC2 instances, is presented in \cite{RasNakWan_PMRBT}. 

%

\item A high-rate MSR code known as the Butterfly code is implemented and evaluated in both Ceph and HDFS in \cite{JuaBlaMat_Butterfly}.  This code is a simplified version of the MSR codes with two parities introduced in \cite{GadMatBla}. 

\item In \cite{VajRamPur_Clay}, the authors evaluate the performance in a Ceph environment, of an MSR code known as the Clay code, and which corresponds to the Ye-Barg code in \cite{YeBar_2}, (and independently rediscovered after in \cite{SasVajKum_arxiv}).  The code is implemented in \cite{VajRamPur_Clay}, from the coupled-layer perspective present in \cite{SasVajKum_arxiv}.   This code is simultaneously optimal in terms of storage overhead and repair bandwidth (as it is an MSR code), and also has the optimal-access (OA) property and the smallest possible sub-packetization level of an OA MSR code. The experimental performance of the Clay code is shown to be match its theoretical performance. 

\een


\section{Contributions and Organization of the Thesis}

The highlighted boxes appearing in the flow chart in Fig.~\ref{fig:flowchart2} represent topics with respect to which this thesis has made a contribution.   
\begin{figure}[ht!]
	\centering
	\includegraphics[width=6in]{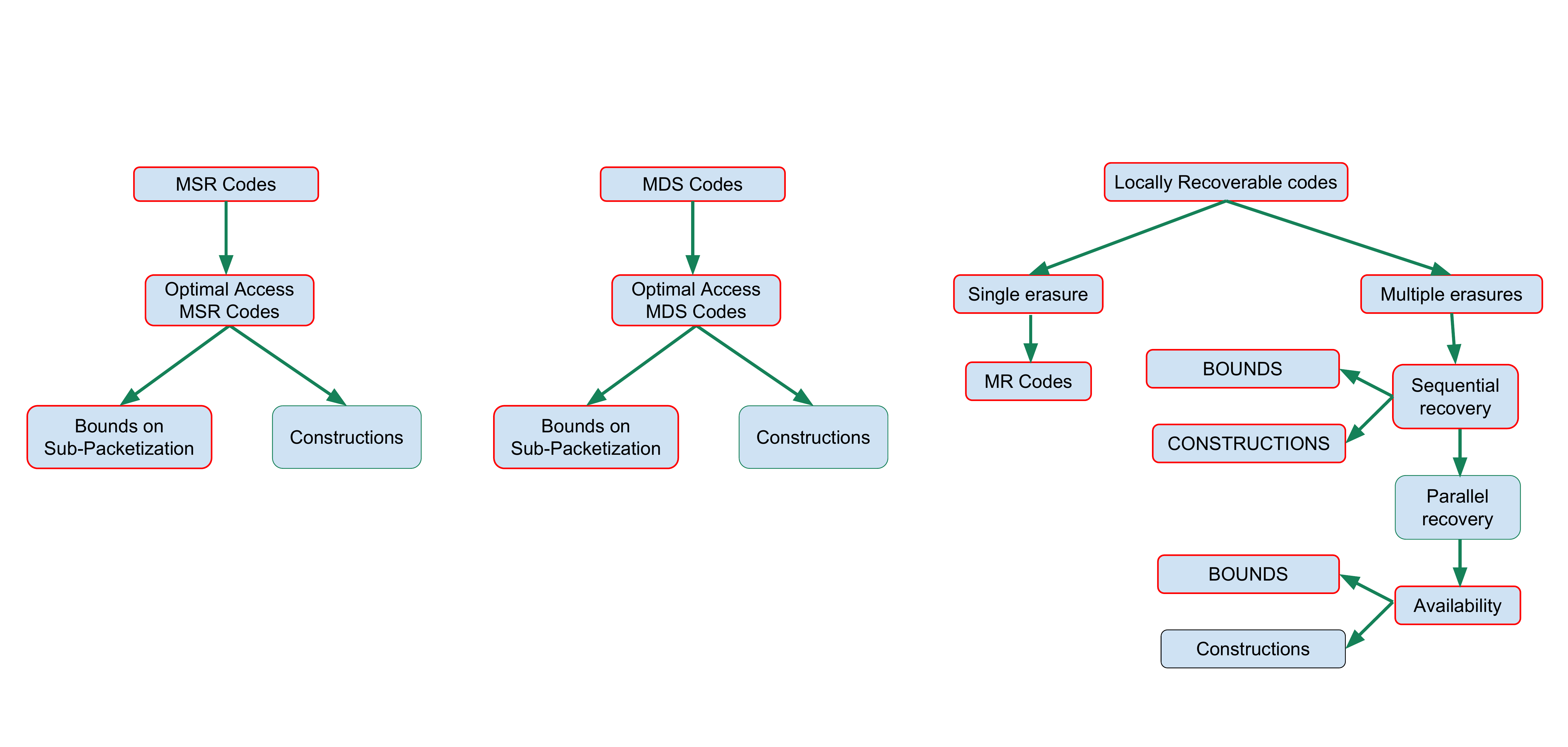}
	\caption{The three flowcharts shown here are extracted from the flowchart in Fig.~\ref{fig:flowchart1}. The highlighted boxes indicate topics to which this thesis has contributed. }\label{fig:flowchart2}
\end{figure}
We now proceed to describe chapter wise, our contributions corresponding to topics in the highlighted boxes.   An overview of the contributions appears in Fig.~\ref{fig:flowchart21}. 

\begin{figure}[ht!]
	\centering
	\includegraphics[width=5in]{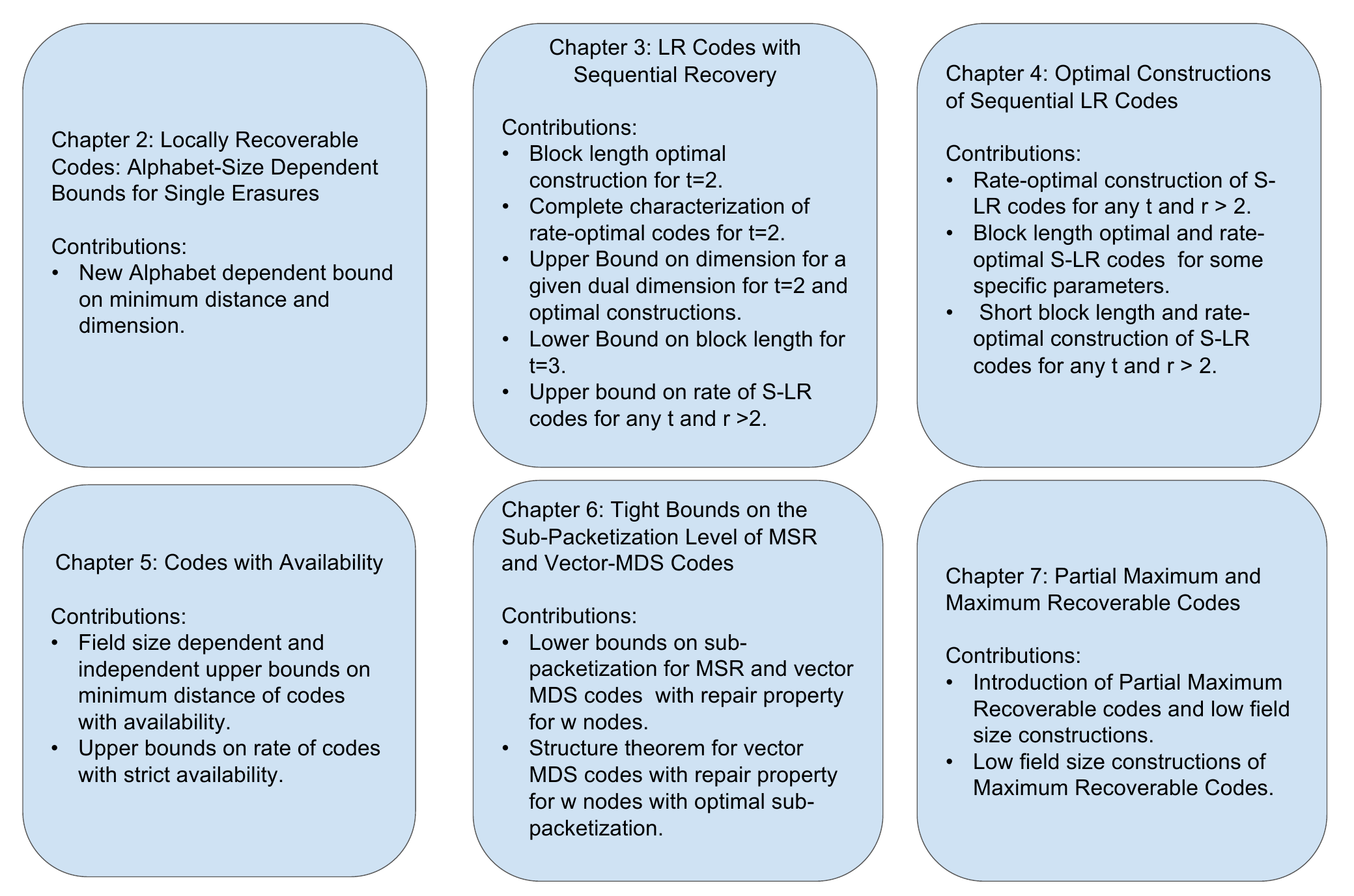}
	\caption{A chapter wise overview of the contributions of the present thesis. The highlighted chapters indicate the chapters containing the principal results of the chapter. The specific principal results appear in boldface. }\label{fig:flowchart21}
\end{figure}

\paragraph{Chapter 2: Locally Recoverable Codes: Alphabet-Size Dependent Bounds for Single Erasures
}
This chapter begins with an overview of LR codes.  Following this, new alphabet-size dependent bounds on both minimum distance and dimension of an LR code that are tighter than existing bounds in the literature, are presented.   

\paragraph{Chapter 3: Tight Bounds on the Rate of LR Codes with Sequential Recovery}
This chapter deals with codes for sequential recovery and contains the principal result of the thesis, namely, a tight upper bound on the rate of a code with sequential recovery for all possible values of the number $t$ of erasures guaranteed to be recovered with locality parameter $r \geq 3$.  Matching constructions are provided in the chapter following, Chapter~\ref{ch:SeqConstr}.    A characterization of codes achieving the upper bound on code rate for the case of $t=2$ erasures is also provided here. The bound on maximum possible code rate assumes that there is no constraint (i.e., upper bound) on the block length of the code or equivalently, on the code dimension.  A lower bound on the block length of codes with sequential recovery from three erasures is also given here. Also given are constructions of codes with sequential recovery for $t=2$ having least possible block length for a given dimension $k$ and locality parameter $r$. An upper bound on dimension for the case of $t=2$ for a given dual dimension and locality parameter $r$ and constructions achieving it are also provided. 
\paragraph{Chapter 4: Matching (Optimal) Constructions of Sequential LR Codes}
In this chapter, we construct codes which achieve the upper bound on rate of codes with sequential recovery derived in Chapter ~\ref{ch:SeqBound} for all possible values of the number $t$ of erasures guaranteed to be recovered with locality parameter $r \geq 3$.  We deduce the general structure of parity check matrix of a code achieving our upper bound on rate. Based on this, we show achievability of the upper bound on code rate via an explicit construction. We then present codes which achieve the upper bound on rate having least possible block length for some specific set of parameters.

\paragraph{Chapter 5: Bounds on the Parameters of Codes with Availability}
This chapter deals with codes with availability. Upper bounds are presented on the minimum distance of a code with availability, both for the case when the alphabet size is constrained and when there is no constraint.   These bounds are tighter than the existing bounds in literature. We next introduce a class of codes, termed {\em codes with strict availability} which are subclass of the codes with availability.   The best-known availability codes in terms of rate belong to this category.  We present upper bounds on the rate of codes with strict availability that are tighter than existing upper bounds on the rate of codes with availability. We present exact expression for maximum possible fractional minimum distance for a given rate for a special class of availability codes as $\ell \rightarrow \infty$ where each code in this special class is a subcode or subspace of direct product of $\ell$ copies of an availability code with parameters $r,t$ for some $\ell$.  We also present a lower bound on block length codes with strict avalability and characterize the codes with strict availability achieving the lower bound on block length. 

\paragraph{Chapter 6: Tight Bounds on the Sub-Packetization Level of MSR and Vector-MDS Codes}
This chapter contains our results on the topic of RG codes. Here, we derive lower bounds on the sub-packetization level an of a subclass of MSR codes known as optimal-access MSR codes.  We also bound the sub-packetization level of optimal-access MDS codes with optimal repair for (say) a fixed number $w$ of nodes. The bounds derived here are tight as there are constructions in the literature that achieve the bounds derived here.  The bounds derived here conversely show that the constructions that have previously appeared in the literature are optimal with respect to sub-packetization level.   We also show that the bound derived here sheds light on the structure of an optimal-access MSR or MDS code.  

\paragraph{Chapter 7: Partial Maximal and Maximal Recoverable Codes}

The final chapter deals with the subclass of LR codes known as Maximal Recoverable (MR) codes. In this chapter we provide constructions of MR codes having smaller field size than the constructions existing in the literature. In particular we modify an existing construction which will result in an MR code with field size of $O(n)$ for some specific set of parameters. We also modify (puncture) an existing construction for $r=2$ to form an MR code which results in reduced field size in comparison with the field size of constructions appearing in the literature.  We also introduce in the chapter, a class of codes termed as Partial Maximal Recoverable (PMR) codes. We provide constructions of PMR codes having small field size. Since a PMR code is in particular an LR code, this also yields a low-field-size construction of LR codes.

\chapter{Locally Recoverable Codes: Alphabet-Size Dependent Bounds for Single Erasures} \label{ch:LRSingleErasure}

This chapter deals with locally recoverable (LR) codes, also known in the literature as codes with locality. Contributions of the thesis in this area include new best-known alphabet-size-dependent bounds on both minimum distance $d_{\min}$ and dimension $k$ for LR codes for $q>2$. For $q=2$, our bound on dimension is the tightest known bound for $d_{\min} \geq 9$.  We begin with some background including a fundamental bound on $d_{\min}$ (Section \ref{LR_dmin_bound}) and a description of two of the better-known and general constructions for this class of codes (Section \ref{sec:LRCodeConst}).   

More recent research has focused on deriving bounds on code dimension and minimum distance, that take into account the size $q$ of the underlying finite field $\mathbb{F}_q$ over which the codes are constructed. We next provide a summary of existing field-size-dependent bounds on dimension and minimum distance (Section \ref{Alphabet_dep_LR_codes}).   Our alphabet-size-dependent bound (Section \ref{LR_Codes_New_Bound}) on minimum distance and dimension makes use of an upper bound on the Generalized Hamming Weights (GHW) (equivalently, Minimum Support Weights (MSW)) derived in \cite{PraLalKum}.  This bound is in terms of a recursively-defined sequence of integers which we refer to here as the Minimum Support Weight Sequence (MSWS).  Our bound also makes use of the notion of {\em shortening of a code}.  Following a presentation of our results, we then summarize existing alphabet-size-dependent constructions (Section \ref{Small_Alph_LR_Codes}). The chapter ends with a summary of the contributions of the thesis on the topic of LR codes for single erasures.  Contributions to the case of LR codes for multiple erasures are contained in subsequent chapters. In this chapter, we will restrict ourselves to only linear codes for most of the discussion.  

\section{Locally Recoverable Codes for Single Erasures} \label{LR_Codes}

In \cite{HanMon}, the authors consider designing codes such that the code designed and codes of short block length derived from the code designed through puncturing operations all have good minimum distance. The requirement of such codes comes from the problem of coding for memory where sometimes you want to read or write only parts of memory.  These punctured codes are what would today be regarded as local codes. The authors derive an upper bound on minimum distance of such codes under the constraint that the code symbols in a local code and code symbols in another local code form disjoint sets and provide a simple parity-splitting construction that achieves the upper bound. Note that this upper bound on minimum distance is without any constraint on field size and achieved for some restricted set of parameters by parity splitting construction which has field size of $O(n)$. In \cite{HuaChenLi}, the authors note that when a single code symbol is erased in an MDS code, $k$ code symbols need to be contacted to recover the erased code symbol where $k$ is the dimension of the MDS code. This led them to design codes called Pyramid Codes which are very simply derived from the systematic generator matrix of an MDS code and which reduce the number of code symbols that is needed to be contacted to recover an erased code symbol.  In \cite{OggDat}, the authors recognize the requirement of recovering a set of erased code symbols by contacting a small set of remaining code symbols and provide a code construction for the requirement based on the use of linearized polynomials.  

In \cite{GopHuaSimYek}, the authors introduce the class of LR codes in full generality, and present an upper bound on minimum distance $d_{\min}$ without any constraint on field size.   This paper along with the paper \cite{HuaSimXu_Azure} (sharing a common subset of authors) which presented the practical application of LR codes in Windows Azure storage, are to a large extent, responsible for drawing the attention of coding theorists to this class of codes.  


The extension to the non-linear case 
appears in \cite{PapDim},\cite{ForYek} respectively.  All of these papers were primarily concerned with local recoverability in the case of a  single erasure i.e., recovering an erased code symbol by contacting a small set of code symbols. More recent research has focused on the multiple-erasure case and multiple erasures are treated in subsequent chapters of this thesis.

Throughout this chapter:
\ben
\item a codeword in an $[n,k]$ linear code will be represented by $\underline{c}=[c_1,c_2,\cdots, c_{n-1},c_n ]$ where $c_i$ denotes the $i$th code symbol.
\item all codes discussed are linear codes and we will use the term nonlinear explicitly when referring to a nonlinear code. 
\item we say a code achieves a bound (an inequality), iff it has parameters such that the bound is satisfied with equality. 
\item The notation $d$ or $d_{\min}$, refers to the minimum distance of a code under discussion.
\een
Let \calc be an $[n,k]$ code over a finite field $\fq$.
Let $G$ be a $(k \times n)$ generator matrix for \calc having columns $\{\ug_i\}_{i=1}^n$, i.e.,
$G = [\ug_1, \ug_2, \cdots, \ug_n]$. 
An information set $E=\{e_1, e_2, \cdots, e_k\}$ is any subset of $[n]$ of size $k$ satisfying:
$\text{rk}(G|_E) = \text{rk}[\ug_{e_1}, \cdots, \ug_{e_k}] = k$. 
\begin{defn}
	An $[n,k]$ code \calc  over a finite field \fq\ is said to be an LR code with information-symbol (IS) locality over \fq\ if there is an information set $E = \{e_1, e_2, \cdots, e_k\}$ such that for every $e_i \in E$, there exists a subset $S_i \subseteq [n]$, with $e_i \in S_i$, such that $|S_i| \leq r+1$ and there is a codeword in the dual code $\calc^{\perp}$ with support exactly equal to $S_i$.
	$\calc$ is said to be an LR code with all-symbol (AS) locality over \fq\ if for every $i \in [n]$, there exists a subset $S_i \subseteq [n]$ with $i \in S_i$, such that $|S_i| \leq r+1$ and there is a codeword in the dual code $\calc^{\perp}$ with support exactly equal to $S_i$. Clearly, an LR code with AS locality is also an LR code with IS locality.  The parameter $r$ appearing above is termed the locality parameter.
\end{defn}        
Throughout this thesis, when we say LR code, it refers to an LR code with all-symbol (AS) locality. When we discuss LR code with information-symbol (IS) locality, we will state it explicitly.
Note that the presence of a codeword in the dual code with support set $S_i$ implies that if the code symbol $c_{e_i}$ is erased then it can be recovered from the code symbols in the set $\{c_j \mid j \in S_i\setminus \{e_i\}\}$.   Recovery from erasures is termed as {\em repair}. The repair is {\em local}, since $|S_i| \leq r+1$ and typically, $r$ is significantly smaller than the block length $n$ of the code.   It is easy to see that every linear code can trivially be regarded as an LR code with locality parameter $k$. The term a local code of an LR code $\mathcal{C}$ refers to the code $\mathcal{C}|_{S_i} = \{ \underline{c}|_{S_i}: \underline{c} \in \mathcal{C} \}$ for some $i \in [n]$ where $\underline{c}|_{A} = [c_{j_1},...,c_{j_{|A|}}]$ for a set $A=\{j_1,...,j_{|A|}\}$ with $j_{\ell} < j_{\ell+1}$,$\forall \ell \in [|A|-1]$..

\subsection{The $d_{\min}$ Bound} \label{LR_dmin_bound}
A major result in the theory of LR codes is the minimum distance bound given in \eqref{Singleton} which was derived for linear codes in \cite{GopHuaSimYek}. An analogous bound for nonlinear codes can be found in \cite{PapDim},\cite{ForYek}.  
\begin{thm}    \cite{GopHuaSimYek}  
	Let \calc be an $[n,k]$ LR code with IS locality over \fq\ with locality parameter $r$ and minimum distance $d_{\min}(n,k,r)$. Then
	\bea
	d_{\min}(n,k,r) & \le &  (n-k+1) - (\left\lceil \frac{k}{r} \right\rceil - 1). \label{Singleton}
	\eea
\end{thm}      
On specializing to the case when $r=k$, \eqref{Singleton} yields the Singleton bound and for this reason, the bound in \eqref{Singleton} is referred to as the Singleton bound for an LR code. Note that since \eqref{Singleton} is derived for an LR code with IS locality, it also applicable to an LR code with AS locality.

\subsection{Constructions of LR Codes} \label{sec:LRCodeConst}

In the following, we will describe two constructions of LR codes having field size of $O(n)$ and achieving the bound \eqref{Singleton}. The first construction called the Pyramid Code construction \cite{HuaChenLi}, allows us to construct for any given parameter set $\{n, k, r\}$, an LR code with IS locality achieving the bound in \eqref{Singleton}. The second construction which appeared in \cite{TamBar_LRC}, gives constructions for LR codes with AS locality, achieving the bound \eqref{Singleton} for any $n,k,r$ under the constraint that $(r+1) \ | \ n$.

\subsubsection{Pyramid Code Construction} \label{sec:pyramid} 

The pyramid code construction technique \cite{HuaChenLi}, allows us to construct for any given parameter set $\{n, k, r\}$ an LR code with IS locality achieving the bound in \eqref{Singleton}. We sketch the construction for the case $k=2r$. The general case $k=ar$, $a>2$ or when $r \nmid k $, follows along similar lines. The construction begins with the systematic generator matrix \gmds\ of an $[n_1,k]$ scalar MDS code $\mathcal{C}_{\text{\tiny MDS}}$ having block length $n_1 = n-1$.  It then reorganizes the sub-matrices of \gmds\ to create the generator matrix \gpyr\ of the pyramid code as shown in the following: 
\bean
\gmds\  = \left[ \begin{array}{cccc}
	I_r & & P_1 & Q_1 \\ 
	& I_r & \underbrace{P_2}_{ (r \times 1)} & \underbrace{Q_2}_{(r \times s)}
\end{array} \right]
& \Rightarrow & 
G_{\text{\tiny PYR}} = \left[\begin{array}{ccccc}
	I_r & P_1 & & & Q_1\\ & & I_r & P_2 & Q_2
\end{array}\right].
\eean
where $s = n_1 - 2r - 1$. 
It is not hard to show that the $[n,k]$ code $\mathcal{C}_{\text{\tiny PYR}}$ generated by $G_{\text{\tiny PYR}}$ is an LR code with IS locality and that $d_{\min}(\mathcal{C}_{\text{\tiny PYR}}) \ \ge \ d_{\min}(\mathcal{C}_{\text{\tiny  MDS}})$. 
It follows that
\bean
d_{\min}(\mathcal{C}_{\text{\tiny PYR}}) & \ge & d_{\min}(\mathcal{C}_{\text{\tiny  MDS}}) = n_1 - k+1 \ = \ (n-k+1)-1, 
\eean
and the code $\mathcal{C}_{\text{\tiny PYR}}$ thus achieves the Singleton bound in \eqref{Singleton}.

\subsubsection{The Tamo-Barg Construction } \label{sec:TB}

\begin{figure}[h!]
	\centering
	\includegraphics[width=2.in, height=1.5in]{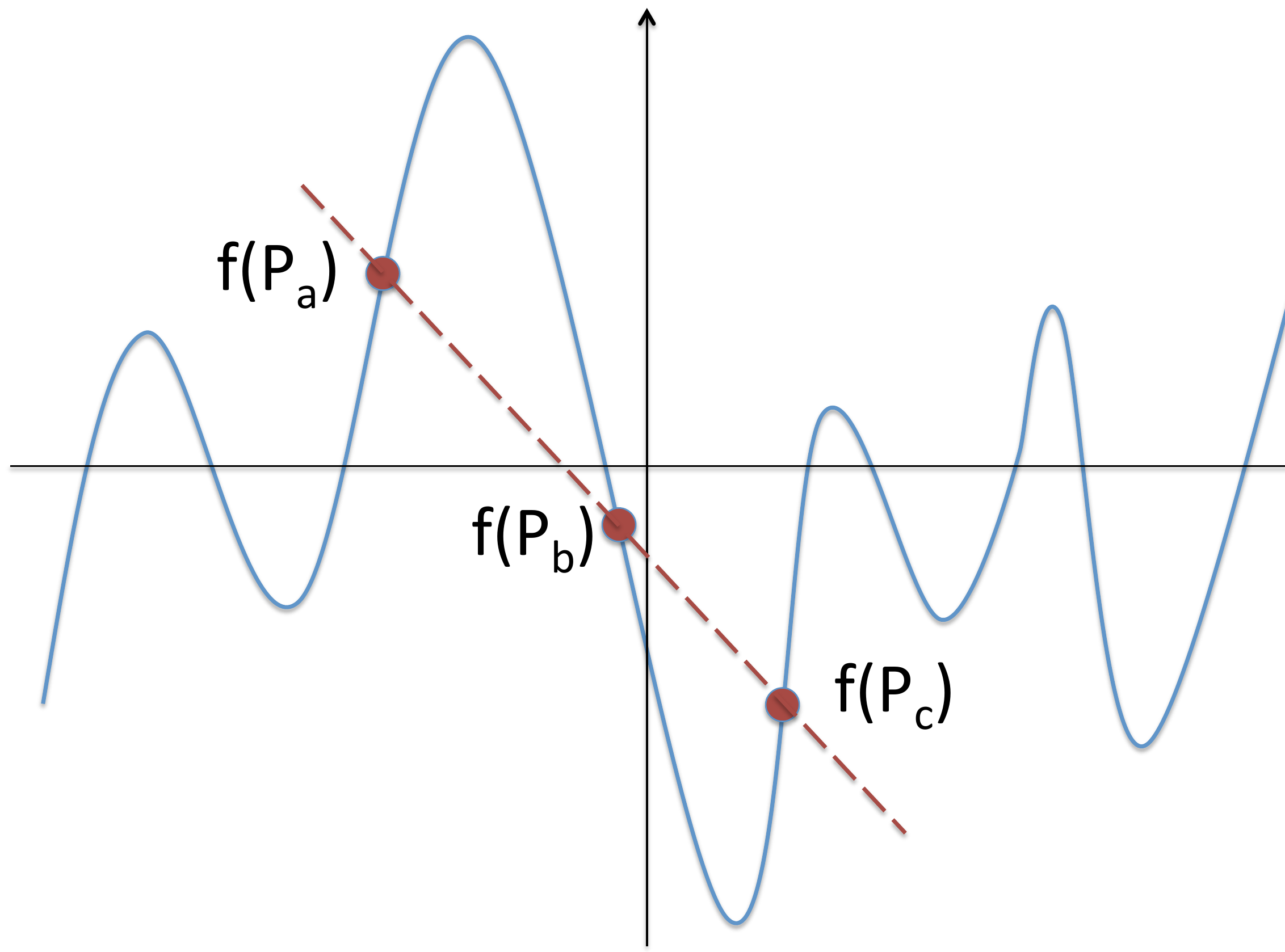}
	\caption{In the T-B construction, code symbols in a local code of length $(r+1)$ correspond to evaluations of a polynomial of degree $\leq (r-1)$. Here, $r=2$ implies that a local code corresponds to evaluation at $3$ points of a linear polynomial.\label{localcode_ch2}}  
\end{figure}

The construction below by Tamo and Barg~\cite{TamBar_LRC}, provides a construction for LR codes with AS locality achieving the bound \eqref{Singleton} for any $n,k,r$ with $(r+1)\ | \ n$. We will refer to this construction as the Tamo-Barg (T-B) construction. 
Let \fq\ be a finite field of size $q$, let $r \geq 2$, $n =m(r+1) \leq q$, with $ m \geq 2$ and $2 \leq k \leq (n-1)$.  Set $k=ar+b, 0 \leq b \leq (r-1)$.  Let $A=\{\theta_1,\theta_2,\cdots,\theta_n\} \subseteq \mathbb{F}_q$  and let $A_i \subset A, 1 \leq i \leq m$, $|A_i|=(r+1), A_i \cap A_j = \phi, i \neq j$, such that $A = \cup_{i=1}^m A_i$ represent a partitioning of $A$.
Let $g(x)$ be a `good' polynomial, by which is meant, a polynomial over \fq\ that is constant on each $A_i$ i.e., $g(x)= c_i, \forall x \in A_i$ for some $c_i \in \fq$ and degree of $g(x)$ is $(r+1)$.  
Let 
\bean
f(x) & = & \sum_{j=0}^{a-1} \sum_{i=0}^{r-1} a_{ij}[g(x)]^jx^i \ + \   \sum_{j=a} \sum_{i=0}^{b-1} a_{ij}[g(x)]^jx^i, 
\eean
where the $a_{ij} \in \fq$ are the message symbols and where the second term is vacuous for $b=0$, i.e., when $r \mid k$.  
Consider the code ${\cal C}$ of block length $n$ and dimension $k$ where the codeword \underline{c} of length $n$ corresponding to a given $k$ message symbols $\{a_{ij}\}$ is obtained by evaluating $f(x)$ at each of the $n$ elements in $A$ after substituting the given values of $k$ message symbols in the expression for $f(x)$. It can be shown that \calc is an LR code with AS locality with locality parameter $r$ and achieves the $d_{\min}$ bound in \eqref{Singleton}. The $i$-th local code corresponds to evaluations of $f(x)$ at elements of $A_i$ (also see Fig \ref{localcode_ch2}). 

An example of how good polynomials may be constructed is given below, corresponding 
to the annihilator polynomial of a multiplicative subgroup $G$ of $\mathbb{F}_q^*$.    
\begin{example}
	Let $H < G \le \fqstar$ be a chain of cyclic subgroups, where $|H|=(r+1), |G|=n$ so that $(r+1)$  $|$  $n$  $|$ $(q-1)$.  Let $n=(r+1)t$.  Let $\{A_i = \gamma_i H \mid i=1,2,\cdots, t\}$ be the $t$ multiplicative cosets of $H$ in $G$, with $\gamma_1$ being the multiplicative identity so that $A_1=H$.  It follows that 
	\bean
	\prod_{\beta \in A_i} (x-\beta) \ = \ x^{r+1}-\gamma_i^{r+1},
	\eean
	so that $x^{r+1}$ is constant on all the cosets of $H$ in $G$ and may be selected as the good polynomial $g(x)$ i.e., $g(x)=x^{r+1}$ is one possible choice of good polynomial based on multiplicative group $H$. 
\end{example}
Further examples may be found in \cite{TamBar_LRC,Sihem_Mesnager,KolBarTamYad}. For constructions meeting the Singleton bound with field size of $O(n)$ and greater flexibility in selecting the value of $r$, please see \cite{LinLimCha}.

\section{Alphabet-Size Dependent Bounds} \label{Alphabet_dep_LR_codes}

This section contains the contributions of the thesis on the topic of LR codes for the case of single erasures.  These include the best-known alphabet-size-dependent bounds on both minimum distance $d_{\min}$ and dimension $k$ for LR codes for $q>2$. For $q=2$, our bound on dimension is the tightest known bound for $d_{\min} \geq 9$.

The bound in equation \eqref{Singleton} as well as the bounds for non-linear and vector codes derived in \cite{PapDim,SilRawKoyVis} hold regardless of the size $q$ of the underlying finite field.  The theorem below which appeared in \cite{CadMaz} takes the size $q$ of the code symbol alphabet into account and provides a tighter upper bound on the dimension of an LR code for a given $(n,r,d)$ that is valid even for nonlinear codes where $d$ is the minimum distance of the code.  The `dimension' of a nonlinear code \calc\ over an alphabet $\mathbb{Q}$ of size $q=|\mathbb{Q}|$ is defined to be the quantity $k=\log_q(| \calc |)$. 

\begin{thm} \cite{CadMaz} \label{thm:CadMaz} 
	Let \calc\  be an $(n,k,d)$ LR code with AS locality and locality parameter $r$ over an alphabet $\mathbb{Q}$ of size $q=|\mathbb{Q}|$.  Then the dimension $k$ of the code must satisfy: 
	\bea
	k \  \leq \  \min_{t \in \mathbb{Z}_{+}} [ tr+ k_{\text{opt}}^{(q)}\left(n-t(r+1),d\right) ],  \label{CadMaz}
	\eea
	where $k_{\text{opt}}^{(q)}\left(n-t(r+1),d\right)$ denotes the largest possible dimension of a code (no locality necessary) over $\mathbb{Q}$ having block length $(n-t(r+1))$ and minimum distance $d$.
\end{thm}
\bpf (Sketch of proof) 
The bound holds for linear as well as nonlinear codes. In the linear case, with $\mathbb{Q}=\mathbb{F}_q$, the derivation proceeds as follows.  Let $G$ be a $(k \times n)$ generator matrix of the locally recoverable code ${\cal C}$.  Then it can be shown that for any integer $t > 0$,  there exists an index set ${\cal I}$ such that $|{\cal I}|=\min(t(r+1),n)$ and $\text{rank} \left( G \mid_{\cal I} \right) \ = \ s \leq tr$ where $ G \mid_{\cal I}$ refers to the set of columns of $G$ indexed by ${\cal I}$. This implies that \calc\ has a generator matrix of the form (after permutation of columns): 
\bean
G & = & 
\left[ \begin{array}{cc} 
	\underbrace{A}_{(s \times |{\cal I}|)} & B \\ 
	\left[ 0 \right]  & D
\end{array} \right] . 
\eean
In turn, this implies that the rowspace of $D$ defines an $[n-t(r+1),k-s \geq k-tr,\geq d]$ code over $\mathbb{F}_q$, if $k-tr > 0$.  
It follows that $k \leq tr + k_{\text{opt}}^{(q)}(n-t(r+1),d)$ and the result follows.  Note that the row space of $D$ corresponds to a shortening $\calc^{S}$ of \calc\ with respect to the coordinates ${\cal I} \subseteq [n]$.  The proof in the general case is a (nontrivial) extension to the nonlinear setting.  \epf 

We next look at a bound on dimension of binary LR codes for a given $(n,r,d)$ that appeared in \cite{WanZhaLin}.  We remark that there is an additional bound on dimension given in \cite{WanZhaLin} for the case when the local codes are disjoint i.e., the case when the support sets $\{S_i\}$, are pairwise disjoint.  However, here we only provide the bound on dimension given in \cite{WanZhaLin}, which applies in full generality, and without the assumption of disjoint local codes. 
\begin{thm} \cite{WanZhaLin}
	For any $[n,k,d]$ linear code  \calc\ that is an LR code with AS locality with locality parameter $r$ over $\mathbb{F}_2$ with $d \geq 5$ and $2 \leq r \leq \frac{n}{2}-2$, we must have:
	\bea
	k \leq \frac{rn}{r+1} - \min\{\log_2(1+\frac{rn}{2}),\frac{rn}{(r+1)(r+2)}\}. \label{Hamming}
	\eea	
\end{thm}
The above bound is obtained by applying a Hamming-bound-type argument to an LR code with AS locality.
In \cite{HuaYaaUchSie}, the authors provide  a bound on the minimum distance of $[n,k]$ LR codes with IS locality\footnote{The bound has an extension to codes with availability as well, see Chapter~\ref{ch:avail} for the definition of availability.},  (the bound thus applies to LR codes with AS  locality as well) that depends on the size $q$ of the underlying finite field $\mathbb{F}_q$:
\begin{thm} 
	For any $[n,k,d]$ linear code  \calc\ that is an LR code with IS locality with locality parameter $r$ over $\mathbb{F}_q$:
	\bea
	d \leq \min_{1 \leq x \leq \lceil \frac{k}{r} \rceil-1,} d^{(q)}(n-(r+1)x,k-rx), \label{dminAlp}
	\eea
	where $d^{(q)}(n,k)$ is the maximum possible minimum distance of a classical (i.e., no locality necessary) $[n,k]$ block code over $\mathbb{F}_q$.	
\end{thm}

We next introduce the notion of Generalized Hamming Weights (also known as Minimum Support Weights) which will be used to derive a new bound on the minimum distance and dimension of an LR code with AS locality, that takes into account the size $q$ of the underlying finite field $\fq$.  The bound makes use of the technique of code shortening and the GHWs of a code provide valuable information about shortened codes. 
%
%
%
%
%
\subsubsection{GHW and the Minimum Support Weight Sequence} 
We will first define the Generalized Hamming Weights of a code, introduced in \cite{Wei}, and also known as  Minimum Support Weights (MSW) (see \cite{HelKloLevYtr}) of a code. In this thesis we will use the term Minimum Support Weight (MSW).
\begin{defn}
	The $i$th Minimum Support Weight (MSW) $d_i$ (equivalently, the $i$th Generalized Hamming Weight) of an $[n,k]$ code $\calc$ is the cardinality of the minimum support of an $i$-dimensional subcode of C, i.e.,
	\bea
	d_i(\calc) = d_i = \min_{\{\mathcal{D} : \ \mathcal{D} < \calc, \ \dim(\mathcal{D})=i \} }|\text{supp}(\mathcal{D})|, \label{MSW}
	\eea
	where the notation $\mathcal{D} < \calc$ denotes a subcode $\mathcal{D}$ of $\calc$ and where $\text{supp}(\mathcal{D})=\cup_{\underline{c} \in \mathcal{D}} \text{supp}(\underline{c})$ (called the support of the code $\mathcal{D}$).
\end{defn}
Although the MSW definition applies to any code, the interest in this thesis, is on its application to a restricted class of codes that we introduce here. 
\begin{defn} [Canonical Dual Code] \label{def:Can_dualcode_ch2}
	By a canonical dual code, we will mean an $[n,m]$ linear code \calc satisfying the following: \calc contains a set $\{\underline{c}_i\}_{i=1}^b$ of $b \leq m$ linearly independent codewords of Hamming weight $\leq r+1$, such that the sets $S_i = supp(\underline{c}_i)$, $1 \leq i \leq b$ cover $[n]$, i.e., 
	\bean
	[n] & \subseteq & \cup_{i=1}^b S_i .
	\eean
\end{defn}
As it turns out, the dual code of an LR code with AS locality (and as we shall see in subsequent chapters, dual of codes with sequential recovery and dual of codes with availability) is an example of a canonical dual code and this is the reason for our interest in the MSWs of this class of codes. 

\begin{thm} \cite{PraLalKum} \label{thm:msw_sequence}
	Let \calc\ be a canonical dual code with parameters $(n,m,r,b)$ and support sets $\{S_i \mid 1 \leq i \leq b\}$ as defined in Definition \ref{def:Can_dualcode_ch2}.  Let $d_i$ denote the $i$th $1 \leq i \leq m$ MSW of \calc.  Let $f_i$ be the minimum possible value of cardinality of the union of any $i$ distinct support sets $\{S_{j} :j \in T \}$, $|T| = i$ , $1 \leq i \leq b$ i.e., 
	\bean
	f_i = \min_{\{T: \ T \subseteq [b], |T|=i\}} |\cup_{j \in T} S_j|.
	\eean
	Let $b_1 \leq b$. Let the integers $\{e_i\}_{i=1}^{b_1}$ be recursively defined as follows: 
	\bea
	e_{b_1} & = & n, \\
	e_{i-1} & = &  \min \{e_i, e_i-\left \lceil \frac{2e_i}{i} \right \rceil + r+1\}, \ \  2 \leq i \leq b_1.  \label{Sequence}
	\eea
	Then 
	\bean
	d_i \ \leq \ f_i & \leq & e_i, \text{ for all } \ \ 1 \leq i \leq b_1. 
	\eean
\end{thm}

Note that in \cite{PraLalKum}, the Theorem \ref{thm:msw_sequence} is proved for the case $b_1=b$, $|S_j|=r+1$, $\forall j \in [b]$ but we observe from the proof of Theorem \ref{thm:msw_sequence} (proof for $f_i \leq e_i$) given in \cite{PraLalKum} that the Theorem \ref{thm:msw_sequence} is also true for any $b_1 \leq b$, $|S_j| \leq r+1$, $\forall j \in [b]$.
We will refer to the sequence $\{e_i\}_{i=1}^{b_1}$ appearing in the theorem above as the {\em Minimum Support Weight (MSW) Sequence} associated to parameter set $(n,b_1,r)$.  In the subsection below, we derive new alphabet-size dependent bounds on minimum distance and dimension, that are expressed in terms of the MSW sequence.

\subsection{New Alphabet-Size Dependent Bound Based on MSW} \label{LR_Codes_New_Bound}

In this subsection, we present field-size dependent bounds on the minimum distance and dimension of an LR code \calc\ with AS locality with $r$ as locality parameter.  The bounds are derived in terms of the MSW sequence associated with the dual $\calc^{\perp}$ of \calc.  The basic idea is to shorten the LR code to a code with $e_i$ ($i$th term of MSW sequence) code symbols set to zero for some $i$.  Theorem~\ref{thm:msw_sequence} provides a lower bound on the dimension of this shortened code. Classical bounds on the parameters of this shortened code are shown to yield bounds on the parameters of the parent LR code.

\begin{thm} \label{thm:dmin1}
	Let $\mathcal{C}$ be an $[n,k]$ LR code with AS locality with locality parameter $r$ over a field $\mathbb{F}_q$ with minimum distance $d$. Let $d^q_{\text{min}}(n,k,r)$ be the maximum possible minimum distance of an $[n,k]$ LR code with AS locality with locality parameter $r$ over a field $\mathbb{F}_q$. Then:
	\bea
	d \leq  \min_{i \in S} \ d^q_{\text{min}}(n-e_i,k+i-e_i,r) \leq \min_{i \in S} d^{(q)}(n-e_i,k+i-e_i),\ \ \text{ and }  \label{Mindist_Bound1} \\
	k \  \leq \  \min_{\{i: e_i < n-d+1 \}} [ e_i-i + k_{\text{opt}}^{(q)}(n-e_i,d) ] , \label{GHW_dimension_Bound}
	\eea
	where 
	\ben
	\item[(i)] $S=\{ i: e_i-i < k, 1 \leq i \leq b_1 \}$, 
	\item[(ii)] $d^{\perp}_i = d_i(\calc^{\perp}) \leq e_i$, $\forall 1 \leq i \leq b_1$
	\item[(iii)]  $b_1=\lceil \frac{n}{r+1} \rceil$, 
	\item[(iv)]  $d^{(q)}(n,k)$ is the maximum possible minimum distance of a classical (i.e., no locality necessary) $[n,k]$ block code over $\mathbb{F}_q$ and \item[(v)] $k_{\text{opt}}^{(q)}(n,d)$ is the largest possible dimension of a code  (i.e., no locality necessary) over $\fq$ having block length $n$ and minimum distance $d$.
	\een
\end{thm}
\begin{proof}
	Since $\calc$ is an LR code with AS locality with locality parameter $r$, we have that $\calc^{\perp}$ is an $[n,n-k]$ canonical dual code with locality parameter $r$ and $b \geq \lceil \frac{n}{r+1} \rceil$.   We explain the reason for the inequality in the symbol $b$.  
	
	Since $\calc$ is an LR code AS locality, for every code symbol $c_{i}$, there is a codeword in $\calc^{\perp}$ of weight $\leq r+1$ whose support contains $i$. So take a codeword in $\calc^{\perp}$ of weight $\leq r+1$ whose support set $S_1$ contains $1$. Next, choose $\ell_2 \in [n]-S_1$ and take a codeword in  $\calc^{\perp}$ of weight $\leq r+1$ whose support set $S_2$ contains $\ell_2$. Repeat this process. At the $i^{th}$ step, choose $\ell_i \in [n]-\cup_{j=1}^{i-1} S_j$ and take a codeword in  $\calc^{\perp}$ of weight $\leq r+1$ whose support set $S_{i}$ contains $\ell_{i}$. Note that the set of $i$ codewords corresponding to support sets $S_1,...,S_i$ form a set of $i$ linearly independent codewords in $\calc^{\perp}$ and this process process can be repeated until $\cup_{j=1}^{b} S_j = [n]$ for some $b$. Since $\cup_{j=1}^{b} S_j = [n]$, we have that  $\sum_{j=1}^{b} |S_j| \geq n$ which implies $b \geq \lceil \frac{n}{r+1} \rceil$ as $|S_j|\leq r+1$.
	
	We now set $b_1 = \lceil \frac{n}{r+1} \rceil$. Hence from Theorem \ref{thm:msw_sequence}, $d_j(\calc^{\perp}) \leq e_j$, $\forall 1 \leq j \leq b_1$. For simplicity, let us write $d_j^{\perp} = d_j(\calc^{\perp})$, $\forall 1 \leq j \leq b_1$. 
	Next, fix $i$ with $1 \leq i \leq b_1$. Let $S'=\{s_1,...,s_{d_i^{\perp}} \}$ be the support of an $i$ dimensional subspace or subcode of $\mathcal{C}^{\perp}$ with the support having cardinality exactly $d_i^{\perp}$ in $\mathcal{C}^{\perp}$. Add $e_i-d_i^{\perp}$ arbitrary extra indices to $S'$ and let the resulting set be $S$. Hence $S' \subseteq S \subseteq [n]$ and $|S|=e_i$. Now shorten the code $\mathcal{C}$ in the co-ordinates indexed by $S$ i.e., take $\mathcal{C}^S=\{\underline{c}|_{S^c} : \underline{c} \in \mathcal{C}, \underline{c}|_S = \underline{0}\}$ where $S^c$ is the compliment of $S$ and $\underline{c}|_{A} = [c_{j_1},...,c_{j_{|A|}}]$ for a set $A=\{j_1,...,j_{|A|}\}$ with $j_{\ell} < j_{\ell+1}$,$\forall \ell \in [|A|-1]$. The resulting code $\mathcal{C}^S$ has block length $n-e_i$, dimension $\geq n-e_i-(n-k-i)=k+i-e_i$ and minimum distance $\geq d$ (if $k+i-e_i > 0$) and the resulting code $\mathcal{C}^S$ is also an LR code with AS locality with locality parameter $r$. Hence:
	\bean
	d \leq  \min_{i \in S} \ d^q_{\text{min}}(n-e_i,k+i-e_i,r) \leq \min_{i \in S} d^{(q)}(n-e_i,k+i-e_i).
	\eean 
	The proof of \eqref{GHW_dimension_Bound} follows from the fact that $\text{dim}(\mathcal{C}^S) \geq k+i-e_i$ and Hence :
	\bea
	k+i-e_i \leq  \text{dim}(\mathcal{C}^S)  \leq k_{\text{opt}}^{(q)}(n-e_i,d)
	\eea
\end{proof}

An example comparison of the upper bounds on dimension of linear LR codes given in \eqref{CadMaz}, \eqref{Hamming} and \eqref{GHW_dimension_Bound} (our bound) is presented in Table~\ref{tab:bd_compare}. The word bound in the following refers to an upper bound.

\begin{table}
	\caption {A comparison of upper bounds on the dimension $k$ of a binary LR code, for given $(n,d,r,q)$.} \label{tab:bd_compare}
	\begin{center}
		\begin{tabular}{ || m{7em} | m{0.5cm}| m{0.5cm} |  m{0.5cm}| m{0.5cm} |  m{0.5cm}|| m{0.5cm} |  m{0.5cm}| m{0.5cm} | m{0.5cm}|| } 
			\hline \hline 
			\multicolumn{6}{|c|}{$n=31$, $q=2$, $d=5$}   \\
			\hline
			$r$ (locality) & 2 & 3 & 4 & 5 & 6 \\ 
			\hline \hline 
			Bound \eqref{CadMaz} & 17 & 19 & 20 & 20 & 20 \\ 
			\hline
			Bound in \eqref{Hamming} & 15 & 18 & 20 & 22 & 23 \\
			\hline
			Bound \eqref{GHW_dimension_Bound} & 16 & 18 & 19 & 20 & 20 \\
			\hline
		\end{tabular}
	\end{center}
\end{table}
Since $e_i \leq i(r+1)$, it can be seen that the bound \eqref{Mindist_Bound1} is tighter than the bound \eqref{dminAlp} when applied to an LR code with AS locality.  For the same reason, the bound \eqref{GHW_dimension_Bound} is tighter than bound \eqref{CadMaz}.  For $q=2$, it is mentioned in \cite{WanZhaLin}, that the bound \eqref{Hamming} is looser than the bound in \eqref{CadMaz} for $d \geq 9$.  Since our bound \eqref{GHW_dimension_Bound} presented here is tighter than the bound appearing in \eqref{CadMaz}, we conclude that for $d \geq 9$, $q=2$, our bound \eqref{GHW_dimension_Bound} is tighter than the bound in \eqref{Hamming}.

Hence our bounds \eqref{Mindist_Bound1},\eqref{GHW_dimension_Bound} are the tightest known bounds on minimum distance and dimension for $q > 2$. For $q=2$, our bound \eqref{GHW_dimension_Bound} is the tightest known bound on dimension for $d \geq 9$. 
We here note that the bounds \eqref{Mindist_Bound1}, \eqref{GHW_dimension_Bound} apply even if we replace $e_i$ with any other upper bound on $i^{th}$ MSW.   The bounds derived here are general in this sense.
For the sake of completeness, in the following we give a survey of existing small alphabet size constructions which are optimal w.r.t bounds appearing in the literature. 

\begin{note}
	Let $0 < R \leq 1$ and let $r$ be a positive integer. Then if $\delta_L(R)=\limsup_{n \rightarrow \infty} \frac{d^q_{\text{min}}(n,nR,r)}{n}$ and $\delta(R)=\limsup_{n \rightarrow \infty} \frac{d^{(q)}(n,nR)}{n}$, it is trivial to observe that:
	\bean
	\delta(R+\frac{1}{r+1}) \leq \delta_L(R) \leq \delta(R)
	\eean
	Trivially, if $r$ is a monotonic function of $n$, then assuming continuity of $\delta(R)$, we get $\delta_L(R) = \delta(R)$. This is because if $r$ is a monotonic function of $n$ then the number of  linearly independent codewords of weight $\leq r+1$ in the dual code needed to satisfy the conditions necessary for an LR code is a negligible fraction of $n$ as $n$ increases. Hence the region of interest in locality is when $r$ is a constant or when $n$ is a small number.
\end{note}

\subsection{Small-Alphabet Constructions} \label{Small_Alph_LR_Codes}

\subsubsection{Construction of Binary Codes} 

Constructions for binary codes that achieve the bound on dimension given in \eqref{CadMaz} for binary codes, appear in \cite{NamSo,SilZeh,HaoXiaChe}. While \cite{HaoXiaChe} and \cite{NamSo} provide constructions for $d_{min}=4$ and $d_{\min}=6$ respectively, the constructions in \cite{SilZeh} handle the case of larger minimum distance but have locality parameter restricted to $r \in  \{ 2,3 \}$.  In \cite{HuaYaaUchSie}, the authors give optimal binary constructions with information and all symbol locality with $d_{min} \in \{3,4\}$. The construction is optimal w.r.t the bound \eqref{dminAlp}.  Constructions achieving the bound on dimension appearing in \cite{WanZhaLin} and the further tightened bound for disjoint repair groups given in \cite{MaGe} for binary codes, appear respectively, in \cite{WanZhaLin,MaGe}. These constructions are for the case $d_{min}=6$. 
In \cite{HaoXiaChe}, the authors present a characterization of binary LR codes that achieve the Singleton bound \eqref{Singleton}. In \cite{ShaKhaArd}, the authors present constructions of binary codes meeting the Singleton bound. These codes are a subclass of the codes characterized in  \cite{HaoXiaChe} for the case $d_{min} \leq 4$.

\subsubsection{Constructions with Small, Non-Binary Alphabet }

In \cite{HaoXiaChe1}, the authors characterize ternary LR codes achieving the Singleton bound \eqref{Singleton}.  In \cite{HaoXiaChe,ShaKhaArd, HaoXia}, the authors provide constructions for codes over a field of size $O(r)$ that achieve the Singleton bound in \eqref{Singleton} for $d_{min} \leq 5$.   Some codes from algebraic geometry achieving the Singleton bound \eqref{Singleton} for restricted parameter sets are presented in \cite{LiMaXin1}.

\subsubsection{Construction of Cyclic LR Codes}

Cyclic LR codes can be constructed by carefully selecting the generator polynomial $g(x)$ of the cyclic code.  We illustrate a key idea behind the construction of a cyclic LR code by means of an example.  
\begin{figure}[h!]
	\centering 		\includegraphics[width=4in]{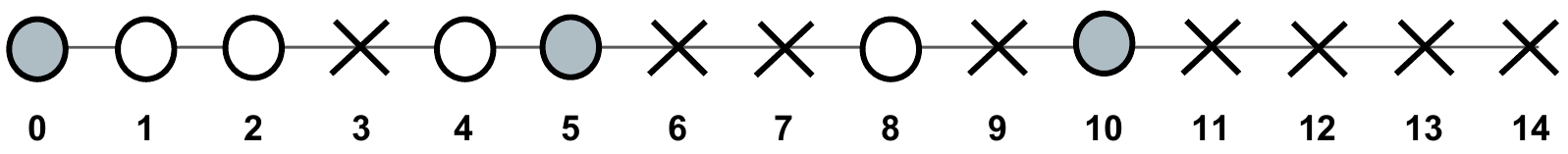}
	\caption{Zeros of the generator polynomial $g(x)=\frac{g_1(x)g_2(x)}{(x+1)}$ of the cyclic code in Example~\ref{eg:zero_train} are identified by circles. The unshaded circles along with the shaded circle corresponding to $\alpha^0=1$ indicate the zeros $\{1,\alpha,\alpha^2,\alpha^4,\alpha^8\}$ of $g_1(x)$ selected to impart the code with $d_{\min} \geq 4$.  The shaded circles indicate the periodic train of zeros $\{1,\alpha^5,\alpha^{10}\}$ introduced to cause the code to be locally recoverable with parameter $(r+1)=5$. The common element $1$ is helpful both to impart increased minimum distance as well as locality.}
	\label{fig:zero_train} 
\end{figure}
%

\begin{eg} \label{eg:zero_train} 
	Let $\alpha$ be a primitive element of $\mathbb{F}_{16}$ satisfying $x^4+x+1=0$.  Let $\calc_1$ be a cyclic $[n=15,k=10]$ code having generator polynomial $g_1(x)=(x+1)(x^4+x+1)$.    
	Since the consecutive powers $\{1,\alpha,\alpha^2\}$ of $\alpha$ are zeros of $g_1(x)$, it follows that $d_{\min}(\calc)\geq 3+1=4$ by the BCH bound.  Suppose we desire to ensure that a code \calc\ having generator polynomial $g(x)$ has $d_{\min} \geq 4$ and in addition, is locally recoverable with parameter $(r+1)=5$, then we do the following.  Set $s=\frac{n}{(r+1)}=3$.   Let $g_2(x) \ = \  \prod_{l=0}^{s-1=2} (x-\alpha^{5l})$ and $g(x) = \text{lcm} \{g_1(x), g_2(x)\} = g_1(x)g_2(x)/(x+1)$.
	It follows that $\sum_{t=0}^{14} c_t \alpha^{5lt} \ = \  0, \ \ l=0,1,2$. Summing over $l$ we obtain:
	\bean
	\sum_{l=0}^2 \sum_{t=0}^{14} c_t \alpha^{5lt} \ =  \ 0  & \Rightarrow & 
	\sum_{t: t= 0 \pmod{3}}c_t \ = \ 0.  
	\eean
	It follows that the symbols $\{c_t \mid t= 0 \pmod{3}\}$ of \calc\ form a local code as they satisfy the constraint of an overall parity-check.  Since the code \calc\ is cyclic the same holds for the code symbols $\{c_{t+\tau} \mid t= 0 \pmod{3}\}$, for $\tau = 0,1,2$.   Thus through this selection of generator polynomial $g(x)$, we have obtained a code that has both locality and $d_{\min} \geq 4$.   The zeros of $g(x)$ are illustrated in Fig.~\ref{fig:zero_train}. The code \calc\ has parameters $[n=15,k=8,d_{min} \geq 4]$ and $r=4$. Note that the price we pay for introduction of locality is a loss in code dimension, equal to the degree of the polynomial $\frac{g_2(x)}{ \text{gcd} \{g_1(x), g_2(x)\}}$.  Thus an efficient code will choose the zeros of $g_1(x),g_2(x)$ for maximum overlap. 
\end{eg}
The above idea of constructing cyclic LR code was introduced in \cite{GopCal} and extended in \cite{TambarGopCal,ZehYak,KimNo,LuoXinYua}. In \cite{KriPurKumTamBarISIT17}, the use of locality for reducing the complexity of decoding a cyclic code is explored. The same paper also makes a connection with earlier work \cite{VardyBeery94} that can be interpreted in terms of locality of a cyclic code.
%
%
In \cite{GopCal} a construction of binary cyclic LR codes for $r=2$ an $d_{min} \in \{2,6,10\}$  achieving a bound derived within the same paper for binary codes is provided.   In \cite{KimNo}, the authors give constructions of optimal binary, ternary codes meeting the Singleton bound \eqref{Singleton} for $d_{min} = 4, r \in \{1,3\}$ and $d_{min}=6, r=2$ as well as a construction of a binary code meeting the bound given in \cite{WanZhaLin} for $d_{min}=6,r=2$  based on concatenating cyclic codes.  A discussion on the locality of classical binary cyclic codes as well as of codes derived from them through simple operations such as shortening, can be found in \cite{HuaYakUchSie,HuaYaaUchSie}. The principal idea here is that any cyclic code has locality $d^{\perp}-1$ where $d^{\perp}$ is the minimum distance of the dual code $\calc^{\perp}$. In \cite{ZehYak}, the authors construct optimal cyclic codes under the constraint that the local code is either a Simplex code or else, a Reed-Muller code.   In \cite{TambarGopCal}, the authors provide a construction of cyclic codes with field size $O(n)$ achieving the Singleton bound \eqref{Singleton} and also study the locality of subfield subcodes as well as their duals, the trace codes. In \cite{LuoXinYua}, constructions of cyclic LR codes with $d_{min} \in \{3,4\}$ for any $q$ and flexible $n$ are provided. 

\section{Summary} \label{Summary}
This chapter dealt with LR codes for the single-erasure case and presented the requisite background as well as the contributions of the thesis in this direction.  The thesis contributions on LR codes for single erasure case correspond to new alphabet-size dependent upper bounds on the minimum distance $d_{\min}$ and dimension $k$ of a linear LR code.  Thus the upper bounds apply to the case of LR codes over a finite field $\mathbb{F}_q$ of fixed size $q$.   A key ingredient in the upper bounds derived here are the bounds on the Generalized Hamming Weights (GHW) derived in \cite{PraLalKum}.  Evidence was presented showing our upper bound on dimension to be tighter in comparison with existing upper bounds in the literature.  

	\chapter{LR Codes with Sequential Recovery} \label{ch:SeqBound}
	
	\section{Introduction} \label{Overview:Seq}
	
	The focus of the present chapter is on LR codes for multiple erasures. We begin by providing motivation for studying the multiple-erasure case (Section~\ref{sec:motivation}).  As there are several approaches towards handling multiple erasures in the literature, we next provide a broad classification of LR codes for multiple erasures (Section~\ref{sec:seq_classification}). The principal contributions of this thesis relate to a particular approach towards the recovery from multiple erasures, termed as sequential recovery.  Section~\ref{sec:seq_recovery} introduces LR codes with sequential recovery and surveys the known literature on the topic.    This is followed by an overview of the contributions of this thesis on the topic of LR codes with sequential recovery in Section~\ref{sec:seq_contributions}. Sections \ref{sec:2_erasures} and \ref{sec:3_erasures} respectively present in detail, the results obtained in this thesis, relating to the case of $2$ and $3$ erasures respectively.  Section~\ref{sec:t_erasures} presents a principal result of this thesis which involves establishing a tight upper bound on the rate of an LR code with sequential recovery for $t \geq 4$ erasures along with a matching construction.  The upper bound derived also proves a conjecture that had previously appeared in the literature.    The final section, Section~\ref{sec:seq_conclusions}, summarizes the contents of the chapter. Throughout this chapter, we use the term weight to denote the Hamming weight.
	
	\subsection{Motivation for Studying Multiple-Erasure LR Codes} \label{sec:motivation} 
	
	Given that the key problems on the topic of LR codes for the single erasure case have been settled, the academic community has turned its attention towards LR codes for multiple erasures.  
	
	\paragraph{Availability} A strong motivation for studying the multiple-erasure case, comes from the notion of {\em availability}.   A storage unit could end up storing data that is in extremely high demand at a certain time instant.  In such situations, regarded by the storage industry as {\em degraded reads}, the storage industry will look to create, on-the-fly replicas of the storage unit's data.  If a code symbol can be recreated in $t$ different ways by calling upon $t$ pairwise disjoint sets $S_i$ of helper nodes, then one could recreate $t$ copies of the data-in-demand in parallel.  But a code which can, for any code symbol recreate in this fashion $t$ simultaneous copies of a code symbol, also has the ability to correct $t$ erasures simultaneously.  This follows because any pattern of $(t-1)$ erasures can affect at most $(t-1)$ of the helper node sets $S_i$ and thus there is still a helper node set remaining that can repair the erased symbol. The problem of designing codes with availability in the context of locality was introduced in \cite{WanZha} and a high rate construction for availability codes appeared in \cite{WanZhaLiu}. For a survey on constructions of availability codes please see Chapter \ref{ch:avail}. Upper bounds on minimum distance and rate of an availabiltiy codes appeared in \cite{TamBarFro}, \cite{HuaYaaUchSie}, \cite{BalKum}, \cite{KruFro}.
	
	\paragraph{Other reasons} Other reasons for being interested in the multiple erasure setting include (a) the increasing trend towards replacing expensive servers with low-cost commodity servers that can result in simultaneous node failures and (b) the temporary unavailability of a helper node to assist in the repair of a failed node.   
	
	\section{Classification of LR Codes for Multiple Erasures} \label{sec:seq_classification}
	
	An overview of the different classes of LR codes that are capable of recovering from multiple erasures proposed in the literature is presented here.  All approaches to recovery from multiple erasures place a constraint $r$ on the number of unerased symbols that are used to recover from a particular erased symbol.  The value of $r$ is usually small in comparison with the block length $n$ of the code and for this reason, one speaks of the recovery as being local.  All the codes defined in this section are over a finite field \fq. A codeword in an $[n,k]$ code will be represented by $[c_1,c_2,...,c_i,...,c_{n-1},c_n ]$. In this chapter, we will restrict ourselves to only linear codes.
	
	\subsubsection{Sequential-Recovery LR Codes} An $(n,k,r,t)$ sequential-recovery LR code (abbreviated as S-LR code) is an $[n,k]$ linear code $\calc$ having the following property: Given a collection of $s \le t$ erased code symbols, there is an ordering $(c_{i_1}, c_{i_2}, \cdots, c_{i_s})$ of these $s$ erased symbols such that for each index $i_j$, there exists a subset $S_j \subseteq [n]$ satisfying (i) $|S_j| \le r$ , (ii) $\ S_j \cap \{i_j, i_{j+1}, \cdots, i_s \} = \phi$,  and (iii) 
	\begin{eqnarray} \label{eq:locality1}
		c_{i_j} & = & \sum \limits_{\ell \in S_j} u_{\ell} c_{\ell}, \ \text{$u_{\ell}$ depends on $(i_j,S_j)$ and $u_{\ell} \in \fq$ } .
	\end{eqnarray} 
	It follows from the definition that an $(n,k,r,t)$ S-LR code can recover from the erasure of $s$ code symbols $c_{i_1}, c_{i_2}, \cdots, c_{i_s}$, for $1 \leq s \leq t$ by using \eqref{eq:locality1} to recover the symbols $c_{i_j}, \ j=1,2,\cdots,s$, in succession. 
	\begin{figure}[ht]
		\centering
		\includegraphics[width=4in]{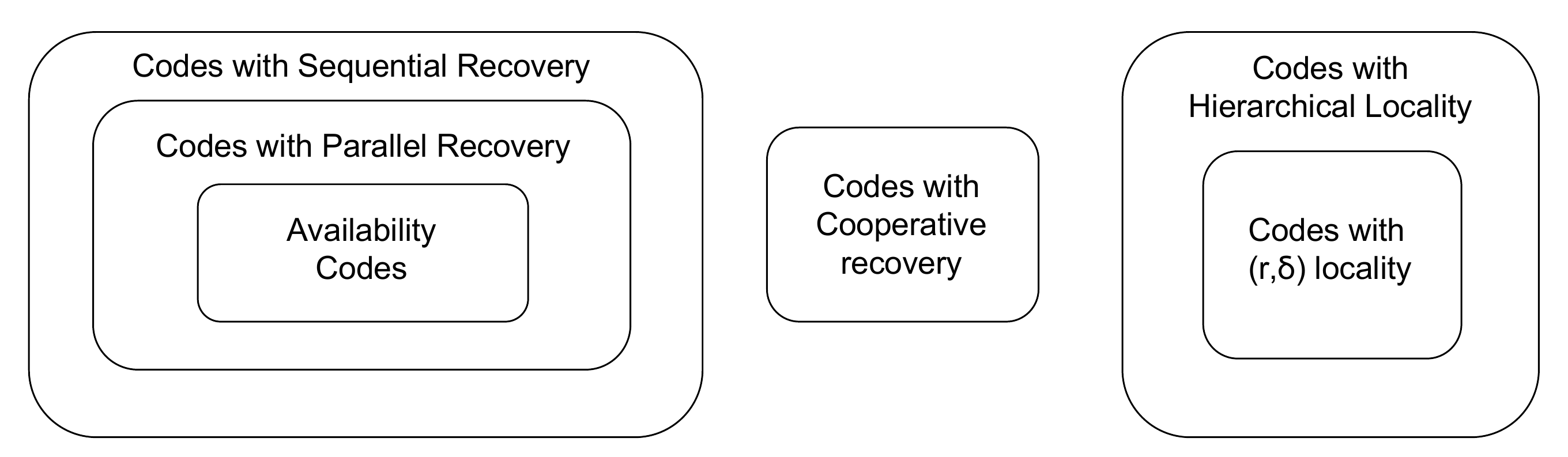}
		\caption{The various code classes of LR codes corresponding to different approaches to recovery from multiple erasures.}
		\label{fig:me_classes}
	\end{figure}
	
	\subsubsection{Parallel-Recovery LR Codes} If in the definition of the S-LR code, we replace the condition (ii) in \eqref{eq:locality1} by the more stringent requirement:
	\bea
	S_j \cap \{i_1, i_2, \cdots, i_s \} = \phi,
	\eea
	then the LR code will be referred to as a \noindent {\em parallel} recovery LR code, abbreviated as P-LR code.  Clearly the class of P-LR codes is a subclass of S-LR codes.  From a practical point of view, P-LR codes are preferred since as the name suggests, the erased symbols can be recovered in parallel. However, this will in general, come at the expense of storage overhead.  We note that under parallel recovery, depending upon the specific code, this may require the same helper (i.e., non-erased) code symbol to participate in the recovery of more than one erased symbol $c_{i_j}$.
	
	\subsubsection{Availability Codes} An $(n,k,r,t)$ \noindent {\em availability} LR code (see \cite{WanZha,WanZhaLiu,TamBarFro,BalKum}), is an $[n,k]$ linear code having the property that in the event of a single but arbitrary erased code symbol $c_i$, there exist $t$ recovery sets $\{R^i_j\}_{j=1}^t$ which are pairwise disjoint and of size $|R^i_j| \leq r$ with $i \notin R^i_j$, such that for each $j, 1 \leq j \leq t$, $c_i$ can be expressed in the form:
	\bean
	c_i = \sum\limits_{\ell \in R^i_j} a_{\ell} c_{\ell}, \  \text{$a_{\ell}$ depends on $(i,R^i_j)$ and $a_{\ell} \in \fq$ } .
	\eean
	An $(n,k,r,t)$ availability code is also an $(n,k,r,t)$ P-LR code.  This follows because the presence of at most $t$ erasures implies, that there will be at least one recovery set for each erased code symbol all of whose symbols remain unerased.
	
	\subsubsection{$(r,\delta)$ Codes} An $[n,k]$ linear code \calc\ is said to have AS $(r,\delta)$ locality (see \cite{PraKamLalKum,SonDauYueLi}), if for each co-ordinate  $i \in [n]$, there exists a subset $S_i \subseteq [n]$, with $i \in S_i$, with 
	\bea
	\label{eq:rdeltaloc}
	\dim(\calc|_{S_i} ) \le r, \ d_{\min}(\calc|_{S_i} ) \ge \delta. 
	\eea
	Recovery from $t$ erasures can also be accomplished by using the codes with $(r,\delta)$ locality, if one ensures that the code has $d_{\min} \geq t+1$.  However in this case, repair is local only in those cases where the erasure pattern is such that the number of erasures $e_i$ within the local code $\calc|_{S_i}$ satisfies $e_i \leq \delta-1$. 
	Thus one may regard $(r,\delta)$ codes as offering probabilistic guarantees of local recovery in the presence of $\leq t$ erasures in exchange for a potential increase in code rate.  Of course, one could always employ an $(r,\delta)$ with local MDS codes (i.e., the code $\calc|_{S_i}$ is MDS) and $\delta \geq t+1$, but this would result in a significant rate penalty. 
	
	\subsubsection{Cooperative Recovery} A \noindent {\em cooperative} recovery $(n,k,r,t)$ LR (C-LR) code (see \cite{RawMazVis}) is an $[n,k]$ linear code such that if a subset $E=\{c_{i_1}, c_{i_2}, \cdots, c_{i_s}\}$, $1 \leq s \le t$ of code symbols are erased then there exists a subset $S_E=\{c_{j_1}, c_{j_2}, \cdots, c_{j_r}\}$ of $r$ other code symbols
	(i.e., $i_a \ne j_b$ for any $a, b$) such that for 
	all $a \in [s]$:
	\bean
	c_{i_{a}} = \sum\limits_{b=1}^r \theta_{b} c_{j_b}, \ \text{$\theta_{b}$ depends on $(i_a,S_E)$ and $\theta_{b} \in \fq$ } .
	\eean
	Clearly an $(n,k,r,t)$ C-LR code is also an $(n,k,r,t)$ P-LR code, but the $r$ in the case of a cooperative LR code will tend to be significantly larger.  One may regard cooperative LR codes as codes that seek to minimize the number of unerased symbols contacted per erased symbol on average, rather than insist that each code symbol be repaired by contacting $r$ other code symbols.
	
	\section{Codes with Sequential Recovery} \label{sec:seq_recovery}
	
	\begin{defn}
		An $(n,k,r,t)$ sequential-recovery LR code (abbreviated as S-LR code) over a finite field \fq\ is an $[n,k]$ linear code $\calc$ over the finite field \fq\ having the following property: Given a collection of $s \le t$ erased code symbols, there is an ordering $(c_{i_1}, c_{i_2}, \cdots, c_{i_s})$ of these $s$ erased symbols such that for each index $i_j$, there exists a subset $S_j \subseteq [n]$ satisfying (i) $|S_j| \le r$ , (ii) $\ S_j \cap \{i_j, i_{j+1}, \cdots, i_s \} = \phi$,  and (iii) 
		\begin{eqnarray} \label{eq:locality}
			c_{i_j} & = & \sum \limits_{\ell \in S_j} u_{\ell} c_{\ell}, \ \text{$u_{\ell}$ depends on $(i_j,S_j)$ and $u_{\ell} \in \fq$ } .
		\end{eqnarray} 
		It follows from the definition that an $(n,k,r,t)$ S-LR code can recover from the erasure of $s$ code symbols $c_{i_1}, c_{i_2}, \cdots, c_{i_s}$, for $1 \leq s \leq t$ by using \eqref{eq:locality} to recover the symbols $c_{i_j}, \ j=1,2,\cdots,s$, in succession. 
	\end{defn}
	
	\subsection{An Overview of the Literature} 
	
	The sequential approach to recovery from erasures, introduced by Prakash et al. \cite{PraLalKum} is one of several approaches to local recovery from multiple erasures as discussed in Section \ref{sec:seq_classification}.   As indicated in Fig.~\ref{fig:me_classes}, Codes with Parallel Recovery and Availability Codes can be regarded as sub-classes of codes with Sequential Recovery.    Among the class of codes which contact at most $r$ other code symbols for recovery from each of the $t$ erasures, codes employing this approach (see \cite{PraLalKum,RawMazVis,SonYue,SonYue_Prdr,BalPraKum,BalPraKum4Era,SongCaiYue_L,BalKinKum_ISIT,BalKinKum_NCC}) have improved rate simply because sequential recovery imposes the least stringent constraint on the LR code.
	
	\paragraph{Two Erasures} Codes with sequential recovery (S-LR code) from two erasures ($t=2$) are considered in \cite{PraLalKum} (see also \cite{SonYue}) where a tight upper bound on the rate and a matching construction achieving the upper bound on rate is provided.  A lower bound on block length and a construction achieving the lower bound on block length is provided in \cite{SonYue}. 
	
	\paragraph{Three Erasures} Codes with sequential recovery from three erasures ($t=3$) can be found discussed in \cite{SonYue,SongCaiYue_L}.  A lower bound on block length as well as a construction achieving the lower bound on block length appears in \cite{SonYue}. 
	
	\paragraph{More Than $3$ Erasures}  A general construction of S-LR codes for any $r,t$ appears in \cite{SonYue_Prdr,SongCaiYue_L}.  Based on the tight upper bound on code rate presented in this chapter, it can be seen that the constructions provided in \cite{SonYue_Prdr,SongCaiYue_L} do not achieve the maximum possible rate of an S-LR code. In \cite{RawMazVis}, the authors provide a construction of S-LR codes for any $r,t$ with rate $\geq \frac{r-1}{r+1}$.  Again, the upper bound on rate presented here shows that $\frac{r-1}{r+1}$ is not the maximum possible rate of an S-LR code.   In the next chapter, we observe that the rate of the construction given in \cite{RawMazVis} is actually $\frac{r-1}{r+1}+\frac{1}{n}$ which equals the upper bound on rate derived here only for two cases: case (i) for $r=1$ and case (ii) for $r \geq 2$ and $t \in \{2,3,4,5,7,11\}$ exactly corresponding to those cases where a Moore graph of degree $r+1$ and girth $t+1$ exist.  In all other cases, the construction given in \cite{RawMazVis} does not achieve the maximum possible rate of an S-LR code.  
	
	In the subsections below, we examine in greater detail, the results in the literature pertaining to the cases $t=2,3$ and general $t$ respectively. 
	
	\subsection{Prior Work: $t=2$ Erasures}
	
	All the results presented below are for the case of $t=2$ erasures. 
	%
	The upper bound on code rate appearing below, can be found in \cite{PraLalKum}.	
	\begin{thm} \cite{PraLalKum} \label{thm:rate_bound} [Bound on Code Rate]
		Let $\mathcal{C}$ be an $(n,k,r,2)$ S-LR code over $\mathbb{F}_q$. Then 
		\begin{eqnarray} \label{eq:rate_bound}
			\frac{k}{n} & \leq & \frac{r}{r+2}.
		\end{eqnarray}
	\end{thm}
	S-LR codes with $t=2$ which achieve the bound \eqref{eq:rate_bound} are said to be {\em rate-optimal}.
	This upper bound on code rate can be rephrased as a lower bound on the block length of the code. 	
	\begin{thm} \cite{PraLalKum} \label{thm:block_length_bound_t2}
		The block length $n$ of an $(n,k,r,2)$ S-LR code $\cal{C}$ over $\mathbb{F}_q$ must satisfy:
		\begin{eqnarray}
			n  \geq  k + \left\lceil \frac{2k}{r}\right\rceil.  \label{eq:n_min_bound}
		\end{eqnarray}		
	\end{thm}
	An S-LR code with $t=2$ achieving the lower bound \eqref{eq:n_min_bound} on block length is said to be {\em block-length optimal}.  In \cite{SonYue}, the authors provide a construction of binary S-LR codes with $t=2$ achieving the lower bound on block length \eqref{eq:n_min_bound} for any $k,r$ such that $\lfloor \frac{k}{r} \rfloor \geq r$.  A rate-optimal construction based on Turan graphs of S-LR codes and appearing in \cite{PraLalKum} is presented below. 
	\begin{const}  \cite{PraLalKum} \label{sec:turan2} 
		Let $r$ be a positive integer. Let $n  =  \frac{(r+\beta)(r+2)}{2}$, with $1 \leq \beta \leq r$ and $\beta|r$. Consider a graph $\mathcal{G}_T$ with $b = \frac{2n}{r+2} = r + \beta$ vertices.  We partition the vertices into $x = \frac{r+\beta}{\beta}$
		partitions, each partition containing $\beta$ vertices.  Next, we place precisely one edge between any two vertices belonging to two distinct partitions. The resulting graph is known as a Tur\'an graph on $b$ vertices with $x$ vertex partitions. The number of edges in this graph is $\frac{x(x-1)\beta^2}{2} = n - b$ and each vertex is connected to exactly $(x-1)\beta = r$ other vertices.
		Let the vertices be labelled from $1$ to $b$ and the edges be labelled from $b + 1$ to $n$ in some random  order.  Now let each edge represent a distinct information bit and each node represent a parity bit storing the parity or binary sum of the information bits represented by edges incident on it. The code thus described corresponds to an $(n=\frac{(r+\beta)(r+2)}{2},k=n-b=\frac{r(r+\beta)}{2},r,2)$ S-LR code over $\mathbb{F}_2$ and hence a rate-optimal code for $t=2$. We will refer to these codes as Tur\'an-graph codes or as codes based on Tur\'an graphs.
	\end{const}
	\begin{example} [Tur\'an Graph] 
		An example Tur\'an graph with $x=3$, $\beta=3$, $b=9$ is shown in Figure \ref{fig:Turan_example}.
		
		\begin{figure}
			\begin{center}
				\includegraphics[width=3in]{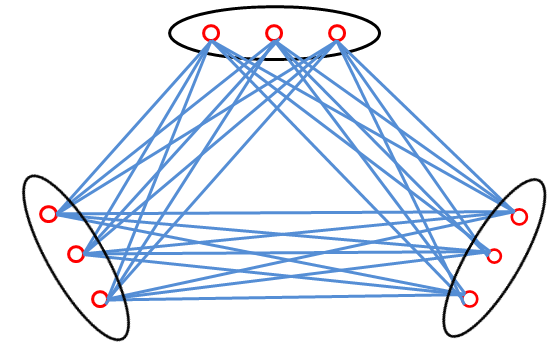} 
				\caption{A picture of the Tur\'an graph with $x=3$, $\beta=3$, $b=9$.}
				\label{fig:Turan_example}
			\end{center}
		\end{figure}
	\end{example}
	\begin{example}
		By definition of a Tur\'an Graph, it can be seen that any complete graph is also a Tur\'an Graph.  
	\end{example}
	\begin{note}
		With respect to the example above, with a complete graph as the example Tur\'an graph, it can be verified that the code obtained by applying Construction \ref{sec:turan2} to complete graph is also an instance of a code  with availability for $t=2$. This is because the sum of all the parity bits represented by nodes of the complete graph is 0. This parity check gives the second recovery set for all the parity bits represented by nodes of the complete graph making the code an availability code for $t=2$.  Since the rate of an availability code cannot exceed that of an S-LR code, it follows that the resultant code is rate optimal as an {\em availability code} for $t=2$.
	\end{note}
	\begin{note} [MSW Optimality of Tur\'an-graph codes]  
		It can be shown (see \cite{PraLalKum}) that the MSW sequence $\{e_i: 1 \leq i \leq b_1 \}$ given in Theorem~\ref{thm:msw_sequence} of Chapter \ref{ch:LRSingleErasure} with $b=b_1 = \lceil \frac{2n}{r+2} \rceil$ provides an upper bound on the MSWs (or equivalently the GHWs) of the dual of an $(n,k,r,2)$ S-LR code.   It was shown in \cite{PraLalKum} that Tur\'an-graph codes achieve this upper bound on MSW. Thus the dual code of an $(n,k,r,2)$ Tur\'an-graph code has the largest possible MSWs compared to the dual of any $(n,k,r,2)$ S-LR code.
	\end{note}
	
	\subsection{Prior Work: $t \geq 3$ Erasures}
	
	\paragraph{Results for $t=3$}  The lower bound on block length for S-LR codes for the case of $t=3$ given in \cite{SonYue} is presented below. 
	\begin{thm} \cite{SonYue} \label{thm:block_length_bound_t3}
		Let $\cal{C}$ denote an $(n,k,r,3)$ S-LR code over a finite field $\mathbb{F}_q$. Then the block length $n$ must satisfy:
		\bea
		n \geq k+\left\lceil\frac{2k+\lceil\frac{k}{r}\rceil}{r}\right\rceil. \label{eq:min_blk_length_t3}
		\eea		
	\end{thm}
	The constructions of S-LR codes for $t=3$ given in \cite{SonYue} include the product code (product of two $[r+1,r]$ single parity check codes) which achieves the maximum possible rate based on the bound \eqref{eq:min_blk_length_t3}.  However, the product code cannot provide codes for all values of $k,r$. To address this, additional constructions were provided in \cite{SonYue} for S-LR codes with $t=3$ achieving the lower bound on block length \eqref{eq:min_blk_length_t3} for almost all $k,r$ such that $\lceil \frac{k}{r} \rceil \geq r$ (for the precise conditions on $k,r$ please see \cite{SonYue}).
	
	\paragraph{Results for General $t$}	
	The following conjecture on the maximum achievable rate of an $(n,k,r,t)$ S-LR code appeared in \cite{SongCaiYue_L}. 
	\begin{conj} \cite{SongCaiYue_L} [Conjecture] \label{Conjecture}
		Let $\cal{C}$ denote an $(n,k,r,t)$ S-LR code over a finite field $\mathbb{F}_q$. Let $ m=\lceil \log_r(k) \rceil$. Then:
		\bean
		\frac{k}{n} \leq \frac{1}{1+\sum_{i=1}^{m}\frac{a_i}{r^i}}, \text{ where } a_i \geq 0, \ \ a_i \in \mathbb{Z}, \ \ \sum_{i=1}^{m} a_i = t.
		\eean
	\end{conj}
	As will be seen, the tight upper bound on rate of an $(n,k,r,t)$ S-LR code derived in Section~\ref{sec:t_erasures} of the present chapter proves the above conjecture by identifying the precise value of the coefficients $a_i$ appearing in the conjecture. Chapter~\ref{ch:SeqConstr} provides constructions of S-LR codes achieving the upper bound on rate derived in Section~\ref{sec:t_erasures} for any $r,t$ with $r \geq 3$.
	
	\section{Contributions of the Thesis to S-LR codes} \label{sec:seq_contributions} 
	
	This chapter contains one of the principal contributions of this thesis, namely the derivation of a tight upper bound (Theorem~\ref{rate_both}) on the rate of an $(n,k,r,t)$ S-LR code.  Proof that this bound is tight (achievable) follows from the matching constructions of S-LR codes provided in the following chapter, Chapter~\ref{ch:SeqConstr}.  There are several other results as well and these are summarized in Fig.~\ref{fig:flowchart_ch3}.  
	
	These include:  
	\ben
	\item Case of S-LR codes with $t=2$:
	\ben
	\item A {\bf rate-optimal construction} of S-LR codes achieving the upper bound on rate given in Theorem \ref{thm:rate_bound}. 
	\item A {\bf block-length-optimal construction} of S-LR codes which achieves the lower bound on block length given in Theorem \ref{thm:block_length_bound_t2}.  This construction extends the range of parameters ($k,r$) of block length optimal constructions compared to block length optimal constructions given in \cite{PraLalKum,SonYue}.
	\item A {\bf characterization} of rate-optimal S-LR codes for $t=2$.
	\item  {\bf An upper bound on the dimension} of an S-LR code over $\mathbb{F}_2$ for the case when one is given the dimension $m$ of the subspace of the dual code spanned by codewords of weight $\leq r+1$.  It may not always be possible to achieve the lower bound on block length given in Theorem \ref{thm:block_length_bound_t2}.  In such situations, this upper bound on code dimension can prove useful.  We also provide here {\bf a construction achieving this upper bound on code dimension} for a family of values of $m$ and locality parameter $r$. 
	\een
	\item Case of S-LR codes with $t=3$:
	\ben
	\item A {\bf lower bound on block length} of an S-LR code over $\mathbb{F}_2$ for a given $k,r$ for $k \leq r^{1.8}-1$.This lower bound is shown to be tighter than the previously known lower bound on block length given in Theorem \ref{thm:block_length_bound_t3}. This is followed by pointing out that the {\bf construction of a short block length S-LR code } given in \cite{BalPraKum} that generalizes an instance of the Tur\'an graph based construction \ref{sec:turan2}  has block length { \bf very close to our lower bound}. We give two specific examples of S-LR codes which achieve our lower bound on block length but does not achieve the lower bound on block length given in Theorem \ref{thm:block_length_bound_t3}.
	\een
	\item Case of S-LR codes for general $t$:
	\ben
	\item An {\bf upper bound on rate} of an S-LR code for any $r,t$ such that $r \geq 3$. This upper bound on rate is achievable and { \bf constructions achieving} it appear in Chapter \ref{ch:SeqConstr}. This upper bound on rate also proves the conjecture \ref{Conjecture} given in \cite{SongCaiYue_L}.
	\een
	\een
	
	\begin{figure}[h!]
		\begin{center}
			\includegraphics[width=5.5in]{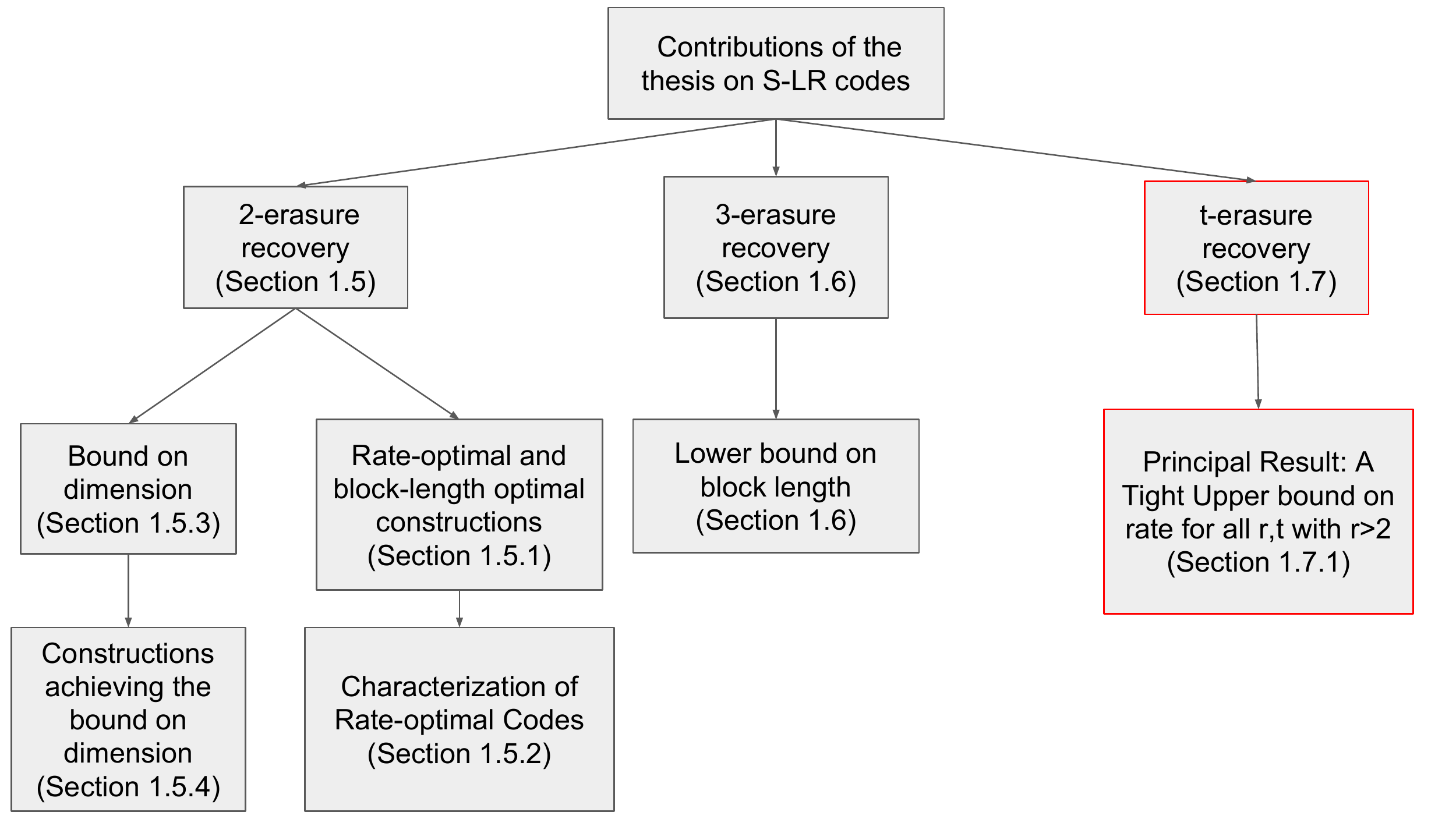} 
			\caption{The principal contributions of the thesis are on the topic of S-LR codes and the contributions here are summarized in the flowchart.   A major result is the tight upper bound on rate of an S-LR code for general $t$, that settles a conjecture.  Constructions achieving the upper bound on rate appear in the chapter following.  }
			\label{fig:flowchart_ch3}
		\end{center}
	\end{figure}
	
	\section{Contributions to S-LR codes with $t=2$} \label{sec:2_erasures}
	
	We have already described the contributions of this thesis to S-LR codes for $t=2$. We give the detailed results in the following.
	
	\subsection{Rate and Block-Length-Optimal Constructions}
	
	We begin by describing a generic, graph-based construction of an S-LR code for $t=2$. Special cases of this construction will yield rate and block-length-optimal constructions. 
	
	\begin{const}   [Generic Graph-based Construction]  \label{const:graph_based}
		Let $\mathcal{G}=(V,E)$ be a graph having vertex set $V$ and edge set $E$. Let each edge in the graph represent an information symbol in the code and each node or vertex, represent a parity symbol which is equal to some linear combination of the information symbols represented by the edges incident on that node, where the linear combination is taken over a finite field $\mathbb{F}_q$ with each coefficient in the linear combination being non-zero.  Each codeword in the resultant systematic code \calc\ is comprised of the information symbols represented by the edges in $E$ and the parity symbols represented by the nodes in $V$.  The dimension of this code is clearly equal to $|E|$ and the block length $n=|V|+|E|$. Since the parity symbol represented by a node $v \in V$ is a non-zero linear combination of precisely $degree(v)$ information symbols, the corresponding parity check involves exactly $\text{degree}(v)+1$ code symbols. Thus the code has locality $\Delta = \max_{v \in V} \text{degree}(v)$. It is straightforward to see that the code can sequentially recover from $2$ erasures (S-LR code with $t=2$).  Thus the code given by this construction is an $(n=|E|+|V|,k=|E|,r=\Delta,t=2)$ S-LR code over $\mathbb{F}_q$.
	\end{const}
	
	We next describe a construction that is a specialization of Construction~\ref{const:graph_based}, in which all nodes with possibly a single exception, have the same degree. This construction yields block-length-optimal codes.
	
	\begin{const}   [Near-Regular Graph Construction]  \label{const:near_reg_graph}
		Let $\{r,k\}$ be positive integers. Let $2k=ar+b , 0\leq b \leq r-1$. Let $\mathcal{G}_1$ be a graph on a set of $m= \lceil \frac{2k}{r} \rceil$ nodes with `$a$' nodes having degree $r$ and for the case when $b>0$, the remaining node having degree $b$. The condition for existence of such a near-regular graph $\mathcal{G}_1$ is given in Remark~\ref{note:near_regular} below. Apply the generic graph-based-construction described in Construction \ref{const:graph_based} to the graph $\mathcal{G}_1$ by replacing $\mathcal{G}$ in Construction \ref{const:graph_based} with $\mathcal{G}_1$ and setting $q=2$.  It follows that the code thus obtained by applying the Construction \ref{const:graph_based} to the graph $\mathcal{G}_1$ with $q=2$ is an $(n=|E|+|V|=k+\lceil \frac{2k}{r} \rceil,|E|=k,r,2)$ S-LR code over $\mathbb{F}_2$. 
	\end{const}
	
	{\bf Block length optimality:}
	Construction~\ref{const:near_reg_graph} yields S-LR codes with $t=2$ which achieve the lower bound on block length given in \eqref{eq:n_min_bound}. Hence Construction~\ref{const:near_reg_graph} yields block-length-optimal codes.
	
	Construction \ref{const:near_reg_graph} yields rate-optimal codes for $t=2$ when $b=0$.  This is described below. 
	\begin{const}  [Regular-Graph Construction] \label{const:reg_graph} 
		In Construction \ref{const:near_reg_graph}, when $b=0$, i.e., when the graph  $\mathcal{G}_1$ is a regular graph, the resultant code is a binary S-LR code with $t=2$ that is rate-optimal. we will refer these codes as \emph{Regular-Graph Codes} or as the \emph{Codes based on Regular-Graphs}.
	\end{const} 
	
	{\bf Rate optimality:}
	Construction~\ref{const:reg_graph}  yields S-LR codes with $t=2$ which achieve the upper bound on rate given in \eqref{eq:rate_bound}. Hence  Construction~\ref{const:reg_graph} provides rate-optimal codes.
	\begin{proof}
		Follows since the code has rate $\frac{r}{r+2}$. 
	\end{proof}
	
	\begin{note} [Existence Conditions for Near-Regular Graphs] \label{note:near_regular} 
		The parameter sets $(k,r)$ for which near-regular graphs $\mathcal{G}_1$ of the form described in Constructions~\ref{const:near_reg_graph} and \ref{const:reg_graph} exist can be determined from the Erd\"os-Gallai theorem \cite{ErdosGallai}.  The conditions are 
		\bea
		\begin{array}{rl} \frac{2k}{r} = m\geq r+1, & b=0,  \\
			\lceil \frac{2k}{r} \rceil = m\geq r+2, & b>0.  \end{array} \label{OptimalConstRange}
		\eea
		As noted earlier, in \cite{SonYue}, the authors provide constructions of binary S-LR codes for $t=2$ achieving the lower bound on block length \eqref{eq:n_min_bound} for any $k,r$ such that $\lfloor \frac{k}{r} \rfloor \geq r$.
		Thus  the Constructions~\ref{const:near_reg_graph} extends the range of $(k,r)$ for which the lower bound on block length \eqref{eq:n_min_bound} is achieved since  Constructions~\ref{const:near_reg_graph} gives binary S-LR codes for $t=2$ achieving the lower bound on block length for any $(k,r)$ such that $ \lceil \frac{2k}{r} \rceil \geq r+2$. For the case $b=0$, Construction~\ref{const:reg_graph} extends the range of constructions of rate-optimal codes for $t=2$ compared to construction \ref{sec:turan2} as construction \ref{sec:turan2} requires that $\frac{2k}{r} = r+\beta$ with $\beta | r$ whereas construction \ref{const:reg_graph} needs only that $\frac{2k}{r} \geq r+1$.
	\end{note}
	
	\begin{example}
		An example $(18,12,4,2)$ S-LR code based on a regular graph and which is hence rate optimal, is shown in Fig.~\ref{fig:regular_example}.  In the example, a codeword takes on the form 
		\bean
		[I_1, \ \cdots, \ I_{12}, \ P_1=I_1+I_6+I_8+I_{12}, \ \cdots \ , \ P_6=I_5+I_6+I_7+I_{11} \ ]. 
		\eean
		
		\begin{figure}[h!]
			\begin{center}
				\includegraphics[width=3in]{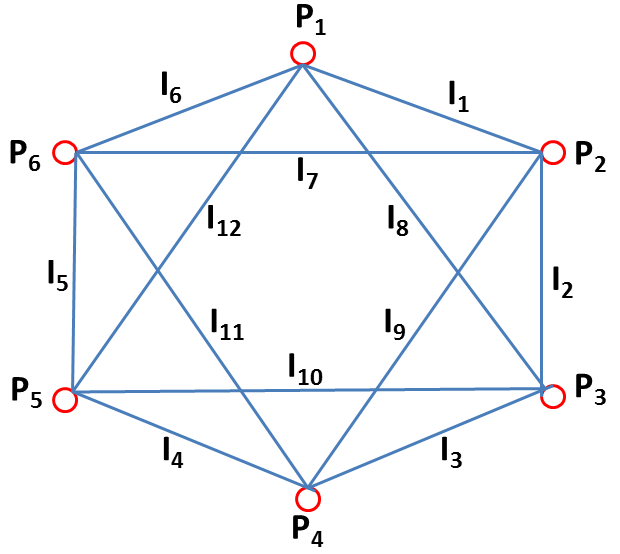} 
				\caption{Example of an $(n=18,k=12,r=4,t=2)$ S-LR code based on a regular. The edges $\{I_1, \cdots , I_{12}\}$ represent information symbols and the nodes $\{P_1, \cdots, P_6\}$ represent parity symbols.}
				\label{fig:regular_example}
			\end{center}
		\end{figure}
	\end{example}
	
	\subsection{Characterizing Rate-Optimal S-LR Codes for $t=2$}
	
	We characterize below rate-optimal S-LR codes for $t=2$.  We begin with Theorem \ref{thm:pc_mx_form} below, which describes a general form of the parity check matrix of a rate-optimal code for $t=2$.  This theorem will be used in our characterization.
	
	\begin{thm} \label{thm:pc_mx_form}  
		Let $\cal{C}$ denote an $(n,k,r,2)$ S-LR code over a finite field $\mathbb{F}_q$. Let $\cal{C}$ be rate-optimal for $t=2$. Let	\begin{eqnarray} \label{eq:C0_local}
			\mathcal{B}_{0} & = & \text{span}\left(\{{ \underline{c}} \in \mathcal{C}^{\perp} : |\text{supp}({\underline{c}})| \leq r+1 \} \right).
		\end{eqnarray} Then 
		\ben
		\item It must be that $r \mid (2k)$, $(n-k) = \frac{2k}{r}$, $\cal{C}^{\perp}=$$\mathcal{B}_{0}$ and 
		\item  Let ${\underline{c}_1,...,\underline{c}_{n-k}}$ be a basis of $\mathcal{B}_0$ such that $|\text{supp}({\underline{c}_i})| \leq r+1$, $\forall i \in [n-k]$. The parity-check matrix $H$ of $\cal{C}$ with $\underline{c}_1,...,\underline{c}_{n-k}$ as its rows, is upto a permutation of columns and scaling of rows, of the form:
		\bea
		[I_{(n-k)}|H'], \label{eq:H_form_rate_optimal}
		\eea
		where $I_{(n-k)}$ denotes the $((n-k) \times (n-k))$ identity matrix  and $H'$ is an $((n-k)\times k)$ matrix with each column having Hamming weight $2$ and each row having Hamming weight $r$. 
		\item There exists a parity-check matrix $H$ (upto a permutation of columns) of $\cal{C}$ of the form given in \eqref{eq:H_form_rate_optimal}.
		\een
	\end{thm}
	The proof of Theorem~\ref{thm:pc_mx_form} above follows as a consequence of the proof of the Theorem \ref{thm:rate_bound} given in \cite{PraLalKum}.  This is because, as can be seen from the proof of the Theorem \ref{thm:rate_bound} given in \cite{PraLalKum}, that for achieving the upper bound on rate given in Theorem \ref{thm:rate_bound}, one needs a parity check matrix where each row has Hamming weight $=r+1$, where there are $n-k$ columns of Hamming weight one and the remaining $k$ columns having weight two.  In the following, we exploit this result to characterize rate-optimal codes for $t=2$.
	
	\begin{thm} \label{thm:structure_2erasure}
		Let $\cal{C}$ denote an $(n,k,r,2)$ S-LR code over a finite field $\mathbb{F}_q$. Let $\cal{C}$ be rate-optimal for $t=2$. Then $\cal{C}$ (after possible permutation of code symbols or co-ordinates) must have one of the forms below:
		
		\ben
		\item[(a)]
		\bean		
		{\cal C}  =  \ \ \underbrace{{\cal C}_1 \times \cdots \times {\cal C}_{\ell}}_{\text{\normalsize direct product of $\ell$ $[r+2,r,3]$ MDS Codes }} 
		\eean
		\item[(b)]
		\bean
		{\cal C} & = & \ \ \underbrace{{\cal C}}_{\text{\normalsize a single regular-graph code over $\mathbb{F}_q$}} 
		\eean
		\item[(c)]
		\bean
		{\cal C} & = & \ \ \underbrace{{\cal C}_1 \times \cdots \times {\cal C}_{\ell}}_{\text{\normalsize direct product of $\ell$ $[r+2,r,3]$ MDS Codes}} \ \ \times \ \ \underbrace{{\cal C}_{\ell+1}}_{\text{\normalsize and a regular-graph code over $\mathbb{F}_q$}}. \\
		\eean
		\een
		
		where the phrase `regular-graph code over $\mathbb{F}_q$' refers to the code obtained by applying the generic graph based construction over $\mathbb{F}_q$ described in Construction~\ref{const:graph_based} to a regular graph $\mathcal{G}_R$ of degree $r$ by replacing $\mathcal{G}$ in Construction~\ref{const:graph_based} with $\mathcal{G}_R$.
	\end{thm}
	
	\begin{proof}  	
		For proof please refer to the Appendix \ref{app:structure_2erasure}.
	\end{proof}
	
	\subsection{Upper Bound on Dimension of an S-LR Code for $t=2$}
	
	Let $\mathcal{C}$ be an $(n,k,r,2)$ S-LR code. Let $\mathcal{B}_{0} = \text{span}\left(\{{ \underline{c}} \in \mathcal{C}^{\perp} : |\text{supp}({\underline{c}})| \leq r+1 \} \right)$. 
	It can be seen that the lower bound on block length \eqref{eq:n_min_bound} is achieved iff $n-k =\text{dim}(\mathcal{B}_{0}) =  \lceil \frac{2k}{r} \rceil$. This implies $\frac{\text{dim}(\mathcal{B}_{0}) r}{2} \geq k >  \frac{(\text{dim}(\mathcal{B}_{0})-1) r}{2}$. Hence if $\frac{(\text{dim}(\mathcal{B}_{0})-1) r}{2} \geq k$ the lower bound \eqref{eq:n_min_bound} is not achievable.
	
	We will now see that for some values of $\text{dim}(\mathcal{B}_{0})$, $\frac{(\text{dim}(\mathcal{B}_{0})-1) r}{2} \geq k$ for binary codes.
	We derive an upper bound on dimension  $k$ of $\mathcal{C}$ for a given $\text{dim}(\mathcal{B}_{0})$ that is tighter than \eqref{eq:n_min_bound} for small values of $\text{dim}(\mathcal{B}_{0})$.
	
	The idea is simply to place as many distinct vectors as possible in the columns of the parity check matrix with $\text{dim}(\mathcal{B}_{0})$ rows under the constraint that the total number of ones in all the columns put together is $\leq \text{dim}(\mathcal{B}_{0}) (r+1)$.

	\begin{thm} \label{thm:DimUpperBound_ch3}
		Let $\mathcal{C}$ be an $(n,k,r,2)$ S-LR code over $\mathbb{F}_2$. Let
		\begin{eqnarray} 
			\mathcal{B}_{0} & = & \text{span}\left(\{ \underline{c} \in \mathcal{C}^{\perp} : |\text{supp}(\underline{c})| \leq r+1 \} \right), \notag
		\end{eqnarray}
		and let $m=\text{dim}(\mathcal{B}_{0})$. Then
		\bea
		k \leq \min_{1 \leq L \leq m} \frac{m(r-L)+ \sum_{i=1}^{L} (L+1-i) {m \choose i}}{L+1}. \label{eq:blocklength_4}
		\eea
	\end{thm}
	\begin{proof}
		For proof please refer to the Appendix \ref{app:upperbounddim}.
	\end{proof}
	
	{\em Tightness of the above bound :}
	Let $1\leq L^* \leq m$. Let $r=\sum_{i=2}^{L^*} {m-1 \choose i-1}+J$ for $0 \leq J \leq {m-1 \choose L^*}-1$: substituting the value of $r$ given above in the RHS (Right Hand Side) of the bound \eqref{eq:blocklength_4} after removing the minimization in the RHS of  \eqref{eq:blocklength_4} over $L$ and setting $L=L^*$:
	\bea
	k  & \leq & \frac{m(r-L^*)+ \sum_{i=1}^{L^*} (L^*+1-i) {m \choose i}}{L^*+1}, \notag \\
	k & \leq & \frac{m(\sum_{i=2}^{L^*} {m-1 \choose i-1}+J)}{L^*+1}-\frac{ mL^*}{L^*+1}+ \frac{\sum_{i=1}^{L^*} (L^*+1-i) {m \choose i}}{L^*+1}, \notag \\
	k & \leq & \frac{(\sum_{i=2}^{L^*}i {m \choose i})+mJ}{L^*+1}-\frac{ mL^*}{L^*+1}+ \frac{\sum_{i=1}^{L^*} (L^*+1-i) {m \choose i}}{L^*+1}, \notag \\
	k & \leq & \sum_{i=2}^{L^*} {m \choose i} +\frac{mJ}{L^*+1}. \label{eq:block_length_5}
	\eea
	Using similar algebraic manipulations, it can be shown that in bound \eqref{eq:blocklength_4} the minimum of the RHS for $r=\sum_{i=2}^{L^*} {m-1 \choose i-1}+J$ for $0 \leq J \leq {m-1 \choose L^*}-1$ is attained when $L=L^*$.  We skip this proof.
	By simple algebraic mainpulations it can be shown that \eqref{eq:block_length_5}, implies $k\leq \frac{(m-1)r}{2}$, for $r=\sum_{i=2}^{L^*} {m-1 \choose i-1}+J$ for $0 \leq J \leq {m-1 \choose L^*}-1$, for large values of $L^*$, for a fixed $m \geq 5$ and hence our bound is tighter than the bound \eqref{eq:n_min_bound}. We now show that the upper bound \eqref{eq:block_length_5} is actually achievable.
	\subsection{Dimension-Optimal Constructions for Given $m,r$:}
	We saw from \eqref{eq:block_length_5} that when  $r=\sum_{i=2}^{L^*} {m-1 \choose i-1}+J$ for $0 \leq J \leq {m-1 \choose L^*}-1$:
	\bea
	k & \leq & \sum_{i=2}^{L^*} {m \choose i} +\frac{mJ}{L^*+1}. \label{eq:block_length_6}
	\eea
	In this section, we give a construction for achieving the bound \eqref{eq:block_length_6} on $k$
	when gcd($L^*+1$,$m$)=$1$ and $(L^*+1) | J$. The construction is carried out in a simple manner as described below: 
	\begin{const}     
		Let $1\leq L^* \leq m$. Let  $r=\sum_{i=2}^{L^*} {m-1 \choose i-1}+J$, for $0 \leq J \leq {m-1 \choose L^*}-1$. Let gcd($L^*+1$,$m$)=$1$ and $(L^*+1) | J$. Let $ \ell = \frac{{m \choose L^*+1}}{m}$. For $2 \leq i \leq L^*$,let $A_i$ be an $m \times {m \choose i}$ matrix such that the set of columns of $A_i$ is equal to the set of all possible distinct $(m \times 1)$ binary vectors of Hamming weight $i$.  Let $S$ denote the set of all possible distinct $(m \times 1)$ binary vectors of Hamming weight $L^*+1$.  Next, define a relation $\sim$ in set $S$ as follows. For $a,b \in S$, $a \sim b$, if $b$ can be obtained by cyclically shifting the co-ordinates of $a$. It can be seen that this relation is an equivalence relation. Let $E_1,...E_{\ell}$ be the equivalence classes, each containing exactly $m$ vectors as gcd($L^*+1$,$m$)=$1$.
		
		Then the desired code is the code having parity-check matrix:
		\bean
		H=\left[\begin{array}{c|c|c|c|c|c|c|c|c}
			I_m & A_2 & A_3 &...& A_{L^*} & E'_{i_1} & E'_{i_2} & ... & E'_{i_{\frac{J}{L^*+1}}}
		\end{array} \right].
		\eean
	\end{const}
	where $E'_i$ is the $m \times m$ matrix with the set of columns of $E'_i$ equal to the set of $(m \times 1)$ vectors in $E_i$ and $S=\{i_1,...,i_{\frac{J}{L^*+1}}\} \subseteq \left[ \ell \right]$ with $|S|=\frac{J}{L^*+1}$.
	Note that the weight of each row of $H$ is exactly $r+1$ and since all the columns of $H$ are distinct, the code is an S-LR code with $t=2$. Note that the code with parity-check matrix $H$ defined as above achieves the bound \eqref{eq:block_length_6}. Hence the code is a $(n=m+\sum_{i=2}^{L^*} {m \choose i} +\frac{mJ}{L^*+1},k=\sum_{i=2}^{L^*} {m \choose i} +\frac{mJ}{L^*+1},r=\sum_{i=2}^{L^*} {m-1 \choose i-1}+J,2)$ S-LR code over $\mathbb{F}_2$ and has maximum possible dimension for the given $m,r$.
	
	\begin{const}     
		Let $1\leq L^* \leq m$. Let  $r=\sum_{i=2}^{L^*} {m-1 \choose i-1}+J$, for $0 \leq J \leq {m-1 \choose L^*}-1$. Let gcd($L^*+1$,$m$)=$1$ and $(L^*+1) | J$. For $2 \leq i \leq L^*$, let $A_i$ be an $m \times {m \choose i}$ matrix such that the set of columns of $A_i$ is equal to the set of all possible distinct $(m \times 1)$ binary vectors of Hamming weight $i$. 
		Let $D$ be the node-edge incidence matrix of a $J$-regular $L^*+1$- uniform simple hypergraph with exactly $m$ nodes.
		Now our final desired code is defined as the code with parity-check matrix:
		\bean
		H=\left[\begin{array}{c|c|c|c|c|c}
			I_m & A_2 & A_3 &...& A_{L^*} & D
		\end{array} \right].
		\eean
	\end{const}
	Note that the weight of each row of $H$ is exactly $r+1$ and since all the columns of $H$ are distinct, the code is an S-LR code for $t=2$. Note that the code with parity-check matrix $H$ defined as above achieves the bound \eqref{eq:block_length_6}. Hence the code is a $(n=m+\sum_{i=2}^{L^*} {m \choose i} +\frac{mJ}{L^*+1},k=\sum_{i=2}^{L^*} {m \choose i} +\frac{mJ}{L^*+1},r=\sum_{i=2}^{L^*} {m-1 \choose i-1}+J,2)$ S-LR code over $\mathbb{F}_2$ and has maximum possible dimension for the given $m,r$.
	
	\section{Contributions to S-LR codes with $t=3$} \label{sec:3_erasures}
	
	In this section, we adopt a slightly different perspective and we consider lower bounding the block length of S-LR codes with given $k,r$ and $t=3$ for binary codes. We first present a lower bound on the block length of binary S-LR codes with given $k,r$ and $t=3$ for $k \leq r^{1.8}-1$. This lower bound is shown to be tighter than the previously known lower bound on block length given in Theorem \ref{thm:block_length_bound_t3}. This is followed by pointing out that the construction of a short block length S-LR code with $t=3$ given in \cite{BalPraKum} that generalizes an instance of the Tur\'an graph based construction \ref{sec:turan2} has block length very close to our lower bound.
	
	\subsubsection{Lower Bound on Block Length} \label{sec:t_eq_3_bound} 
	
	As noted earlier in Theorem \ref{thm:block_length_bound_t3} (due to \cite{SonYue}):
	\bea
	\scalemath{1}{	n\geq k+\left\lceil\frac{2k+\lceil\frac{k}{r}\rceil}{r}\right\rceil }. \label{eq:song_3}
	\eea
	Constructions of S-LR codes with $t=3$ achieving the above bound \eqref{eq:song_3} were provided for almost all $k,r$ such that $\lceil \frac{k}{r} \rceil \geq r$ (almost the regime ($k \geq r^2$). The maximum possible rate for $t=3$ is achieved by the product code (product of 2 $[r+1,r,2]$ single parity check codes). Here again, we present a new lower bound on block length of binary S-LR codes with $t=3$. Simulation shows that this new lower bound is tighter than equation \eqref{eq:song_3} for the regime $k \leq r^{1.8}-1$.  We also provide a few sporadic examples where our new lower bound turns out to be tight.   
	\begin{lem} \label{lem:min_len3}
		Let $\mathcal{C}$ denote a binary code whose minimum distance $d_\text{min} \geq 4$. Let us assume in addition, that $\mathcal{C}$ possesses an $(M \times N)$ parity-check matrix $H$ over $\mathbb{F}_2$ such that each column of $H$ has Hamming weight equal to $2$. The matrix $H$ need not be full rank. Then, $N$ satisfies: 
		\begin{equation}
			N \leq \frac{(M+3)(M+1)}{4}+1.
		\end{equation}
	\end{lem}
	\begin{proof}
		For proof please refer to the Appendix \ref{app:lem_min_len3}
	\end{proof}
	
	\begin{thm} \label{thm:min_len3}
		Let $\mathcal{C}$ denote an $(n,k,r,3)$ S-LR code over $\mathbb{F}_2$. The block length $n$ of the code $\mathcal{C}$ satisfies the lower bound:
		\bea
		n \geq k+\min_{s_1 \in \mathbb{N} \cup \{0\}} \max\{f_1(s_1),f_2(s_1),s_1\} ,\label{eq:min_len}
		\eea
		\bean 
		\text{where} \  f_1(s_1) & = &  \left\lceil \frac{-(2r-5)+\sqrt{(2r-5)^2 +4(6k+s_1^2-5s_1)}}{2} \right\rceil, \\
		f_2(s_1) & = & \left\lceil \frac{-(4r-4+2s_1)+\sqrt{(4r-4+2s_1)^2 +4(12k+3s_1^2-4s_1-7)}}{2} \right\rceil,
		\eean
		and $\mathbb{N}$ denotes the set of natural numbers.
	\end{thm}
	\begin{proof}
		For proof please refer to the Appendix \ref{app:min_len3}.
	\end{proof}
	
	Fig. \ref{fig:min_len_3_comparison} shows a comparison between the two lower bounds \eqref{eq:song_3} and \eqref{eq:min_len} on the block-length $n$ of a binary $(n,k,r,3)$ S-LR code for the case when $k=20$. Simulation results show that the new lower bound in \eqref{eq:min_len} is tighter than the previously known lower bound in \eqref{eq:song_3} for $r \leq k \leq r^{1.8}-1$, $2 \leq r \leq 200$.
	\begin{figure}[h!]
		\centering
		\includegraphics[width=4.5in]{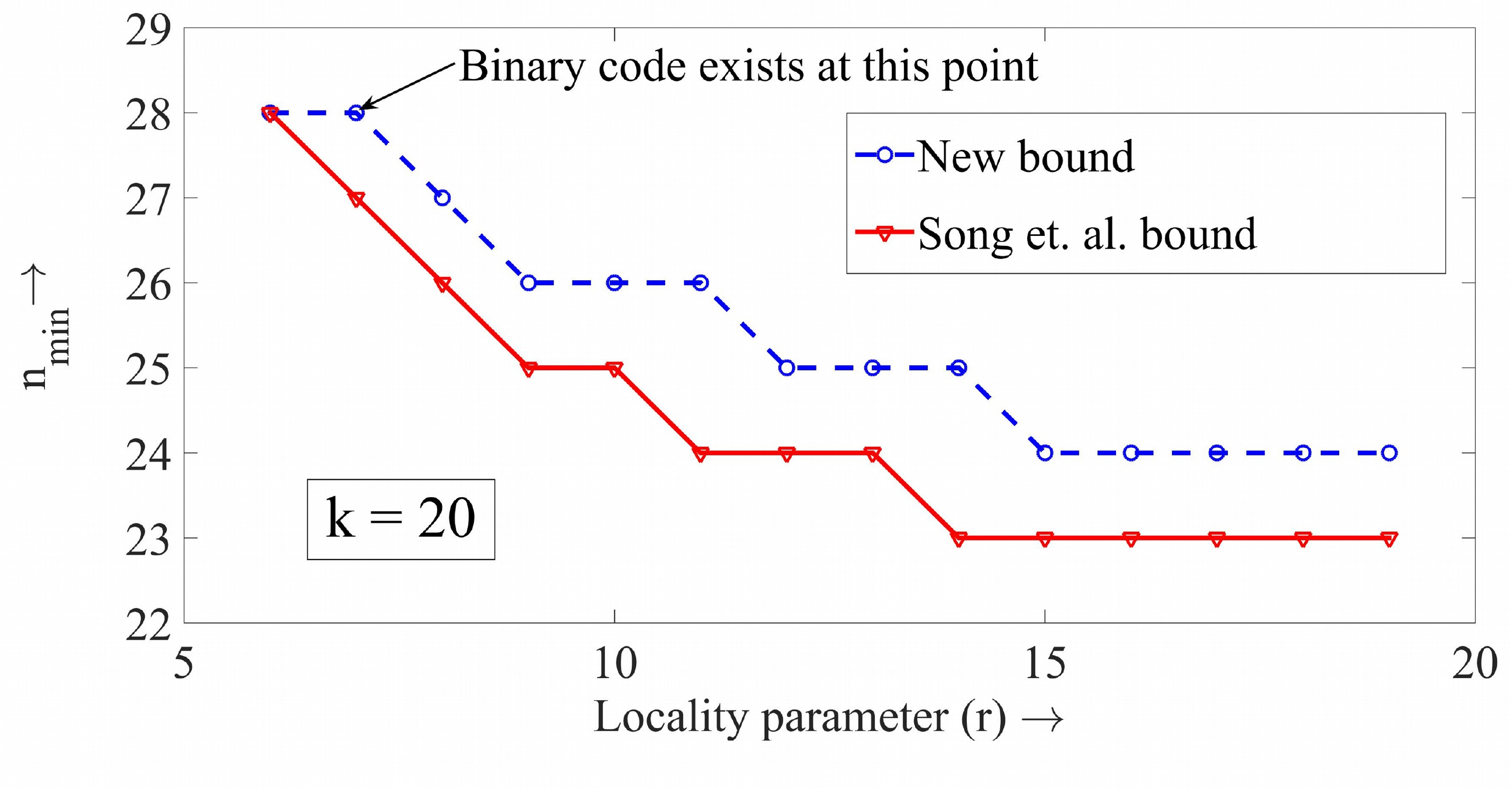}
		\caption{Comparing the new lower bound on block-length $n$ of a binary $(n,k,r,3)$ S-LR code given in \eqref{eq:min_len} with the bound in \eqref{eq:song_3} (by Song et. al. \cite{SonYue}) for $k=20$. Here, $\text{n}_\text{min}$ denotes the lower bound on block length obtained from the respective bounds.}
		\label{fig:min_len_3_comparison}
	\end{figure}
	
	Two example S-LR codes with $t=3$ which achieve the lower bound in \eqref{eq:min_len} are provided below.   The construction in Examples \ref{eg:3_erasure_1} is based on the proof of Theorem~\ref{thm:min_len3} whereas the construction in Example \ref{eg:3_erasure_2}, is based on the hyper graph-based construction which appeared in \cite{BalPraKum}. In \cite{BalPraKum}, authors showed that the construction of S-LR codes with $t=3$ based on hyper graph given in \cite{BalPraKum}, has block length differing from our lower bound on block length \eqref{eq:min_len} by atmost 2 for an infinite set of parameters $k,r$. This shows that the lower bound presented here is tight in a significant number of cases.    Table \ref{table:3_erasure_eg} compares the block-lengths of the codes in the two examples with the lower bounds on block-length appearing in \eqref{eq:song_3} (due to \cite{SonYue} Song et al) and \eqref{eq:min_len} (the new lower bound presented here).
	\begin{table}[h!] 
		\centering
		
		\begin{tabular}{||c|c|c|c|c|c|| }
			\hline
			$ $ & $k$ & $r$ & Bound~\eqref{eq:song_3} (Song et. al. \cite{SonYue}) & Bound~\eqref{eq:min_len} (new bound) & $n$\\
			\hline \hline
			Example~\ref{eg:3_erasure_1} & 5 & 3 & 9 & 10 & 10\\
			\hline
			Example~\ref{eg:3_erasure_2} & 8 & 4 & 13 & 14 & 14\\
			\hline
		\end{tabular}
		\caption{Comparing the block-length $n$ of the codes in examples \ref{eg:3_erasure_1}, \ref{eg:3_erasure_2}  against the lower bound on block-length obtained from \eqref{eq:song_3} and \eqref{eq:min_len}.}
		\label{table:3_erasure_eg}
	\end{table}
	
	\begin{example} \label{eg:3_erasure_1}
		$n=10,k=5,r=3,t=3$ : \eqref{eq:song_3} gives $n \geq 9$ for $k=5,r=3$.  The lower bound \eqref{eq:min_len} derived here gives $n \geq 10$ for $k=5,r=3$. It can be seen that the binary code associated with the parity-check matrix given below in \eqref{eq:H1_ch3} is a binary S-LR code with $n=10,k=5,r=3,t=3$ and based on \eqref{eq:min_len} has the least possible block length $n$ for $k=5,r=3,t=3$.
		\bea
		H & = & 
		\left[ \begin{array}{cccccccccc} 
			1 & 0 & 0 & 0 & 1 & 0 & 0 & 0 & 1 & 1 \\
			0 & 1 & 0 & 0 & 0 & 1 & 0 & 0 & 1 & 0 \\
			0 & 0 & 1 & 0 & 0 & 0 & 1 & 0 & 1 & 1 \\
			0 & 0 & 0 & 1 & 0 & 0 & 0 & 1 & 0 & 1\\
			0 & 0 & 0 & 0 & 1 & 1 & 1 & 1 & 0 & 0\\
		\end{array} \right]. \label{eq:H1_ch3}
		\eea
	\end{example}
	
	\begin{example} \label{eg:3_erasure_2}
		$n=14,k=8,r=4,t=3$: the lower bound \eqref{eq:min_len} derived here in gives $n \geq 14$ whereas \eqref{eq:song_3} gives $n \geq 13$. It can be seen that the binary code associated with the parity-check matrix given below in \eqref{eq:n_eq_14_code} is a binary S-LR code with $n=14,k=8,r=4,t=3$ and based on \eqref{eq:min_len} has the least possible block length $n$ for $k=8,r=4,t=3$. 
		\bea
		\scalemath{1}{
			H  =  \left[  \begin{array}{cccccccccccccc}
				1 & 0 & 0 & 0 & 0 & 0 & 1 & 1 & 1 & 1 & 0 & 0 & 0 & 0 \\ 
				0 & 1 & 0 & 0 & 0 & 0 & 0 & 0 & 0 & 0 & 1 & 1 & 1 & 1 \\ 
				0 & 0 & 1 & 0 & 0 & 0 & 1 & 1 & 0 & 0 & 1 & 1 & 0 & 0 \\ 
				0 & 0 & 0 & 1 & 0 & 0 & 0 & 0 & 1 & 1 & 0 & 0 & 1 & 1 \\ 
				0 & 0 & 0 & 0 & 1 & 0 & 1 & 0 & 1 & 0 & 1 & 0 & 1 & 0 \\ 
				0 & 0 & 0 & 0 & 0 & 1 & 0 & 1 & 0 & 1 & 0 & 1 & 0 & 1
			\end{array} \right] .
			\label{eq:n_eq_14_code}
		} 
		\eea
	\end{example}
	

	\section{Contributions to S-LR codes with any value of $t$} \label{sec:t_erasures}
	
	Contributions include a tight upper bound on rate of S-LR codes and a construction achieving the upper bound for any value of $r,t$ with $r \geq 3$. The construction achieving the upper bound is presented in the chapter following, Chapter~\ref{ch:SeqConstr}. Only the upper bound is presented in this chapter.   
	
	\subsection{A Tight Upper Bound on Rate of S-LR codes} \label{sec:tight_bound}
	In this section we provide an upper bound on rate of an $(n,k,r,t)$ S-LR code for any $r \geq 3$ and any $t$. The cases of even $t$ and odd $t$ are considered separately.  The proof proceeds by deducing the structure of parity-check matrix of an S-LR code. Constructions achieving our upper bound on rate for any $(r,t)$ with $r \geq 3$ are provided in the next chapter, Chapter~\ref{ch:SeqConstr} which establishes that the upper bound on rate derived here gives the exact expression for maximum possible rate of an S-LR code with parameters $(r,t)$ with $r \geq 3$. The matching constructions also make use of the structure of parity-check matrix deduced here. The upper bound on rate presented here also proves the conjecture by Song et al given in \ref{Conjecture}. 
	\begin{thm} \label{rate_both} \textbf{Rate Bound}: 
		Let $\mathcal{C}$ denote an $(n,k,r,t)$ S-LR code over a finite field $\mathbb{F}_q$. Let $r \geq 3$. Then 
		\begin{align}
			\frac{k}{n} & \leq  \frac{r^{s+1}}{r^{s+1} + 2 \sum_{i=0}^{s} r^i}\hspace{0.5cm}&\text{for even $t$,} \label{Thm1}\\
			\frac{k}{n} & \leq  \frac{r^{s+1}}{r^{s+1} + 2 \sum_{i=1}^{s} r^i + 1}\hspace{0.5cm}&\text{for odd $t$,}\label{Thm2}
		\end{align}
		where $s = \lfloor \frac{t-1}{2} \rfloor$.
	\end{thm}
	\begin{proof} 
		We begin by setting $\mathcal{B}_0 = \text{span}(\{\underline{c} \in \mathcal{C}^{\perp} : w_H(\underline{c}) \leq r+1 \})$ where $w_H(\underline{c})$ denotes the Hamming weight of the vector $\underline{c}$.
		Let $m$ be the dimension of $\mathcal{B}_0$. Let ${\underline{c}_1,...,\underline{c}_m}$ be a basis of $\mathcal{B}_0$ such that $w_H(\underline{c}_i) \leq r+1$, $\forall i \in [m]$. For all $i \in [m]$, We represent $\underline{c}_i$ as a $(1 \times n)$ vector where $j$th component of the $(1 \times n)$ vector is the $j$th codesymbol of the codeword $\underline{c}_i$, $\forall j \in [n]$.
		Let $H_1={[{\underline{c}_1}^T,...,{\underline{c}_m}^T]}^T$.
		It follows that $H_1$ is a parity check matrix of an $(n,n-m,r,t)$ S-LR code as its row space contains every codeword of Hamming weight at most $r+1$ which is present in $\mathcal{C}^\perp$. Also, 
		\bean
		\frac{k}{n} \leq 1 - \frac{m}{n}.
		\eean
		The idea behind the next few arguments in the proof is the following. S-LR codes with high rate will tend to have a larger value of $n$ for a fixed $m$.  On the other hand, the Hamming weight of the matrix $H_1$ (i.e., the number of non-zero entries in the matrix) is bounded above by $m(r+1)$.  	 It follows that to make $n$ large, one would like the columns of $H_1$ to have as small a weight as possible. It is therefore quite natural to start building $H_1$ by picking many columns of weight $1$, then columns of weight $2$ and so on.  As one proceeds by following this approach, it turns out that the matrix $H_1$ is forced to have a certain sparse, block-diagonal, staircase form and an understanding of this structure is used to derive the upper bound on code rate.\\
		The cases of $t$ being an even integer and an odd integer are considered separately. We form linear inequalities which arise from the structure of the matrix $H_1$ and derive the required upper bounds on the rate from these linear inequalities. See Appendix \ref{Appendix_rate_bounds} for detailed proofs for both the cases.
	\end{proof}
	
	\paragraph{Conditions for equality in \eqref{Thm1}: $t$ even case:}
	Note that for achieving the upper bound on rate given in \eqref{Thm1}, an S-LR code must have a parity check matrix $H_1$ (upto a permutation of columns) of the form given in \eqref{Hform} with parameters such that the inequalities given in \eqref{Ineq1},\eqref{Ineq2},\eqref{Ineq3},\eqref{Ineq_star1} (in Appendix) become equalities with $p=0$ and $D$ must be an empty matrix i.e., no columns of weight $\geq 3$ (because once all these inequalities become equalities, the sub matrix of $H_1$ obtained by restricting $H_1$ to the columns with weights 1,2 will have each row of weight exactly $r+1$ and hence no non-zero entry can occur outside the columns having weights 1,2 for achieving the upper bound on rate). Hence it can be seen that an $(n,k,r,t)$ S-LR code achieving the upper bound on rate \eqref{Thm1} must have a parity check matrix (upto a permutation of columns) of the form given in \eqref{eq:Hmatrixteven_ch3}.
	
	\bea
	H_1 = \left[
	\begin{array}{c|c|c|c|c|c|c|c}
		D_0 & A_1 & 0 & 0 & \hdots & 0 & 0 & 0  \\
		\cline{1-8}
		0 & D_1 & A_2 & 0 & \hdots & 0 & 0 & 0 \\
		\cline{1-8}
		0 & 0 & D_2 & A_3 & \hdots & 0 & 0 & 0  \\
		\cline{1-8}
		0 & 0 & 0 & D_3 & \hdots & 0 & 0 & 0 \\
		\cline{1-8}
		\vdots & \vdots & \vdots & \vdots & \hdots & \vdots & \vdots & \vdots \\
		\cline{1-8}
		0 & 0 & 0 & 0 & \hdots & A_{s-1} & 0 & 0  \\
		\cline{1-8}
		0 & 0 & 0 & 0 & \hdots & D_{s-1} & A_{s} & 0 \\
		\cline{1-8}
		0 & 0 & 0 & 0 & \hdots & 0 & D_{s} & C \\
	\end{array} 
	\right],  \label{eq:Hmatrixteven_ch3}
	\eea
	
	\paragraph{Conditions for equality in \eqref{Thm2}: $t$ odd case:}
	Note that for achieving the upper bound on rate given in \eqref{Thm2}, an S-LR code must have a parity check matrix $H_1$ (upto a permutation of columns) of the form given in \eqref{Hform2} with parameters such that the inequalities given in \eqref{Ineq1O},\eqref{Ineq2O},\eqref{Ineq3O},\eqref{OddSpe},\eqref{Ineq_star2} (in Appendix) become equalities with $p=0$ and $D$ must be an empty matrix i.e., no columns of weight $\geq 3$ (because once all these inequalities become equalities, the sub matrix of $H_1$ obtained by restricting $H_1$ to the columns with weights 1,2 will have each row of weight exactly $r+1$ and hence no non-zero entry can occur outside the columns having weights 1,2 for achieving the upper bound on rate). Note that for achieving the upper bound on rate, $C$ must also be an empty matrix. This is because inequality \eqref{OddSpe} must become an equality which implies that $D_s$ is a matrix with each row of weight $r+1$ and we also saw that $p=0$. Hence $C$ must be a zero matrix which implies $C$ is an empty matrix. Hence it can be seen that an $(n,k,r,t)$ S-LR code achieving the upper bound on rate \eqref{Thm2} must have a parity check matrix (upto a permutation of columns) of the form given in \eqref{eq:Hmatrixtodd_ch3}.  
	
	\bea
	H_1 & = & \left[
	\begin{array}{c|c|c|c|c|c|c}
		D_0 & A_1 & 0 & 0 & \hdots & 0 & 0   \\
		\cline{1-7}
		0 & D_1 & A_2 & 0 & \hdots & 0 & 0  \\
		\cline{1-7}
		0 & 0 & D_2 & A_3 & \hdots & 0 & 0   \\
		\cline{1-7}
		0 & 0 & 0 & D_3 & \hdots & 0 & 0   \\
		\cline{1-7}
		\vdots & \vdots & \vdots & \vdots & \ddots & \vdots & \vdots  \\
		\cline{1-7}
		0 & 0 & 0 & 0 & \hdots & A_{s-1} & 0   \\
		\cline{1-7}
		0 & 0 & 0 & 0 & \hdots & D_{s-1} & A_{s}   \\
		\cline{1-7}
		0 & 0 & 0 & 0 & \hdots & 0 & D_{s} 
	\end{array} \right] \label{eq:Hmatrixtodd_ch3},
	\eea

	It may be noted here that our upper bound on rate given in Theorem \ref{rate_both}, for the special cases of $t=2,3$, matches with the upper bound given in \cite{PraLalKum},\cite{SongCaiYue_L} respectively. For $t \geq 4$, our upper bound on rate given in Theorem \ref{rate_both} is new and hence a new contribution to S-LR codes. In the remainder of the thesis, S-LR codes with parameters $(r,t)$ achieving the upper bounds \eqref{Thm1} or \eqref{Thm2} depending on $t$ will be referred to as ``rate-optimal codes".
	
	%
	%
	\begin{note} \label{blk_length_rem}
		We now make a remark on blocklength of the rate-optimal codes. It can be seen that, for even $t$, with upper bound on rate given in \eqref{Thm1} being achieved with equality and for $k$ to be an integer with $r \geq 3$, $2n$ needs to be an integer multiple of $r^{s+1} + 2 \sum_{i=0}^{s} r^i$. Similarly for odd $t$, for $k$, to be an integer and with upper bound on rate given in \eqref{Thm2} being achieved with equality, $n$ needs to be an integer multiple of $r^{s+1}+2\sum_{i = 1}^{s}r^i+1$.
	\end{note}
	
	Note that the upper bound on rate of an S-LR code derived here is applicable in a much larger setting. The upper bound on rate derived in Theorem \ref{rate_both}, also yields an upper bound on the rate of a code with minimum distance $\geq t+1$ and has a parity check matrix $H$ with each row of weight $\leq r+1$. This is true as at no point in the derivation of Theorem \ref{rate_both}, did we use the sequential-recovery property apart from the property that the minimum distance $\geq t+1$ for a code with parity check matrix $H_1$. The problem of deriving an upper bound on rate in this general setting without the sequential-recovery property is addressed in the literature in \cite{BurKriLitMil,BenLit,IceSam,Fro}. But these papers are interested in deriving an upper bound on the rate of a code with parity check matrix with each row of weight $\leq r+1$ with a relative minimum distance $\delta$. But here in this work, we are not interested in having a non-zero relative minimum distance. It is enough if the minimum distance $\geq t+1$ irrespective of the relative minimum distance. Hence we obtain a different bound in Theorem \ref{rate_both} whereas the expressions for upper bound on rate in \cite{BurKriLitMil,BenLit,IceSam,Fro} are in terms of the expression for binary entropy.  Apart from the above, we prove in next chapter that the upper bound on rate in Theorem \ref{rate_both} is achievable and hence is tight whereas the tightness of the bounds in \cite{BurKriLitMil,BenLit,IceSam,Fro} is unknown. 
	
	\section{Summary} \label{sec:seq_conclusions} 
	
	This chapter dealt with LR codes for the multiple erasures case.  Following an overview of this topic, the chapter focused on the case of      codes with sequential recovery from multiple erasures (S-LR codes) as some of the principal contributions of the thesis  relate to S-LR codes.   These results include (a) new rate-optimal and block-length-optimal constructions of S-LR codes for the $t=2$ case, (b) a characterization of rate-optimal S-LR codes for $t=2$, (c) a new upper bound on dimension for S-LR codes and constructions of S-LR codes achieving this bound for the $t=2$ case, (d) a lower bound on the block length of S-LR codes for $t=3$ case, which is tighter than the existing bounds in the literature for $r \leq k \leq r^{1.8}-1$, $2 \leq r \leq 200$. We then present our main result, namely, an upper bound on the rate of an S-LR code for any $(r,t)$ with $r \geq 3$.  Constructions of binary codes achieving the upper bound on rate derived here for $r \geq 3$ are deferred to the next chapter. 
	%
	\begin{appendices}
		\chapter{Proof of Theorem \ref{thm:structure_2erasure}} \label{app:structure_2erasure}
		\begin{proof} Throughout the proof: the term code symbols of $\mathcal{C}$ with indices in a set $S=\{j_1,...,j_{\ell}\}$ refers to code symbols $c_{j_1},...,c_{j_{\ell}}$. Whenever we say a row of $A$ denoted by $R$, $R$ is used to refer to either the vector formed by the row of the matrix $A$ or the index of the row in matrix $A$ depending on context. The term columns $C_1,...C_{\ell}$ of $A$ refers to the $C_1^{th}$ column,$C_2^{th}$ column,...,$C_{\ell}^{th}$ column of the matrix $A$.
			
			From the description of general form of parity check matrix of rate-optimal code given in Theorem \ref{thm:pc_mx_form}, equation \eqref{eq:H_form_rate_optimal}, there exists (after possible permutation of code symbols of $\mathcal{C}$) a parity-check matrix $H$ for the rate-optimal code $\mathcal{C}$ of the form (Note that this is true over any finite field $\mathbb{F}_q$):
			\beqn
			H=[I|H'],
			\eeqn
			where $I$ denote the $((n-k) \times (n-k))$ identity matrix and $H'$ is a $(n-k)\times k$ matrix with each columns having weight $2$ and each row having weight $r$.
			Consider $2$ rows of $H$ denoted by $R_1$ and $R_2$. Let $supp(R_1) \cap supp(R_2) = \{C_1,C_2,...,C_s\}$. Since all columns in $H'$ has Hamming weight exactly $2$, the columns $C_1,C_2,...,C_s$ of $H$ will have non zero entries in rows $R_1$ and $R_2$ only. Let $A$ denote the $2\times s$ sub matrix obtained by restricting $H$ to the rows $R_1$ and $R_2$ and the columns $C_1,C_2,...,C_s$. In order to recover from any instance of $2$ erasures in the code symbols of $\mathcal{C}$ with indices in $\{C_1,C_2...C_s\}$, any two columns of $A$ must be linearly independent. Thus the $2\times s$ sub matrix $A$ is a generator matrix of an $MDS$ code of block length $s$ and dimension 2. This also says that any vector, obtained by a non-zero linear combination of the two rows of $A$ will have a Hamming weight at least $s-1$.\\
			Let us consider two extreme cases:\\
			\emph{Case 1: $s=r$:}
			In this case, the $r+2$ code symbols of $\mathcal{C}$ with indices in $supp(R_1) \cup supp(R_2)$ form an $[r+2,r]$ MDS code. The sub matrix of $H$ obtained by restricting $H$ to the columns $supp(R_1) \cup supp(R_2)$ and all the rows apart from the rows $R_1,R_2$ is a zero matrix.\\
			\emph{Case 2: $s\leq1$:} $|supp(R_1) \cap supp(R_2)| \leq 1$.\\
			If these are the only two cases that can occur for any pair of rows of $H$ denoted by say $R$ and $R'$ i.e.,$|supp(R) \cap supp(R')| \in \{0,1,r\}$ for any 2 distinct rows of $H$ denoted by $R,R'$, then the set of code symbols can be partitioned into two sets, one set of symbols forming a regular-graph code over $\mathbb{F}_q$ and the other set of symbols forming direct product of $[r+2,r]$ MDS codes, with no linear constraint involving code symbols from both the sets i.e., the code $\mathcal{C}$ will be a direct product $\mathcal{C}_R \times \mathcal{C}_M$ (after possible permutation of code symbols of $\mathcal{C}$) where $\mathcal{C}_R$ is a regular-graph code over $\mathbb{F}_q$ and $\mathcal{C}_M$ is a direct product of $[r+2,r]$ MDS codes. Hence the code $\mathcal{C}$ has the form stated in the theorem.
			
			Now, we will prove that  for any 2 distinct rows of $H$ denoted by $R,R'$, $1<s<r$ is not possible, where $s= |supp(R) \cap supp(R')|$.
			
			Wolog assume that $1<s<r$ for the pair of rows $R_1$ and $R_2$ i.e., $1 < |supp(R_1) \cap supp(R_2)|=s < r$. As denoted before, $C_i,C_j$ are some two elements of $supp(R_1) \cap supp(R_2)$ for $1 \leq i \neq j \leq s$. Assume that the code symbols of $\mathcal{C}$ with indices in $\{C_i, C_j\}$ are erased. In order for $\mathcal{C}$ to be an S-LR code with $t=2$, a linear combination of $R_1$ and $R_2$ and some of the remaining rows of $H$ must result in a vector $\mathbf{v}$ with the following properties.
			\ben
			\item Hamming weight of $\mathbf{v}$ is less than or equal to $(r+1)$.
			\item $\mathbf{v}$ has the value zero in the $C_i^{th}$ coordinate and has a non zero value in the $C_j^{th}$ coordinate, or vice versa.
			\een
			Assume that a linear combination of $\ell$ rows of $H$ denoted by $R_1,R_2,R_3...R_{\ell}$ results in $\mathbf{v}$ where each coefficient in the linear combination is non-zero. Let $s_{ij} = |supp(R_i)\cap supp(R_j)|$, $\forall 1 \leq i \neq j \leq \ell$. Clearly, $s_{12}=s$. If $s_{ij}>0$, we have already shown that the $2\times s_{ij}$ sub matrix of $H$, obtained by restricting $H$ to the rows $R_i$ and $R_j$ and the columns in $supp(R_i)\cap supp(R_j)$ is a generator matrix of an MDS code of block length $s_{ij}$ and dimension 2 and any non-zero linear combination of the two rows of the sub matrix gives a vector of Hamming weight at least $s_{ij}-1$. Thus the Hamming weight of $\mathbf{v}$ is at least $\ell + \sum_{\{(i,j) : 1 \leq i<j\leq \ell,s_{ij}>0\}}^{}(s_{ij}-1) + f$, where the first term $\ell$ comes from the identity part of $H$ (i.e., columns $1$ to $n-k$ of $H$) and the term $f$ comes from the single weight columns in the sub matrix $L$ of $H$ obtained by restricting $H$ to the rows $\{R_1,R_2,R_3,...,R_{\ell}\}$ and the columns $n-k+1$ to $n$ of $H$. Let $S = \{(i,j) : 1 \leq i<j\leq \ell,s_{ij}>0\}$. Since the Hamming weight of $\mathbf{v}$, must be upper bounded by $r+1$,
			\begin{align}
				\ell + \sum_{(i,j) \in S}^{}(s_{ij}-1) + f &\leq r+1, \notag\\ 
				\ell + \sum_{(i,j) \in S}^{}s_{ij} - {\ell \choose 2} + f &\leq r+1, \notag\\ 
				\ell + 2 \sum_{(i,j) \in S}^{} s_{ij} - {\ell \choose 2} + f &\leq r+1 + \sum_{(i,j) \in S}^{} s_{ij}. \label{eq:wt1}
			\end{align}
			Also, by counting the non zero entries in $L$ column wise and row wise
			\begin{align}
				f+ 2\sum_{(i,j) \in S}^{}s_{ij} &= \ell r, \label{eq:wt2}\\
				\sum_{(i,j) \in S}^{}s_{ij} &\leq \frac{\ell r}{2}. \label{eq:wt3}
			\end{align}
			
			Substituting \eqref{eq:wt2} and \eqref{eq:wt3} in \eqref{eq:wt1} gives: 
			\bea
			\ell + \ell r-{\ell \choose 2} \leq r+1 + \frac{\ell r}{2}, \notag \\
			\ell-r-1  \leq   \frac{\ell (\ell- 1-r)}{2} \label{ineq_rate_char_ch3}
			\eea
			Now assuming $\ell < r+1$, \eqref{ineq_rate_char_ch3} gives $\ell \leq 2$.
			Hence assuming $\ell > 2$, we get
			\bean
			r+1 \leq \ell.
			\eean
			Hence we get $\ell \geq r+1$, when $\ell > 2$. But when $\ell \geq r+1$, the first $n-k$  co-ordinates of $\mathbf{v}$ will have atleast $r+1$ non-zero values  (because the columns $1$ to $n-k$ of $H$ form an identity matrix), making the Hamming weight of $\mathbf{v}$  strictly greater than $r+1$ as $\mathbf{v}$ must also have a non zero $C_j^{th}$ or $C_i^{th}$ coordinate as $n-k+1 \leq C_i,C_j \leq n$. This is a contradiction as Hamming weight of $\mathbf{v}$ is $\leq r+1$. Hence, if $1<s<r$, $\ell>2$ is not possible. 
			Now, assume $\ell=2$ i.e., a linear combination of $R_1$ and $R_2$ should give $\mathbf{v}$. As seen before, the Hamming weight of a linear combination of $R_1$ and $R_2$ is at least $(s-1)+2(r+1-s)$ (weight $s-1$ comes from the coordinates indexed by $supp(R_1)\cap supp(R_2)$, and weight $2(r+1-s)$ comes from the remaining coordinates indexed by $(supp(R_1)\cup supp(R_2))-(supp(R_1)\cap supp(R_2))$). Since weight of $\mathbf{v}$ is $\leq r+1$, We need,
			\begin{align*}
				(s-1)+2(r+1-s) &\leq r+1,\\
				s &\geq r,
			\end{align*}
			which is not possible as  $1<s<r$. Hence $\ell=2$ is also not possible. Hence putting together we have $\ell \geq 2$ not possible for $1<s<r$, but we need to linearly combine $\ell \geq 2$ rows to get $\mathbf{v}$. Hence $1<s<r$ is not possible.
		\end{proof}
		\chapter{Proof of Theorem \ref{thm:DimUpperBound_ch3}} \label{app:upperbounddim}
		\begin{proof}
			Let $\{\underline{c}_1,...,\underline{c}_m\}$ be a basis of $\mathcal{B}_{0}$ with $\text{w}_\text{H}(c_i) \leq r+1$, where $\text{w}_\text{H}(\underline{c}_i)$ denotes the Hamming weight of $\underline{c}_i$, $\forall 1 \leq i \leq m$. Let
			\bean
			H & = & 
			\left[ \begin{array}{c} 
				\underline{c}_1  \\
				\vdots \\
				\underline{c}_m   \end{array} \right].
			\eean
			Let $1 \leq L \leq m$. Let $s_i$ denote the total number of columns of $H$ having Hamming weight $i$ for $1 \leq i \leq L$. Then by simple counting of non-zero entries of $H$ row wise and column wise, we get:
			\bea
			\sum_{i=1}^{L} i s_i + (L+1)(n-\sum_{i=1}^{L} s_i) \leq m(r+1), \notag \\
			(L+1)n \leq m(r+1)+ \sum_{i=1}^{L} (L+1-i) s_i. \label{eq:blocklength_1}
			\eea
			Since $H$ is a parity-check matrix of an $(n,n-m,r,2)$ S-LR code over $\mathbb{F}_2$, all the columns of $H$ must be distinct. Hence $s_i \leq {m \choose i}$ for $1 \leq i \leq L$. Also we know that $n \geq k+m$. Hence substituting $s_i \leq {m \choose i}$ and $n \geq k+m$ in \eqref{eq:blocklength_1}:
			\bea
			(L+1)(k+m) \leq m(r+1)+ \sum_{i=1}^{L} (L+1-i) {m \choose i}, \notag \\
			k \leq \frac{m(r-L)+ \sum_{i=1}^{L} (L+1-i) {m \choose i}}{L+1}. \label{eq:blocklength_2}
			\eea
			\eqref{eq:blocklength_2} gives an upper bound on $k$ which is applicable for every $1\leq L \leq m$. Taking minium over $1\leq L \leq m$ gives:
			\bea
			k \leq \min_{1 \leq L\leq m} \frac{m(r-L)+ \sum_{i=1}^{L} (L+1-i) {m \choose i}}{L+1} \label{eq:blocklength_3}
			\eea
		\end{proof}
		\chapter{Proof of Lemma~\ref{lem:min_len3}} \label{app:lem_min_len3}
		\begin{proof}
			Each column of the $(M \times N)$ binary matrix $H$ has Hamming weight equal to $2$. As a first step, do the following. Permute the rows of $H$ such that the $2$ non-zero coordinates of the first column appear in the first two rows. Let the resulting matrix be $H'$. Let the Hamming weights of first two rows of $H'$ be $w_1$ and $w_2$.

			Now, keeping the first column fixed, permute the columns of $H'$ such that the $w_1$ non zero coordinates of the first row appear in the first $w_1$ columns and the $w_2-1$ non-zero coordinates (except the first coordinate) of the second row appear in columns $w_1+1$ through $w_1+w_2-1$. 
			
			Note that except the first column, there is no column with non zero entries in both of the first $2$ rows. This follows from the fact that $d_\text{min}\geq 4$.
			
			If a column has non-zero elements in the first row and $g^\text{th}$ row, $g \geq 3$, then the remaining columns cannot have non-zero entries in the second row and the $g^\text{th}$ row. This is due to the fact that these two columns (one column with non-zero entry in the first row and $g^\text{th}$ row and the other column with non-zero entry in the second row and $g^\text{th}$ row) along with the first column of $H'$ forms a set of $3$ linearly dependent vectors, which would imply $d_\text{min} < 4$.
			
			Hence, if a column has non-zero elements in the first row and $g^\text{th} (g>2)$ row, then a column with non-zero entries in the second and $g^\text{th}$ row cannot occur and if a column (apart from the first column) has a one in the first row, there are only $M-2$ possible coordinates where it can have the second non-zero entry. Similarly if a column (apart from the first column) has a one in the second row, there are only $M-2$ possible coordinates where it can have the second non-zero entry. Hence, we have
			\beq
			w_1+w_2-2 \leq M-2. \label{eq:dmin3_4_ineq}
			\eeq

			Remove the first $w_1+w_2-1$ columns of $H'$ and then remove the first $2$ rows of $H'$. Let the resulting matrix be $H_{M-2}$. In $H'$, since none of columns other then the first $w_1+w_2-1$ columns of $H'$ has ones in any of the first 2 rows, the new matrix $H_{M-2}$ is also a parity-check matrix of a code with $d_\text{min} \geq 4$ and each column of weight $2$ with $M-2$ rows. As a second step, repeat the same arguments as first step on $H_{M-2}$. Let the resulting matrix after permutation of rows and columns of $H_{M-2}$ in the second step be $H'_{M-2}$. Let $w_3$ and $w_4$ denote the Hamming-weights of the row one and row two of $H'_{M-2}$ respectively. Using similar arguments as in the case of \eqref{eq:dmin3_4_ineq}, we get
			\beqn
			w_3+w_4-2 \leq M-4.
			\eeqn
			
			Now remove the first $w_3+w_4-1$ columns of $H'_{M-2}$ and then remove the first two rows of $H'_{M-2}$ as before and let the resulting matrix be $H_{M-4}$. Now as a third step, repeat the same arguments on $H_{M-4}$. Repeat the procedure until all columns are exhausted i.e., all the columns are removed in the process. Let $e$ denote the total number of steps in the process (we are removing atleast one column in each step. So the process will terminate.). 
			
			As we progress, we are removing $w_1+w_2-1$ columns in the first step and $w_3+w_4-1$ columns in the second step and so on. Hence, the number of columns of $H$ is given by:
			\begin{equation*}
				N=(w_1+w_2-1)+(w_3+w_4-1)+\dots+(w_{2e-1}+w_{2e}-1).
			\end{equation*}
			In the $i^\text{th}$ step we get:
			\beq
			w_{2i-1}+ w_{2i}-2 \leq M-2i, \hspace{0.5cm} 1 \leq i \leq e. \label{eq:min_len3_lem_1}
			\eeq
			Since $2$ rows are removed in each step, the number of steps is upper bounded as follows:
			\beq
			e \leq \left\lfloor \frac{M+1}{2} \right\rfloor. \label{eq:min_len3_lem_2}
			\eeq
			
			From \eqref{eq:min_len3_lem_1} and \eqref{eq:min_len3_lem_2}, we have
			\bean
			(w_1+w_2-1)+(w_3+w_4-1)+\dots+(w_{2e-1}+w_{2e}-1) \\
			\leq \left( \sum_{\ell=1}^{\lfloor \frac{M+1}{2} \rfloor}  (M-(2 \ell-1)) \right) + 1.
			\eean
			Hence,
			\begin{align*}
				N &\leq \left( \sum_{\ell=1}^{\lfloor \frac{M+1}{2} \rfloor}  (M-(2 \ell-1)) \right) + 1 \\
				&= M \left\lfloor \frac{M+1}{2} \right\rfloor -\left\lfloor \frac{M+1}{2} \right\rfloor\left(\left\lfloor \frac{M+1}{2} \right\rfloor+1\right)+\left\lfloor \frac{M+1}{2} \right\rfloor + 1 \\
				&\leq M \left(\frac{M+1}{2}\right) -\left(\frac{M+1}{2}-1\right)\left(\frac{M+1}{2}\right)+\frac{M+1}{2}+ 1 \\
				&= (M +2) \left(\frac{M+1}{2}\right) -\frac{(M+1)(M+1)}{4}+ 1 \\
				&= \left(M+\frac{3}{2}\right) \left(\frac{M+1}{2}\right) -\frac{M(M+1)}{4}+ 1 \\
				&= \frac{(M+3)(M+1)}{4}+1.
			\end{align*}
		\end{proof}
		\chapter{Proof of Theorem \ref{thm:min_len3}} \label{app:min_len3}
		\begin{proof}
			Let, 
			\bean
			\mathcal{B}_{0} & = & \text{span}\left(\{\underline{c} \in \mathcal{C}^{\perp} : |\text{supp}(\underline{c})| \leq r+1  \}\right).
			\eean
			Let $\{\underline{c}_1,...,\underline{c}_m\}$ be a basis of $\mathcal{B}_0$ with $|\text{supp}(\underline{c}_i)| \leq r+1$, $\forall i \in [m]$. Let
			\bean
			H & = & 
			\left[ \begin{array}{c} 
				\underline{c}_1  \\
				\vdots \\
				\underline{c}_m   \end{array} \right].
			\eean
			Let $m'=n-k$.
			Extend the basis $\{\underline{c}_1, \cdots ,\underline{c}_m\}$ of $\mathcal{B}_0$ to a basis of $\cal{C^{\perp}}$ and form a $(m' \times n)$ parity-check matrix $H'$ of $\cal{C}$, with this basis of $\mathcal{C}^{\perp}$ as its rows. Hence 
			\bean
			H' & = & 
			\left[ \begin{array}{c} 
				H  \\
				H_1   \end{array} \right],
			\eean
			where the rows of $H_1$ are precisely those vectors added to the basis of $\mathcal{B}_0$ while extending it to a basis of $\cal{C^{\perp}}$.
			\bea
			n=k+m' \geq k+m. \label{eq:dmin3_ineq_1}
			\eea
			
			$H$ is a parity-check matrix of an $(n,n-m,r,3)$ S-LR code as the row space of $H$ contains all the codewords of $\cal{C^{\perp}}$ which has weight $\leq r+1$.
			We will consider the $(n,n-m,r,3)$ S-LR code defined by the $(m \times n)$ parity-check matrix $H$ and derive a lower bound on $m$ as a function of $k$ and $r$. Using \eqref{eq:dmin3_ineq_1} and the derived lower bound on $m$, we obtain a lower bound on $n$.
			
			Let $s_1,s_2$ be the number of columns of $H$ having weights $1$ and $2$ respectively. Then by counting non zero entries of $H$ row wise and column wise, we get:
			\bea
			s_1+2 s_2 +3(n-s_1-s_2) & \leq & m(r+1) \notag \\
			3n-m(r+1)-2s_1 & \leq & s_2. \label{eq:bound_s2}
			\eea
			Permute the columns and rows of $H$ such that:
			\bean
			H & = & 
			\left[ \begin{array}{cc} 
				D_{s_1} & A  \\
				0 & B   \end{array} \right],
			\eean	
			where $D_{s_1}$ is a $(s_1 \times s_1)$ diagonal matrix with non zero diagonal entries and the set of first $s_1$ columns are the set of all columns of $H$ having weight one. Note that in the first $s_1$ columns having weight one there can't be 2 columns with non zero entry at exactly the same row as this would imply $d_\text{min} \leq 2$ (where $d_\text{min}$ is the minimum distance of the code defined by the parity-check matrix $H$). Now the $s_2$ columns having weight $2$ are to the right of $D_{s_1}$. In these $s_2$ columns, we cannot have a column with two non zero entries in the first $s_1$ rows, as this would imply $d_\text{min} \leq 3$, which is a contradiction since the code defined by parity-check matrix $H$ is also an $(n,n-m,r,3)$ S-LR code which requires $d_\text{min} \geq 4$.
			
			Let $f_1$ denote number of columns of weight $2$ with exactly one non zero entry in the first $s_1$ rows and $f_2$ denote the number of columns of weight $2$ with both non zero entries in the last $m-s_1$ rows.
			\begin{align}
				s_2 &=f_1+f_2 \label{s2:1}\\
				f_1 &\leq s_1(m-s_1)  \\
				f_2 &\leq N(m-s_1,2,4), 
			\end{align}
			where $N(m-s_1,2,4)$ is the maximum number of columns in a parity-check matrix (of a code with minimum distance $\geq 4$) with $m-s_1$ rows and each column having a Hamming weight 2. Here, $f_2 \leq N(m-s_1,2,4)$ because the sub matrix of $H$, obtained by restricting $H$ to these $f_2$ columns and the last $m-s_1$ rows is a parity-check matrix (of a code with minimum distance $\geq 4$) with $m-s_1$ rows and each column having a Hamming weight of 2.
			Restricting to binary codes, it is straightforward to see that,
			$N(m-s_1,2,4) \leq {m-s_1 \choose 2}$.
			Using Lemma~\ref{lem:min_len3}, this can be tightened to:
			\begin{align*}
				N(m-s_1,2,4) &\leq \frac{(m-s_1+3)(m-s_1+1)}{4}+1.
			\end{align*}
			Hence from \eqref{s2:1},
			\begin{align}
				s_2  &\leq s_1(m-s_1) + {m-s_1 \choose 2} \label{eq:bin_1}\\
				s_2  &\leq s_1(m-s_1) + \frac{(m-s_1+3)(m-s_1+1)}{4}+1. \label{eq:bin_2}
			\end{align}
			
			Inequalities \eqref{eq:binary_1} and \eqref{eq:binary_2} given below are obtained by substituting the above upper bounds \eqref{eq:bin_1} and \eqref{eq:bin_2} on $s_2$ in \eqref{eq:bound_s2}.
			\begin{align}
				3n-m(r+1)-2s_1  &\leq s_1(m-s_1) + {m-s_1 \choose 2} \label{eq:binary_1}\\
				3n-m(r+1)-2s_1  &\leq s_1(m-s_1) + \notag \\ &\frac{(m-s_1+3)(m-s_1+1)}{4}+1. \label{eq:binary_2}
			\end{align}
			
			\eqref{eq:binary_1} (On using $n \geq k+m$) leads to:
			\bean
			m^2 + m(2r-5)-(6k+s_1^2-5s_1) \geq 0,
			\eean
			which along with the fact $m \geq 0$ and the fact that the derivative of $m^2 + m(2r-5)-(6k+s_1^2-5s_1)$ at the negative root of $m^2 + m(2r-5)-(6k+s_1^2-5s_1)$ is $\leq 0$, shows that $m$ is atleast the positive root of $m^2+m(2r-5)-(6k+s_1^2-5s_1)$. Hence,
			\begin{align*}
				m &\geq f_1(s_1) = \left\lceil \frac{-(2r-5)+\sqrt{(2r-5)^2 +4(6k+s_1^2-5s_1)}}{2} \right\rceil 
			\end{align*}
			which when added with $k$ gives a required lower bound on the block-length $n$.
			
			\eqref{eq:binary_2} (On using $n \geq k+m$) leads to:
			\bean
			m^2 + m(4r-4+2s_1)-(12k+3s_1^2-4s_1-7) \geq 0, 
			\eean
			which along with the fact $m \geq 0$ and the fact that the derivative of $m^2 + m(4r-4+2s_1)-(12k+3s_1^2-4s_1-7)$ at the negative root of $m^2 + m(4r-4+2s_1)-(12k+3s_1^2-4s_1-7)$ is $\leq 0$,  shows that $m$ is atleast the positive root of  $m^2 + m(4r-4+2s_1)-(12k+3s_1^2-4s_1-7)$. Hence,
			\begin{align*}
				m &\geq f_2(s_1) = \left\lceil \frac{-(4r-4+2s_1)+\sqrt{(4r-4+2s_1)^2 +4(12k+3s_1^2-4s_1-7)}}{2} \right\rceil
			\end{align*}
			Hence using $m \geq f_1(s_1)$, $m \geq f_2(s_1)$ and $m \geq s_1$ we get
			\bean
			m \geq \min_{s_1 \in \mathbb{N} \cup \{0\}} \max(f_1(s_1),f_2(s_1),s_1).
			\eean
			Using $n \geq k+m$, we get 
			\begin{align*}
				n \geq k+\min_{s_1 \in \mathbb{N} \cup \{0\}} \max(f_1(s_1),f_2(s_1),s_1),
			\end{align*}
			where $\mathbb{N}$ denote the set of natural numbers. 
		\end{proof}
		\chapter{Proof of theorem \ref{rate_both}}\label{Appendix_rate_bounds}
		\begin{proof}
			Throughout this proof: we denote any set of the form $\{x_1,....,x_{\ell}\}$ by $\{x_j\}$ for any $\ell$. The term $w$-weight column of $A$ refers to a column of the matrix $A$ with Hamming weight $w$. The term sum of row weights of a matrix $A$, refers to the sum of weight of each row of the matrix $A$. Similar definition applies for the term sum of column weights. For an $x \times y$ matrix $A$ and for some $E_1=\{m_1,..,m_{|E_1|}\} \subseteq [x]$,$E_2=\{n_1,...,n_{|E_2|}\} \subseteq [y]$ with $m_{\ell} < m_{\ell+1}$, $\forall \ell \in [|E_1|-1]$ and $n_{\ell} < n_{\ell+1}$, $\forall \ell \in [|E_2|-1]$ we denote by $A|_{E_1,E_2}$, the $(|E_1| \times |E_2|)$ submatrix of $A$ with $(i,j)$th entry of $A|_{E_1,E_2}$ equal to $(m_i,n_j)$th entry of $A$. The term empty matrix refers to a $0 \times L$ or $L \times 0$ or $0 \times 0$ matrix, for some $L$. If a matrix is not an empty matrix, it will be termed as a non-empty matrix. 
			\section{case i: $t$ an even integer}
			Recall that $H_1={[{\underline{c}_1}^T,...,{\underline{c}_m}^T]}^T$ where ${\underline{c}_1,...,\underline{c}_m}$ are $m$ linearly independent codewords with $w_H(\underline{c}_i) \leq r+1$, $\forall i \in [m]$. It can be seen that the matrix $H_1$ after permutation of rows and columns can be written in the form given in \eqref{Hform}.
			\bea
			H_1 = \left[
			\scalemath{1}{
				\begin{array}{c|c|c|c|c|c|c|c|c}
					D_0 & A_1 & 0 & 0 & \hdots & 0 & 0 & 0 & \\
					\cline{1-8}
					0 & D_1 & A_2 & 0 & \hdots & 0 & 0 & 0 & \\
					\cline{1-8}
					0 & 0 & D_2 & A_3 & \hdots & 0 & 0 & 0 & \\
					\cline{1-8}
					0 & 0 & 0 & D_3 & \hdots & 0 & 0 & 0 &  \\
					\cline{1-8}
					\vdots & \vdots & \vdots & \vdots & \hdots & \vdots & \vdots & \vdots & D\\
					\cline{1-8}
					0 & 0 & 0 & 0 & \hdots & A_{\frac{t}{2}-2} & 0 & 0 & \\
					\cline{1-8}
					0 & 0 & 0 & 0 & \hdots & D_{\frac{t}{2}-2} & A_{\frac{t}{2}-1} & 0 & \\
					\cline{1-8}
					0 & 0 & 0 & 0 & \hdots & 0 & D_{\frac{t}{2}-1} &  & \\
					\cline{1-7}
					0 & 0 & 0 & 0 & \hdots & 0 & 0 & C & \\
				\end{array} 
			}
			\right], \ \ \label{Hform}
			\eea
			%
			where 
			\begin{enumerate}
				\item Rows are labeled by the integers $1,...,m$  with top most row of $H_1$ labelled $1$ and next row labeled $2$ and so on. Similarly columns are labeled by the integers $1,...,n$ with leftmost column of $H_1$ labeled $1$ and next column labeled $2$ and so on, 
				\item $A_i$ is a $\rho_{i-1} \times a_i$ matrix for $1\leq i \leq \frac{t}{2}-1$, $D_i$ is a $\rho_{i} \times a_i$ matrix for $0 \leq i \leq \frac{t}{2}-1$ for some $\{\rho_j\},\{a_j\}$. Note that any of $a_i$, $\rho_i$ is allowed to take the value 0 for any $i$,
				\item $D_0$ is a matrix with each column having weight 1 and each row having weight at least 1. The first $a_0$ columns of $H_1$ contains the columns of $D_0$. The set of first $a_0$ columns of $H_1$ is equal to the set of all those columns of $H_1$ which has weight $1$,
				\item  $\{A_j\}$,$\{D_j\}$,$\{B_j= \left[\frac{A_j}{D_j}\right]\}$ are  such that for $1\leq i \leq \frac{t}{2}-1$, each column of $B_i$ has weight 2, each column of $A_i$ has weight at least 1 and each row of $D_i$ has weight at least 1 and each column of $D_i$ has weight at most 1,
				\item $C$ is a matrix with each column having weight 2. The set of columns of the matrix $D$ is equal to the set of all those columns of $H_1$ which has weight $\geq 3$.
			\end{enumerate}
			Let $J = \min{\{\{j: 1 \leq j \leq \frac{t}{2}-1, A_j,D_j \text{ are empty matrices }\} \cup \{\frac{t}{2} \} \}}$. We set $J=0$, if $D_0$ is an empty matrix. We now redefine $H_1$ based on the value of $J$.
			We set $A_i$, $D_i$ to be empty matrices and set $a_i=0, \rho_i=0$, $\forall J \leq i \leq \frac{t}{2}-1$. 
			Let $E_2  \subseteq \{\sum_{j=0}^{J-1}a_j+1,...,n\}$ such that $E_2$ is the set of labels of all the 2-weight columns of $H_1$ apart from those 2-weight columns of $H_1$ containing the columns of $B_1,...,B_{J-1}$. Let $E_1 = \{\sum_{j=0}^{J-1}{\rho_j}+1,...,m\}$. If $J=\frac{t}{2}$ then $H_1$ is defined by \eqref{Hform}. If $J < \frac{t}{2}$, we redefine $C=H_1|_{E_1,E_2}$.   If $J < \frac{t}{2}$, the matrix $H_1$ can be written in the form given in \eqref{Hform_111} and hence defined by \eqref{Hform_111}.
			\bea
			H_1 = \left[
			\scalemath{1}{
				\begin{array}{c|c|c|c|c|c|c|c|c}
					D_0 & A_1 & 0 & 0 & \hdots & 0 & 0 & 0 & \\
					\cline{1-8}
					0 & D_1 & A_2 & 0 & \hdots & 0 & 0 & 0 & \\
					\cline{1-8}
					0 & 0 & D_2 & A_3 & \hdots & 0 & 0 & 0 & \\
					\cline{1-8}
					0 & 0 & 0 & D_3 & \hdots & 0 & 0 & 0 &  \\
					\cline{1-8}
					\vdots & \vdots & \vdots & \vdots & \hdots & \vdots & \vdots & \vdots & D\\
					\cline{1-8}
					0 & 0 & 0 & 0 & \hdots & A_{J-2} & 0 & 0 & \\
					\cline{1-8}
					0 & 0 & 0 & 0 & \hdots & D_{J-2} & A_{J-1} & 0 & \\
					\cline{1-8}
					0 & 0 & 0 & 0 & \hdots & 0 & D_{J-1} & 0 & \\
					\cline{1-8}
					0 & 0 & 0 & 0 & \hdots & 0 & 0 & C & \\
				\end{array} 
			}
			\right]. \ \ \label{Hform_111}
			\eea
			Irrespective of the value of $J$, let the number of columns in $C$ be denoted as $a_{\frac{t}{2}}$. If $C$ is an empty matrix then we can clearly set $a_{\frac{t}{2}}=0$.
			The entire derivation of upper bound on rate is correct and all the inequalities in the derivation will hold with $a_i=0$, $\rho_i=0$, $\forall J \leq i \leq \frac{t}{2}-1$. Let $S=\{j: 1 \leq j \leq J-1, D_{j} \text{ is an empty matrix}, A_{j} \text{ is a non-empty matrix} \} $. Let $\ell = \min{\{ S \cup \{J\} \}}$. If $\ell < J$, then the proof of the following Lemma \ref{claim1} (since the proof of Lemma \ref{claim1} proceeds by induction starting with the proof of the lemma for $D_0$ first and then proceeding to $A_1$ and then $D_1$ and so on and hence we prove the lemma for $A_{\ell}$ first and then proceed to $D_{\ell}$ and since $D_0,A_j,D_j$ must be non-empty matrices $\forall 1 \leq j \leq \ell-1$) will imply that each column of $A_{\ell}$ has weight 1 which will imply that $D_{\ell}$ cannot be an empty matrix. Hence the case $A_{\ell}$, a non-empty matrix and $D_{\ell}$, an empty matrix cannot occur.
			Although we have to prove the following Lemma \ref{claim1} for $A_i,D_i$, $\forall 1 \leq i \leq J-1$, $D_0$, we assume all $D_0,$ $A_i,D_i$, $\forall 1 \leq i \leq \frac{t}{2}-1$ to be non-empty matrices and prove the lemma. Since the proof of the lemma is by induction, the induction can be made to stop after proving the lemma for $A_{J-1},D_{J-1}$ (induction starts by proving the lemma for $D_0$ as mentioned before) and the proof is unaffected by it.
			\begin{lem}\label{claim1}
				For $1 \leq i \leq \frac{t}{2}-1$, $A_i$ is a matrix with each column having weight 1.
				For $0 \leq j \leq \frac{t}{2}-1$, $D_j$ is a matrix with each row having weight 1 and each column having weight 1. 
			\end{lem}
			%
			%
			%
			\begin{proof}
				Let $d_{\min}$ be the minimum distance of the code with parity check matrix $H_1$. We use the fact that $d_{\min} \geq t+1$ to prove the lemma.\\
				It is enough to show that:
				
				\begin{itemize}
					\item For $1 \leq i \leq \frac{t}{2}-1$, $A_i$ is a matrix with each column having  weight 1.
					\item For $0 \leq i \leq \frac{t}{2}-1$, $D_i$ is a matrix with each row having weight 1. 
				\end{itemize}
				This is because the property that $A_i$ is a matrix with each column having weight 1 combined with the fact that each column of $B_i=\left[\frac{A_i}{D_i}\right]$ has weight 2  implies that $D_i$ is a matrix with each column having weight 1 and $D_0$ by definition is a matrix with each column having weight 1. \\
				Let us denote the column of the matrix $H_1$ with label $j$ by $\underline{h}_j$, $1 \leq j\leq n$. Let us show the lemma by induction as follows:\\
				
				Induction Hypothesis: 
				\begin{itemize}
					\item We induct on a variable denoted by $i$.
					\item  Property $P_i$: any $m \times 1$ vector having weight at most 2 with support contained in $\{\sum_{\ell=0}^{i-1}\rho_{\ell}+1,..,\sum_{\ell=0}^{i-1}\rho_{\ell}+\rho_i\}$  can be written as some linear combination of vectors  $\underline{h}_{p_1},...,\underline{h}_{p_\psi}$ for some $\{p_1,...,p_\psi\} \subseteq \{1,...,\sum_{\ell=0}^{i}a_{\ell}\}$ and for some $0 < \psi \leq 2(i+1)$.
					\item Let us assume as induction hypothesis that the property $P_i$ is true and the Lemma \ref{claim1} is true for $A_1,...A_i$, $D_0,...D_i$.
				\end{itemize}
				Initial step $i=0$ and $i=1$:
				\begin{itemize}	
					\item We show that each row of $D_0$ has Hamming weight exactly 1.\\
					Suppose there exists a row of $D_0$ with Hamming weight more than 1; let the support set of the row be $\{i_1,i_2,...\}$. Then the columns $\underline{h}_{i_1},\underline{h}_{i_2}$ of $H_1$ can be linearly combined to give a zero column. This contradicts the fact that $d_{\min} \geq t+1, t > 0$ and $t$ is even.
					Hence, all rows of $D_0$ have Hamming weight exactly 1.
					\item If $t = 2$, then the lemma is already proved. So let $t \geq 4$.		
					\item We show that each column of $A_1$ has Hamming weight exactly 1.\\
					Suppose $j^{th}$ column of $A_1$ $\text{for some } 1 \leq j \leq a_1$ has Hamming weight 2; let the support of the column be $\{j_1,j_2\}$ in $A_1$. Then the column $\underline{h}_{a_0 + j}$ in $H_1$ along with the 2 column vectors of $H_1$ say $\underline{h}_{p_1},\underline{h}_{p_2}$ where $\{p_1,p_2\} \subseteq \{1,...,a_0\}$ where $\underline{h}_{p_1}$ has exactly one non-zero entry in  $j_1^{th}$ co-ordinate and $\underline{h}_{p_2}$ has exactly one non-zero entry $j_2^{th}$ co-ordinate, can be linearly combined to give a zero column again leading to a contradiction on minimum distance. Such columns with one column having only one non-zero entry exactly in $j_1^{th}$ co-ordinate and another column having only one non-zero entry exactly in $j_2^{th}$ co-ordinate with column labels in $\{1,...,a_0\}$ exist due to the 1-weight columns in the matrix $D_0$.
					\item The above argument also shows that any $m \times 1$ vector having Hamming weight at most 2 with support contained in $\{1,..,\rho_0\}$ can be written as some linear combination of at most 2 column vectors of $H_1$ say $\underline{h}_{p_1},...,\underline{h}_{p_\psi}$ for some $\{p_1,...,p_\psi\} \subseteq \{1,...,a_0\}$ ($\psi = 1 \text{ or }2$). Hence Property $P_0$ is true.		
					\item We now show that each row of $D_1$ has Hamming weight exactly 1.
					Suppose $j^{th}$ row of $D_1$ has Hamming weight more than 1; let the support set of the row be $\{l_1,l_2,...,l_z\}$ in $D_1$. Now there is some linear combination of columns $\underline{h}_{a_0+l_1}$ and $\underline{h}_{a_0+l_2}$ in $H_1$ that gives a zero in $(\rho_0+j)^{th}$ coordinate and thus this linear combination has support contained in $\{1,...,\rho_0\}$ with Hamming weight at most 2. Now applying Property $P_0$ on this linear combination implies that there is a non-empty set of at most $4$ linearly dependent columns in $H_1$ leading to a contradiction on minimum distance.		
					\item Now we show that Property $P_{1}$ is true: We have to prove that any $m \times 1$ vector with Hamming weight at most $2$ with support contained in $\{\rho_0+1,..,\rho_0+\rho_1\}$ can be written as linear combination of at most $2(1+1)=4$ column vectors of $H_1$ say $\underline{h}_{p_1},...,\underline{h}_{p_\psi}$ for some $\{p_1,...,p_\psi\} \subseteq \{1,...,\sum_{l=0}^{1}a_l\}$ and $0 < \psi \leq 4$. This can be easily seen using arguments similar to ones presented before.
					Let an $m \times 1$ vector have non-zero entries exactly in coordinates $\rho_0+j_1,\rho_0+j_2$ or $\rho_0+j_1$. Then this vector can be linearly combined with at most 2 column vectors in $H_1$ say $\underline{h}_{y_1},...,\underline{h}_{y_\phi}$ where $\{y_1,...,y_\phi\} \subseteq \{ a_0+1,...,a_0+a_{1} \}$ ($\phi = 1 \text{ or }2$) (2 columns $\underline{h}_{y_1},\underline{h}_{y_2}$ with first and second column having a non-zero entry in co-ordinates $\rho_0+j_1,\rho_0+j_2$ respectively or a column $\underline{h}_{y_1}$ with a non-zero entry in $(\rho_0+j_1)^{th}$ co-ordinate. These columns $\underline{h}_{y_1},\underline{h}_{y_2}$ or $\underline{h}_{y_1}$ exist due to $D_1$. Note that $D_1$ is a matrix with each row and column  having weight exactly 1.) to form a $m \times 1$ vector with Hamming weight at most 2 with support contained in $\{ 1,...,\rho_0 \}$ which in turn can be written as linear combination of at most $2$ column vectors in $H_1$ say $\underline{h}_{z_1},...,\underline{h}_{z_\theta}$ for some $\{z_1,...,z_\theta\} \subseteq \{1,...,a_0\}$ ($\theta = 1$ or $2$) by property $P_0$. Hence the given $m \times 1$ vector is written as linear combination of at most $2(1+1)=4$ column vectors in $H_1$ say $\underline{h}_{w_1},...,\underline{h}_{w_\gamma}$ for some $\{w_1,...,w_\gamma\} \subseteq \{1,...,\sum_{l=0}^{1}a_l\}$ and $0 < \gamma \leq 4$.
					
				\end{itemize}
				Induction step :
				\begin{itemize}
					\item Let us assume by induction hypothesis that Property $P_i$ is true and the Lemma \ref{claim1} is true for $A_1,...A_i$, $D_0,...D_i$ for some $i \leq \frac{t}{2}-2$ and prove the induction hypothesis for $i+1$. For $t=4$, the initial step of induction completes the proof of the Lemma \ref{claim1}. Hence assume $t>4$.
					\item Now we show that each column of $A_{i+1}$ has Hamming weight exactly 1: suppose $j^{th}$ column of $A_{i+1}$ for some $1 \leq j \leq a_{i+1}$ has Hamming weight 2; let the support of the column be $j_1,j_2$ in $A_{i+1}$. It is clear that the corresponding column vector  $\underline{h}_{\sum\limits_{l=0}^{i}a_l+j}$ in $H_1$ is a vector with support contained in $\{\sum_{l=0}^{i-1}\rho_l+1,..,\sum_{l=0}^{i-1}\rho_l+\rho_i\}$ and Hamming weight 2. Now applying Property $P_i$ on this column vector $\underline{h}_{\sum\limits_{l=0}^{i}a_l+j}$ implies that there is a non-empty set of at most $2(i+1)+1$ columns in $H_1$ which are linearly dependent; hence contradicts the minimum distance as $2(i+1)+1 \leq t-1$. Hence each column of $A_{i+1}$ has Hamming weight exactly 1.
					\item Now we show that each row of $D_{i+1}$ has Hamming weight exactly 1: suppose $j^{th}$ row of $D_{i+1}$ has Hamming weight more than 1; let the support set of the row be $\{l_1,...,l_z\}$ in $D_{i+1}$. Now some linear combination of columns $\underline{h}_{\sum\limits_{j=0}^{i}a_j+l_1}$ and $\underline{h}_{\sum\limits_{j=0}^{i}a_j+l_2}$ in $H_1$ will make the resulting vector have a $0$ in $(\sum\limits_{l=0}^{i}\rho_l+j)^{th}$ coordinate and the resulting vector also has Hamming weight at most 2 with support contained in $\{\sum_{l=0}^{i-1}\rho_l+1,..,\sum_{l=0}^{i-1}\rho_l+\rho_i\}$ and hence applying Property $P_i$ on this resulting vector implies that there is a non-empty set of at most $2(i+1)+2$ columns in $H_1$ which are linearly dependent; hence contradicts the minimum distance as $2(i+1)+2 \leq t$; thus proving that each row of $D_{i+1}$ has Hamming weight exactly 1.
					\item Now we show that Property $P_{i+1}$ is true: We have to prove that any $m \times 1$ vector with Hamming weight at most $2$ with support contained in $\{\sum_{l=0}^{i}\rho_l+1,..,\sum_{l=0}^{i}\rho_l+\rho_{i+1}\}$ can be written as linear combination of at most $2(i+2)$ column vectors of $H_1$ say $\underline{h}_{p_1},...,\underline{h}_{p_\psi}$ for some $\{p_1,...,p_\psi\} \subseteq \{1,...,\sum_{l=0}^{i+1}a_l\}$ and $0 < \psi \leq 2(i+2)$. This can be easily seen using arguments similar to ones presented before.
					Let an $m \times 1$ vector have non-zero entries in coordinates $\sum_{l=0}^{i}\rho_l+j_1,\sum_{l=0}^{i}\rho_l+j_2$ or $\sum_{l=0}^{i}\rho_l+j_1$. Then this vector can be linearly combined with at most 2 column vectors in $H_1$ say $\underline{h}_{y_1},...,\underline{h}_{y_\phi}$ where $\{y_1,...,y_\phi\} \subseteq \{ \sum_{l=0}^{i}a_l+1,...,\sum_{l=0}^{i}a_l+a_{i+1} \}$ with $\phi = 1$ or $2$,  (2 columns $\underline{h}_{y_1},\underline{h}_{y_2}$ with first column and second column having a non-zero entry in co-ordinates $\sum_{l=0}^{i}\rho_l+j_1,\sum_{l=0}^{i}\rho_l+j_2$ respectively or a column $\underline{h}_{y_1}$ with a non-zero entry in $(\sum_{l=0}^{i}\rho_l+j_1)^{th}$ co-ordinate. These columns $\underline{h}_{y_1},\underline{h}_{y_2}$ or $\underline{h}_{y_1}$ exist due to $D_{i+1}$. Note that $D_{i+1}$ is a matrix with each row and column having weight exactly $1$.) to form a $m \times 1$ vector with Hamming weight at most 2 with support contained in $\{ \sum_{l=0}^{i-1}\rho_l+1,\sum_{l=0}^{i-1}\rho_l +\rho_i \}$ which in turn can be written as linear combination of at most $2(i+1)$ column vectors in $H_1$ say $\underline{h}_{z_1},...,\underline{h}_{z_\theta}$ for some $\{z_1,...,z_\theta\} \subseteq \{1,...,\sum_{l=0}^{i}a_l\}$ and $0 < \theta \leq 2(i+1)$ by property $P_i$. Hence the given $m \times 1$ vector is written as linear combination of at most $2(i+2)$ column vectors in $H_1$ say $\underline{h}_{w_1},...,\underline{h}_{w_\gamma}$ for some $\{w_1,...,w_\gamma\} \subseteq \{1,...,\sum_{l=0}^{i+1}a_l\}$ and $0 < \gamma \leq 2(i+2)$.
				\end{itemize}
			\end{proof}
			
			By Lemma \ref{claim1}, after permutation of columns of $H_1$ (in \eqref{Hform} or \eqref{Hform_111} depending on $J$) within the columns labeled by the set $\{\sum_{l=0}^{j-1}a_l+1,...\sum_{l=0}^{j-1}a_l+a_{j}\}$ for $0 \leq j \leq J-1$, the matrix $D_j,0 \leq j \leq J-1$ can be assumed to be a diagonal matrix with non-zero entries along the diagonal and hence $\rho_{i}=a_{i}$, $\forall 0 \leq i \leq \frac{t}{2}-1$. 
			
			Since the sum of the column weights of $A_i, 1 \leq i \leq \frac{t}{2}-1$ must equal the sum of the row weights and since each row of $A_i$ for $i \leq J-1$ can have weight atmost $r$ and not $r+1$ due to 1-weight rows in $D_{i-1}$, and since for $\frac{t}{2}-1 \geq i \geq J$, $A_i$ is an empty matrix and we have set $a_i=0$, we obtain:
			\bea
			\text{For } 1 \leq i \leq \frac{t}{2}-1 : & & \notag \\
			\rho_{i-1} r & \geq & a_i, \notag \\
			a_{i-1} r & \geq & a_i. \label{Ineq1}
			\eea
			For some $p \geq 0$,
			\bea
			\sum_{i=0}^{\frac{t}{2}-1} \rho_i + p = \sum_{i=0}^{\frac{t}{2}-1} a_i + p =m.  \label{Ineq2}
			\eea 
			By equating sum of row weights of $C$, with sum of column weights of $C$, we obtain:
			\bea
			2 a_{\frac{t}{2}} & \leq & (a_{\frac{t}{2}-1}+p) (r+1)-a_{\frac{t}{2}-1}. \label{Ineq3}	    
			\eea
			Note that if $C$ is an empty matrix then also the inequality \eqref{Ineq3} is true as we would have set $a_{\frac{t}{2}}=0$. If $J < \frac{t}{2}$ and $C$ a non-empty matrix then the number of rows in $C$ is $p$ with each column of $C$ having weight 2, hence the inequality \eqref{Ineq3} is still true.\\
			Substituting \eqref{Ineq2} in \eqref{Ineq3} we get:
			\bea
			2 a_{\frac{t}{2}} & \leq & (m-\sum_{i=0}^{\frac{t}{2}-2} a_i) (r+1)-(m-\sum_{i=0}^{\frac{t}{2}-2} a_i-p), \notag \\
			2 a_{\frac{t}{2}} & \leq & (m-\sum_{i=0}^{\frac{t}{2}-2} a_i) r+p. \label{Ineq4}
			\eea
			
			By equating sum of row weights of $H_1$, with sum of column weights of $H_1$, we obtain:
			\bea
			m(r+1) & \geq & a_0 + 2(\sum_{i=1}^{\frac{t}{2}} a_i) + 3 (n-\sum_{i=0}^{\frac{t}{2}} a_i), \label{Ineq_star1} 
			\eea
			If $J < \frac{t}{2}$ then $a_i=0$, $\forall J \leq i \leq \frac{t}{2}-1$. If $C$ is an empty matrix then $a_{\frac{t}{2}}=0$. Hence the inequality \eqref{Ineq_star1} is true irrespective of whether $J=\frac{t}{2}$ or $J< \frac{t}{2}$ (even if $C$ is an empty matrix).
			\bea
			\text{From \eqref{Ineq_star1} :} & &  \notag \\
			m(r+1) & \geq & 3n-2a_0 -(\sum_{i=1}^{\frac{t}{2}} a_i). \label{Ineq5} 
			\eea
			
			Our basic inequalities are \eqref{Ineq1},\eqref{Ineq2},\eqref{Ineq3},\eqref{Ineq_star1}. We manipulate these 4 inequalities to derive the bound on rate.
			
			Substituting \eqref{Ineq2} in \eqref{Ineq5} we get:
			\bea
			m(r+1) & \geq & 3n-a_0 -a_{\frac{t}{2}}-(m-p), \notag \\
			m(r+2) & \geq & 3n+p-a_0 -a_{\frac{t}{2}}. \label{Ineq6} 
			\eea
			Substituting \eqref{Ineq4} in \eqref{Ineq6}, we get:
			\bea
			m(r+2) & \geq & 3n+p-a_0 - \left( \frac{(m-\sum_{i=0}^{\frac{t}{2}-2} a_i) r+p}{2} \right), \notag \\
			m(r+2+\frac{r}{2}) & \geq & 3n+\frac{p}{2}-a_0 + \left(\sum_{i=0}^{\frac{t}{2}-2} a_i \right) \frac{r}{2}. \label{Ineq7}
			\eea
			
			From \eqref{Ineq2}, for any $0 \leq j \leq \frac{t}{2}-2$:
			\bea
			a_{\frac{t}{2}-2-j} = m-\sum_{i=0}^{\frac{t}{2}-2-j-1} a_i - \sum_{i=\frac{t}{2}-2-j+1}^{\frac{t}{2}-1} a_i - p. \label{Ineq10}
			\eea
			Subtituting \eqref{Ineq1} for $\frac{t}{2}-2-j+1 \leq i \leq \frac{t}{2}-1$ in \eqref{Ineq10}, we get: 
			\bea
			a_{\frac{t}{2}-2-j} & \geq & m-\sum_{i=0}^{\frac{t}{2}-2-j-1} a_i - \sum_{i=\frac{t}{2}-2-j+1}^{\frac{t}{2}-1} a_{\frac{t}{2}-2-j} r^{i-(\frac{t}{2}-2-j)} - p, \notag
			\eea
			\bea
			a_{\frac{t}{2}-2-j} & \geq & m-\sum_{i=0}^{\frac{t}{2}-2-j-1} a_i - \sum_{i=1}^{j+1} a_{\frac{t}{2}-2-j} r^{i} - p, \notag
			\eea
			\bea
			a_{\frac{t}{2}-2-j} & \geq & \frac{m-\sum_{i=0}^{\frac{t}{2}-2-j-1} a_i-p}{ 1+ \sum_{i=1}^{j+1} r^{i} }. \label{Ineq11}
			\eea
			Let,
			\bea
			\delta_0 & = & \frac{r}{2}, \label{Ineq12} \\
			\text{For $j_1 \geq 0$: }\delta_{j_1+1} & = & \delta_{j_1} - \frac{\delta_{j_1}}{1+\sum_{i=1}^{j_1+1}r^i}. \label{Ineq13}
			\eea
			Let us prove the following inequality by induction for $0 \leq J_1 \leq \frac{t}{2}-2$,
			\bea
			m(r+2+\delta_{J_1}) & \geq & 3n + p\left( \frac{1}{2}+\delta_{J_1} - \frac{r}{2}\right)-a_0  + \left(\sum_{i=0}^{\frac{t}{2}-2-J_1} a_i \right) \delta_{J_1}. \label{Ineq9}
			\eea
			\eqref{Ineq9} is true for $J_1=0$ by \eqref{Ineq7}. Hence \eqref{Ineq9} is proved for $t=4$ and the range of $J_1$ is vacuous for $t=2$. Hence assume $t>4$. Hence let us assume \eqref{Ineq9} is true for $ J_1$ such that $ \frac{t}{2}-3 \geq J_1 \geq 0$ and prove it for $J_1+1$. 
			Substituting \eqref{Ineq11} for $j=J_1$ in \eqref{Ineq9}, we get:
			\bea
			m(r+2+\delta_{J_1}) & \geq & 3n + p\left( \frac{1}{2}+\delta_{J_1} - \frac{r}{2}\right)-a_0 \\
			& +& \left(\sum_{i=0}^{\frac{t}{2}-2-J_1-1} a_i \right) \delta_{J_1} + \left( \frac{m-\sum_{i=0}^{\frac{t}{2}-2-J_1-1} a_i-p}{ 1+ \sum_{i=1}^{J_1+1} r^{i} } \right) \delta_{J_1}, \notag \\
			m(r+2+\delta_{J_1}-\frac{\delta_{J_1}}{1+ \sum_{i=1}^{J_1+1} r^{i} }) & \geq & 
			3n + p\left( \frac{1}{2}+\delta_{J_1} -\frac{\delta_{J_1}}{1+ \sum_{i=1}^{J_1+1} r^{i} }- \frac{r}{2}\right)-a_0  + \notag \\ 
			& & \left(\sum_{i=0}^{\frac{t}{2}-2-J_1-1} a_i \right) \left(\delta_{J_1} - \frac{\delta_{J_1}}{1+ \sum_{i=1}^{J_1+1} r^{i} }\right). \label{Ineq141}
			\eea
			Substituing \eqref{Ineq13} in \eqref{Ineq141}, we obtain
			\bea
			m(r+2+\delta_{J_1+1}) & \geq & 
			3n + p\left( \frac{1}{2}+\delta_{J_1+1}- \frac{r}{2}\right)-a_0  + \left(\sum_{i=0}^{\frac{t}{2}-2-J_1-1} a_i \right) \delta_{J_1+1}. \label{Ineq14}
			\eea
			Hence \eqref{Ineq9} is proved  for any $0 \leq J_1 \leq \frac{t}{2}-2$ for $t \geq 4$. Hence writing \eqref{Ineq9} for $J_1 = \frac{t}{2}-2$ for $t\geq 4$, we obtain:
			\bea
			m(r+2+\delta_{\frac{t}{2}-2}) & \geq & 3n + p\left( \frac{1}{2}+\delta_{\frac{t}{2}-2} - \frac{r}{2}\right)-a_0  + (a_0) \delta_{\frac{t}{2}-2}, \notag \\
			m(r+2+\delta_{\frac{t}{2}-2}) & \geq & 3n + p\left( \frac{1}{2}+\delta_{\frac{t}{2}-2} - \frac{r}{2}\right)+ a_0 (\delta_{\frac{t}{2}-2}-1). \label{Ineq15} 
			\eea
			It can be seen that $\delta_{j_1}$ for $r \geq 2$ has a product form as:
			\bea
			\delta_{j_1} = \frac{r}{2}\left(\frac{r^{j_1+1}-r^{j_1}}{r^{j_1+1}-1}\right). \label{Ineq16}
			\eea
			Hence for $r \geq 3$, $t \geq 4$:
			\bean
			\delta_{\frac{t}{2}-2}=\frac{r}{2}\left(\frac{r^{\frac{t}{2}-1}-r^{\frac{t}{2}-2}}{r^{\frac{t}{2}-1}-1}\right) > 1.
			\eean
			Hence we can substitute \eqref{Ineq11} for $j=\frac{t}{2}-2$ in \eqref{Ineq15} :
			\bea
			m(r+2+\delta_{\frac{t}{2}-2}) & \geq & 3n + p\left( \frac{1}{2}+\delta_{\frac{t}{2}-2} - \frac{r}{2}\right) \\ 
			&+& \left( \frac{m-p}{1+\sum_{i=1}^{\frac{t}{2}-1}r^i} \right) (\delta_{\frac{t}{2}-2}-1), \notag \\
			\scalemath{0.7}{m(r+2+\delta_{\frac{t}{2}-2}-\frac{\delta_{\frac{t}{2}-2}}{1+\sum_{i=1}^{\frac{t}{2}-1}r^i}+\frac{1}{1+\sum_{i=1}^{\frac{t}{2}-1}r^i}) }
			& \geq & \scalemath{0.7}{3n +  p\left( \frac{1}{2}+\delta_{\frac{t}{2}-2}-\frac{\delta_{\frac{t}{2}-2}}{1+\sum_{i=1}^{\frac{t}{2}-1}r^i} + \frac{1}{1+\sum_{i=1}^{\frac{t}{2}-1}r^i}- \frac{r}{2}\right).} \label{Ineq17}
			\eea
			Substituting \eqref{Ineq13} in \eqref{Ineq17}, we obtain:
			\bea
			m\left(r+2+\delta_{\frac{t}{2}-1}+ \frac{1}{1+\sum_{i=1}^{\frac{t}{2}-1}r^i}\right) & \geq & 3n + p\left( \frac{1}{2}+\delta_{\frac{t}{2}-1}+ \frac{1}{1+\sum_{i=1}^{\frac{t}{2}-1}r^i}- \frac{r}{2}\right). \label{Ineq18}
			\eea
			Using \eqref{Ineq16}, we obtain:
			\bean
			\left( \frac{1}{2}+\delta_{\frac{t}{2}-1}+ \frac{1}{1+\sum_{i=1}^{\frac{t}{2}-1}r^i}- \frac{r}{2}\right) > 0.
			\eean
			Hence \eqref{Ineq18} implies:
			\bea
			m\left(r+2+\delta_{\frac{t}{2}-1}+ \frac{1}{1+\sum_{i=1}^{\frac{t}{2}-1}r^i}\right) & \geq & 3n.  \label{Ineq19}
			\eea
			\eqref{Ineq19} after some algebraic manipulations gives the required upper bound on $1-\frac{m}{n}$ and hence gives the required upper bound on $\frac{k}{n}$ as stated in the theorem. Note that although the derivation is valid for $r\geq 3$, $t \geq 4$, the final bound given in the theorem is correct and tight for $t =2$. The upper bound on rate for $t=2$ can be derived specifically by substituting $a_0 \leq m$ in \eqref{Ineq7} and noting that $p \geq 0$.\\
			
			\section{case ii: $t$ an odd integer}
			Again it can be seen that the matrix $H_1$ after permutation of rows and columns can be written in the form given in \eqref{Hform2}.
			\bea
			H_1 & = & \left[
			\begin{array}{c|c|c|c|c|c|c|c|c}
				D_0 & A_1 & 0 & 0 & \hdots & 0 & 0 & 0 & \\
				\cline{1-8}
				0 & D_1 & A_2 & 0 & \hdots & 0 & 0 & 0 & \\
				\cline{1-8}
				0 & 0 & D_2 & A_3 & \hdots & 0 & 0 & 0 & \\
				\cline{1-8}
				0 & 0 & 0 & D_3 & \hdots & 0 & 0 & 0 &  \\
				\cline{1-8}
				\vdots & \vdots & \vdots & \vdots & \ddots & \vdots & \vdots & \vdots & D\\
				\cline{1-8}
				0 & 0 & 0 & 0 & \hdots & A_{s-1} & 0 & 0 & \\
				\cline{1-8}
				0 & 0 & 0 & 0 & \hdots & D_{s-1} & A_{s} & 0 & \\
				\cline{1-8}
				0 & 0 & 0 & 0 & \hdots & 0 & D_{s} &  & \\
				\cline{1-7}
				0 & 0 & 0 & 0 & \hdots & 0 & 0 & C & \\
			\end{array} \right] \label{Hform2},
			\eea
			where 
			\begin{enumerate}
				\item Rows are labeled by the integers $1,...,m$  with top most row of $H_1$ labelled $1$ and next row labeled $2$ and so on. Similarly columns are labeled by the integers $1,...,n$ with leftmost column of $H_1$ labeled $1$ and next column labeled $2$ and so on, 
				\item $A_i$ is a $\rho_{i-1} \times a_i$ matrix for $1\leq i \leq s$, $D_i$ is a $\rho_{i} \times a_i$ matrix for $0 \leq i \leq s$ for some $\{\rho_j\},\{a_j\}$. Note that any of $a_i$, $\rho_i$ is allowed to take the value 0 for any $i$,
				\item $D_0$ is a matrix with each column having weight 1 and each row having weight at least 1. The first $a_0$ columns of $H_1$ contains the columns of $D_0$. The set of first $a_0$ columns of $H_1$ is equal to the set of all those columns of $H_1$ which has weight $1$,
				\item  $\{A_j\}$,$\{D_j\}$,$\{B_j= \left[\frac{A_j}{D_j}\right]\}$ are  such that for $1\leq i \leq s$, each column of $B_i$ has weight 2, each column of $A_i$ has weight at least 1 and each row of $D_i$ has weight at least 1 and each column of $D_i$ has weight at most 1,
				\item $C$ is a matrix with each column having weight 2. The set of columns of the matrix $D$ is equal to the set of all those columns of $H_1$ which has weight $\geq 3$.
			\end{enumerate}
			Let $J = \min{\{\{j: 1 \leq j \leq s, A_j,D_j \text{ are empty matrices }\} \cup \{ s+1 \} \}}$. We set $J=0$, if $D_0$ is an empty matrix. We now redefine $H_1$ based on the value of $J$.
			We set $A_i$, $D_i$ to be empty matrices and set $a_i=0, \rho_i=0$, $\forall J \leq i \leq s$. 
			Let $E_2  \subseteq \{\sum_{j=0}^{J-1}a_j+1,...,n\}$ such that $E_2$ is the set of labels of all the 2-weight columns of $H_1$ apart from those 2-weight columns of $H_1$ containing the columns of $B_1,...,B_{J-1}$. Let $E_1 = \{\sum_{j=0}^{J-1}{\rho_j}+1,...,m\}$. If $J=s+1$ then $H_1$ is defined by \eqref{Hform2}. If $J < s+1$, we redefine $C=H_1|_{E_1,E_2}$.   If $J < s+1$, the matrix $H_1$ can be written in the form given in \eqref{Hform_11} and hence defined by \eqref{Hform_11}.
			\bea
			H_1 = \left[
			\scalemath{1}{
				\begin{array}{c|c|c|c|c|c|c|c|c}
					D_0 & A_1 & 0 & 0 & \hdots & 0 & 0 & 0 & \\
					\cline{1-8}
					0 & D_1 & A_2 & 0 & \hdots & 0 & 0 & 0 & \\
					\cline{1-8}
					0 & 0 & D_2 & A_3 & \hdots & 0 & 0 & 0 & \\
					\cline{1-8}
					0 & 0 & 0 & D_3 & \hdots & 0 & 0 & 0 &  \\
					\cline{1-8}
					\vdots & \vdots & \vdots & \vdots & \hdots & \vdots & \vdots & \vdots & D\\
					\cline{1-8}
					0 & 0 & 0 & 0 & \hdots & A_{J-2} & 0 & 0 & \\
					\cline{1-8}
					0 & 0 & 0 & 0 & \hdots & D_{J-2} & A_{J-1} & 0 & \\
					\cline{1-8}
					0 & 0 & 0 & 0 & \hdots & 0 & D_{J-1} & 0 & \\
					\cline{1-8}
					0 & 0 & 0 & 0 & \hdots & 0 & 0 & C & \\
				\end{array} 
			}
			\right], \ \ \label{Hform_11}
			\eea
			Irrespective of the value of $J$, let the number of columns in $C$ be denoted as $a_{s+1}$. If $C$ is an empty matrix then we can clearly set $a_{s+1}=0$.
			The entire derivation of upper bound on rate is correct and all the inequalities in the derivation will hold with $a_i=0$, $\rho_i=0$, $\forall J \leq i \leq s$. Let $S=\{j: 1 \leq j \leq J-1, D_{j} \text{ is an empty matrix}, A_{j} \text{ is a non-empty matrix} \} $. Let $\ell = \min{\{ S \cup \{J\} \}}$. If $\ell < J$, then the proof of the following Lemma \ref{claim2} (since the proof of Lemma \ref{claim2} proceeds by induction starting with the proof of the lemma for $D_0$ first and then proceeding to $A_1$ and then $D_1$ and so on and hence we prove the lemma for $A_{\ell}$ first and then proceed to $D_{\ell}$ and since $D_0,A_j,D_j$ must be non-empty matrices $\forall 1 \leq j \leq \ell-1$) will imply that each column of $A_{\ell}$ has weight 1 which will imply that $D_{\ell}$ cannot be an empty matrix. Hence the case $A_{\ell}$, a non-empty matrix and $D_{\ell}$, an empty matrix cannot occur.
			Although we have to prove the following Lemma \ref{claim2} for $A_i,D_i$, $\forall 1 \leq i \leq J-1$, $D_0$, we assume all $D_0,$ $A_i,D_i$, $\forall 1 \leq i \leq s$ to be non-empty matrices and prove the lemma. Since the proof of the lemma is by induction, the induction can be made to stop after proving the lemma for $A_{J-1},D_{J-1}$ (induction starts by proving the lemma for $D_0$ as mentioned before) and the proof is unaffected by it.
			\begin{lem}\label{claim2}
				
				For $1\leq i \leq s$, $A_i$ is a matrix with each column having weight 1.
				For $0 \leq i \leq s-1$, $D_i$ is a matrix with each row and each column having weight 1. $D_{s}$ is a matrix with each column having weight 1.

			\end{lem}
			\begin{proof}
				Proof is exactly similar to the proof of Lemma \ref{claim1} and proceeds by induction. So we skip the proof.
			\end{proof}
			
			By Lemma \ref{claim2}, after permutation of columns of $H_1$ (in \eqref{Hform2} or \eqref{Hform_11} depending on $J$)  within the columns labeled by the set $\{\sum_{l=0}^{j-1}a_l+1,...\sum_{l=0}^{j-1}a_l+a_{j}\}$ for $0 \leq j \leq \min(J-1,s-1)$, the matrix $D_j,0 \leq j \leq \min(J-1,s-1)$ can be assumed to be a diagonal matrix with non-zero entries along the diagonal and hence $\rho_{i}=a_{i}$, for $0 \leq i \leq s-1$. 
			
			Since the sum of the column weights of $A_i, 1 \leq i \leq s$ must equal the sum of the row weights and since each row of $A_i$ for $i \leq J-1$ can have weight atmost $r$ and not $r+1$ due to weight one rows in $D_{i-1}$, and since for $J \leq i \leq s$, $A_i$ is an empty matrix and we have set $a_i=0$, we obtain:
			\bea
			\text{For } 1 \leq i \leq s : & & \notag \\
			\rho_{i-1} r & \geq & a_i, \notag \\
			a_{i-1} r & \geq & a_i. \label{Ineq1O}
			\eea
			For some $p \geq 0$,
			\bea
			\rho_{s}+\sum_{i=0}^{s-1} \rho_i + p = \rho_{s}+\sum_{i=0}^{s-1} a_i + p =m.  \label{Ineq2O}
			\eea 
			By equating sum of row weights of $[\frac{D_{s}}{0} | C]$, with sum of column weights of $[\frac{D_{s}}{0} | C]$, we obtain:
			\bea
			2 a_{s+1}+a_{s} & \leq & (\rho_{s}+p) (r+1). \label{Ineq3O}	    
			\eea
			If $J \leq s$ then the number of rows in $C$ is $p$ with each column of $C$ having weight 2 and $a_{s}=\rho_{s}=0$ (and $a_{s+1}=0$ if $C$ is also an empty matrix), hence the inequality \eqref{Ineq3O} is true. If $J=s+1$ and $C$ is an empty matrix then also the inequality \eqref{Ineq3O} is true as we would have set $a_{s+1}=0$.
			
			Substituting \eqref{Ineq2O} in \eqref{Ineq3O}:
			\bea
			2 a_{s+1} & \leq & (m-\sum_{i=0}^{s-1} a_i) (r+1)-a_{s}. \label{Ineq4O}
			\eea
			By equating sum of row weights of $D_{s}$, with sum of column weights of $D_{s}$, we obtain (Note that if $D_{s}$ is an empty matrix then also the following inequality is true as we would have set $a_{s}=0$):
			\bea
			a_{s} \leq \rho_{s}(r+1). \label{OddSpe}
			\eea
			By equating sum of row weights of $H_1$, with sum of column weights of $H_1$, we obtain
			\bea
			m(r+1) & \geq & a_0 + 2(\sum_{i=1}^{s+1} a_i) + 3 (n-\sum_{i=0}^{s+1} a_i), \label{Ineq_star2}
			\eea 
			If $J \leq s$ then $a_i=0$ $\forall J \leq i \leq s$. If $C$ is an empty matrix then $a_{s+1}=0$. Hence the inequality \eqref{Ineq_star2} is true irrespective of whether $J=s+1$ or $J \leq s$ (even if $C$ is an empty matrix).
			\bea
			m(r+1) & \geq & 3n-2a_0 -(\sum_{i=1}^{s+1} a_i). \label{Ineq5O} 
			\eea
			
			Our basic inequalities are \eqref{Ineq1O},\eqref{Ineq2O},\eqref{Ineq3O},\eqref{OddSpe},\eqref{Ineq_star2}. We manipulate these 5 inequalities to derive the bound on rate.
			
			Substituting \eqref{Ineq4O} in \eqref{Ineq5O}:
			\bea
			m(r+1) & \geq & 3n-2a_0 -(\sum_{i=1}^{s} a_i) - \left(\frac{(m-\sum_{i=0}^{s-1} a_i) (r+1)-a_{s}}{2} \right). \label{Ineq6O}
			\eea
			For $s=0$, \eqref{Ineq6O} becomes:
			\bea
			m(r+1) & \geq & 3n-2a_0 - \left(\frac{m (r+1)-a_{0}}{2} \right), \notag \\        m\frac{3(r+1)}{2} & \geq & 3n-\frac{3}{2}a_0. \label{Ineq7O} 
			\eea
			Substituting  \eqref{OddSpe} in \eqref{Ineq7O}:
			\bea
			m\frac{3(r+1)}{2} & \geq & 3n-\frac{3}{2}\rho_0 (r+1), \label{Ineq8O} 
			\eea
			Substituting  \eqref{Ineq2O} in \eqref{Ineq8O}:
			\bea
			m\frac{3(r+1)}{2} & \geq & 3n-\frac{3}{2}(m-p) (r+1), \notag \\
			\text{Since $p \geq 0$, } 3m(r+1) & \geq & 3n. \label{Ineq9O}
			\eea
			\eqref{Ineq9O} implies,
			\bea
			\frac{k}{n} \leq \frac{r}{r+1}. \label{Ineq10O}
			\eea
			\eqref{Ineq10O} proves the bound \eqref{Thm2} for $s=0$. Hence from now on we assume $s \geq 1$.\\
			For $s \geq 1$, \eqref{Ineq6O} implies:
			\bea
			m \frac{3(r+1)}{2} & \geq & 3n+ a_0 \left( \frac{r+1}{2}-2 \right) +(\sum_{i=1}^{s-1} a_i) \left( \frac{r+1}{2}-1 \right) - \frac{a_{s}}{2}. \label{Ineq11O} 
			\eea
			Substituting  \eqref{Ineq1O} in \eqref{Ineq11O} and since $r \geq 3$:    
			\bea
			m \frac{3(r+1)}{2} & \geq & 3n+ \frac{a_{s}}{r^{s}} \left( \frac{r+1}{2}-2 \right) +(\sum_{i=1}^{s-1} \frac{a_{s}}{r^{s-i}}) \left( \frac{r+1}{2}-1 \right) - \frac{a_{s}}{2},  \notag \\
			m \frac{3(r+1)}{2} & \geq & 3n+ a_{s} \left( \left(\sum_{i=1}^{s} \frac{1}{r^i} \right) \left(\frac{r+1}{2}-1\right) -\frac{1}{r^{s}}-\frac{1}{2} \right),  \notag \\
			m \frac{3(r+1)}{2} & \geq & 3n- a_{s} \left(\frac{3}{2 r^{s}}\right). \label{Ineq12O}
			\eea
			Rewriting \eqref{Ineq2O}:
			\bea
			\rho_{s}+\sum_{i=0}^{s-1} a_i + p =m.  \label{Ineq13O}
			\eea
			Substituting \eqref{OddSpe},\eqref{Ineq1O} in \eqref{Ineq13O}:
			\bea
			\rho_{s}+\sum_{i=0}^{s-1} a_i + p &=& m, \notag \\
			\frac{a_{s}}{r+1}+\sum_{i=0}^{s-1} \frac{a_{s}}{r^{s-i}} & \leq & m-p, \notag \\
			a_{s} \leq \frac{m-p}{\frac{1}{r+1}+\sum_{i=1}^{s} \frac{1}{r^i}}, \notag \\
			a_{s} \leq \frac{(m-p)(r+1)}{1+\frac{(r^{s}-1)(r+1)}{(r^{s})(r-1)} }. \label{Ineq14O}
			\eea
			Substituting \eqref{Ineq14O} in \eqref{Ineq12O}:
			\bea
			m \frac{3(r+1)}{2} & \geq & 3n- \frac{(m-p)(r+1)}{1+\frac{(r^{s}-1)(r+1)}{(r^{s})(r-1)} } \left(\frac{3}{2 r^{s}}\right), \notag \\
			\text{Since $p \geq 0$, } m \frac{3(r+1)}{2} \left( 1+\frac{1}{r^{s}+\frac{(r^{s}-1)(r+1)}{(r-1)}} \right) & \geq & 3n. \label{Ineq15O} 
			\eea
			\eqref{Ineq15O} after some algebraic manipulations gives the required upper bound on $1-\frac{m}{n}$ and hence gives the required upper bound on $\frac{k}{n}$ as stated in the theorem. \\
		\end{proof}

	\end{appendices}

\chapter{Optimal Constructions of Sequential LR Codes} \label{ch:SeqConstr}
In this chapter, a construction of rate-optimal S-LR codes is presented for any $r,t$ with $r \geq 3$.  The construction presented here will establish the tightness of the upper bound on rate derived in Chapter~\ref{ch:SeqBound}.    
The starting point of the construction is the staircase form of the parity-check (p-c) matrix $H$ of a rate-optimal S-LR code derived in Chapter \ref{ch:SeqBound}, equations  \eqref{eq:Hmatrixteven_ch3}, \eqref{eq:Hmatrixtodd_ch3}.   It will be shown that this forces the code to possess a certain, simple, graphical representation ${\cal G}$.  In ${\cal G}$, the edges represent code symbols and nodes, the parity checks.  While the constraints imposed by the staircase structure of the parity-check matrix $H$ are necessary for a rate-optimal S-LR code, they are not sufficient to guarantee (sequential) recovery from $t$ erasures.  It turns out that however, that adding the further constraint that ${\cal G}$ have girth $\geq (t+1)$, leads to a necessary and sufficient condition.   Section~\ref{sec:graphical_rep} presents the graphical representation ${\cal G}$ dictated by the staircase form of the p-c matrix. The section following, Section~\ref{sec:complete_code} shows how to construct this graph ${\cal G}$ that satisfies in addition, the girth requirement, and this completes the construction of rate-optimal S-LR codes since the code follows from the graph.  Section~\ref{sec:Moore_ch4} discusses the construction of S-LR codes that are not only rate-optimal, but which are also optimal in terms of having the shortest possible block length. This leads to a class of graphs known as Moore graphs. The final section present a summary. 

Throughout this chapter whenever we refer to a code, it refers to an S-LR code with parameters $(r,t)$ defined in Chapter \ref{ch:SeqBound}. We use the term nodes or vertices to indicate the vertices of a graph. We use the term vertex set to denote the set of all vertices in a graph and the term edge set to denote the set of all edges in a graph. The vertex set of a graph $G$ is denoted by $V(G)$. Throughout this thesis, for a code $\mathcal{C}$, a parity check refers to a codeword $\underline{c}$ in the dual code $\mathcal{C}^{\perp}$ or to the equation $ G\underline{c}^T = \underline{0}$ for a generator matrix $G$ of $\mathcal{C}$.
\section{A Graphical Representation for the Rate-Optimal Code}  \label{sec:graphical_rep} 

We show in this section, that the staircase form of the parity-check (p-c) matrix forced on a rate-optimal code (Chapter \ref{ch:SeqBound}, equations  \eqref{eq:Hmatrixteven_ch3}, \eqref{eq:Hmatrixtodd_ch3}) leads to a tree-like graphical representation of the code.
The structure is slightly different for $t$ odd and $t$ even.  We begin with the $t$-even case. 
\subsection{$t$ Even case} 

In the case $t$ even, the p-c matrix of a rate-optimal code can be put into the form (Chapter \ref{ch:SeqBound}, equation  \eqref{eq:Hmatrixteven_ch3}): 
\bea
\heven = \left[
\begin{array}{c|c|c|c|c|c|c|c}
	D_0 & A_1 & 0 & 0 & \hdots & 0 & 0 & 0  \\
	\cline{1-8}
	0 & D_1 & A_2 & 0 & \hdots & 0 & 0 & 0 \\
	\cline{1-8}
	0 & 0 & D_2 & A_3 & \hdots & 0 & 0 & 0  \\
	\cline{1-8}
	0 & 0 & 0 & D_3 & \hdots & 0 & 0 & 0 \\
	\cline{1-8}
	\vdots & \vdots & \vdots & \vdots & \hdots & \vdots & \vdots & \vdots \\
	\cline{1-8}
	0 & 0 & 0 & 0 & \hdots & A_{s-1} & 0 & 0  \\
	\cline{1-8}
	0 & 0 & 0 & 0 & \hdots & D_{s-1} & A_{s} & 0 \\
	\cline{1-8}
	0 & 0 & 0 & 0 & \hdots & 0 & D_{s} & C \\
\end{array} 
\right],  \label{eq:H_Eq_ch4}
\eea
where $s= \lfloor \frac{t-1}{2} \rfloor$, or equivalently, $t=2s+2$. 
Our goal is to show that \heven\ forces the code to have a certain graphical representation.  We note first that, each column in \heven, apart from the columns associated to diagonal sub-matrix $D_0$ have (Hamming) weight $2$.  To bring in symmetry, we add an additional row to \heven\ at the top, which has all $1$s in the columns associated to $D_0$ and zeros elsewhere to obtain the matrix \hevena, shown in \eqref{eq:stair_H_even_aug_ch4}.  
\bea
\hevena & = & 
\begin{array}{c} \ \\  V_{\infty} \\ V_0 \\ V_1 \\ V_2 \\ V_3 \\ \vdots \\ V_{s-2} \\ V_{s-1} \\ V_{s} \end{array} 
\left[
\begin{array}{c|c|c|c|c|c|c|c}
	E_0 & E_1 & E_2 & E_3 & \cdots & E_{s-1} & E_{s} & E_{s+1} \\ \hline \hline 
	\underline{1}^t & \underline{0}^t & \underline{0}^t & \hdots  & \hdots & \hdots & \underline{0}^t & \underline{0}^t  \\
	\cline{1-8}
	D_0 & A_1 & 0 & 0 & \hdots & 0 & 0 & 0  \\
	\cline{1-8}
	0 & D_1 & A_2 & 0 & \hdots & 0 & 0 & 0 \\
	\cline{1-8}
	0 & 0 & D_2 & A_3 & \hdots & 0 & 0 & 0  \\
	\cline{1-8}
	0 & 0 & 0 & D_3 & \hdots & 0 & 0 & 0 \\
	\cline{1-8}
	\vdots & \vdots & \vdots & \vdots & \ddots & \vdots & \vdots & \vdots \\
	\cline{1-8}
	0 & 0 & 0 & 0 & \hdots & A_{s-1} & 0 & 0  \\
	\cline{1-8}
	0 & 0 & 0 & 0 & \hdots & D_{s-1} & A_{s} & 0 \\
	\cline{1-8}
	0 & 0 & 0 & 0 & \hdots & 0 & D_{s} & C \\
\end{array} 
\right].  \label{eq:stair_H_even_aug_ch4}
\eea
Since each column of \hevena\ has weight $2$, the matrix has a natural interpretation as the node-edge incidence matrix of a graph where the incidence matrix of the graph is obtained by replacing every non-zero entry of \hevena\ with $1$. Hence the vertices of the graph are in one to one correspondence with the rows of the matrix \hevena\ and edges of the graph are in one to one correspondence with the columns of the matrix \hevena\ where an edge corresponding to a column with non-zero entries at rows $i,j$ connects the nodes corresponding to the rows $i,j$ of \hevena.\ The nodes of the graph corresponding to those rows of \hevena\ containing the rows of $D_i$ will be denoted by \vi\ and 
similarly, the edges corresponding to those columns of \hevena\ containing the columns of $D_j$ will be denoted by $E_j$. The edges associated with those columns of \hevena\ containing the columns of $C$ will be denoted by $E_{s+1}$.  We use $V_{\infty}$ to denote the node associated with the row at the very top of \hevena\ (see \eqref{eq:stair_H_even_aug_ch4}). Each node except the node $V_{\infty}$ has degree $(r+1)$.

\subsubsection{Unravelling the Graph}

We can unravel the structure of this graph as follows. Fig.~\ref{fig:tree_graph_t_even} shows the graph for the case $(r+1)=4, t=6,s=2$ with $|V_0|=3$.  Since each row of \hevena, apart from the top row, has weight $(r+1)$, it follows that in the resultant graph, every node except the node $V_{\infty}$ has degree $(r+1)$.  Node \vinfty\ has degree $a_0= |V_0|$, since $D_0$ is a diagonal matrix.  
The $a_0$ edges originating from \vinfty\ are terminated in the $a_0$ nodes making up \vzero. We will use $E_{0}$ to denote this collection of edges. There are $r$ other edges that emanate from each node in \vzero, each of these edges is terminated at a distinct node in \vone. We use \eone\ to denote this collection of edges.  Each of the other $r$ edges that emanate from each node in \vone, terminate in a distinct node in \vtwo.  We use \etwo\ to denote this collection of edges.  We continue in this fashion, until we reach the nodes in \vs\ via edge-set \es.  Here, the $r$ other edges outgoing from each node in \vs\ are terminated among themselves.  We use \espone\ to denote this last collection of edges.  As can be seen, the graph has a tree-like structure, except for the edges (corresponding to edge-set \espone) linking the leaf nodes $V_s$ at the very bottom. 

We use \gzero\ to denote the overall graph and use \gi\ to denote the restriction of \gzero\ to node-set $V_s \cup V_{s-1}  \cdots \cup V_i$ i.e., \gi\ denotes the subgraph of \gzero\ induced by the nodes $V_s \cup V_{s-1}  \cdots \cup V_i$ for $1 \leq i \leq s$. Thus the graphs are nested:
\bean
\gs \subseteq {\cal G}_{s-1} \subseteq \cdots \subseteq \gtwo\ \subseteq \gone\ \subseteq \gzero.
\eean
Fig.~\ref{fig:tree_graph_t_even} identifies the graphs $\gtwo\ \subseteq \gone\ \subseteq \gzero$ for the case $t=6$. 

\subsubsection{Connecting the Graph to the Code} 

Each node in the graph \gzero\ is associated to a row of the p-c matrix \heven\ of the code and hence to a p-c. The one exception is the fictitious node $V_{\infty}$ which does not correspond to a p-c. When the p-c matrix \heven\ is over $\mathbb{F}_2$, the top row of \hevena\ is the binary sum of all other rows of \hevena.\ Hence in case when \heven\ is over $\mathbb{F}_2$, \hevena\ also represent a p-c matrix of the code and the node $V_{\infty}$ is also associated to a p-c. Also, each edge in the graph \gzero\ is associated to a unique code symbol as the edges are associated with columns of the matrix \hevena.\
The structure of \gzero\ is mostly determined once we specify $a_0$ as \gzero\ with edges in \espone\ removed is just a tree with $V_{\infty}$ as the root node. Hence the only freedom lies in selecting the pairs of nodes in node-set \vs\ that are linked by the edges in edge-set \espone.  The p-c matrix requires however, that these edges be selected such that each node in \vs\ is of degree $(r+1)$. 

\begin{figure}[ht]
	\centering
	\includegraphics[scale = 0.5]{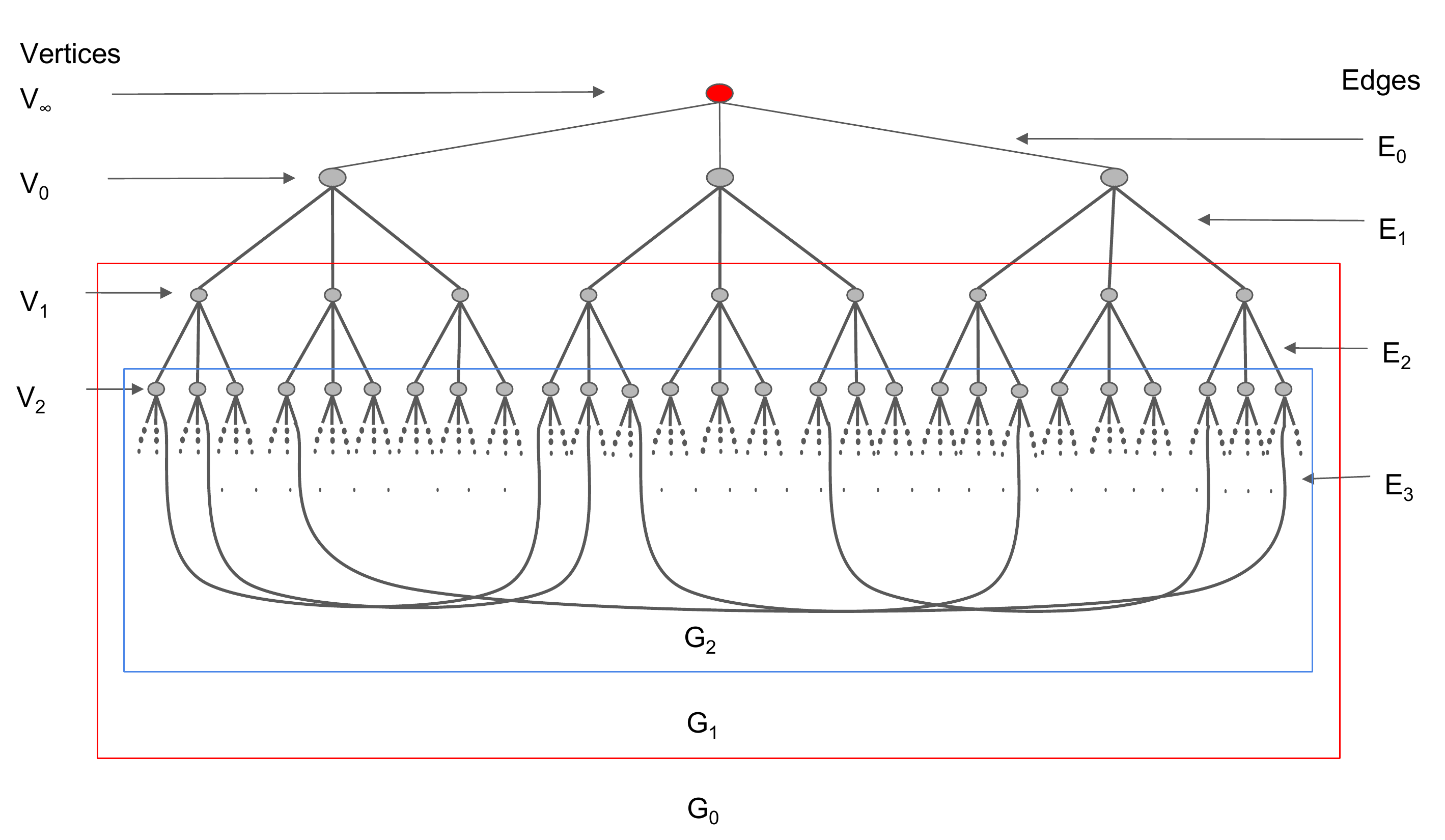}
	\caption{Graphical representation induced by the staircase p-c matrix \eqref{eq:stair_H_even_aug_ch4} for the case $(r+1)=4, \ t=6, \ s=2$ with $|V_0|=3$.}
	\label{fig:tree_graph_t_even}
\end{figure}

\subsection{$t$ Odd Case} 

In the case $t$ odd, the p-c matrix of a rate-optimal code can be put into the form (Chapter \ref{ch:SeqBound}, equation \eqref{eq:Hmatrixtodd_ch3}): 
\bea
\hodd & = & \left[
\begin{array}{c|c|c|c|c|c|c}
	D_0 & A_1 & 0 & 0 & \hdots & 0 & 0   \\
	\cline{1-7}
	0 & D_1 & A_2 & 0 & \hdots & 0 & 0  \\
	\cline{1-7}
	0 & 0 & D_2 & A_3 & \hdots & 0 & 0   \\
	\cline{1-7}
	0 & 0 & 0 & D_3 & \hdots & 0 & 0   \\
	\cline{1-7}
	\vdots & \vdots & \vdots & \vdots & \ddots & \vdots & \vdots  \\
	\cline{1-7}
	0 & 0 & 0 & 0 & \hdots & A_{s-1} & 0   \\
	\cline{1-7}
	0 & 0 & 0 & 0 & \hdots & D_{s-1} & A_{s}   \\
	\cline{1-7}
	0 & 0 & 0 & 0 & \hdots & 0 & D_{s} 
\end{array} \right] \label{eq:stair_H_odd_ch4},
\eea
where $s= \lfloor \frac{t-1}{2} \rfloor$, or equivalently, $t=2s+1$. 
Our goal once again, is to show that \hodd\ forces the code to have a certain graphical representation.  We add here as well, an additional row to \hodd\ at the top, which has all $1$s in the columns associated to $D_0$ and zeros elsewhere to obtain the matrix \hodda, shown in \eqref{eq:stair_H_odd_aug_ch4}.  
\bea
\hodda = 
\begin{array}{c} \ \\  V_{\infty} \\ V_0 \\ V_1 \\ V_2 \\ V_3 \\ \vdots \\ V_{s-2} \\ V_{s-1} \\ V_{s} \end{array} 
\left[
\begin{array}{c|c|c|c|c|c|c}
	E_0 & E_1 & E_2 & E_3 & \cdots & E_{s-1} & E_{s}  \\ \hline \hline 
	\underline{1}^t & \underline{0}^t & \underline{0}^t & \hdots  & \hdots & \hdots & \underline{0}^t  \\
	\cline{1-7}
	D_0 & A_1 & 0 & 0 & \hdots & 0 & 0  \\
	\cline{1-7}
	0 & D_1 & A_2 & 0 & \hdots & 0 & 0  \\
	\cline{1-7}
	0 & 0 & D_2 & A_3 & \hdots & 0 & 0   \\
	\cline{1-7}
	0 & 0 & 0 & D_3 & \hdots & 0 & 0  \\
	\cline{1-7}
	\vdots & \vdots & \vdots & \vdots & \ddots & \vdots & \vdots  \\
	\cline{1-7}
	0 & 0 & 0 & 0 & \hdots & A_{s-1} & 0   \\
	\cline{1-7}
	0 & 0 & 0 & 0 & \hdots & D_{s-1} & A_{s}  \\
	\cline{1-7}
	0 & 0 & 0 & 0 & \hdots & 0 & D_{s}  \\
\end{array} 
\right].  \label{eq:stair_H_odd_aug_ch4}
\eea
Since each column of \hodda\ also has weight $2$, the matrix again has an interpretation as the node-edge incidence matrix of a graph  where the incidence matrix of the graph is obtained by replacing every non-zero entry of \hodda\ with $1$. We retain the earlier notation with regard to node sets \vi\ and node \vinfty\ and edge sets \ej\ but now w.r.t the matrix \hodda\ (see \eqref{eq:stair_H_odd_aug_ch4}). Here also each node except the node $V_{\infty}$ has degree $(r+1)$.


\subsubsection{Unravelling the Graph}

We can unravel the structure of this graph exactly as in the case of $t$ even. Fig.~\ref{fig:tree_graph_t_odd} shows the graph for the case $(r+1)=4, t=7,s=3$ with $|V_0|=4$.  Differences compared to the case $t$ even, appear only when we reach the nodes \vsmone\ via edge-set \esmone.  Here, the $r$ other edges outgoing from each node in \vsmone\ are terminated in the node set \vs.  We use \es\ to denote this last collection of edges.  As can be seen, the graph has a tree-like structure, except for the edges (corresponding to edge-set \es) linking nodes in \vsmone\ and \vs.\  The restriction of the overall graph to $\vsmone\ \cup \ \vs$ i.e., the subgraph induced by the nodes $\vsmone\ \cup \vs$ can be seen to be a bipartite graph \gsmone\ where each node in \vsmone\ has degree $r$  while each node in \vs\ has degree $r+1$. 

We use \gzero\ to denote the overall graph and use \gi\ to denote the restriction of \gzero\ to node-set $V_s \cup V_{s-1}  \cdots \cup V_i$ i.e., \gi\ denotes the subgraph of \gzero\ induced by the nodes $V_s \cup V_{s-1}  \cdots \cup V_i$ for $1 \leq i \leq s-1$. Thus the graphs are nested:
\bean
{\cal G}_{s-1} \subseteq {\cal G}_{s-2} \subseteq \cdots \subseteq \gtwo\ \subseteq \gone\ \subseteq \gzero.
\eean
Fig.~\ref{fig:tree_graph_t_odd} identifies the graphs $\gtwo\ \subseteq \gone\ \subseteq \gzero$ for the case $t=7$. 

\begin{note} \label{note:bipartite_count} 
	We note here that the bipartite graph \gsmone\ must be bi-regular with nodes of degree $r$ above, corresponding to nodes in $V_{s-1}$ and nodes of degree $(r+1)$ corresponding to nodes below (i..e, nodes in $V_s$). From the tree-like structure of the graph \gzero\ it follows that the number of nodes in $V_{s-1}$ and $V_s$ are respectively given by 
	\bea
	|V_{s-1}| & = & a_0 r^{s-1} \label{eq:Vsminus1_ch4} \\
	|V_s| & = & a_0 r^{s-1}\frac{r}{r+1}. \label{eq:Vs_ch4}
	\eea
	Since $r,r+1$ are co-prime, this forces $a_0$ to be a multiple of $(r+1)$.  
\end{note}
\begin{thm}\label{thm:t_odd_azero}
	For $t$ odd, $a_0$ must be a multiple of $(r+1)$. 
\end{thm} 

%
%

\subsubsection{Connecting the Graph to the Code} \label{g0specification}

Each node in the graph \gzero\ is associated to a row of the p-c matrix \hodd\ of the code and hence to a p-c.  The one exception is the fictitious node $V_{\infty}$ which does not correspond to a p-c.  When the p-c matrix \hodd\ is over $\mathbb{F}_2$, the top row of \hodda\ is the binary sum of all other rows of \hodda.\ Hence in case when \hodd\ is over $\mathbb{F}_2$,  \hodda\ also represent a p-c matrix of the code and the node $V_{\infty}$ is also associated to a p-c. Also, each edge in the graph \gzero\ is associated to a unique code symbol as the edges are associated with columns of the matrix \hodda.\
The structure of \gzero\ is mostly determined once we specify $a_0$ as \gzero\ with edges in \es\ removed is just a tree with $V_{\infty}$ as the root node. Hence the only freedom lies in selecting the edges that make up the bipartite graph \gsmone.  

\begin{figure}[ht]
	\centering
	\includegraphics[scale = 0.5]{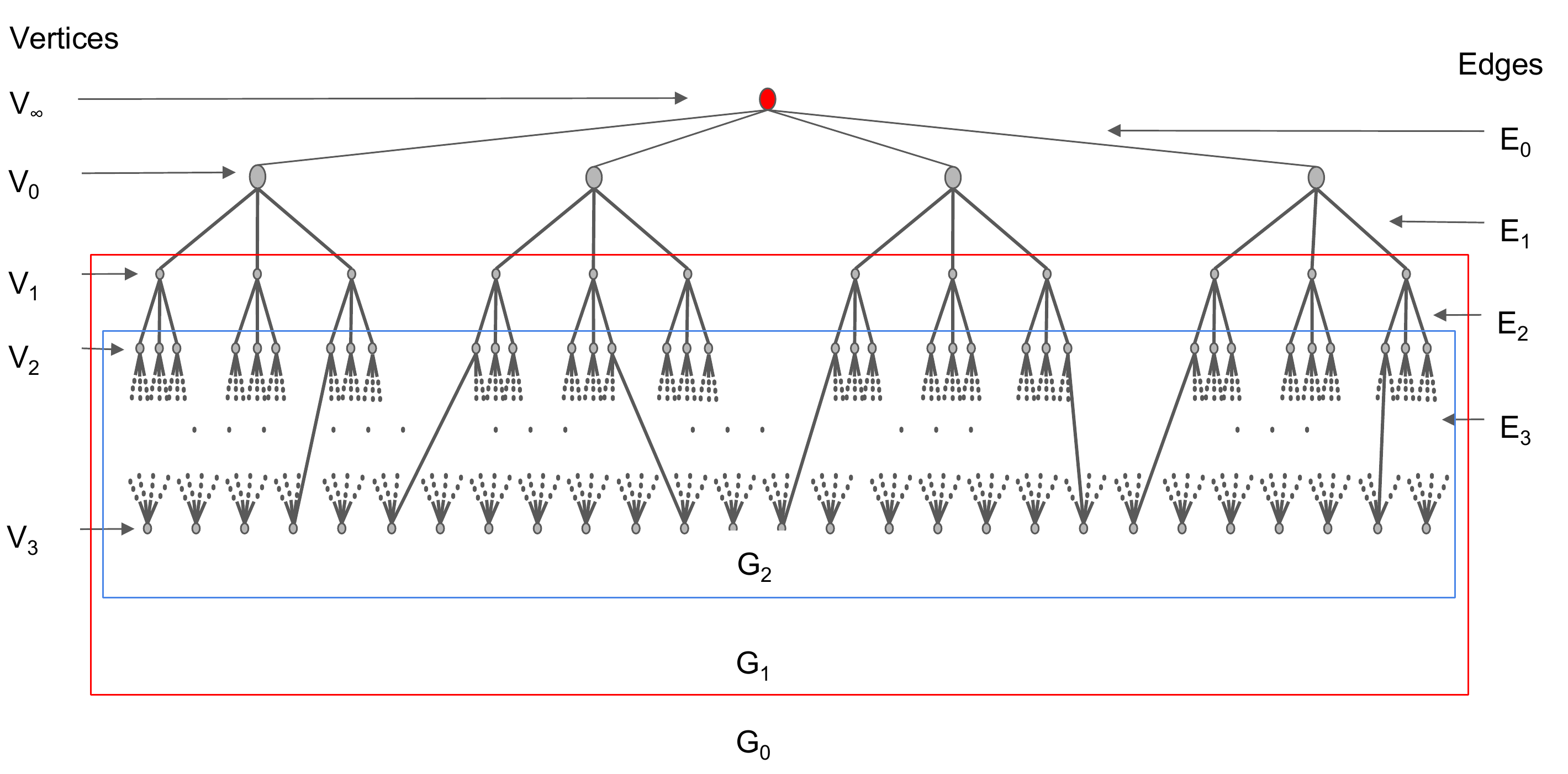}
	\caption{Graphical representation induced by the staircase p-c matrix \eqref{eq:stair_H_odd_aug_ch4} for the case  $(r+1)=4, \  t=7, \ s=3$ with $|V_0|=4$.}
	\label{fig:tree_graph_t_odd}
\end{figure}

\section{Girth Requirement} \label{sec:girth_req} 

\begin{thm} \label{thm:girthreq}
	For both $t$ even and $t$ odd,
	\ben	
	\item The code associated to graph \gzero\ with each node representing a p-c over  $\mathbb{F}_2$ can recover sequentially from $t$ erasures iff \gzero\ has girth $\geq (t+1)$.  
	\item The code associated to graph \gzero\ with each node representing a p-c over $\mathbb{F}_q$ with $q>2$, can recover sequentially from $t$ erasures if \gzero\ has girth $\geq (t+1)$.  
	\een
\end{thm}
\begin{proof}
	Let us assume that there is an erasure pattern involving $\ell \leq t$ erased code symbols and that it is not possible to recover from this erasure pattern sequentially.  These $\ell$ erasures correspond to $\ell$ distinct edges $\{e_i\}_{i=1}^{\ell}$ of the graph \gzero.  Let $J=\{e_i \mid 1 \leq i \leq \ell\}$ and let us restrict our attention to the subgraph \gsub\ of \gzero\ with edge set of \gsub\ exactly the set $J$ with vertex set of \gsub\ exactly the set ${\cal U}$ of nodes that the edges in $J$ are incident upon. We note that in graph \gsub\, every node in ${\cal U} \setminus \{\vinfty\}$ must have degree $\geq 2$ because in \gsub\ a node in ${\cal U} \setminus \{\vinfty\}$ of degree one would imply that the code word in the row of p-c matrix corresponding to the degree one node can be used to recover the erased code symbol corresponding to the edge incident on it in \gsub.
	\ben
	\item First let us assume that $\vinfty \not \in {\cal U}$. We start with edge $e_1$, this must be linked to a p-c node $U_1 \in {\cal U}$ which is linked to a second erased symbol $e_2$ and so on, as degree of each node in \gsub\ is $\geq 2$.  In this way, we can create a path in \gsub\ with distinct edges.  But since there are only a finite number of nodes, this must eventually force us to revisit a previous node, thereby establishing that the graph \gzero\ has girth $\leq \ell \leq t$. 
	\item Next, for the case when \vinfty\ is a node in \gsub, we start at an erased edge incident upon node \vinfty\ and move to the node at the other end of the edge.  Since that node has degree $\geq 2$, there must be an edge corresponding to a second erased symbol that it is connected to the node and so on.  Again the finiteness of the graph will force us to revisit either \vinfty\ or else, a previously-visited node proving once again that an unrecoverable erasure pattern indicates a cycle and hence the graph \gzero\ has girth $\leq \ell \leq t$. 
	\een
	We have thus established that having a girth $\geq (t+1)$ will guarantee recovery from $\leq t$ erasures.  For $q=2$, it is easy to see that a girth of $t+1$ is necessary since if the girth is $\leq t$, then the set of erasures lying on a cycle of length $\leq t$ is uncorrectable regardless of whether or not the nodes associated with this cycle includes \vinfty. 
\end{proof}

\begin{thm} \label{LowerBoundOna0}
	For the graph \gzero\ to have girth $(t+1)$, the degree $a_0$ of \vinfty\ or equivalently, the number $a_0$ of nodes in $V_0$, has the lower bound, $a_0 \geq r+1$.
\end{thm}
\begin{proof}
	{\bf Case $t$ odd:} \ As shown in Theorem~\ref{thm:t_odd_azero}, $a_0$ must be in fact be a multiple of $(r+1)$.  
	
	{\bf Case $t$ even:}   Let the set containing all those nodes of $V_s$ which are at distance atmost $s$ from a particular node $v \in V_0$ be denoted by $N_v$ where distance between vertex $w_1,w_2$ is measured by the number of edges in the shortest path between $w_1,w_2$. Note that $N_v \cap N_w = \emptyset$, $v \neq w$. Since every node in $N_v$ has a path of length (measured in terms of number of edges) $s$ edges leading to $v$ and hence these paths does not involve any edge from \espone, no two nodes in $N_v$ can be connected by an edge in \espone\ in \gzero.  For, if two nodes in $N_v$ were connected by an edge in \espone\ then there would be a cycle of length at most $2s+1 < t+1$ present as $s:=\lfloor \frac{t-1}{2} \rfloor$.  Further, a node in $N_v$ cannot connect to two nodes in $N_w$ via edges in \espone, $\forall v \neq w, w \in V_0$, for the same reason.  It would once again imply the presence of a cycle of length $2s+2 < t+1$. Hence each node in $N_v$ must connects to $r$ nodes with $i$th node belonging to $N_{w_i}$ , $1 \leq i \leq r$ respectively for some set of $r$ distinct nodes $\{w_1,...,w_r\} \subseteq V_0-\{v\}$.   It follows that there must be at least $r+1$ distinct nodes in $V_0$.  Thus $a_0 \geq r+1$. 
	
\end{proof}

\section{Code Construction by Meeting Girth Requirement} \label{sec:complete_code}


As noted in Section~\ref{sec:girth_req}, to complete the construction of rate-optimal code over $\mathbb{F}_2$, we need to ensure that the graph \gzero\ has girth $\geq t+1$ for both $t$ even and $t$ odd. Since the code is over $\mathbb{F}_2$, the node-edge incidence matrix of the graph \gzero\ with girth $\geq t+1$ will directly give the required p-c matrix for a rate-optimal code. Hence in the following we focus only in designing \gzero\ with girth $\geq t+1$. The steps followed are outlined below.

\subsection{Step $1$ : Construction and Coloring of the Base Graph} 

\ben
\item We begin by deleting the edges in the graph \gzero\ connecting $V_{\infty}$ to the nodes in $V_0$.  One is then left with a graph where the nodes in $V_0$ have degree $r$ and all the remaining nodes have degree $(r+1)$.  In particular, every node has degree $\leq (r+1)$.  We shall call the resultant graph the {\em base graph} \ \gbase.  Note that if \gbase\ has girth $\geq t+1$, the construction ends here but we do not need \gbase\ to have girth $\geq t+1$, as we will modify \gbase\ to construct another graph $\gexp_0$ having same incidence matrix as  \eqref{eq:stair_H_even_aug_ch4} for $t$ even, \eqref{eq:stair_H_odd_aug_ch4} for $t$ odd which will have girth $\geq t+1$.
\item By Vizing's theorem \cite{Viz}, the edges of \gbase\ can be colored using $\ell (\leq (r+2))$ colors in such a way that each edge is colored with some color from the set of $\ell$ colors and for every $v \in V(\gbase)$, all the edges incident on $v$ are of a different colors.  However it is possible to color the edges of \gbase\ using $\ell = (r+1)$ colors:  
\bit
\item {\bf Case $t$ even:} \ In this case, by selecting $a_0=4$ for $s \geq 2$, i.e., $t \geq 6$, and through careful selection of the edges connecting nodes in $V_s$, one can ensure that the base graph can be colored using $(r+1)$ colors (See Appendix \ref{Gbasecolouring_ch4}); thus the base graph will in this case have 
\bea
\nbase\ \ = \ a_0(1+r+\cdots +r^s) & = & 4\left( \frac{r^{s+1}-1}{r-1} \right) \label{eq:tevenblklength}
\eea
nodes. 
For $s=0$, $t=2$, we can choose \gbase\ to be a complete graph of degree $r+1$ and the complete construction of rate-optimal code ends here as the graph \gbase\ has girth $t+1=3$. For $s=1$, $t=4$, we can construct \gbase\ with $a_0=2r$ with girth $6 > t+1=5$ and the complete construction of rate-optimal code ends here. We skip the description for $t=4$ case here.

\item {\bf Case $t$ odd:} \ In this case, by selecting $a_0=(r+1)$ and through careful selection of the edges connecting nodes in $V_{s-1}$ with nodes in $V_s$, one can ensure that the base graph can be colored using $(r+1)$ colors (See Appendix \ref{Gbasecolouring_ch4}); thus the base graph will have in this case, a total of 
\bean
\nbase\ \ = \ a_0(1+r+\cdots +r^{s-1}+r^{s-1} \frac{r}{r+1}) & = & (r+1)\left( \frac{r^{s}-1}{r-1}\right) +r^s \label{eq:toddblklength}
\eean
nodes. 

\eit
The coloring is illustrated in Fig.~\ref{fig:base_graph_coloring} for the case $t=5$ and $r=3$. 
\een
\begin{figure}[ht]
	\centering
	\includegraphics[scale = 0.35]{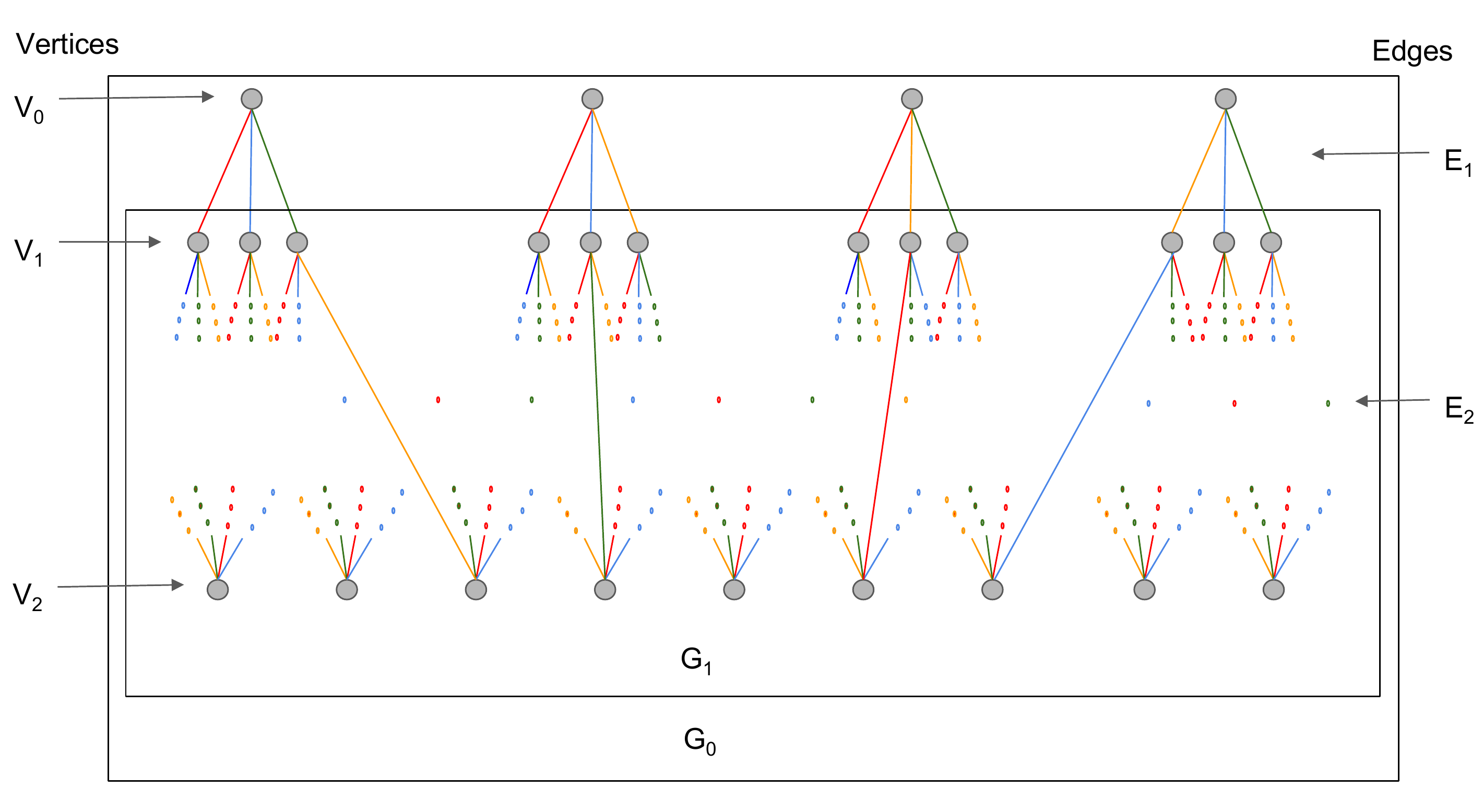}
	\caption{An example base graph with associated coloring of the edges using $(r+1)$ colors.  Here $t=5$, $r=3$, $a_0=r+1=4$ so the base graph can be constructed such that we can color the edges with $r+1=4$ colors.}
	\label{fig:base_graph_coloring}
\end{figure}
In summary, the base graph is graph \gzero\ with node $V_{\infty}$ removed and has all nodes of degree $ \leq (r+1)$ and can be colored using $(r+1)$ colors; these are the only properties of the base graph that are carried forward to the next steps of the construction; we will number the colors $1$ through $r+1$ and speak of color $i$ as the $i$th color.   The steps that follow, are the same for both $t$ even and $t$ odd.  

\subsection{Step $2$ : Construction and Coloring of the Auxiliary Graph} 

Next, we begin by identifying a second graph which we shall call the {\em auxiliary graph} \aux.  The following properties are required of the auxiliary graph \aux:
\ben
\item \aux\ has an even number $\naux=2m$ of nodes and every node in \aux\ has degree $\geq (r+1)$, 
\item \aux\ is of girth $\geq t+1$, 
\item \aux\ should permit a coloring of a subset of the edges of \aux\ using the same $(r+1)$ colors used to color the base graph in such a way that 
\bit
\item Every vertex of \aux\ contains an edge incident upon it of color $i$, for any $i$ such that $1 \leq i \leq r+1$; 
\item For $1 \leq i \leq r+1$, if \auxi\ denotes the subgraph of \aux,\ with edge set of \auxi\ exactly equal to the set of edges in \aux\ of color $i$, with $V(\auxi) = V(\aux)$, then \auxi\ is a bipartite graph corresponding to a perfect matching of \aux\, i.e., there are $|\aux|/2$ nodes on either side of the bipartite graph and every vertex of \auxi\ is of degree $1$. 
\eit 
This is illustrated in Fig.~\ref{fig:aux_graph_coloring}.
\begin{figure}[ht]
	\centering
	\includegraphics[angle=270,scale = 0.65]{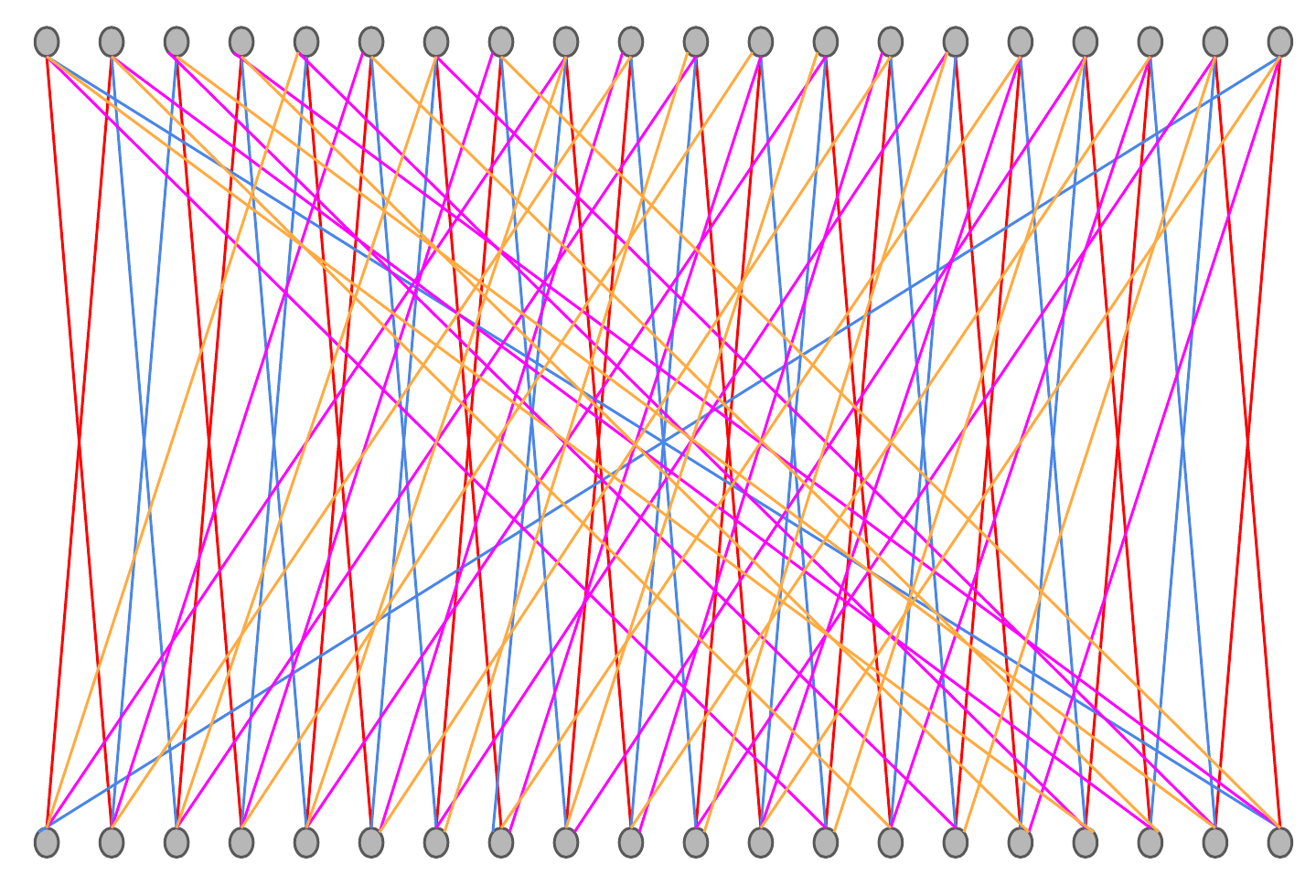}
	\caption{An example auxiliary graph \aux\ with associated coloring of the edges using $(r+1)$ colors.  Here $r=3$, so there are $r+1=4$ colors.  This graph is a regular bipartite graph of degree $r+1=4$ with $\naux=40$ vertices with girth $\geq 6$.  }
	\label{fig:aux_graph_coloring}
\end{figure}
\een
It follows from Hall's Theorem (1935) \cite{dieRich} that an $(r+1)$-regular bipartite graph $G$ of girth $\geq t+1$ can be colored with $r+1$ colors in such a way that $G$ satisfies the conditions necessary for an auxiliary graph \aux.\ It can also be shown that an $(r+1)$-regular graph with girth $\geq t+1$ with $N_R$ nodes can be converted to an $(r+1)$-regular bipartite graph with girth $\geq t+1$ with $2N_R$ nodes. We skip this description.
Since $(r+1)$-regular graph with girth $\geq t+1$ can be constructed with small number of nodes (close to smallest possible) \cite{Lubotzky1988,X_Dahan,Mor,DavSarVal,LazUstWol}, we can construct the auxiliary graph \aux\ with necessary conditions with small number of nodes.
\subsection{Step $3$ : Using the Auxiliary Graph to Expand the Base Graph} 
In this step, we use the graph \aux\ to expand the graph \gbase, creating in the process, a new graph \gexp\ as follows:
\ben
\item We start with \gbase\ and replace each vertex $v \in V(\gbase)$ with \naux\ vertices $W_{v} = \{(v,u) \mid u \in \aux \}$, each corresponding to a vertex in \aux\ i.e.,  $V(\gexp) = V(\gbase) \times V(\aux) = \{(v,u) \mid u \in \aux , v \in \gbase \}$. The resultant graph will be termed as the expanded graph \gexp.\ We will now define the edge set of \gexp.\ The edge set of \gexp\ is defined in such a way that every edge in \gexp\ is colored with some color $j \in [r+1]$.
\item For every $v,x \in V(\gbase)$ and for every $u,w \in V(\aux)$:

We add the edges $((v,u),(x,w)),((v,w),(x,u))$ in \gexp\ with both edges of color $i$ iff 
\ben
\item in \gbase, vertex $v$ and vertex $x$ were connected by an edge of color $i$ and
\item $(u,w) \in \auxi$ i.e., the nodes $u$ and $w$ are connected by an edge of color $i$ in \aux.
\een
\een

\begin{thm}
	\gexp\ has girth $\geq t+1$.
\end{thm}
\begin{proof}
	This follows simply because corresponding to every path traversed through \gexp with successive edges in the path with some sequence of colors, there is a corresponding path in \aux\ with successive edges in the path corresponding to the same sequence of colors. The edge coloring of the base graph ensures that we never retrace our steps in the auxiliary graph i.e., any two successive edges in the path in auxiliary graph are not the same.  It follows that since \aux\ has girth $\geq t+1$, the same must hold for \gexp. 
\end{proof} 
Let $V'_i = V_i \times V(\aux)$. We add a node $V'_{\infty}$ to \gexp\ and connect it to each node in $V'_0 \subseteq V(\gexp)$ through an edge. Call the resulting graph $\gexp_0$.
It is clear from the construction that the graph $\gexp_0$ is graph \gzero\ with girth $\geq t+1$ but now with $V'_i$ replacing $V_i$ in the description of \gzero\ because: 
\ben
\item \gbase\ is \gzero\ with node $V_{\infty}$ removed,
\item $\gexp_0$ is just an expanded version of \gbase\ i.e, a node $v \in V(\gbase)$ is expanded into the set of nodes $W_{v} \subseteq V(\gexp_0)$,
\item Edges are defined in $\gexp_0$ preserving the tree-like structure of \gbase\ as the sub graph of $\gexp_0$ induced by $W_{v} \cup W_{x}$ is isomorphic to \auxi\ with all nodes having degree one for very edge $(v,x)$ in \gbase\ of color $i$. 
\een
Hence $\gexp_0$ has the same node-edge incidence matrix as  \eqref{eq:stair_H_even_aug_ch4} over $\mathbb{F}_2$  for $t$ even, \eqref{eq:stair_H_odd_aug_ch4} over $\mathbb{F}_2$ for $t$ odd with $V'_i$ replacing $V_i$ and also has girth $\geq t+1$. This node-edge incidence matrix gives the required p-c matrix for a rate-optimal code over $\mathbb{F}_2$ with parameters $(r,t)$. Hence we have constructed the required p-c matrix for a rate-optimal code over $\mathbb{F}_2$ for both $t$ even and $t$ odd, for any $r \geq 3$.
Hence we conclude our construction of rate-optimal code for both $t$ even and $t$ odd.

\section{S-LR Codes that are Optimal with Respect to Both Rate and Block Length} \label{sec:Moore_ch4}

In this section, we begin by presenting a construction for S-LR codes given in \cite{RawMazVis}. We improve upon the lower bound  for the rate of this construction provided in \cite{RawMazVis} and show that for a given $(r,t)$, there is a unique block length for which this improved lower bound on rate is equal to the right hand side of the upper bound on rate derived in Chapter \ref{ch:SeqBound} in Theorem \ref{rate_both}. For a given $(r,t)$, codes based on this construction with the unique block length for which the improved lower bound is equal to the maximum possible rate turn out to correspond to codes based on a graph known as Moore graph with degree $r+1$ and girth $t+1$. Unfortunately, for $r \geq 2$ (see \cite{DynCageSur}), Moore graphs are known to  exist only when $t \in \{2,3,4,5,7,11\}$. 


\begin{const} \label{con:Raw_bipartite} (\cite{RawMazVis})
	Consider an $(r+1)$-regular bipartite graph $G$ having girth $\geq t+1$. Let the number of nodes in $G$ be $N$. Let $H$ be the $(N \times \frac{N(r+1)}{2})$ node-edge incidence matrix of the graph $G$ with each row representing a distinct node and each column representing a distinct edge. The code $\mathcal{C}$ with p-c matrix $H$ thus defined is an S-LR code with parameters $(n = \frac{N(r+1)}{2},k \geq n-N,r,t)$ over $\mathbb{F}_2$. The constrution takes $G$ as input and constructs code $\mathcal{C}$ as output.
\end{const}
\begin{proof}
	(sketch of proof)  
	Let $\mathcal{C}$ be the code obtained as the output of the construction \ref{con:Raw_bipartite} with an $(r+1)$-regular bipartite graph $G$ with girth $\geq t+1$ as input. The code $\mathcal{C}$ is an S-LR code with parameters $(r,t)$, simply because a set of erased symbols with least cardinality which cannot be recovered through sequential recovery must correspond to a set of linearly dependent columns in $H$ with least cardinality and hence corresponds to a set of edges forming a cycle in $G$.  Since $G$ has girth $\geq t+1$, the number of edges in this cycle must be $>t$ and hence the number of erased symbols is $>t$.  The code parameters follow from a simple calculation. 
\end{proof}
The graph $G$ described in Construction \ref{con:Raw_bipartite} need not have the tree-like structure of \gzero.\ 
Let $\mathcal{C}$ be the code obtained as the output of the construction \ref{con:Raw_bipartite} with an $(r+1)$-regular bipartite graph $G$ with girth $\geq t+1$ as input.
Since the graph $G$ need not have the tree-like structure of \gzero, it may be impossible for the code $\mathcal{C}$ to have a p-c matrix similar to \eqref{eq:stair_H_even_aug_ch4} for $t$ even, \eqref{eq:stair_H_odd_aug_ch4} for $t$ odd. We will now see that it is possible to write a p-c matrix for $\mathcal{C}$ similar to \eqref{eq:stair_H_even_aug_ch4} for $t$ even, \eqref{eq:stair_H_odd_aug_ch4} for $t$ odd iff $G$ is a Moore graph. It follows from Construction \ref{con:Raw_bipartite}, as was observed in \cite{RawMazVis}, that the rate of the code $\mathcal{C}$ is $\geq \frac{r-1}{r+1}$.   But we will shortly provide a precise value for the rate of this code. 
\begin{defn} (\textbf{Connected Component})
	Let $G$ be a graph. Then a connected component of $G$ is a subgraph $G_1$ such that $G_1$ is connected as a graph and moreover, there is no edge in $G$, connecting a vertex in $V(G_1)$ to a vertex in $V(G)\setminus V(G_1)$. 
\end{defn}
Of course if $G$ is a connected graph then there is just a single connected component, namely the graph $G$ itself. 
\begin{thm} \label{thm:rate}
	Let $G$ be an $(r+1)$-regular bipartite graph with girth $\geq t+1$. Let the graph $G$ be connected with exactly $N$ nodes. The code $\mathcal{C}$ obtained as the output of the construction \ref{con:Raw_bipartite} with the graph $G$ as input
	is an S-LR code with parameters $(n = \frac{N(r+1)}{2},k = n-N+1,r,t)$ over $\mathbb{F}_2$ and hence having rate given by:
	\bea
	\frac{r-1}{r+1}+\frac{1}{n}. \label{RateofRawatConst}
	\eea
\end{thm}
\begin{proof}
	Let $H$ be the node-edge incidence matrix of the graph $G$.
	From the description of Construction \ref{con:Raw_bipartite}, the matrix $H$ is a p-c matrix of the code $\mathcal{C}$.
	The p-c matrix $H$, has each row of Hamming weight $(r+1)$ and each column of weight $2$.  It follows that the sum of all the rows of $H$ is the all-zero vector.  Thus the rank of $H$ is $\leq N-1$.  
	
	Next, let $\ell$ be the smallest integer such that a set of $\ell$ rows of $H$ add up to the all-zero vector. Let $M$ be the set of nodes in $G$ corresponding to a set of $\ell$ rows $\underline{r}_1,...,\underline{r}_{\ell}$ in $H$ such that $\sum_{i=1}^{\ell} \underline{r}_i=\underline{0}$.  We note that any edge $(u,v)$ in $G$ with $u \in M$ will be such that $v \in M$ and similarly if $v \in M$ then $u \in M$. Let $S=\cup_{i=1}^{\ell}supp(\underline{r}_i)$, it follows that the subgraph of $G$ with vertex set equal to $M$ and the edge set equal to the edges associated to columns of $H$ indexed by $S$ form a connected component of the graph $G$.  But since $G$ is connected, $\ell=|M|=N$ and hence $S=[n]$.  It follows that any set of $N-1$ rows of $H$ is linearly independent.  Hence the rank of $H$ equals $N-1$.   The parameters of $\mathcal{C}$ are thus given by:
	\bean
	\text{block length} \ n & = & \frac{N(r+1)}{2} \\
	\text{dimension} \ k & = & \frac{N(r+1)}{2} - (N-1) \\
	\text{rate} \ R & = & 1 - \frac{2(N-1)}{N(r+1)} =  1 - \frac{2}{r+1} + \frac{2}{N(r+1)}  =  \frac{r-1}{r+1}+\frac{1}{n}.
	\eean
\end{proof}
We note here that while the Construction~\ref{con:Raw_bipartite} made use of regular bipartite graphs, the bipartite requirement is not a requirement as in the argument above, we only used the fact that the graph $G$ is regular.   We collect together the above observations concerning rate and sufficiency of the regular-graph requirement into a (slightly) modified construction. 

\begin{const} \label{con:modified} (modified version of the construction in \cite{RawMazVis})
	Let $G$ be a connected, regular graph of degree $(r+1)$ and of girth $\geq t+1$ having exactly $N$ vertices. Let $H$ be the $(N \times \frac{N(r+1)}{2})$ node-edge incidence matrix of the graph $G$ with each row representing a distinct node and each column representing a distinct edge. The code $\mathcal{C}$ with p-c matrix $H$ is an S-LR code having parameters $(n = \frac{N(r+1)}{2},k=n-(N-1),r,t)$ over $\mathbb{F}_2$. The constrution takes $G$ as input and constructs code $\mathcal{C}$ as output.
\end{const}

For the rest of this section: Let $r,t$ be arbitrary positive integers. Let $G$ be a connected, regular graph of degree $(r+1)$ and of girth $\geq t+1$ having exactly $N$ vertices. Let $\mathcal{C}$ be the S-LR code having parameters $(n = \frac{N(r+1)}{2},k=n-(N-1),r,t)$ over $\mathbb{F}_2$ obtained as the output of the construction \ref{con:modified} with the graph $G$ as input.

Clearly, the rate of the code $\mathcal{C}$ is maximized by minimizing the block length $n=\frac{N(r+1)}{2}$ of the code, or equivalently, by minimizing the number of vertices $N$ in $G$. Thus there is interest in regular graphs of degree $r+1$, having girth $\geq t+1$ with the least possible number of vertices.   This leads us to the Moore bound and Moore graphs. 
%
%

\begin{thm} \label{thm:Moore} \textbf{(Moore Bound)} (\cite{DynCageSur})
	The number of vertices $N'$ in a regular graph of degree $r+1$ and girth $\geq t+1$ satisfies the lower bound :
	\bean
	N' \ \geq N_{r,t} & := &  1+\sum_{i=0}^{s}(r+1)r^i,  \ \  \ \ \text{ for $t=2s+2$ even }, \label{eq:Moore:Even} \\
	N' \ \geq N_{r,t} & := & 2\sum_{i=0}^{s}r^i,  \ \ \ \ \ \ \ \ \ \ \ \ \ \ \  \text{ for $t=2s+1$ odd }.  \label{eq:Moore:Odd}
	\eean
\end{thm}
\begin{defn} \label{def:Mooregraph} \textbf{(Moore graphs)}
	A regular graph with degree $r+1$ with girth atleast $t+1$ with number of vertices $N'$ satisfying $N' = N_{r,t}$ is called a Moore graph.
\end{defn}

\begin{lem} \label{MooreOptimality}
	The rate $R$ of the code $\mathcal{C}$ with block length $n=\frac{N(r+1)}{2}$ satisfies:
	\bea
	R = \frac{r-1}{r+1}+\frac{2}{N(r+1)} & \leq & \frac{r-1}{r+1}+\frac{2}{N_{r,t}(r+1)} \label{eq:rate_bound_Moore}
	\eea
	
	The inequality \eqref{eq:rate_bound_Moore} will become equality iff $G$ is a Moore graph.
\end{lem}
\begin{proof}
	The Lemma follows from Theorem \ref{thm:rate} and the Moore bound given in Theorem \eqref{thm:Moore}.
\end{proof}
It turns out interestingly, that the upper bound on rate of the code $\mathcal{C}$ given by the expression $\frac{r-1}{r+1}+\frac{2}{N_{r,t}(r+1)}$ (Lemma \ref{MooreOptimality}) is numerically, precisely equal to the right hand side of inequality \eqref{Thm1} (for $t$ even), \eqref{Thm2} (for $t$ odd). Note that the inequality \eqref{Thm1} (for $t$ even), \eqref{Thm2} (for $t$ odd) directly gives an upper bound on rate of an S-LR code. As a result, we have the following Corollary.
\begin{cor}
	The S-LR code $\mathcal{C}$ is a rate-optimal code iff $G$ is a Moore graph.
\end{cor}
\begin{cor}
	Let $(r,t)$ be such that a Moore graph exists. If $G$ is the Moore graph then the code $\mathcal{C}$ is not only rate optimal, it also has the smallest block length possible for a binary rate-optimal code for the given parameters $r,t$. 
\end{cor}
\begin{proof}
	From the discussions in Section \ref{sec:graphical_rep}, a rate-optimal code must have the graphical representation \gzero.\ From Theorem \ref{thm:girthreq}, \gzero\ must have girth $\geq t+1$, if the rate-optimal code is over $\mathbb{F}_2$. Hence from Theorem \ref{LowerBoundOna0}, the number $a_0$ of vertices in $V_0$ satisfies the lower bound 
	\bea
	a_0 & \geq & r+1. \label{eq:a0}
	\eea
	Since the value of $a_0,r$ through numerical computation  gives the number of edges in \gzero, it can be seen that \gzero\ with girth $\geq t+1$ and $a_0=(r+1)$ leads to a binary rate-optimal code having block length equal to $\frac{N_{r,t} (r+1)}{2}$. The code $\mathcal{C}$ also has block length equal to $\frac{N_{r,t} (r+1)}{2}$ when $G$ is a Moore graph. Hence the corollary follows from \eqref{eq:a0} and noting that the number of edges in the graph \gzero\ grows strictly monotonically with $a_0$ and hence $a_0=r+1$ correspond to least possible block length.
\end{proof}

%
In the following we fix $t=4,r=6$. An example Moore graph with $t=4,r=6$, known as the Hoffman-Singleton graph (with the removal of a single vertex corresponding to eliminating an unneeded, linearly dependent parity check) is shown in the Fig.~\ref{fig:Moore}. If  $G$ is this Moore graph then the code \calc is a rate-optimal code with least possible block-length with $t=4,r=6,n=175$.

\begin{figure}[ht]
	\centering
	\includegraphics[scale = 0.4]{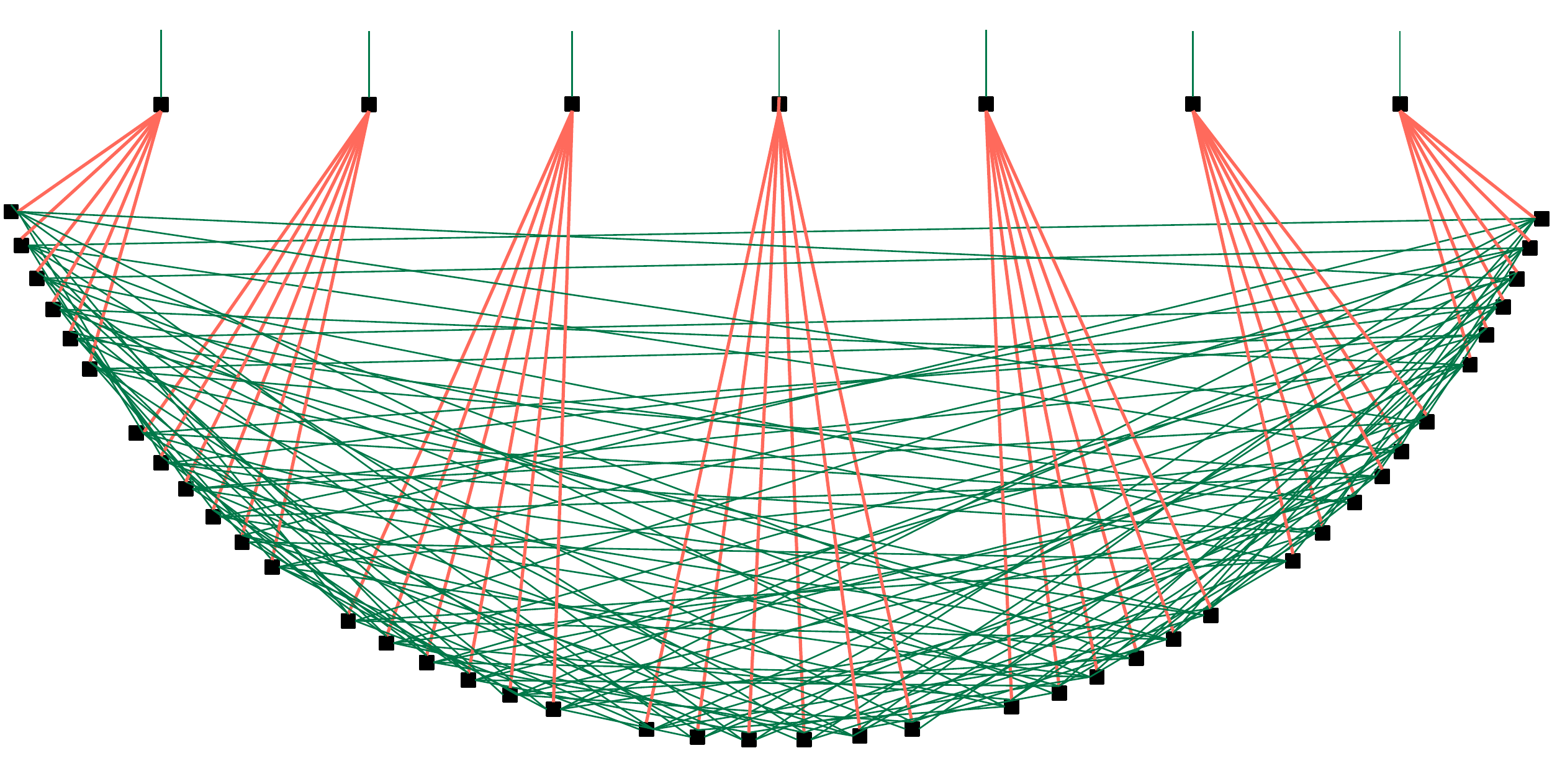}
	\caption{An example Moore graph for $t=4,r=6$ called the Hoffman- Singleton graph is shown in the figure. If  $G$ is this Moore graph then the code \calc is a rate-optimal code with least possible block-length with $t=4,r=6,n=175$.}
	\label{fig:Moore}
\end{figure}

\begin{appendices}
	
	\chapter{Coloring \gbase\ with $(r+1)$ colours} \label{Gbasecolouring_ch4}
	
	In this section of the appendix, we show how to construct the base graph \gbase\ in such a way that its edges can be colored using $(r+1)$ colors.    In our construction, for $t$ odd, we set $a_0=(r+1)$ which is the smallest possible value fo $a_0$ by Theorem \ref{thm:t_odd_azero}.  For the case of $t$ even, we set $a_0=4$.   We begin with a key ingredient that we make use of in the construction, namely, that the edges of a regular bipartite graph of degree $d$ can be colored with $d$ colors. 
	
	\begin{thm} \label{bipartiteColouring}
		The edges of a regular bipartite graph $G$ of degree $d$ can be coloured with exactly $d$ colours such that each edge is associated with a colour and adjacent edges of the graph does not have the same colour. 
	\end{thm}
	\begin{proof}
		By Hall's Theorem (1935) \cite{dieRich}, there exists a perfect matching in the $d$-regular bipartite graph $G$. We first identify this perfect matching and then remove the edges corresponding to this matching from $G$ to form a $(d-1)$-regular bipartite graph to which we can once again apply Hall's Theorem and so on until we have partitioned the edges of the bipartite graph $G$ into the disjoint union of $d$ perfect matchings.  
		To conclude the proof, we simply choose $d$ different colors say $\{1,2,...,d\}$ and color each edge in $i$th perfect matching with color $i$, $\forall 1 \leq i \leq d$. 
	\end{proof}
	
	\section{The Construction of \gbase\ for $t$ Odd}
	
	The aim here is to show that we can construct the base graph \gbase\ and color the edges of it using exactly $(r+1)$ colors such that adjacent edges does not have the same color.  
	
	In the case of $t$ odd, we set $a_0=r+1$.  From Remark~\ref{note:bipartite_count}, it follows that if we set $a_0=(r+1)$, then 
	\bean
	|V_{s-1}|r & = & |V_{s}|(r+1),
	\eean
	and hence from \eqref{eq:Vsminus1_ch4}, \eqref{eq:Vs_ch4}, it is possible to connect nodes in $V_{s-1}$ to nodes in $V_s$ so that \gsmone\ is a bipartite graph with the two sets of nodes in \gsmone\  being equal to $V_{s-1},V_s$ respectively where each node in $V_{s-1}$ is of degree $r$ and each node in $V_{s}$ is of degree $(r+1)$. Since the graph \gbase\ is completely specified once \gsmone\ and $a_0$ are specified as mentioned in Section \ref{g0specification}, \gbase\ can be constructed with $a_0=r+1$. Let us recall that if we add a node $V_{\infty}$ to \gbase\ and connect $V_{\infty}$ to all the nodes in $V_0$, we will recover the graph \gzero.  It is easily seen that \gzero\ is an $(r+1)$-regular graph.  By grouping together nodes in alternate layers in \gzero\ i.e., by letting $U_1 = V_{\infty} \cup V_1 \cup V_3 ...$ and $U_2 = V_{0} \cup V_2 \cup V_4 ...$, it can be verified that \gzero\ is in fact, an $(r+1)$-regular {\em bipartite} graph with node-set $U_1$ on the left and node-set $U_2$ to the right.  Hence by Theorem \ref{bipartiteColouring}, the edges of the graph \gzero\ and hence the edges of the graph \gbase\ can be colored with $r+1$ colors.
	
	\section{The Construction of \gbase\ for $t$ Even}
	
	Let $T_v$ be the tree (a subgraph of \gbase) with root node $v \in V_0$ formed by all the paths of length at most $s$ starting from $v$ i.e., $T_v$ is the subgraph of \gbase\ induced by the vertices which are at distance atmost $s$ from $v$ where we remove in this induced subgraph all the edges which are in $E_{s+1}$. 
	Let $V(T_v)$ be the vertex set of $T_v$. The nodes $V(T_v) \cap V_i$ are at depth $i$ in $T_v$ with the root node $v$ at depth 0.  We now color the edges of the tree $T_v$ with the $r+1$ colors $\{1,...,r+1\}$. It is clear that such a coloring of edges of $T_v$ can be done.  There are $r$ edges $\{e_1,...,e_r\}$ incident on $v$. Let the color of $e_i$ be $i$, $\forall i \in [r]$. Hence there is no edge of color $r+1$ incident on $v$. 
	\ben
	\item Let $j$ be one of the $r$ colors, $1 \leq j \leq r$. Let $X^j_i$ be the largest subset of $V(T_v) \cap V_i$ (nodes at depth $i$) where each node in $X^j_i$ is connected by an edge of color $j$ to a distinct node at depth $(i-1)$, i.e., a node in $V(T_v) \cap V_{i-1}$. Let $|X^j_i| = x^j_i$. Let $Y_i$ be the largest subset of $V(T_v) \cap V_i$ (nodes at depth $i$)   where each node in $Y_i$ is connected by an edge of color $(r+1)$ to a distinct node in $V(T_v) \cap V_{i-1}$ (nodes at depth $i-1$). Let $|Y_i|=y_i$. It is clear that $x^1_i = x^2_i = ... = x^r_i$. We set $x^j_i = x_i$. It can be verified that:
	\bean
	y_i & = & \left\{ \begin{array}{rl} x_i-1 & \text{if $i$ is odd}, \\  x_i+1 & \text{if $i$ is even} \end{array} \right. .
	\eean
	Hence $|V(T_v) \cap V_i| = r^i = x_i r + y_i$. It follows that 
	\bean
	x_i & = & \left\{ \begin{array}{rl} \frac{r^i+1}{r+1} & \text{if $i$ is odd}, \\  \frac{r^i-1}{r+1} & \text{if $i$ is even} \end{array} \right. .
	\eean
	Since the set $X^j_s$ depends on $v$, from now on we denote it as $X^j_{s,v}$. Since the set $Y_s$ depends on $v$, from now on we denote it as $Y_{s,v}$.    
	\item Next, set $a_0=4$.  Let $V_0=\{v_1,v_2,v_3,v_4\}$. Hence there are $4$ such trees $T_{v_1},...,T_{v_4}$. Color the edges of the tree $T_{v_{\ell}}$ with $r+1$ colors as before such that there is no edge of color $r+1$ incident on $v_{\ell}$, $\forall \ell \in [4]$. An edge in $E_{s+1}$ which is incident on any node in $X^j_{s,v_{\ell}} (\subseteq V_s)$ cannot be colored with color $j$ but every color from $[r+1] \setminus \{ j\}$ can be used to color it, $\forall \ell \in [4]$, $\forall j \in [r]$.  Similarly an edge in $E_{s+1}$ which is incident on any node in $Y_{s,v_{\ell}} (\subseteq V_s)$ cannot be colored with color $(r+1)$, $\forall \ell \in [4]$. We can connect the set of nodes $\cup_{\ell=1}^{4}X^j_{s,v_{\ell}}$ of size $4x_s$ (or the set of nodes $\cup_{\ell=1}^{4}Y_{s,v_{\ell}}$ nodes of size $4y_s$) to form bipartite graph of degree $r$.  This is possible because we can first construct a regular graph of degree $r$ with $2x_s$ or $2y_s$ nodes since $\min\{2x_s,2y_s\} \geq (r+1)$, and then create two copies of the vertex set of this regular graph to create the left and right nodes of a bipartite graph $G$.  The edges of the bipartite graph can then be formed by connecting vertices in accordance with the regular graph, i.e., if nodes $w_1$ and $w_2$ were connected in the regular graph, then node $w_1$ on the left is connected to node $w_2$ on the right in the bipartite graph and vice versa. Hence a bipartite graph $G$ of degree $r$ with $4x_s$ or $4y_s$ nodes is constructed. Now connect the $4x_s$ nodes in $\cup_{\ell=1}^{4}X^j_{s,v_{\ell}}$ in accordance with this bipartite graph $G$ which also has $4x_s$ nodes and color the edges by colors from $[r+1]\setminus \{j\}$, $\forall j \in [r]$. This is possible by Theorem \ref{bipartiteColouring}. Similarly connect the $4y_s$ nodes in $\cup_{\ell=1}^{4}Y_{s,v_{\ell}}$ in accordance with this bipartite graph $G$ which also has $4y_s$ nodes and color the edges of this graph by colors from $[r+1]\setminus \{r+1\}$. This is again possible by Theorem \ref{bipartiteColouring}. The edges of these $(r+1)$ bipartite graphs form the edge set $E_{s+1}$. The construction is complete as we have constructed \gbase\ and connected the nodes in $V_s = (\cup_{j=1}^{r} \cup_{\ell=1}^{4}X^j_{s,v_{\ell}}) \cup (\cup_{\ell=1}^{4}Y_{s,v_{\ell}})$ according to an edge set $E_{s+1}$ and colored the edges of \gbase\ using $(r+1)$ colors.      
	\een
\end{appendices}

\chapter{Codes with Availability} \label{ch:avail}

Codes with availability are codes which are such that a desired code symbol that is unavailable, can be reconstructed in $t$ different ways from the remaining code symbols, i.e., reconstructed by employing $t$ disjoint recovery sets, each recovery set being of size $\leq r$.   The desired symbol here will belong to either a collection of code symbols that form an information set, or else, the entire set of code symbols. The corresponding codes will be termed as {\em information-symbol availability} or {\em all-symbol availability} codes respectively.  The motivation here is that the unavailable symbol could in practice represent `hot' data, i.e., data such as a newly-released music file, that is in high demand.   

This requirement can be phrased in terms of the  parity check (p-c) matrix of the code.  Corresponding to any desired symbol $c_i$, the rowspace of the p-c matrix will contain $t$ rows whose support is of the form $\{ \{i\} \cup A_i \} \subseteq [n]$, $|A_i| \leq r$, $A_i \cap A_j = \phi, i \neq j$.   Important problems on this topic are determination of the rate $R$ and minimum distance $d_{\min}$ of such a code for given $(r,t)$, as well as providing constructions of codes that are efficient with respect to these two metrics.  It turns out that some of the best-performing codes satisfy a more stringent criterion, which we term here as {\em strict availability}.  A strict availability code is an availability code possessing p-c matrix $H$ with a sub matrix $H_a$ formed by a subset of rows of $H$ with $H_a$ such that each row has Hamming weight equal to $(r+1)$ and each column has Hamming weight $t$. The rate and minimum distance of codes with strict availability are also explored. 

Codes with availability are similar to codes with single-step, majority-logic decoding.  The only difference is that in the case of an availability code, each recover set is restricted to have weight $\leq r$ whereas under majority logic decoding, no such constraint is present.  

\paragraph{Organization of the Chapter} The chapter is organized as follows.  We begin by formally defining codes with availability and strict availability (SA) in Section~\ref{sec:avail_defns}.  Section \ref{sec:constructions} examines constructions for availability codes in the literature.  Sections~\ref{sec:avail_alph_dep_ch5} and \ref{sec:avail_alph_ind_ch5} examine alphabet-size dependent and alphabet-size independent bounds respectively. Section \ref{sec:Asym_ch5} examines optimal tradeoff between rate and fractional minimum distance for a special class of availability codes. A lower bound on block length of codes with strict availability is given in Section~\ref{sec:strict_avail_ch5}.  The final section, Section~\ref{sec:avail_summary} presents a summary.  

\paragraph{Contributions}  Contributions of the thesis on the topic of availability codes include: 
\ben
\item Improved field-size dependent upper bounds (Section \ref{sec:NewAlphBounddmin_ch5}) on the minimum distance $d_{\min}$ (for given $(n,k,r,t)$) and dimension (for given $(n,d_{\min},r,t)$) of an availability code.
\item Improved field-size independent upper bound (Section  \ref{sec:Newdminbound_ch5}) on the minimum distance of an availability code and improved upper bound on rate (Section \ref{sec:SubSecTranspose_ch5}) of an SA code.
\item Exact expression for maximum possible fractional minimum distance (Section \ref{sec:Asym_ch5}) for a given rate for a special class of availability codes as $\ell \rightarrow \infty$ where each code in this special class is a subcode or subspace of direct product of $\ell$ copies of an availability code with parameters $r,t$ for some $\ell$.
\item A lower bound on block length of an SA code (Section \ref{sec:strict_avail_ch5}) and showed that SA codes achieving the lower bound on block length exists iff a Balanced Incomplete Block Design (BIBD) with certain parameters exist. We then present two SA codes based on two well-known BIBDs. The codes based on the two BIBDs have maximum possible rate (for the minimum possible block length) assuming the validity of a well-known conjecture in the literature of BIBDs.  
\een

\section{Definitions} \label{sec:avail_defns} 

In this chapter, we will restrict ourselves to linear codes.   Throughout the chapter, a codeword in an $[n,k]$ linear code will be represented by $\underline{c}=[c_1,c_2,\cdots, c_{n-1},c_n ]$ where $c_i$ denotes the $i$th code symbol. Throughout this chapter, we will use the term weight to refer to Hamming Weight.	 

\begin{defn}
	An $(n,k,r,t)$ \noindent all-symbol (AS) {\em availability} code over a finite field $\fq$, is an $[n,k]$ linear code $\mathcal{C}$ over $\fq$ having the property that in the event of a single but arbitrary erased code symbol $c_i$, (i) there exist $t$ recovery sets $\{R^i_j\}_{j=1}^t$ which are pair-wise disjoint and of size $|R^i_j| \leq r$ with $R^i_j \subseteq [n]-\{i\}$ such that (ii) for each $j, 1 \leq j \leq t$, $c_i$ can be expressed in the form:
	\bea
	c_i = \sum\limits_{\ell \in R^i_j} a_{\ell} c_{\ell}, \  \text{$a_{\ell}$ depends on $(i,R^i_j)$ and $a_{\ell} \in \fq$ }. \label{eq:Av:property_ch5}
	\eea
\end{defn}
If in the definition above, instead of $c_i$ being an arbitrary erased code symbol, we were to restrict $c_i$ to be an erased code symbol belonging to an {\em information set} $\{c_j : j \in I \}$ with $I \subseteq [n], \text{dim}(\mathcal{C}|_{I})=k$ of the code, then the code $\mathcal{C}$ would be called an information symbol (IS) availability code.   This chapter and the literature on the subject of availability codes is concerned for the most part, with AS availability codes.  For this reason throughout the remainder of this chapter, when we use the term an availability code, we will mean an AS-availability code. The parameter $r$ is termed as the locality parameter.

For an availability code, let the set of codewords in the dual code associated with the recovery sets $R^i_j$, $\forall i \in [n], j \in [t]$ form the rows of a matrix $D$. The parity check matrix of an availability code can be written in the form  $H=[H_a^T \ H_b^T]^T$ where the rows of $H_a$ are all the distinct codewords in the rows of the matrix $D$ and where the matrix $H_b$ is a full rank matrix containing the remaining codewords in the dual code such that $\text{rank}(H)=n-k$. Hence $H_a$ does not necessarily have linearly independent rows unless otherwise specified. Clearly the Hamming  weight of each row of $H_a$ is $\leq (r+1)$ and each column of $H_a$, has weight $\geq t$.

\begin{defn}
	Codes with Strict Availability (SA codes) are simply the subclass of availability codes with parity check matrix $H=[H_a^T \ H^T_b]^T$ as described above such that each row of $H_a$ has weight equal to $(r+1)$ and each column of $H_a$ has weight equal to $t$.  Thus the number $m$ of rows of $H_a$ must satisfy $m(r+1)=nt$.  Further, it is clear that if the support sets of the rows in $H_a$ having a non-zero entry in the $i^\text{th}$ column are given respectively by $S^{(i)}_j, j=1,2,\cdots t$, then we must have by the disjointness of the recovery sets, that $S^{(i)}_j \cap S^{(i)}_l =   \{i \} , \forall\ 1 \leq j \neq l \leq t$. Each code symbol $c_i$ in an SA code is thus protected by a collection of $t$ `orthogonal' parity checks, each of weight $(r+1)$. An SA code with parameters $n,k,r,t$ will also be called as an $(n,k,r,t)$ SA code.
\end{defn}

\section{Constructions of Availability Codes} \label{sec:constructions}

In this section, we provide an overview of some efficient constructions of availability codes to be found in the literature. 


\paragraph{The Product Code}  Consider the $[(r+1)^t, r^t]$ product code in $t$ dimensions.  Clearly this is an $(n=(r+1)^t,k=r^t,r,t)$ availability code, having rate $R=(\frac{r}{r+1})^t$. 

\paragraph{The  Wang et al. Construction}  For any given parameter pair $(r,t)$,  Wang et al. \cite{WanZhaLiu} provide  a construction for an $(n,k,r,t)$ availability code which is defined through its parity-check matrix $H=H_a$. Let $S$ be a set of $\ell=(r+t)$ elements.  Then in the construction, rows of $H$  are indexed by distinct subsets of $S$ of cardinality  $(t-1)$ and columns are indexed by distinct subsets of $S$ of cardinality $t$.  We set $h_{ij}=1$ if the $i$-th subset of $S$ of cardinality $(t-1)$ belongs to the $j$-th subset of $S$ of cardinality $t$ and zero otherwise.  Thus $H$ is of size ${\ell \choose t-1} \times {\ell \choose t}$. It is easy to verify that each row of $H$ has constant row weight $(r+1)$ and each column of $H$ has constant weight $t$.  It turns out that the rank of $H$ is given by ${\ell-1 \choose t-1}$ and that $H$ defines an $(n,k,r,t)$ availability code, having parameters: $n= {\ell \choose t}, k= {\ell \choose t} - {\ell-1 \choose t-1}$ and rate $R= \frac{r}{r+t}$. 
Thus this code provides improved rate in comparison with the product code.  Since $ {r+t \choose t} \ <  \ (r+1)^t$,  the code has smaller block length as well. 

\paragraph{Direct-Sum Construction}  It is shown in \cite{KadCal} that the direct sum of $m$ copies of the $[7,3]$ Simplex code yields an SA code with parameters $(7m,3m,2,3)$ having maximum possible rate for a binary SA code for $n=7m,r=2,t=3$.

\paragraph{Availability Codes Derived from Steiner Triple Systems and Cyclic Availability Codes}  We will see that the incidence matrix of Steiner Triple System (STS) yields parity check matrix for an availability code for $t=3$ and specific $r$ of the form, $2^s-2$. These results were obtained by \cite{WanZhaLin1} in parallel with work in \cite{BalPraKum} carried out as part of the present thesis, dealing with code construction based on STS and described later in this chapter, see Section~\ref{sec:strict_avail_ch5}. 

It was also shown in \cite{WanZhaLin1} that one can puncture the incidence matrix of an STS and derive availability codes for multiple, different values of $r$ for $t=3$.   Also present in \cite{WanZhaLin1} is a construction of a cyclic availability code having parameters: $[n=q^m-1, k=q^{m-1}-1, d_{\min}=\max(q+1,t+1)]$ with $r=m-1$ and $t=em$ over $\mathbb{F}_p$ where $q=p^e$ for some prime number $p$ and positive number $e$. 

\paragraph{Availability Codes Based on Majority Logic Decodable Codes}
In \cite{HuaYaaUchSie}, authors point out constructions for availability codes based on majority logic decodable codes. One of these codes based on projective geometry is presented in this chapter in Section~\ref{sec:strict_avail_ch5} where we present this construction as an example of a minimum block length construction of SA codes based on BIBD.

\paragraph{Anticode based Construction}
In \cite{SilZeh}, the authors present constructions which are optimal w.r.t the alphabet-dependent bounds given in \eqref{eq:CadMaz1_ch5} and the Griesmer bound for linear codes given in Theorem 1 of \cite{BelGueMenSal} for $r \in \{2,3\}$ and large $t$. The authors take the Simplex code and generate various codes from it by puncturing the generator matrix of the code at positions corresponding to the columns of the generator matrix of an anticode \cite{Far}. By selecting multiple anticodes, they generate various availability codes from the Simplex code. They consider both the binary and the $q$-ary Simplex code. They show that their constructions for some parameters are optimal either w.r.t the alphabet-dependent bounds given in \eqref{eq:CadMaz1_ch5} or else, the Griesmer bound  of coding theory.

\paragraph{Constructions of Availability Codes having Large Minimum Distance}

By an availability code with large minimum distance we mean here an availability code where the minimum distance is a fraction of the block length. Constructions with minimum distance $\geq n-k-\lceil \frac{2k}{r-1} \rceil + 1$ for $t=2$ appear  in \cite{TamBar_LRC}. In \cite{TamBar_LRC}, the authors also provide a general method of constructing codes with parameters $(r,t)$ with explicit constructions provided for $t=2$. A special case of this general construction given in \cite{TamBar_LRC}, was provided in \cite{BhaTha} having parameters $(n=(r+1)^t,k=r+1,r,t,d_{\min}=n-(r+1))$ which meets the bound \eqref{eq:TamoBargDmin_ch5} given in \cite{TamBarFro}. Note that this construction is of low rate.   For other constructions where the minimum distance being a fraction of block length, please refer to \cite{TamBarFro,BarTamVla,KruFro}. 

\section{Alphabet-Size Dependent Bounds on Minimum Distance and Dimension} \label{sec:avail_alph_dep_ch5}

The bounds that are discussed in this section, are bounds that are alphabet-size dependent.  {\em For brevity, throughout this section, we will just say upper bound on minimum distance when we mean a alphabet-size dependent bound on minimum distance etc. }  

We begin with a survey of existing bounds on the minimum distance and dimension of an availability code. All the existing results follow an approach based on shortening.  We next present a new improved upper bound on the minimum distance and dimension of an availability code. This new bound makes full use of the shortening-based approach and results in the tightest-known upper bounds on the minimum distance and dimension of an availability code. 

\subsection{Known Bounds on Minimum Distance and Dimension}

In this section, we present a survey of upper bounds on minimum distance and dimension of an availability code existing in the literature. All the bounds presented here are based on code-shortening. Under this approach, the shortening is carried out in such a way that the necessary constraints on an LR code have either disappeared or are present in weaker form.  Bounds on classical codes are then applied to this shortened code, which can be translated into bounds on the parent, availability code.

The theorem below which appeared in \cite{CadMaz} provides an upper bound on the dimension of an availability code with $t=1$ that is valid even for nonlinear codes.  The `dimension' of a nonlinear code \calc\ over an alphabet $\mathbb{Q}$ of size $q=|\mathbb{Q}|$ is defined to be the quantity $k=\log_q(| \calc |)$. 
\ben
\item Let $d^{(q)}(n,k)$ be the maximum possible minimum distance of a classical $[n,k]$ block code (not necessarily an LR code) over $\mathbb{F}_q$,
\item Let $k^{(q)}(n, d)$ be the largest possible dimension of a code (not necessarily an LR code) over $\mathbb{F}_q$ having block length $n$ and minimum distance $d$.	 
\een
\begin{thm} \cite{CadMaz} \label{thm:CadMaz1_ch5}
	For any $(n,k,d_{\min})$ code  \calc\ that is an availability code with $t=1$ and with locality parameter $r$ over an alphabet $\mathbb{Q}$ of size $q=|\mathbb{Q}|$, 
	\bea
	k \  \leq \  \min_{\ell \in \mathbb{Z}_{+}} [ \ell r+ k^{(q)}(n-\ell (r+1),d_{\min}) ],  \label{eq:CadMaz1_ch5}
	\eea
\end{thm}
\bpf (Sketch of proof) 
The bound holds for linear as well as nonlinear codes. In the linear case, with $\mathbb{Q}=\mathbb{F}_q$, the derivation proceeds as follows.  Let $G$ be a $(k \times n)$ generator matrix of the availability code ${\cal C}$.  Then it can be shown that for any integer $\ell > 0$,  there exists an index set ${\cal I}$ such that $|{\cal I}|=\min(\ell(r+1),n)$ and $\text{rank} \left( G \mid_{\cal I} \right) \ = \ s \leq \ell r$ where $ G \mid_{\cal I}$ refers to the set of columns of $G$ indexed by ${\cal I}$. This implies that \calc\ has a generator matrix of the form (after permutation of columns): 
\bean
G & = & 
\left[ \begin{array}{cc} 
	\underbrace{A}_{(s \times |{\cal I}|)} & B \\ 
	\left[ 0 \right]  & D
\end{array} \right] . 
\eean
In turn, this implies that the rowspace of $D$ defines an $[n-\ell(r+1),k-s \geq k-\ell r,\geq d_{\min}]$ code over $\mathbb{F}_q$, if $k-\ell r > 0$.  
It follows that $k \leq \ell r + k^{(q)}(n-\ell (r+1),d_{\min})$ and the result follows.  Note that the row space of $D$ corresponds to a shortening $\calc^{S}$ of \calc\ with respect to the coordinates ${\cal I} \subseteq [n]$.  The proof in the general case is a (nontrivial) extension to the nonlinear setting.  \epf 
\paragraph{Key Ingredient of the Shortening-Based Bounds and Generalizations}

Note that a key ingredient in the above proof is showing that there exists a set of $\ell (r+1)$ code symbols in the code, corresponding to which there exist $\ell$ linearly independent codewords of weight $\leq r+1$ in the dual code, each of whose support is wholly contained in the support of these $\ell (r+1)$ code symbols.  The same ingredient is also present in the bounds derived in \cite{HuaYaaUchSie,KruFro} which we will next present.   

As a contribution of the present thesis on this topic, we will generalize this idea in Section \ref{sec:NewAlphBounddmin_ch5} and provide the best-known estimate of the number of linearly independent codewords of weight $\leq r+1$ in the dual code, each having support contained in the support of a carefully chosen set of (say) $x$ code symbols. This estimate will yield the tightest-known bound todate.

In \cite{HuaYaaUchSie}, the authors provide  an upper bound on the minimum distance and dimension of a $[n,k,d_{\min}]$ code with IS availability (the bound thus also applies to AS  availability codes as well) over an alphabet of size $q$:
\begin{thm} 
	\cite{HuaYaaUchSie} For any $[n,k,d_{\min}]$ code  \calc\ that is an information-symbol availability code with parameters $(r,t)$ over $\mathbb{F}_q$:
	\bea
	d_{\min} \leq \min_{\{(x,y) : 1 \leq x \leq \lceil \frac{k-1}{(r-1)t+1} \rceil, x \in \mathbb{Z}^{+} , y \in [t]^x, A(r,x,y) < k \}} d^{(q)}(n-B(r,x,y),k-A(r,x,y)), \label{eq:dminAlp:ch5} \\
	k \leq \min_{\{(x,y): 1 \leq x \leq \lceil \frac{k-1}{(r-1)t+1} \rceil, x \in \mathbb{Z}^{+} , y \in [t]^x, A(r,x,y) < k \}} (A(r,x,y)+k^{(q)}(n-B(r,x,y),d_{\min})), \label{eq:dimenAlp:ch5} 
	\eea
	where
	\bean
	A(r,x,y) &=& \sum_{j=1}^{x} (r-1)y_j + x, \\
	B(r,x,y) &=& \sum_{j=1}^{x} r y_j + x.
	\eean
\end{thm}

The above bound is derived by finding a set of $B(r,x,y)$ code symbols such that there are at least $B(r,x,y)-A(r,x,y)$ linearly independent codewords in the dual each of whose support is completely contained in the support of these $B(r,x,y)$ code symbols. This is the key ingredient, we had alluded to above.

In \cite{BalKumAvail}, as part of the contribution of the thesis on this topic, we provide the bound below on the minimum distance and dimension of a $[n,k,d_{\min}]$ availability code over an alphabet of size $q$:
\begin{thm} 
	\cite{BalKumAvail} For any $(n,k,r,t)$ availability code \calc  over $\mathbb{F}_q$ with minimum distance $d_{\min}$:
	\bea
	d_{\min} & \leq & \min_{\{i: 
		\ 1+(r-1) i < k\}} d^{(q)}(n-1-i r,k-1-(r-1)i),  \label{eq:dminKF_ch5} \\
	k & \leq & \min_{\{i: \ i r+1 \leq n-d_{\min}\}} (1+(r-1)i +k^{(q)}(n-1-i r,d_{\min})). \label{eq:dimenKF_ch5} 
	\eea
\end{thm}
\footnote{       The above bounds \eqref{eq:dminKF_ch5},\eqref{eq:dimenKF_ch5} were first derived by us in \cite{BalKumAvail} (see version 1). Subsequently, the same bounds were derived in \cite{KruFro}.} 
The above bounds \eqref{eq:dminKF_ch5},\eqref{eq:dimenKF_ch5} are derived by finding a set of $1+\ell r$ code symbols such that there are at least $\ell$ linearly independent codewords in the dual each of whose support is completely contained in the support of these $1+\ell r$ code symbols.  Once again, this is the key ingredient, we had alluded to above.

We will now generalize these ideas in this following and provide the best-known estimate of the number of linearly independent codewords of weight $\leq r+1$ in the dual code, each having support contained in the support of a carefully chosen set of (say) $x$ code symbols. This estimate will give the tightest known bound and is based on this shortening approach. This is estimated as Minimum Support Weight  (MSW) Sequence defined in the following which gives an upper bound $e_i$ on the size of a set of code symbols that can be chosen so that we are guaranteed that there exists $i$ linearly independent codewords of weight $\leq r+1$ in the dual code  each of whose support is contained in the support of this set of code symbols. Before we look at MSW sequence, we define Generalized Hamming Weight (GHW) of a code for which MSW Sequence acts as an upper bound.
As explained before, our new bound makes use of the technique of code shortening and the GHWs of a code provide valuable information about shortened codes. 
%
%
%
%
%
\subsubsection{GHW and the Minimum Support Weight Sequence} 
We will first define the Generalized Hamming Weights of a code, introduced in \cite{Wei}, and also known as  Minimum Support Weights (MSW) (see \cite{HelKloLevYtr}) of a code. In this thesis we will use the term Minimum Support Weight (MSW).
\begin{defn}
	The $i^{th}$ Minimum Support Weight (MSW) $d_i$  (equivalently, the $i$th Generalized Hamming Weight) of an $[n,k]$ code $\calc$ is the minimum possible value of the cardinality of the support of an $i$-dimensional subcode of C, i.e.,
	\bea
	d_i(\calc) = d_i = \min_{\{\mathcal{D} : \ \mathcal{D} < \calc, \ \dim(\mathcal{D})=i \} }|\text{supp}(\mathcal{D})|, \label{eq:MSW1_ch5}
	\eea
	where the notation $\mathcal{D} < \calc$ denotes a subcode $\mathcal{D}$ of $\calc$ and where $\text{supp}(\mathcal{D})=\cup_{\underline{c} \in \mathcal{D}} \text{supp}(\underline{c})$ (called the support of the code $\mathcal{D}$).
\end{defn}
Although the MSW definition applies to any code, the interest in this thesis, is on its application to a restricted class of codes that we re-introduce (it is already introduced in Chapter \ref{ch:LRSingleErasure}) here. 
\begin{defn} [Canonical Dual Code] \label{def:Can_dualcode_ch5}
	By a canonical dual code, we will mean an $[n,m_{cd}]$ linear code \calc satisfying the following: \calc contains a set $\{\underline{c}_i\}_{i=1}^b$ of $b \leq m_{cd}$ linearly independent codewords of Hamming weight $\leq r+1$, such that the sets $S_i = supp(\underline{c}_i)$, $1 \leq i \leq b$ cover $[n]$, i.e., 
	\bean
	[n] & \subseteq & \cup_{i=1}^b S_i .
	\eean
\end{defn}
As it turns out, the dual code of an availability code is an example of a canonical dual code and this is the reason for our interest in the MSWs of this class of codes. 

\begin{thm} \cite{PraLalKum} \label{thm:msw_sequence:ch5}
	Let \calc\ be a canonical dual code with parameters $(n,m_{cd},r,b)$ and support sets $\{S_i \mid 1 \leq i \leq b\}$ as defined in Definition \ref{def:Can_dualcode_ch5}. Let $d_i$ denote the $i$th $1 \leq i \leq m_{cd}$ MSW of \calc. Let $f_i$ be the minimum possible value of cardinality of the union of any $i$ distinct support sets $\{S_{j} :j \in T \}$, $|T| = i$ , $1 \leq i \leq b$ i.e., 
	\bean
	f_i = \min_{\{T: \ T \subseteq [b], |T|=i\}} |\cup_{j \in T} S_j|.
	\eean
	Let $b_1 \leq b$. Let the integers $\{e_i\}_{i=1}^{b_1}$ be recursively defined as follows: 
	\bea
	e_{b_1} & = & n, \\
	e_{i-1} & = &  \min \{e_i, e_i-\left \lceil \frac{2e_i}{i} \right \rceil + r+1\}, \ \  2 \leq i \leq b_1.  \label{eq:Sequence1_ch5}
	\eea
	Then 
	\bean
	d_i \ \leq \ f_i & \leq & e_i, \text{ for all } \ \ 1 \leq i \leq b_1. 
	\eean
\end{thm}
Note that in \cite{PraLalKum}, the Theorem \ref{thm:msw_sequence:ch5} is proved for the case $b_1=b$, $|S_j|=r+1$, $\forall j \in [b]$ but we observe from the proof of Theorem \ref{thm:msw_sequence:ch5} (proof for $f_i \leq e_i$) given in \cite{PraLalKum} that the Theorem \ref{thm:msw_sequence:ch5} is also true for any $b_1 \leq b$, $|S_j| \leq r+1$, $\forall j \in [b]$.
We will refer to the sequence $\{e_i\}_{i=1}^{b_1}$ appearing in the Theorem \ref{thm:msw_sequence:ch5} above as the {\em Minimum Support Weight (MSW) Sequence} associated to parameter set $(n,b_1,r)$.  


Hence this MSW sequence gives an upper bound $e_i$ on the size of a set of code symbols that can be chosen so that we are guaranteed that there exists $i$ linearly independent codewords of weight $\leq r+1$ in the dual code  each of whose support is contained in the support of this set of code symbols.
In the subsection below, we derive a new alphabet-size dependent bound on minimum distance and dimension, that are expressed in terms of the MSW sequence.

\subsection{New Alphabet-Size Dependent Bound on Minimum Distance and Dimension Based on MSW} \label{sec:NewAlphBounddmin_ch5}
In this subsection, we present a new field-size dependent bound on the minimum distance and dimension of an availability code \calc with parameters $(n,k,r,t)$.  The bound is derived in terms of the MSW sequence associated with the dual code $\calc^{\perp}$ of \calc.  The basic idea is to shorten the availability code to a code with shortened block length  with size of the set of code symbols set to zero equal to the $i$th term $e_i$ in the MSW sequence for some $i$.  Theorem~\ref{thm:msw_sequence:ch5} provides a lower bound on the dimension of this shortened code. Classical bounds on the parameters of this shortened code are shown to yield bounds on the parameters of the parent availability code. 

\begin{thm} \label{thm:dmin1:ch5}
	Let 
	\bean
	\rho(r,t)&=&\begin{cases}
		\frac{r}{r+1} & \text{ if } t=1\\
		\frac{r}{r+2} & \text{ if } t=2\\
		\frac{r^2}{(r+1)^2} & \text{ if } t=3 \\
		\frac{1}{\prod_{j=1}^{t}(1+\frac{1}{jr})} & \text{ if } t>3	,				
	\end{cases}
	\eean
	Let $\mathcal{C}$ be an $(n,k,r,t)$ availability code over a field $\mathbb{F}_q$ with minimum distance $d_{\min}$. Let $d^q_{\text{min}}(n,k,r,t)$ be the maximum possible minimum distance of an $(n,k,r,t)$ availability code over the field $\mathbb{F}_q$. Then:
	\bea
	d_{\min} \leq  \min_{i \in S} \ d^q_{\text{min}}(n-e_i,k+i-e_i,r,t) \leq \min_{i \in S} d^{(q)}(n-e_i,k+i-e_i), \ \label{eq:Mindist_Bound1:ch5}
	\eea
	\bea
	k \  \leq \  \min_{\{i: \  e_i < (n-d_{\min}+1) \}} [ e_i-i + k^{(q)}(n-e_i,d_{\min}) ] . \label{eq:GHW_dimension_Bound:ch5}
	\eea
	where
	\ben
	\item[(i)]   $b_1   = \lceil n(1-\rho(r,t)) \rceil$,
	\item[(ii)]  $S=\{ i: e_i-i < k, 1 \leq i \leq b_1 \}$,
	\item[(iii)] $d^{\perp}_i = d_i(\calc^{\perp}) \leq e_i$, $\forall 1 \leq i \leq b_1$.
	\een
\end{thm}
\begin{proof}
	Since $\calc$ is an availability code with parameters $(r,t)$, we have that $\calc^{\perp}$ is an $[n,m_{cd} = n-k]$ canonical dual code with locality parameter $r$ and $b \geq \lceil n(1-\rho(r,t)) \rceil $. $b \geq \lceil n(1-\rho(r,t)) \rceil $ because $\rho(r,t)$ is an upper bound on maximum possible rate of an availability code with parameters $(r,t)$ (due to equations \eqref{Thm1}, \eqref{Thm2}, \eqref{eq:TamoBargRate_ch5}). 
	We now set $b_1 = \lceil n(1-\rho(r,t)) \rceil $. Hence from Theorem \ref{thm:msw_sequence:ch5}, $d_j(\calc^{\perp}) \leq e_j$, $\forall 1 \leq j \leq b_1$. For simplicity, let us write $d_j^{\perp} = d_j(\calc^{\perp})$, $\forall 1 \leq j \leq b_1$. 
	Next, fix $i$ with $1 \leq i \leq b_1$. Let $S'=\{s_1,...,s_{d_i^{\perp}} \}$ be the support of an $i$ dimensional subspace or subcode of $\mathcal{C}^{\perp}$ with the support having cardinality exactly $d_i^{\perp}$ in $\mathcal{C}^{\perp}$. Add $e_i-d_i^{\perp}$ arbitrary extra indices to $S'$ and let the resulting set be $S$. Hence $S' \subseteq S \subseteq [n]$ and $|S|=e_i$. Now shorten the code $\mathcal{C}$ in the co-ordinates indexed by $S$ i.e., take $\mathcal{C}^S=\{\underline{c}|_{S^c} : \underline{c} \in \mathcal{C}, \underline{c}|_S = \underline{0}\}$ where $S^c$ is the compliment of $S$ and $\underline{c}|_{A} = [c_{j_1},...,c_{j_{|A|}}]$ for a set $A=\{j_1,...,j_{|A|}\}$ with $j_{\ell} < j_{\ell+1}$,$\forall \ell \in [|A|-1]$.  The resulting code $\mathcal{C}^S$ has block length $n-e_i$, dimension $\geq n-e_i-(n-k-i)=k+i-e_i$ and minimum distance $\geq d_{\min}$ (if $k+i-e_i > 0$) and the resulting code $\mathcal{C}^S$ is also an availability code with parameters $(r,t)$. Hence:
	\bean
	d_{\min} \leq  \min_{i \in S} \ d^q_{\text{min}}(n-e_i,k+i-e_i,r,t) \leq \min_{i \in S} d^{(q)}(n-e_i,k+i-e_i).
	\eean 
	The proof of \eqref{eq:GHW_dimension_Bound:ch5} follows from the fact that $\text{dim}(\mathcal{C}^S) \geq k+i-e_i$ and the fact
	\bea
	k+i-e_i \leq  \text{dim}(\mathcal{C}^S)  \leq k^{(q)}(n-e_i,d_{\min}).
	\eea
\end{proof}
In the following, we will analyze MSW sequence and form an upper bound on MSW sequence. Based on this upper bound, we will conclude that our upper bound on minimum distance and dimension given in \eqref{eq:Mindist_Bound1:ch5},\eqref{eq:GHW_dimension_Bound:ch5} are the tightest known bounds.

\paragraph{Analysis of MSW sequence}

\begin{lem}  \label{lem:LemmaOnei_ch5}
	Let $\{e_i: 1\leq i \leq b_1 \}$ be as defined in Theorem \ref{thm:msw_sequence:ch5}. Let $e_{b_1} = n \leq (b_1-\ell_1)r$ for some $\ell_1$. Then for $1 \leq i \leq b_1$,
	\bea
	e_{i} &\leq& ir+\frac{(b_1-i)i}{(b_1-1)}-\frac{(i-1)i}{b_1(b_1-1)} \ell_1 r. \label{eq:Recur1_ch5} 
	\eea
\end{lem}
\begin{proof}
	
	We will prove the upper bound on $e_i$ inductively.
	Let by induction hypothesis,
	\bea
	e_{b_1-i} &\leq& (b_1-i)r+\frac{(b_1-i)i}{(b_1-1)}-\frac{(b_1-i-1)(b_1-i)}{b_1(b_1-1)} \ell_1 r. \label{eq:Recur} 
	\eea
	For $i=0$, the inequality \eqref{eq:Recur} becomes:
	\bea
	e_{b_1} &\leq & b_1 r- \ell_1 r.
	\eea
	which is true. Hence we proved the intial condition of induction.
	
	Substituting the above expression \eqref{eq:Recur} for $e_{b_1-i}$ in the recursion given in \eqref{eq:Sequence1_ch5}  without ceiling:
	\bea
	e_{b_1-i-1}&\leq&\frac{b_1-i-2}{b_1-i} e_{b_1-i} + r+1, \notag \\
	e_{b_1-i-1}&\leq&\frac{b_1-i-2}{b_1-i} ((b_1-i)r+\frac{(b_1-i)i}{(b_1-1)}-\frac{(b_1-i-1)(b_1-i)}{b_1 (b_1-1)} \ell_1 r) + r+1, \notag 
	\eea
	\bea
	e_{b_1-i-1}&\leq& ((b_1-i-1)r+\frac{(b_1-i-2)i}{(b_1-1)}+1-\frac{(b_1-i-2)(b_1-i-1)}{b_1(b_1-1)} \ell_1 r),  \notag \\
	e_{b_1-i-1}&\leq& ((b_1-i-1)r+\frac{(b_1-i-1)(i+1)}{(b_1-1)}-\frac{(b_1-i-2)(b_1-i-1)}{b_1(b_1-1)} \ell_1 r).  \notag
	\eea
	
	Hence the inequality \eqref{eq:Recur} is proved by induction.
	
	Hence we proved that:
	\bea
	e_{i} &\leq& ir+\frac{(b_1-i)i}{(b_1-1)}-\frac{(i-1)i}{b_1 (b_1-1)} \ell_1 r. 
	\eea
	
\end{proof}

\begin{cor} \label{cor:CorollaryOnei_ch5} Let $\{e_i: 1\leq i \leq b_1 \}$ be as defined in Theorem \ref{thm:msw_sequence:ch5} and let $t \geq 2$ and $n \leq (b_1-\ell_1)r$ for some $\ell_1$. Then the following inequalities hold:
	
	\ben
	\item $e_i < ir+1$ whenever $i \geq \frac{b_1 + x}{x+1}$ where $x=\frac{\ell_1 r }{b_1}$ or whenever $i \geq \frac{b_1 +r-\frac{1}{1-\rho(r,t)}}{r-\frac{1}{1-\rho(r,t)}+1}$. 
	\item $e_i \leq ir+1$, $\forall i \geq 1,$ $r\geq 2$. 
	\een
	where $b_1, \rho(r,t)$ are as defined in the Theorem \ref{thm:dmin1:ch5}.
\end{cor}
\begin{proof}
	From the inequality \eqref{eq:Recur1_ch5}, it can seen that $e_i < ir+1$ whenever $i \geq \frac{b_1 +x}{x+1}$ where $x=\frac{\ell_1 r }{b_1}$. It can also be seen from the inequality $n \leq (b_1-\ell_1)r$ that $x \leq r-\frac{n}{b_1}$. Choosing $x = r-\frac{1}{1-\rho(r,t)}$, we get  $e_i < ir+1$ whenever $i \geq \frac{b_1 +r-\frac{1}{1-\rho(r,t)}}{r-\frac{1}{1-\rho(r,t)}+1}$.
	Further it can seen by using the recursion \eqref{eq:Sequence1_ch5} that $e_i \leq ir+1$, $\forall i \geq 1,t \geq 2,$ $r\geq 2$. This is because if you substitute $e_i \leq ir+1$ in the recursion \eqref{eq:Sequence1_ch5} (with ceiling), you will get $e_{i-1} \leq (i-1)r+1$.
\end{proof}

\paragraph{Tightness of the Bounds Given in Theorem \ref{thm:dmin1:ch5}}
In the following the word bound refers to an upper bound. From the Corollary \ref{cor:CorollaryOnei_ch5}, it can be seen that the bound \eqref{eq:Mindist_Bound1:ch5} is tighter than the bounds \eqref{eq:dminAlp:ch5},\eqref{eq:dminKF_ch5}.  For the same reason, the bound \eqref{eq:GHW_dimension_Bound:ch5} is tighter than the bounds \eqref{eq:CadMaz1_ch5}, \eqref{eq:dimenAlp:ch5}, \eqref{eq:dimenKF_ch5} .
Hence our bounds \eqref{eq:Mindist_Bound1:ch5},\eqref{eq:GHW_dimension_Bound:ch5} are the tightest known bounds on minimum distance and dimension of an availability code over $\fq$.
We here note that the bounds \eqref{eq:Mindist_Bound1:ch5}, \eqref{eq:GHW_dimension_Bound:ch5} apply even if we replace $e_i$ with any other upper bound on $i^{th}$ MSW. Hence the bounds \eqref{eq:Mindist_Bound1:ch5}, \eqref{eq:GHW_dimension_Bound:ch5} are general in that sense. We give better upper bounds on $i^{th}$ MSW than $e_i$ in \cite{BalKum} under restricted conditions. We do not include this result here to keep things simple.

\section{Bounds for Unconstrained Alphabet Size} \label{sec:avail_alph_ind_ch5}

In this section, we present our results on upper bound on rate of an SA code and upper bound on minimum distance of an availability code. Our upper bound on rate of an SA code is the tightest known bound as $r$ increases for a fixed $t$. Our upper bound on minimum distance of an availability code is the tightest known bound on minimum distance of an availability code for all parameters. 

\subsection{Upper Bounds on Rate for unconstrained alphabet size}

In this section, we first give a survey of upper bounds on rate of an availability code from the existing literature. We will then present a new upper bound on rate of an SA code. Our upper bound on rate of an SA code is the tightest known bound as $r$ increases for a fixed $t$. This upper bound on rate of an SA code is based on the idea that transpose of the matrix $H_a$ in the parity check matrix of an SA code with parameters $(r,t)$ is also a parity check matrix of an SA code with parameters $r'= t-1$ ,$t'=r+1$ ($r'$ is the locality parameter and $t'$ is the number of disjoint recovery sets available for a code symbol). 

\subsubsection{Known Bounds on Rate} \label{sec:KnownRateBounds:ch5}

The following upper bound on the rate of a code with availability was given in \cite{TamBarFro}. The bound is derived based on a graphical approach.
\begin{thm} \cite{TamBarFro} If \calc\ is an $(n,k,r,t)$ availability code over a finite field \fq,\ then its rate $R$ must satisfy: 
	\bea
	R \ = \ \frac{k}{n} \leq \frac{1}{\prod_{j = 1}^{t}(1+\frac{1}{jr})}. \label{eq:TamoBargRate_ch5}
	\eea
\end{thm} 

\bean
\text{Let }R(r,t) & = &  sup_{\{(k,n) : (n,k,r,t) \text{ SA code exists} \}} \frac{k}{n}. 
\eean
After our results were published in \cite{BalKum} (applicable for any field size), the following upper bound on rate of an $(n,k,r,t=3)$ SA code over $\mathbb{F}_2$ was published in \cite{KadCal}:

\bea
R(r,3) \leq \frac{r-2}{r+2}+ \frac{3}{r+1} H_2 \left(\frac{1}{r+2} \right). \label{eq:Kadhe_ch5}
\eea
Comparison of our bound \eqref{eq:Rate_2_ch5} with the above bound \eqref{eq:Kadhe_ch5} appears later.

For $r=2$ and $t\geq 74$, the tightest known upper bound on rate of an $(n,k,r=2,t)$ SA code over $\mathbb{F}_2$ is given in \cite{KadCal} :
\bea
R(2,t) \leq H_2 \left( \frac{1}{t+1} \right). \label{eq:SwaCald1_ch5}
\eea
Our bound in \eqref{eq:Rate_2_ch5} is the tightest known bound as $r$ increases for a fixed $t$, it does not give good bounds for small $r$ like $r=2$ for large $t$.

\subsubsection{A Simple New Upper Bound on Rate of an $(n,k,r,t)$ SA Code:}\label{sec:SubSecTranspose_ch5}
Here we derive a new upper bound on rate of an $(n,k,r,t)$ SA code using a very simple transpose trick. This upper bound on rate of an SA code is the tightest known bound as $r$ increases for a fixed $t$.
\begin{thm}
	
	\bea
	R(r,t) &=& 1-\frac{t}{r+1}+\frac{t}{r+1} R(t-1,r+1), \label{eq:rate_11_ch5} \\
	R(r,t) &\leq& 1-\frac{t}{r+1}+\frac{t}{r+1} \frac{1}{\prod_{j=1}^{r+1}(1+\frac{1}{j(t-1)})}. \label{eq:Rate_2_ch5}
	\eea
	
\end{thm}
\begin{proof}
	Let $R^{(n,q)}(r,t)$ be the maximum achievable rate of an SA code for a fixed $n,r,t$ over the field $\mathbb{F}_q$. If $(n,k,r,t)$ SA code doesn't exist for any $k > 0$ for fixed $n,r,t,q$ then we define $R^{(n,q)}(r,t)$ as follows. $R^{(n,q)}(r,t)= 0$ if there exists a parity check matrix $H=H_a$ satisfying the conditions of an SA code but with $\text{rank}(H_a) = n$ and we set  $R^{(n,q)}(r,t)= -\infty$ if there is no parity check matrix $H=H_a$ satisfying the conditions of an SA code for any value of $\text{rank}(H_a)$. Let us choose $n,r,t,q$ such that $R^{(n,q)}(r,t) \geq 0$. Let $\mathcal{C}$ be an $(n,k,r,t)$ SA code over the field $\mathbb{F}_q$ with rate $R^{(n,q)}(r,t)$. By definition $\mathcal{C}$ is the null space of an $m \times n$ matrix $H=H_a$ with each column having weight $t$ and each row having weight $r+1$ ($H_b$ is an empty matrix as $\mathcal{C}$ is a code with maximum possible rate for the given $n,r,t,q$). This matrix $H$ contains all $t$ orthogonal parity checks protecting any given symbol. Now the null space of $H^T$ (transpose of $H$) corresponds to an $(m,k_2,t-1,r+1)$ SA code over the field $\mathbb{F}_q$ for some $k_2$. Hence we have the following inequality:
	\bean
	rank(H) &=& n(1-R^{(n,q)}(r,t)), \\
	rank(H) &=& rank(H^T) \geq m(1-R^{(m,q)}(t-1,r+1)).
	\eean
	Hence we have (Let $S= \{ (\ell,q) : \ell \geq 0, \ell \in \mathbb{Z}_+,  q=p^w , \text{$p$ is a prime and } w \in \mathbb{Z}_+ \}$):
	\bean
	m(1-R^{(m,q)}(t-1,r+1)) & \leq & n(1-R^{(n,q)}(r,t)). \\
	\text{Using } m(r+1)=nt : \\
	\frac{t}{r+1}(1-R^{(m,q)}(t-1,r+1)) & \leq & (1-R^{(n,q)}(r,t)), \\
	R^{(n,q)}(r,t) & \leq & 1-\frac{t}{r+1}+\frac{t}{r+1}R^{(m,q)} (t-1,r+1), \\
	R^{(n,q)}(r,t) & \leq & 1-\frac{t}{r+1} + \frac{t}{r+1} (sup_{\{(m,q) \in S\}} \ \ R^{(m,q)}(t-1,r+1)), \\
	sup_{\{(n,q) \in S\}}\ \ R^{(n,q)}(r,t) & \leq & 1-\frac{t}{r+1} +\frac{t}{r+1} (sup_{\{(m,q) \in S\}} \ \  R^{(m,q)}(t-1,r+1)), \\
	R(r,t) & \leq & 1-\frac{t}{r+1}+\frac{t}{r+1} R(t-1,r+1).
	\eean
	Now swapping the roles of $H$ and $H^T$ in the above derivation i.e., we take an $(m,k_1,t-1,r+1)$ SA code $\mathcal{C}$ over the field $\mathbb{F}_q$ with rate $R^{(m,q)}(t-1,r+1)$ (Note that $R^{(m,q)}(t-1,r+1) \geq 0$) and repeat the above argument in exactly the same way. By doing so we get :
	\bean
	R(r,t) \geq 1-\frac{t}{r+1}+\frac{t}{r+1} R(t-1,r+1).
	\eean
	Hence we get:
	\bea
	R(r,t) = 1-\frac{t}{r+1}+\frac{t}{r+1} R(t-1,r+1). \label{eq:rate_111_ch5}
	\eea
	Now substituting the rate bound  $R(t-1,r+1) \leq \frac{1}{\prod_{j=1}^{r+1}(1+\frac{1}{j(t-1)})}$  given in \eqref{eq:TamoBargRate_ch5} into \eqref{eq:rate_111_ch5}, we get:
	\bea
	R(r,t) \leq 1-\frac{t}{r+1}+\frac{t}{r+1} \frac{1}{\prod_{j=1}^{r+1}(1+\frac{1}{j(t-1)})}.
	\eea
\end{proof}
\begin{figure}[h!]
	\centering
	\includegraphics[width=5in]{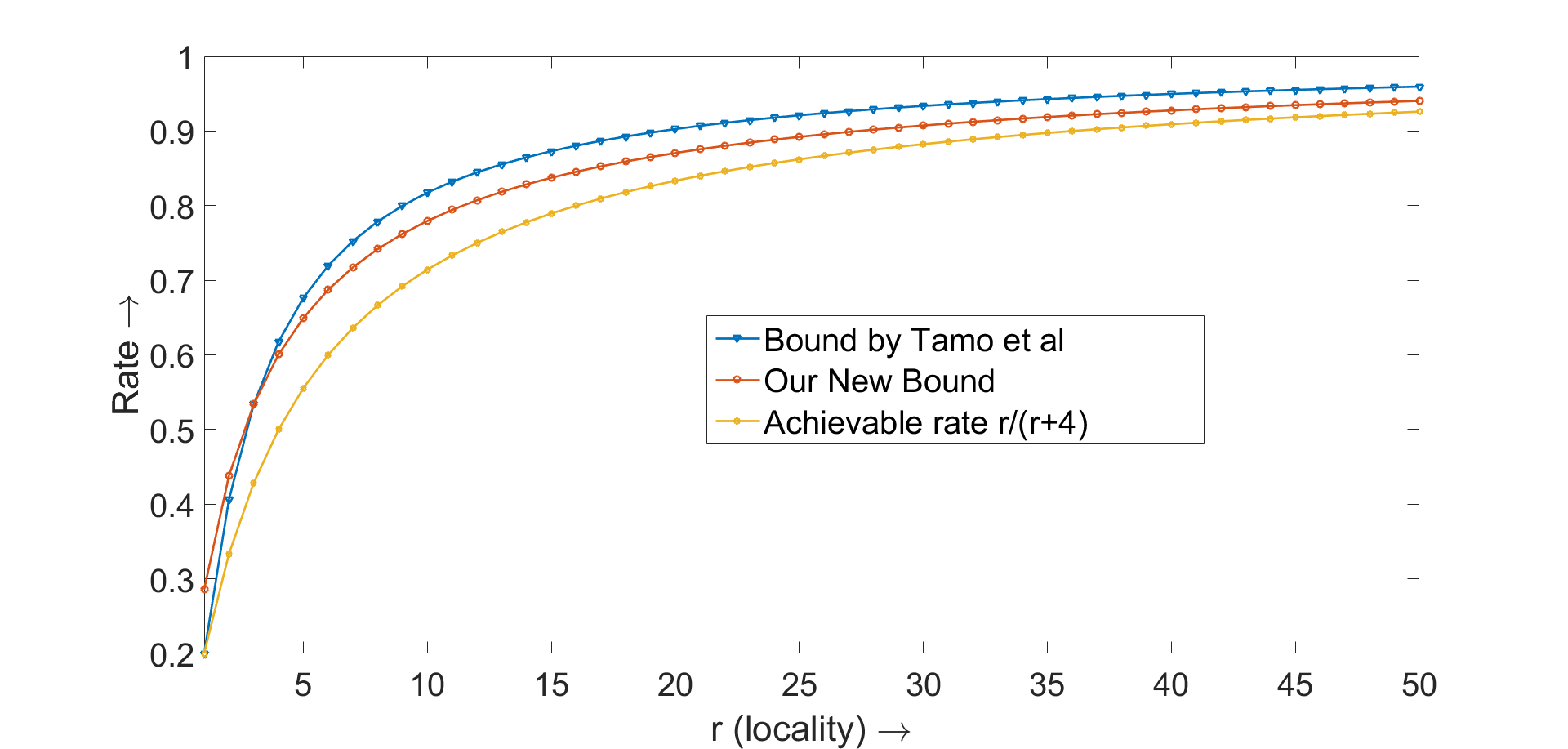}
	\caption[Example plot of locality parameter $r$ vs rate]{Plotting locality parameter $r$ vs rate for $t=4$. Here we compare our new bound \eqref{eq:Rate_2_ch5} with the bound given in \eqref{eq:TamoBargRate_ch5} (Bound by Tamo et al) and the achievable rate $\frac{r}{r+4}$ given by the Wang et al. Construction explained in Section \ref{sec:constructions}.  Note that our bound \eqref{eq:Rate_2_ch5} is tighter than the bound given in \eqref{eq:TamoBargRate_ch5}.}
	\label{Rate_4}
\end{figure}
\begin{note} $\bf{Tightness \ of \ the \ bound }$:\\
	In the following the word bound refers to an upper bound and the word rate bound refers to an upper bound on rate.
	The bound on $R(r,t)$ given in \eqref{eq:Rate_2_ch5} becomes tighter than the bound $R(r,t) \leq \frac{1}{\prod_{j=1}^{t}(1+\frac{1}{jr})}$ given in  \eqref{eq:TamoBargRate_ch5} as $r$ increases for a fixed $t$ (We skip a formal proof of this. This can be seen by approximating \eqref{eq:TamoBargRate_ch5} by $\frac{r}{r+\log(t)}$ for large $r$ for fixed $t$ and by the fact that the bound \eqref{eq:TamoBargRate_ch5} follows the inequality $ \frac{1}{\prod_{j=1}^{t}(1+\frac{1}{jr})} \leq \frac{r}{r+\log(t)}$. Simulation shows that our bound given in \eqref{eq:Rate_2_ch5} is tighter than the bound $R(r,t) \leq \frac{1}{\prod_{j=1}^{t}(1+\frac{1}{jr})}$ given in  \eqref{eq:TamoBargRate_ch5} for $r \geq c t$ for $c=1.1$ and $1\leq t\leq 10000$). As an example lets calculate $R(r,2)$. From \eqref{eq:Rate_2_ch5}:
	\bean
	R(r,2) & \leq & 1-\frac{2}{r+1}+\frac{2}{r+1} \frac{1}{\prod_{j=1}^{r+1}(1+\frac{1}{j})}, \\
	R(r,2) & \leq & 1-\frac{2}{r+1}+\frac{2}{r+1} \frac{1}{r+2}, \\
	R(r,2) & \leq & \frac{r}{r+2}.
	\eean
	The above abound on $R(r,2)$ obtained from  \eqref{eq:rate_11_ch5} and  \eqref{eq:Rate_2_ch5} is a tight bound as it is known that rate $\frac{r}{r+2}$ is achievable for $t=2$ and any $r$ using a complete graph code (\cite{PraLalKum}) and hence clearly tighter than $R(r,2) \leq \frac{r^2}{(r+1)(r+\frac{1}{2})}$ given in \eqref{eq:TamoBargRate_ch5} . 
	We show a plot of the bound given in \eqref{eq:Rate_2_ch5} for $t=4$ in  Fig.~\ref{Rate_4}. The plot shows that the bound given by \eqref{eq:Rate_2_ch5} is tighter than the bound given in  \eqref{eq:TamoBargRate_ch5} for $t=4, r>2$. \\
	Even though our bound given in \eqref{eq:Rate_2_ch5} becomes tighter as $r$ increases for a fixed $t$ than the bound in  \eqref{eq:TamoBargRate_ch5}, the bound given in  \eqref{eq:TamoBargRate_ch5}  is for $(n,k,r,t)$ availability codes but the bound in \eqref{eq:rate_11_ch5} and  \eqref{eq:Rate_2_ch5} is applicable only for $(n,k,r,t)$ SA codes. But we also would like to draw attention to the fact that most of the high rate constructions known in literature for $(n,k,r,t)$ availability codes are also $(n,k,r,t)$ SA codes. In fact the most general high rate construction for $(n,k,r,t)$ availability code is given in \cite{WanZhaLiu} (the Wang et al. Construction explained in Section \ref{sec:constructions}) and it is also an $(n,k,r,t)$ SA code. Hence there is a very good reason to think that when it comes to rate-optimality $(n,k,r,t)$ SA codes will give good candidate codes.
	
	\paragraph{Rate Bound comparison for $t=3$:}
	
	Note that the bound \eqref{eq:Kadhe_ch5} is derived in \cite{KadCal} using our equation \eqref{eq:rate_11_ch5} by substituting the upper bound on rate for $(n,k,r=2,t)$ SA codes given in \eqref{eq:SwaCald1_ch5}. Note that the bounds \eqref{eq:Kadhe_ch5},\eqref{eq:SwaCald1_ch5} are specific to binary codes whereas our bounds in \eqref{eq:rate_11_ch5} and  \eqref{eq:Rate_2_ch5} and the bound we derived in \cite{BalKum} for $t=3$ are applicable over any finite field. But nonetheless, we compare our bounds with \eqref{eq:Kadhe_ch5} for $t=3$. 
	
	When $n \geq {r+3 \choose 3}$, the following conclusions are true for SA codes.
	\ben 
	\item Rate bound we derived in \cite{BalKum} for $t=3$ (we are skipping this result from this thesis as it involves relatively complicated analysis. For interested reader please see \cite{BalKum}.) is tighter than \eqref{eq:Kadhe_ch5} for $2 \leq r \leq 31$.
	\item Rate bound in \eqref{eq:Rate_2_ch5} is tighter than \eqref{eq:Kadhe_ch5} for $2 \leq r \leq 72$.
	\item Rate bound we derived in \cite{BalKum} for $t=3$ is tighter than \eqref{eq:Rate_2_ch5} for $2 \leq r \leq 14$.
	\item All these 3 bounds \eqref{eq:Rate_2_ch5}, bound in \cite{BalKum} for $t=3$, \eqref{eq:Kadhe_ch5}, are tighter than \eqref{eq:TamoBargRate_ch5}.
	\item Hence the conclusion is that for $2 \leq r \leq 14$, rate bound we derived in \cite{BalKum} for $t=3$ is the tightest known bound and for  $15 \leq r \leq 72$, \eqref{eq:Rate_2_ch5} is the tightest known bound and for $r>72$ , \eqref{eq:Kadhe_ch5} (applicable only for binary codes) is the tightest known bound. On plotting it can be seen that even for $r>72$, the difference between bound in  \eqref{eq:Kadhe_ch5} and bound in \eqref{eq:Rate_2_ch5} is very small around $10^{-3}$.
	\item For $q>2$ (non-binary codes), the conclusion is that for $2 \leq r \leq 14$, rate bound we derived in \cite{BalKum} for $t=3$ is the tightest known bound and for  $r>14$, \eqref{eq:Rate_2_ch5} is the tightest known bound. 
	\een
	
\end{note}

\subsection{Upper Bounds on Minimum Distance of an Availability Code for Unconstrained Alphabet Size}	
\subsubsection{Known Bounds on Minimum Distance} \label{sec:KnowndminBounds:ch5}
Let $d_{\text{min}}(n,k,r,t) $ be the largest possible minimum distance of an $(n,k,r,t)$ availability code.
In \cite{WanZha}, the following upper bound on the minimum distance of an information symbol availability code with parameters $(n,k,r,t)$ (and hence the upper bound is also applicable to the case of all-symbol availability codes as well) was presented:
\bea
d_{\text{min}}(n,k,r,t) \leq n-k+2 - \left \lceil \frac{t(k-1)+1}{t(r-1)+1} \right \rceil .   \label{eq:WangDmin_ch5}
\eea
This bound was derived by adopting the approach employed in Gopalan et.al \cite{GopHuaSimYek} to bound the minimum distance of an LR code.   An improved upper bound on minimum-distance for (All Symbol) availability codes appears in \cite{TamBarFro}:
\bea
d_{\text{min}}(n,k,r,t) \leq n-\sum_{i=0}^{t} \left \lfloor \frac{k-1}{r^i} \right \rfloor. \label{eq:TamoBargDmin_ch5}
\eea
The following upper bound on minimum distance of an availability code appeared in \cite{KruFro}.
\bea
d_{\text{min}}(n,k,r,t) \leq n-k+1 - \left \lfloor \frac{k-2}{r-1} \right \rfloor. \label{eq:KruFro_ch5}
\eea

\subsubsection{A New Upper Bound on Minimum Distance of an availability code with Unconstrained Alphabet Size} \label{sec:Newdminbound_ch5}
Here we derive an upper bound on minimum distance of an availability code (not necessarily an SA code.). This bound is the tightest known bound on minimum distance of an availability code for all parameters.
This bound is derived by taking supremum w.r.t field size on both  sides (first on right hand side and then on left hand side) of \eqref{eq:Mindist_Bound1:ch5}. 
\begin{thm} \label{thm:ThmMinDist}
	The field-size dependency of the bound \eqref{eq:Mindist_Bound1:ch5} can be removed and written as:
	\bea
	d_{\min}(n,k,r,t) & \leq &  \min_{i \in S}  \  d_{\min}(n-e_i,k+i-e_i,r,t), \notag \\
	& \leq & \min_{i \in S}  \  \ \   n-k-i+1-\sum_{j=1}^{t} \left \lfloor \frac{k+i-e_i-1}{r^j} \right \rfloor, \ \   \ \  \ \label{eq:Mindist_Bound2}
	\eea
	where $S= \{ i: e_i-i < k, 1 \leq i \leq b_1 \}$ and $e_i,1 \leq i \leq b_1 $ are the MSW sequence defined in Theorem \ref{thm:msw_sequence:ch5} and $b_1$ is same as that defined in Theorem \ref{thm:dmin1:ch5}. 
\end{thm}
\begin{proof}
	Let $F=\{q: q = p^w , \text{$p$ is a prime number}, w \in \mathbb{Z}_{+} \}$.
	From \eqref{eq:Mindist_Bound1:ch5} :
	\bean
	d^q_{\min}(n,k,r,t) & \leq &  \min_{i \in S} \ d^q_{\text{min}}(n-e_i,k+i-e_i,r,t), \\
	d^q_{\min}(n,k,r,t) & \leq &  \min_{i \in S} \ \sup_{q \in F} \  d^q_{\text{min}}(n-e_i,k+i-e_i,r,t), \\
	\sup_{q \in F} \  d^q_{\min}(n,k,r,t) & \leq &  \min_{i \in S} \ \sup_{q \in F} \  d^q_{\text{min}}(n-e_i,k+i-e_i,r,t), \\
	d_{\min}(n,k,r,t) & \leq &  \min_{i \in S} \ d_{\text{min}}(n-e_i,k+i-e_i,r,t).
	\eean
	The second inequality in \eqref{eq:Mindist_Bound2} follows from \eqref{eq:TamoBargDmin_ch5}.
\end{proof}
\begin{note}\textbf{Tightness of the Bound}:
	From the Corollary \ref{cor:CorollaryOnei_ch5},  $e_i < ir+1$ whenever $i \geq \frac{b_1 +r-\frac{1}{1-\rho(r,t)}}{r-\frac{1}{1-\rho(r,t)}+1}$ and $e_i \leq ir+1$, $\forall i \geq 1,t \geq 2,$ $r\geq 2$. Hence it can be the seen that \eqref{eq:Mindist_Bound2} is the tightest known upper bound on minimum distance of an availability code for all parameters. This is because:
	\ben
	\item For $(r-1) \  | \  (k-1)$, the right hand side of the bounds given in \eqref{eq:WangDmin_ch5},\eqref{eq:TamoBargDmin_ch5},\eqref{eq:KruFro_ch5} are lower bound by the following expression:
	
	\bean
	n-k+1- \left \lfloor \frac{k-2}{r-1} \right \rfloor. \label{eq:dminBoundSimplified2_ch5}
	\eean 
	
	\item For all values of $r,k$, the right hand side of the bounds given in \eqref{eq:WangDmin_ch5},\eqref{eq:TamoBargDmin_ch5},\eqref{eq:KruFro_ch5} are lower bound by the following expression:
	
	\bean
	n-k+1- \left \lfloor \frac{k-1}{r-1} \right \rfloor. \label{eq:dminBoundSimplified1_ch5}
	\eean
	
	\een
	We plot our bound for $t=3$ in Fig.~\ref{fig:dmin_t3}. It can be seen from the plot in Fig.~\ref{fig:dmin_t3} that our upper bound given in \eqref{eq:Mindist_Bound2} is tigher than the upper bounds \eqref{eq:WangDmin_ch5},\eqref{eq:TamoBargDmin_ch5},\eqref{eq:KruFro_ch5}.
\end{note}

We derived another upper bound on minimum distance of an availability code in \cite{BalKum} based on our own calculation of an upper bound on $i$th MSW which is better than $e_i$. This bound uses the rank of $H_a$ and an upper bound on weight of any column of $H_a$. This bound is tighter than the bounds \eqref{eq:WangDmin_ch5},\eqref{eq:TamoBargDmin_ch5},\eqref{eq:KruFro_ch5} without any constraint i.e., even after maximizing over all possible values of rank of $H_a$ and all possible values of  upper bound on weight of any column of $H_a$. But it is tighter than the bound \eqref{eq:Mindist_Bound2} only when we do calculation of our upper bound on $i$th MSW given in \cite{BalKum} knowing specifically the rank of $H_a$ and the upper bound on weight of any column of the matrix $H_a$. We do not mention this bound to keep things simple. Interested reader please refer to \cite{BalKum}.
\begin{figure}[h!]
	\centering
	\includegraphics[width=5in]{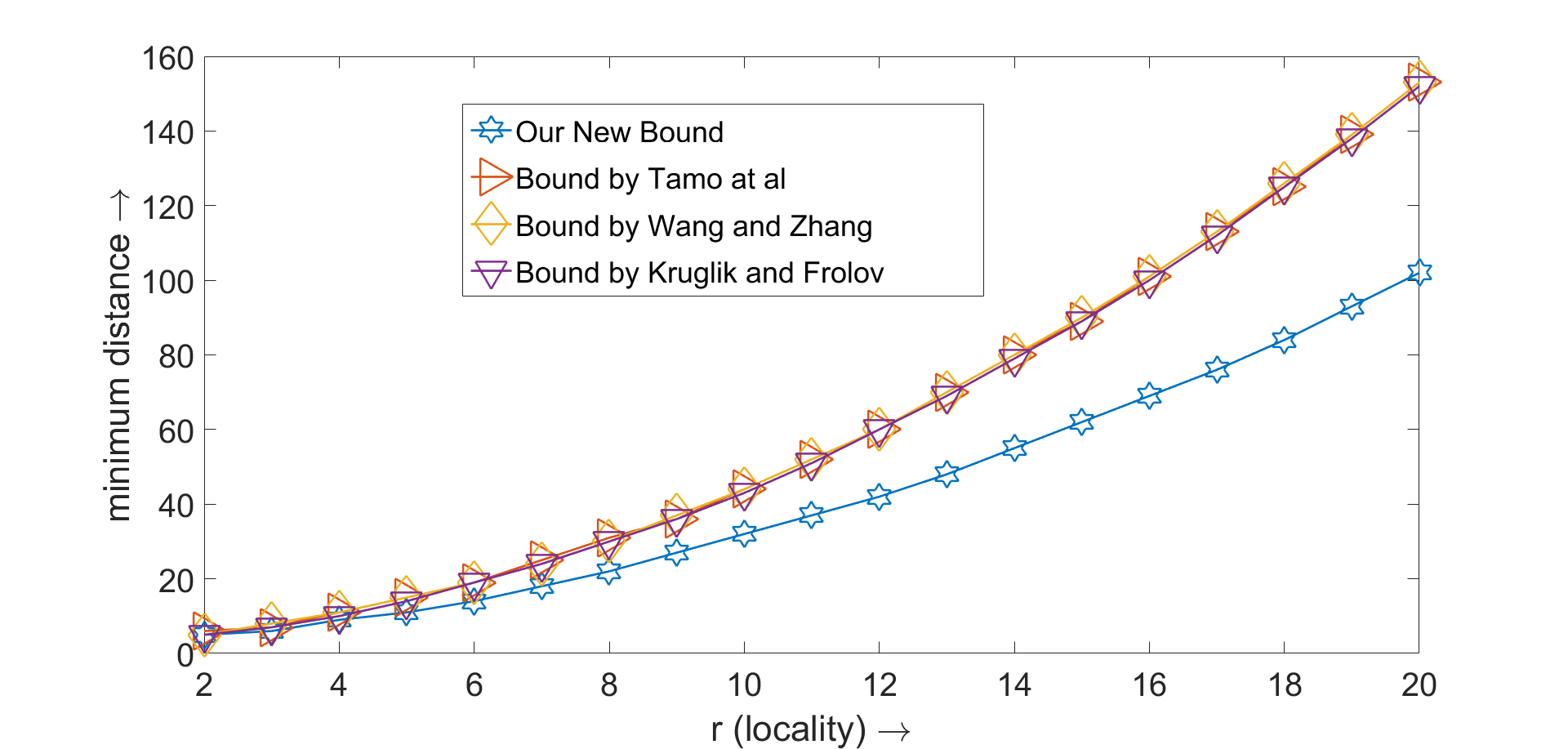}
	\caption[Example plot of locality parameter vs minimum distance]{Plotting locality parameter $r$ vs minimum distance for $t=3$ with $n= {r+3 \choose 3},k=\frac{nr}{r+3}$. Here we are comparing our new bound \eqref{eq:Mindist_Bound2} with the bounds \eqref{eq:TamoBargDmin_ch5} (Bound by Tamo et al),\eqref{eq:WangDmin_ch5} (Bound by Wang and Zhang), \eqref{eq:KruFro_ch5} (Bound by Kruglik and Frolov). In the plot, bounds \eqref{eq:TamoBargDmin_ch5},\eqref{eq:WangDmin_ch5},\eqref{eq:KruFro_ch5} are very close to each other. Note that our bound \eqref{eq:Mindist_Bound2} is tighter than the bounds \eqref{eq:TamoBargDmin_ch5},\eqref{eq:WangDmin_ch5},\eqref{eq:KruFro_ch5}.}
	\label{fig:dmin_t3}
\end{figure}

\section{A Tight Asymptotic Upper Bound on Fractional Minimum Distance for a Special class of Availability Codes} \label{sec:Asym_ch5}
Let $\mathcal{C}_{\ell}$ be an availability code with parameters $(n,k,r,t)$ over a finite field $\mathbb{F}_{q_{\ell}}$ with parity check matrix $H$ as described in the following where $\ell$ will also be defined in the following. We have already seen that parity check matrix $H$ of $\mathcal{C}_{\ell}$ can be written in the form $H=[H^T_a \  H^T_b]^T$. In this section, we assume $H_a$ (writing only linearly independent rows) can be written as follows:
\bea
H_a & = & \left[
\begin{array}{c|c|c|c|c}
	H_c & 0 & \hdots & 0 & 0 \\
	\hline
	0 & H_c  & \hdots & 0 & 0 \\
	\hline
	\vdots & \vdots & \vdots & \vdots & \vdots \\
	\hline
	0 & 0  & \hdots & 0 & H_c 
\end{array} \right] \label{eq:DirectProduct_ch5},
\eea
where $H_a$ is an $(\ell (n_c-k_c) \times \ell n_c)$ matrix and $H_c$ is an $(n_c-k_c) \times n_c$ matrix and $H_c$ is a parity check matrix of an availability code with parameters $(n_c,k_c,r,t)$ with rate $R_c$ over $\mathbb{F}_{q_c}$ ($\mathbb{F}_{q_{\ell}}$ is an extension field of $\mathbb{F}_{q_c}$).
Let $H_b$ be a $(s \times \ell n_c)$ matrix with rank $s$ with entries from $\mathbb{F}_{q_{\ell}}$ where $n-k = \ell (n_c-k_c) + s$. Hence the parity check matrix $H$ is an $((n-k) \times n)$ matrix with $n-k = \ell (n_c-k_c) + s$ and $n = \ell n_c$. Let $R_{max} = \sup_{\{(k_1,n_1):(n_1,k_1,r,t) \text{ availability code exists }\}}\frac{k_1}{n_1}$. Note that $R_{max}$ depends on $(r,t)$. We assume throughout this section that $R_c > 0$.

\begin{thm} \label{thm:dminupperbound_ch5}
	The code $\mathcal{C}_{\ell}$ over $\mathbb{F}_{q_{\ell}}$ with $((n-k) \times n)$ parity check matrix $H=[H_a^T \ \ H_b^T]^T$ with $H_a$ defined by \eqref{eq:DirectProduct_ch5} as described above is an availability code with parameters $(n,k,r,t)$ with minimum distance $d_{\min}$ satisfying:
	\bea
	d_{\min} & \leq &    n  - \frac{k}{R_c} + n_c \frac{1-R_c}{R_c} +1. \label{eq:DminUpperBound_ch5}
	\eea
\end{thm}
\begin{proof}
	Recall that $n-k = \ell (n_c-k_c) + s$. Since the code $\mathcal{C}_{\ell}$ has minimum distance $d_{\min}$, any set of $d_{\min}-1$ columns of $H$ must be linearly independent. If we take the first $d_{\min}-1$ columns of $H$ then its rank is upper bounded by $\frac{(n_c - k_c) (d_{\min}-1)}{n_c}+n_c-k_c + s $. Hence:
	\bean
	d_{\min}-1 & \leq &  \frac{(n_c - k_c) (d_{\min}-1)}{n_c}+n_c-k_c + s, \\
	d_{\min}-1 & \leq &  \frac{n_c}{k_c} (n_c-k_c + s), \\
	d_{\min}-1 & \leq &  \frac{n_c}{k_c} (n_c-k_c + n \frac{k_c}{n_c} - k), \\
	d_{\min}-1 & \leq &    n  - k\frac{n_c}{k_c} + (n_c-k_c)\frac{n_c}{k_c},  \\
	d_{\min} & \leq &    n  - \frac{k}{R_c} + n_c \frac{1-R_c}{R_c} +1.
	\eean
\end{proof}

\begin{cor} \label{cor:AsymUpperBound_ch5}
	Let $0 < R \leq R_{max}$. Let $\{\mathcal{C}_{\ell} : \ell \geq 1 \}$ be such that $R=\lim_{\ell \rightarrow \infty}$ rate($\mathcal{C}_{\ell}$) with each code $\mathcal{C}_{\ell}$ as defined before defined by $((n-k) \times n)$ parity check matrix $H=[H_a^T \ \ H_b^T]^T$ over $\mathbb{F}_{q_{\ell}}$ with $H_a$ defined by \eqref{eq:DirectProduct_ch5} where $n,k,H_a,H_b,n_c,k_c,q_c,H_c$ depends on $\ell$. Let $\delta = \limsup_{\ell \rightarrow \infty}\frac{d^{\ell}_{\min}}{n}$ where $d^{\ell}_{\min}$ is the minimum distance of $\mathcal{C}_{\ell}$ and $n=\ell n_c$. Then,
	\bea
	\delta \leq 1  - \frac{R}{R_{max}}. \label{eq:UpperBound_ch5}
	\eea
	
\end{cor}
\begin{proof}
	Proof follows from Theorem \ref{thm:dminupperbound_ch5}.
\end{proof}

\begin{thm} \label{thm:dminlowerbound_ch5}
	For $q_{\ell} > {n \choose k}$, there exists a choice of $H_b$ such that the code $\mathcal{C}_{\ell}$ over $\mathbb{F}_{q_{\ell}}$ as defined before with $((n-k) \times n)$ parity check matrix $H=[H_a^T \ \ H_b^T]^T$ with $H_a$ defined by \eqref{eq:DirectProduct_ch5} is an availability code with parameters $(n,k,r,t)$ with minimum distance $d_{\min}$ satisfying:
	\bea
	d_{\min} \geq   n \frac{R_c}{R_{max}} - \frac{k}{R_{max}} +1.  \label{eq:DminLowerBound_ch5}
	\eea
\end{thm}
\begin{proof}
	For an $(x \times y)$ matrix  $B$ and $E \subseteq [y]$, let $B|_{E}$ refer to the $(x \times |E|)$ sub-matrix of $B$ containing exactly the columns indexed by $E$.
	Since the code $\mathcal{C}_{\ell}$ has minimum distance $d_{\min}$, any set of $d_{\min}-1$ columns of $H$ must be linearly independent. Let $w = \lceil (1-R_{max}) (d_{\min}-1) \rceil$. Recall that $H_b$ is an $(s \times n)$ matrix. If we take any $E \subseteq [n]$, $|E|=d_{\min}-1$, then $rank(H|_{E}) \geq L=\min( w + s,d_{\min}-1)$. This is because $rank(H_a|_{E}) \geq w$. Now choose the entries of $H_b$ to be variables taking values from $\mathbb{F}_{q_{\ell}}$. We will now choose the values of these variables so that $rank(H|_{E}) \geq L$. Let $A$  be a $(L \times (d_{\min}-1))$ matrix with first $w$ rows of $A$ containing some $w$ linearly independent rows from $(H_a)|_{E}$ and the last $L-w$ rows of $A$ containing the first $L-w$ rows of $(H_b)|_{E}$. The rank of the matrix $H|_{E}$ can be made at least $L$ by making sure that $rank(A) \geq L$ by choosing the values of variables in $H_b$ appropriately. Note that $A$ is a submatrix of $H|_{E}$ restricted to these $L$ rows. This condition can be formulated as a condition requiring the determinant of a $(L \times L)$ square sub-matrix of $A$ to be non-zero where the first $w$ rows of the sub-matrix are chosen to be linearly independent. This property of the first $w$ rows of the sub-matrix is possible since the first $w$ rows of $A$ are linearly independent. Each such condition corresponds to making a non-zero polynomial in variables introduced in $H_b$ to take a non-zero value. We can multiply all these non-zero polynomials corresponding to each such subset $E$ and choose values for variables in $H_b$ so that this product of all these polynomials takes a non-zero value by choosing $q_{\ell} > {n \choose k}$ by the Combinatorial Nullstellensatz Theorem \cite{Nog}. Hence: $d_{\min}-1 \geq L=\min(w + s,d_{\min}-1)$ and since by above argument we can make any set of $d-1$ columns to have rank at least $\min(\lceil (1-R_{max}) (d-1) \rceil +s,d-1)$, we can choose largest $d$ such that $d_1-1 \leq \lceil (1-R_{max}) (d_1-1) \rceil + s$ for $d> d_1$ and $(d-1) > \lceil (1-R_{max}) (d-1) \rceil + s$. Hence we can choose $d_{\min}-1 \geq d-2 \geq \lceil (1-R_{max}) (d-2) \rceil+ s$. Hence:
	\bean
	d-2 & \geq & (1-R_{max}) (d-2)+ s, \\
	d-2 & \geq & (1-R_{max}) (d-2)+  n \frac{k_c}{n_c} - k, \\
	d_{\min} & \geq & d-1 \geq  n \frac{R_c}{R_{max}} - \frac{k}{R_{max}} +1.
	\eean
\end{proof}

\begin{cor} \label{cor:dminlower_ch5}
	Let $\epsilon >0$. For sufficiently large $n_c,q_c$ there exists a choice of $H_c$ with $R_c \geq R_{max} - \epsilon$ and for $q_{\ell} > {n \choose k}$ with $\mathbb{F}_{q_{\ell}}$ being an extension field of $\mathbb{F}_{q_c}$ there exists a choice of $H_b$ such that the code $\mathcal{C}_{\ell}$ over $\mathbb{F}_{q_{\ell}}$ as defined before with $((n-k) \times n)$ parity check matrix $H=[H_a^T \ \ H_b^T]^T$  with $H_a$ defined by \eqref{eq:DirectProduct_ch5} is an availability code with parameters $(n,k,r,t)$ with minimum distance $d_{\min}$ satisfying:
	\bea
	d_{\min} \geq   n \frac{R_{max}-\epsilon}{R_{max}} - \frac{k}{R_{max}} +1.  \label{eq:DminLowerBound_asympt_ch5}
	\eea
\end{cor}
\begin{proof}
	The proof follows from applying Theorem \ref{thm:dminlowerbound_ch5} with $H_c$ such that $R_c \geq R_{max}-\epsilon$. This is always possible by definition of $R_{max}$ and by choosing $n_c,q_c$ appropriately.
\end{proof}

\begin{cor} \label{cor:Existence_ch5}
	Let $0 < R \leq R_{max}$. Then there exists a sequence of codes $\{\mathcal{C}_{\ell} : \ell \geq 1 \}$ with each code $\mathcal{C}_{\ell}$ as defined before defined by $((n-k) \times n)$ parity check matrix $H=[H_a^T \ \ H_b^T]^T$ over $\mathbb{F}_{q_{\ell}}$ with $H_a$ defined by \eqref{eq:DirectProduct_ch5} where $n,k,H_a,H_b,n_c,k_c,q_c,H_c$ depends on $\ell$ such that $R=\lim_{\ell \rightarrow \infty}$ rate($\mathcal{C}_{\ell}$) and 
	\bea
	\delta = 1  - \frac{R}{R_{max}}. \label{eq:Bound}
	\eea
	where $\delta = \limsup_{\ell \rightarrow \infty}\frac{d^{\ell}_{\min}}{n}$ where $d^{\ell}_{\min}$ is the minimum distance of $\mathcal{C}_{\ell}$ and $n = \ell n_c$.
\end{cor}
\begin{proof}
	We choose different $H_c,H_b$ for each code $\mathcal{C}_{\ell}$. Note that $n_c,q_c$ depends on $\ell$ here.  For each $\ell$, we choose $n_c,q_c$ to be sufficiently large such that there exists $H_c$ with $R_c \geq R_{max} - \epsilon_{\ell}$ for $\epsilon_{\ell} > 0$ for some $\{\epsilon_{\ell} : \ell \geq 1 \}$ such that $\lim_{\ell \rightarrow \infty} \epsilon_{\ell} = 0$. 
	Hence for each $\ell$, by choosing $n_c,q_c$ to be sufficiently large and $q_{\ell} > {n \choose k}$ with $\mathbb{F}_{q_{\ell}}$ being an extension field of $\mathbb{F}_{q_c}$, by Corollary \ref{cor:dminlower_ch5}, we can choose $\mathcal{C}_{\ell}$ such that $R_c \geq R_{max} - \epsilon_{\ell}$ and $d^{\ell}_{\min}$ satisfies \eqref{eq:DminLowerBound_asympt_ch5} with $\epsilon=\epsilon_{\ell}$. Since $\lim_{\ell \rightarrow \infty} \epsilon_{\ell} = 0$, it is trivial to make sure $R=\lim_{\ell \rightarrow \infty}$ rate($\mathcal{C}_{\ell}$).
	By Corollary \ref{cor:AsymUpperBound_ch5} and Corollary \ref{cor:dminlower_ch5}, for the chosen sequence of codes $\{\mathcal{C}_{\ell} : \ell \geq 1 \}$:
	\bean
	1  - \frac{R}{R_{max}} \geq \delta \geq  \frac{R_{max} - \epsilon_{\ell}}{R_{max}} - \frac{R}{R_{max}} 
	\eean	
	Since $\lim_{\ell \rightarrow \infty} \epsilon_{\ell} = 0$ the equation \eqref{eq:Bound} follows.
\end{proof}

\begin{note}
	Note that Corollary \ref{cor:AsymUpperBound_ch5} and Corollary \ref{cor:Existence_ch5} completely characterizes the optimal tradeoff between rate and fractional minimum distance for the class of codes defined by $\mathcal{C}_{\ell}$ as $\ell \rightarrow \infty$.
\end{note}

\section{Minimum Block-Length Codes with Strict Availability}  \label{sec:strict_avail_ch5} 

In this section, we present a new lower bound on block length of an SA code. We also show that the lower bound on block length is achieved if and only if BIBD (Balanced Incomplete Block Design) \cite{ColDin} with certain parameters exist. We then present two well-known BIBD. Based on our observation, the SA codes based on these two BIBD has least possible block length and maximum possible rate (for the given block length) based on a famous conjecture.

From definition, parity check matrix of an SA code \calc over a finite field $\mathbb{F}_q$ can be written in the form $H=[H^T_a  \ \ H^T_b]^T$ where $H_a$  is an $(m \times n)$ matrix with each row of weight $r+1$ and each column of weight $t$ which includes the $t$ orthogonal parities ($t$ parity checks corresponding to $t$ disjoint recovery sets $\{R_j^i : 1 \leq j \leq t \}$) for each code symbol $c_i$. The matrix $H_a$ may also be viewed as parity check matrix of a $(d_v=t, d_c=(r+1))$-regular LDPC code.  The corresponding Tanner graph of the code must necessarily have no cycles of length $4$. Let $(i,j)^{th}$ entry of $H_a$ be $h^a_{ij}$.
Let $A$ be an $m \times n$ matrix over $\Re$ with entries from $\{0,1\}$  given by $a_{ij} = 1$ when $h^a_{ij} \neq 0$ and $0$ else.   

\begin{thm}
	Let \calc\ be an $(n,k,r,t)$ SA code. Then:
	\bea
	n \geq  (r+1)^2 - \frac{(r+1)r}{t}.  \label{eq:lb_n_ch5} 
	\eea
\end{thm}

\begin{proof}
	Let $H_a$ and $A$ be $m \times n$ matrices as defined above in this section. From the requirements on an SA code, we must have: 
	\bean
	\scalemath{0.95}{{m \choose 2} \ \geq \ \sum_{j>i} \left( \sum_{l=1}^na_{i,l}a_{j,l} \right) 
		=  \sum_{l=1}^n \left(\sum_{j>i} a_{il}a_{jl} \right)  =  n {t \choose 2} .}
	\eean
	Using the relation $nt=m(r+1)$, we obtain 
	\bean
	\scalemath{0.95}{
		m  \geq  (t-1)(r+1) + 1, \ \ \ 
		n \geq  (r+1)^2 - \frac{(r+1)r}{t}. } 
	\eean
\end{proof}
Our interest is in the minimum-block-length case, where \eqref{eq:lb_n_ch5} holds with equality and for which a necessary condition is that $t \mid r(r+1)$. 

\begin{cor}
	For $1 \leq j \leq n$, let us define the sets $B_j \subseteq [m]$ by 
	\bean
	i \in B_j & \text{ iff } &  a_{i,j}=1 \text{ or equivalently, $h^a_{ij} \neq 0$}.
	\eean
	where $a_{i,j}$ and $h^a_{ij}$ are as defined before in this section. When equality holds in \eqref{eq:lb_n_ch5}, the sets $\{ B_j \}_{j=1}^n$ form a $(v,b,\hat{r},\hat{k}, \lambda)$ balanced incomplete block design (BIBD) having parameters 
	\bea
	v \ = \ m \ ,b \ = \ n, \ \hat{r} = (r+1), \ \hat{k}=t, \ \ \lambda=1. \label{eq:BIBD_ch5}
	\eea
	Conversely a BIBD with the parameter values as in \eqref{eq:BIBD_ch5} will yield an SA code \calc with parameters $(r,t)$ having block length satisfying \eqref{eq:lb_n_ch5} with equality where the parity check matrix of the code \calc is given by the incidence matrix of the BIBD.  The rate $R$ of \calc clearly satisfies $R \geq 1 -\frac{t}{(r+1)}$.  
\end{cor}
\begin{proof}
	Proof follows from the fact that equality holds in \eqref{eq:lb_n_ch5} iff the inner product of every pair of distinct rows of $A$ is exactly equal to $1$.
\end{proof}
%
%
\begin{example}\label{eg:ProjPlane_ch5}
	Let $Q=2^s, s \geq 2$ and let $PG(2,Q)$ denote the projective plane over $\mathbb{F}_Q$ having $Q^2+Q+1$ points and $Q^2+Q+1$ lines. Let the lines and points be indexed by $1,...,Q^2+Q+1$ in some order. Each line contains $Q+1$ points and there are $Q+1$ lines passing through a point. Set $n=Q^2+Q+1$. Let $H= H_a$ be the $(n \times n)$ parity-check matrix of a binary code \calc given by $h_{ij}=1$ if the $i^\text{th}$ point lies on the $j^\text{th}$ line and $h_{ij}=0$ otherwise. Then it is known that $H$ has rank $3^s+1$ over $\mathbb{F}_2$ (\cite{Smith}), and that \calc has minimum distance $d_{\min}=Q+2$, thus \calc is a binary $(n,k=n-(3^s+1),r=Q,t=Q+1)$ SA code. A plot comparing the rate of this code $\calc$ with the upper bound on rate in \eqref{eq:TamoBargRate_ch5} as a function of the parameter $s$, is shown in Fig.~\ref{fig:bibd_TB_bound}.  
\end{example}
While this code is well-known in the literature on LDPC codes, the intention here is to draw attention to the fact that this code is an SA having minimum block length. The parameters of a minimum-block-length code obtained by a similar construction involving lines in the affine plane and $(r,t)=(Q,Q)$ are given by $[n,k,d_{\min}]=[Q^2+Q,Q^2+Q-3^s,\geq Q+1]$, where $Q=2^s, s \geq 2$. 

\begin{conj}[Hamada-Sachar Conjecture \cite{AssKey}] 
	Every Projective plane of order $p^s$, p a prime, has p rank at least ${p+1 \choose 2}^s + 1$ with equality if and only if its desarguesian.
\end{conj}
The above conjecture is as yet unproven, but if true, would show the projective-plane code described in example \ref{eg:ProjPlane_ch5} to have minimum block length and maximum rate among all SA binary codes with $n=Q^2+Q+1, r=Q$ and $t=Q+1$ where  $Q=2^s$. 

\begin{example}\label{eg:sts}
	Consider the code $\calc$ obtained by making use of the Steiner Triple System (STS) associated to the point-line incidence matrix of $(s-1)$ dimensional projective space $PG(s-1,2)$ over $\mathbb{F}_2$ where this incidence matrix defines the parity check matrix $H$ of the code $\calc$. Hence once again the rows of $H$ correspond to points in the projective space and the columns to lines (2 dimensional subspace). Let $m=2^s-1$. From \cite{Doyen}, $\text{rank}(H)=m-s$. The code $\calc$ is a binary SA code having parameters $n=\frac{m(m-1)}{6},k=\frac{m(m-1)}{6}-m+s,d_{\min}=4$, $r=\frac{2^s-2}{2^2-2}-1 = 2^{s-1}-2,t=3$.   The corresponding comparison with the rate bound in \eqref{eq:TamoBargRate_ch5} appears in Fig.~\ref{fig:bibd_TB_bound}.    
\end{example} 
\begin{conj}[Hamada's Conjecture \cite{ColDin}]
	The p-rank of any design $D$ with parameters of a geometric design $G$ in PG(n,q) or AG(n,q) $(q=p^m)$ is at least the p-rank of $G$ with equality if and only if $D$ is isomoprhic to $G$.
\end{conj}
In \cite{Doyen}, this conjecture has been shown to hold true for the Steiner Triple system appearing in Example \ref{eg:sts}. Thus the code in the example has the minimum block length and maximum rate among all binary SA codes with $n=\frac{m(m-1)}{6},r=\frac{m-1}{2}-1,t=3$ where $m=2^s-1$.

\begin{figure}[ht!]
	\centering
	\begin{minipage}[c]{0.48\textwidth}
		\centering
		\includegraphics[width=3in]{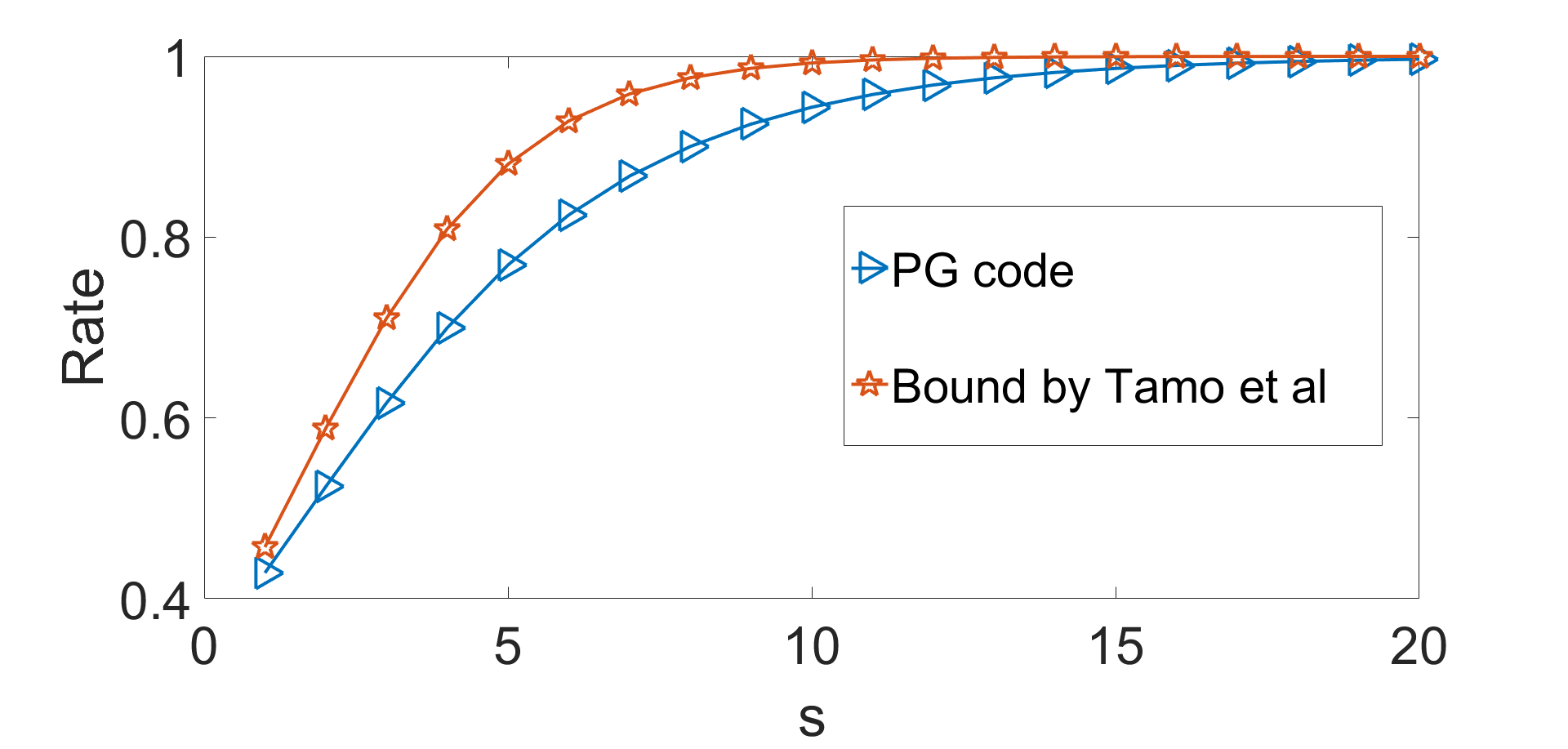}
	\end{minipage}
	\hspace{0.02\textwidth}
	\begin{minipage}[c]{0.48\textwidth}
		\includegraphics[width=3in]{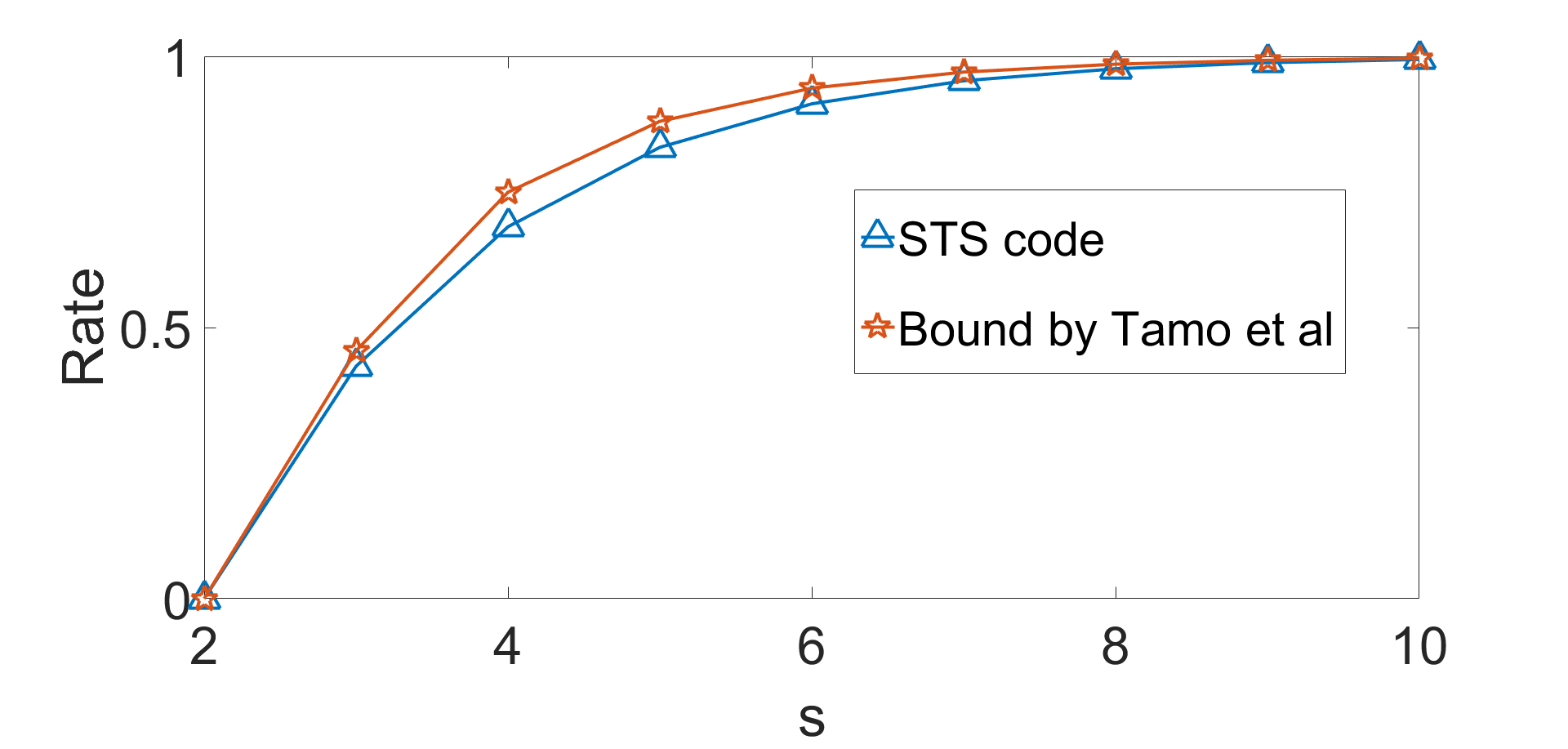}
	\end{minipage}
	\caption{Comparing the rates of the projective plane (PG) based codes (Example \ref{eg:ProjPlane_ch5}) and Steiner-triple-system (STS) based codes (Example \ref{eg:sts}) with the bound \eqref{eq:TamoBargRate_ch5} (Bound by Tamo et al) in \cite{TamBarFro}.}
	\label{fig:bibd_TB_bound}
\end{figure}

\section{Summary and Contributions} \label{sec:avail_summary}

In this chapter, we presented new field size dependent and field size independent upper bounds on minimum distance and dimension of an availability code based on shortening approach and MSW sequence. We also presented a new upper bound on rate of an SA code. Before presenting new results we presented the appropriate results from literature for comparison. We showed evidence that the new bounds presented in this chapter are the tightest known bounds. We then completely characterized the optimal tradeoff between rate and fractional minimum distance for a special class of availability codes defined by $\mathcal{C}_{\ell}$ as $\ell \rightarrow \infty$. We finally presented a lower bound on block length of an SA code and showed that codes achieving the lower bound on block length exists iff BIBD with certain parameters exists. We then presented two well known BIBD and SA codes based on them. These codes based on the two BIBD have maximum possible rate (for the minimum possible block length) based on a famous conjecture.

\chapter{Tight Bounds on the Sub-Packetization Level of MSR and Vector-MDS Codes} \label{ch:Subpkt}
In this chapter we will begin by introducing a class of codes called Regenerating (RG) codes. These codes are vector codes with block length $n$ and dimension $B$ where each vector code symbol have $\alpha$ components. The variable $\alpha$ is also known as sub-packetization. Each vector code symbol is stored in a distinct node or a disk. A special property of RG codes is that it can replace the contents of a failed node or a vector code symbol by downloading $\beta$ symbols from each node from a set of $d$ other nodes. This property is called repair property. These $\beta$ symbols could be any function of $\alpha$ symbols stored in the node. The RG codes are such that this value $d \beta$ is the smallest possible value. Another property of RG codes is that it can recover the $B$ message symbols by using the contents of any $k$ vector code symbols or contents of any $k$ nodes. This property is called data collection property. A special case of RG codes corresponding to the smallest possible value of $\alpha$ for a given value of $n,k,d,B$ is called Minimum Storage Regenerating (MSR) codes. It turns out that MSR codes are also vector MDS codes. In an MSR code, repair property holds for any failed node but if we restrict the repair property to any node belonging to a fixed subset of $w$ nodes, we call it MDS codes with optimal repair of $w$ nodes. Another special case of RG codes corresponding to the smallest possible value of $\beta$ for a given value of $n,k,d,B$ is called Minimum Bandwidth Regenerating (MBR) codes. It turns out that the value of $\alpha$ cannot be too small for the design of an MSR code or MDS code with optimal repair of $w$ nodes. Hence it is an interesting problem to find a lower bound on $\alpha$ as $\alpha$ plays an important role in the encoding and decoding complexity as well as in complexity in implementing the repair property. In this chapter, we will give a tight lower bound on $\alpha$ for an MSR or MDS code with optimal repair of $w$ nodes for the case when repair property is such that we directly use $\beta$ out of $\alpha$ symbols from each of $d$ nodes and carry out repair property. This repair property is called optimal access repair or repair by help-by-transfer. The fact that our bound is tight is shown by pointing out constructions from literature which achieves our lower bound. We also present a theorem on structure of these codes when $\alpha$ is exactly equal to our lower bound. We point out through an example that the structure we deduced is present in existing codes. We also very briefly describe a construction of MSR code with minimum possible value of $\alpha$ and field size of $O(n)$ and point out to our paper for further details. Below is a summary of our results in this chapter.

\paragraph{Organization of the Chapter} The chapter is organized as follows.  We begin by formally defining RG codes and MSR codes in Section~\ref{sec:RG_MSR_ch6}. Section \ref{sec:LowerBoundMSR_ch6} presents new lower bounds on $\alpha$ for MSR codes. Section \ref{sec:MDScodes_ch6} presents new lower bounds on $\alpha$ for MDS codes with optimal access repair of $w$ nodes. Section \ref{sec:struct_ch6} presents a theorem on structure of MDS codes with optimal access repair of $w$ nodes. In Section \ref{sec:MSRconstr_ch6}  we briefly gave an intuitive outline of a new construction of an optimal access MSR code with optimal sub-packetization.

\paragraph{Contributions}  Contributions of the thesis on the topic of MSR codes and MDS codes with optimal access repair of a set of $w$ nodes: 
\ben
\item A tabular summary of the new lower bounds on sub-packetization-level $\alpha$ derived here appears in Table~\ref{tab:results} and amounts to a summary of the main results presented in this chapter.
\item We prove new lower bounds on sub-packetization-level $\alpha$ of MSR codes and vector MDS codes with optimal repair of $w$ nodes with optimal access repair property. We first derive the lower bounds for the case of MSR codes for $d=n-1$ (Sections \ref{sec:dnGMSR_ch6},\ref{sec:dnMSR_ch6} ) and extend it to general $d$ (\ref{sec:arbdMSR_ch6}). We then derive the lower bounds for the case of MDS codes with optimal access repair of $w$ nodes for $d=n-1$ (Section  \ref{sec:dnMDS_ch6} ) and extend it to general $d$ (Section \ref{sec:arbdMDS_ch6}). We show that our lower bounds on $\alpha$ are tight by comparing with existing code constructions. 
\item We study the structure of a vector MDS code (Section \ref{sec:struct_ch6}) with optimal access repair of a set of $w$ nodes, and which achieve our lower bound on sub-packetization level $\alpha$. It turns out interestingly, that such a code must necessarily have a coupled-layer structure, similar to that of the Ye-Barg code \cite{YeBar_2}. 
\item  Finally we give a very brief intuitive overview of a new construction of optimal access MSR codes (Section \ref{sec:MSRconstr_ch6} ) with $O(n)$ field size with optimal sub-packetization for some specific parameters. This construction of ours appeared in \cite{VajBalKum}.
\een

Our approach to lower bounds on $\alpha$ is along the lines of that adopted in \cite{TamWanBru_access_tit} but with the difference that here we consider non-constant repair subspaces and consider all-node repair and also consider MDS codes with optimal access repair of $w$ nodes.

\section{Regenerating Codes} \label{sec:RG_MSR_ch6}

Data is stored in distributed storage by distributing it across disks or nodes. So one of the important problems in distributed storage is to repair a node on its failure. The coding theory community has come up with two types of coding techniques for this. They are called Regenerating (RG)  codes and Locally Recoverable (LR) codes. The focus in a Regenerating (RG) code is on minimizing the amount of data download needed to repair a failed node, termed the {\em repair bandwidth} while LR codes seek to minimize the number of helper nodes contacted for node repair, termed the {\em repair degree}. In a different direction, coding theorists have also re-examined the problem of node repair in RS codes and have come up with new and more efficient repair techniques. 
There are two principal classes of RG codes, namely Minimum Bandwidth Regenerating (MBR) and Minimum Storage Regeneration (MSR). These two classes of codes are two extreme ends of a tradeoff known as the storage-repair bandwidth (S-RB) tradeoff. MSR codes tries to minimize the storage overhead whereas MBR codes tries to minimize the repair bandwidth. There are also codes that correspond to the interior points of this tradeoff. The theory of regenerating codes has also been extended in several directions.

\begin{figure}[ht!]
	\centering
	\begin{minipage}[c]{0.43\textwidth}
		\begin{center}
			Parameters: $(\ (n,k,d), \ (\alpha, \beta), \ B, \ \mathbb{F}_q \ )$
		\end{center}
		\vspace*{-0.1in} 
		\begin{center}
			\begin{minipage}{1.4in}
				\begin{center}
					\includegraphics[trim= 0in  0in 0in 0in, width=1.4in]{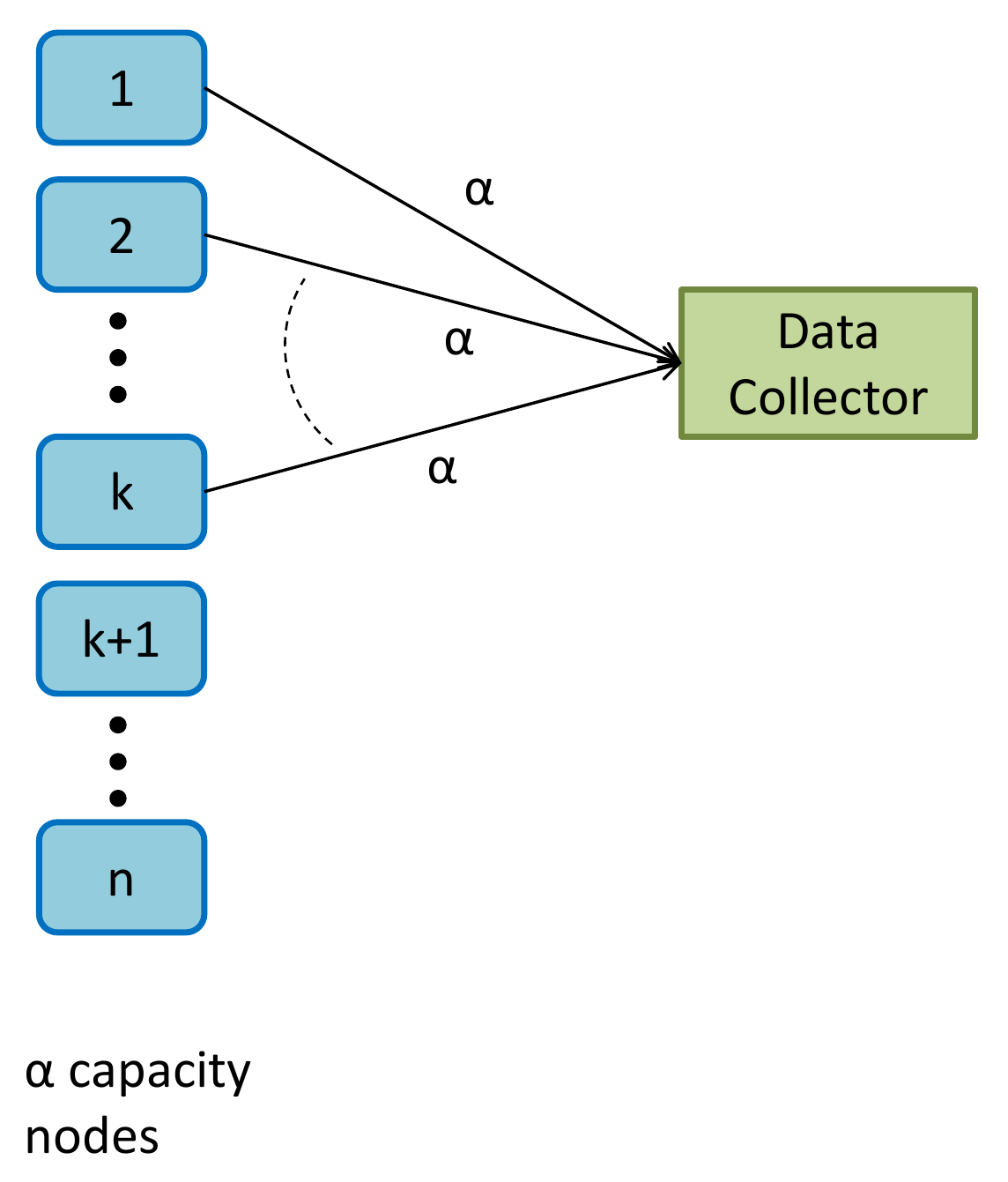} 
				\end{center}
				
			\end{minipage} \hspace*{0.2in} 
			\begin{minipage}{1.in}
				\begin{center}
					\includegraphics[trim= 0in  0in 0in 0in, width=1.1in]{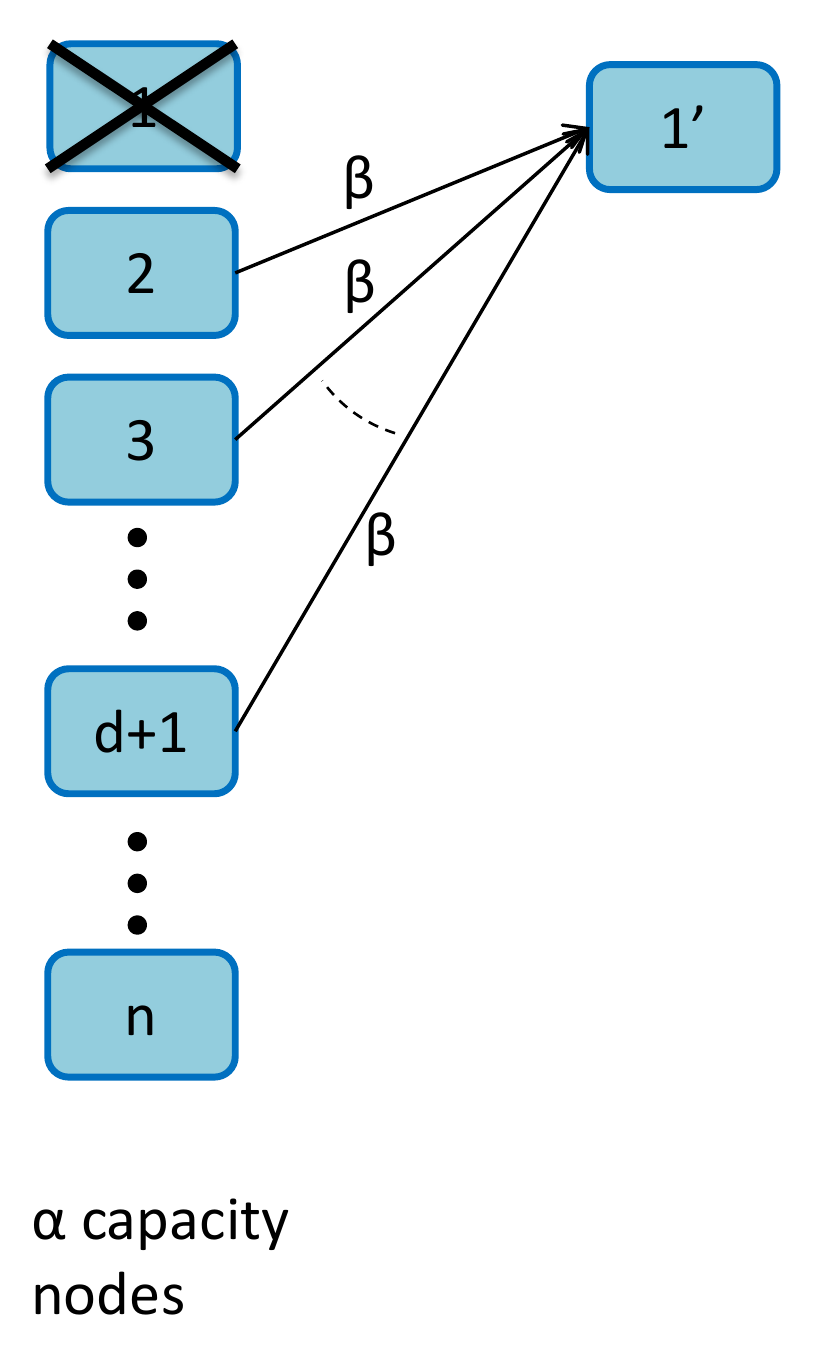} 
				\end{center}
			\end{minipage}
		\end{center}
		\caption{An illustration of the data collection and node repair properties of a regenerating code.} 
	\end{minipage}
	\hspace{0.02\textwidth}
	\begin{minipage}[c]{0.4\textwidth}
		\begin{center}
			\includegraphics[trim= 0in  0in 0in 0in, width=3in]{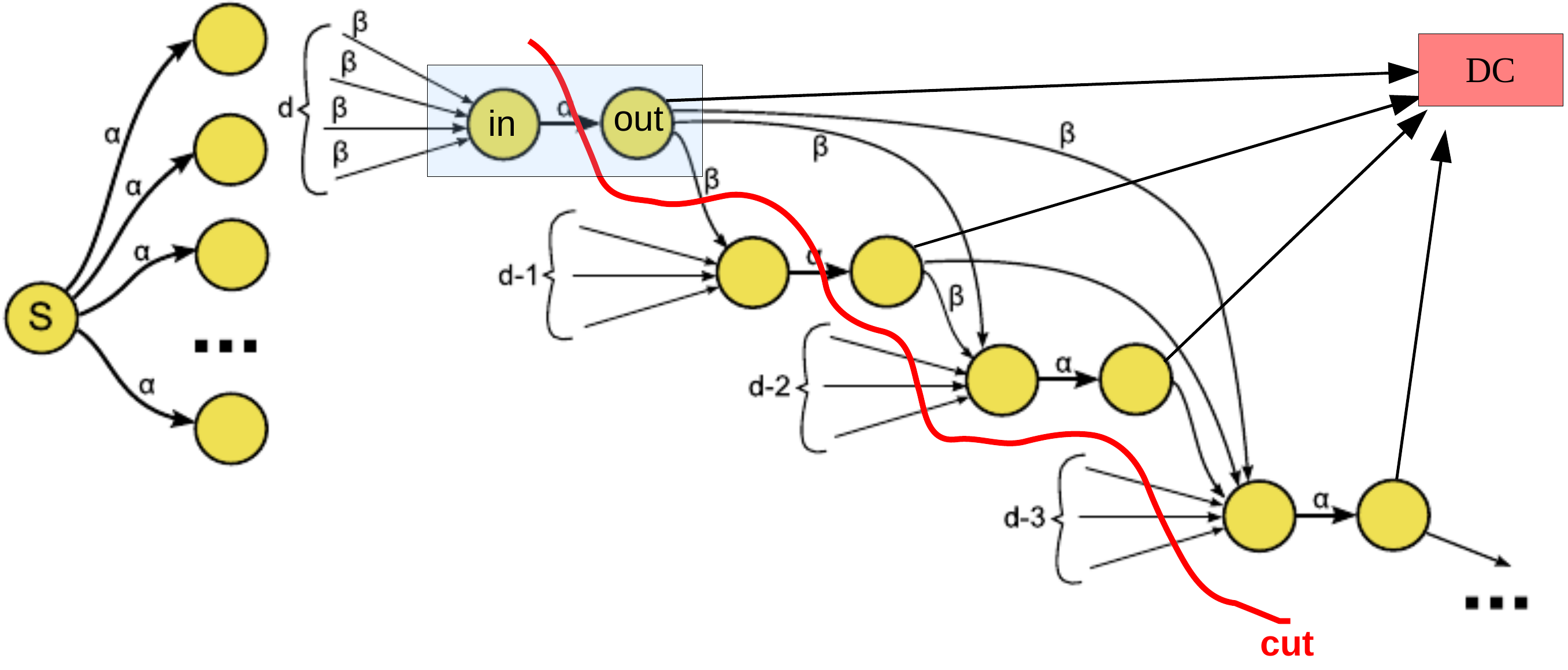} 
		\end{center}
		\caption{The graph behind the cut-set file size bound.\label{fig:cutsetbound}}
	\end{minipage}
\end{figure}

\begin{defn}[\cite{DimGodWuWaiRam}]
	Let \fq\ denote a finite field of size $q$. Then a regenerating (RG) code  \calc\ over \fq\ having integer parameter set $((n,k,d), (\alpha, \beta), B)$ where $1\le k \le n-1$, $k \le d \le n-1$, $\beta \le \alpha$, maps a file $\uu \in \fqB$ on to a collection $\{\underline{c}_i\}_{i=1}^n$ where $\underline{c}_i \in \fqalpha$ using an encoding map
	\bean
	E(\underline{u}) = [\underline{c}_1^T, \underline{c}_2^T, \cdots, \underline{c}_n^T]^T 
	\eean
	with the $\alpha$ components of $\underline{c}_i$ stored on the $i$-th node in such a way that the following two properties are satisfied:
	{\em Data Collection: }The message \uu\ can be uniquely recovered from the contents $\{c_{i_j}\}_{j=1}^k$ of any $k$ nodes.\\
	{\em Node Repair: } If the $f$-th node storing $\underline{c}_f$ fails, then a replacement node (the node which replaces the failed node $f$) can 
	\ben
	\item contact any subset $D \subseteq [n] \setminus \{f\}$ of the remaining $(n-1)$ nodes of size $|D| = d$,
	\item Each node $h \in D$ (called a helper node)  maps the $\alpha$ symbols $\underline{c}_{h}$ stored in it on to a collection of $\beta$ repair symbols $\underline{a}_{h,f}^D \in \fqbeta$,
	\item pool together the $d\beta$ repair symbols $\{\underline{a}_{h,f}^D : h \in D \}$ thus computed to use them to create a replacement vector $\hat{\underline{c}}_f \in \fqalpha$ whose $\alpha$ components are stored in the replacement node, in a such a way that the contents of the resultant nodes, 	with the replacement node replacing the failed node, once again forms a regenerating code. 
	\een
\end{defn}

A regenerating code is said to be {\em exact-repair} (ER) regenerating code if the contents of the replacement node are exactly same as that of the failed node, ie., $\hat{\underline{c}}_f = \underline{c}_f$. Else the code is said to be {\em functional-repair} (FR) regenerating code. A regenerating code is said to be linear if
\ben
\item $E(\uu_1 + \theta \uu_2) = E(\uu_1) + \theta E(\uu_2)$, $\uu_1, \uu_2 \in \fqB, \theta \in \fq$ and
\item the map mapping the contents $\underline{c}_h$ of the $h$-th helper node on to the corresponding $\beta$ repair symbols $\underline{a}_{h,f}^D$ is linear over \fq.
\een

Thus a regenerating code is a code over a vector alphabet \fqalpha\ and the quantity $\alpha$ is termed the {\em sub-packetization level} of the regenerating code. The total number $d\beta$ of \fq\ symbols $\{\underline{a}_{h,f}^D : h \in D \}$ transferred for repair of failed node $f$ is called the {\em repair bandwidth} of the regenerating code. Note that this number $d \beta$ is same for any failed node $f \in [n]$. The rate of the regenerating code is given by $R = \frac{B}{n\alpha}$. Its reciprocal $\frac{n\alpha}{B}$ is the {\em storage overhead}.

\subsection{Cut-Set Bound}

Let us assume that \calc\ is a functional-repair regenerating code having parameter set: $((n,\ k, \ d) \ ,(\alpha, \ \beta), \ B)$. Since an exact-repair regenerating code is also a functional-repair code, this subsumes the case when \calc\ is an exact-repair regenerating code. Over time, nodes will undergo failures and every failed node will be replaced by a replacement node. Let us assume to begin with, that we are only interested in the behaviour of the regenerating code over a finite-but-large number $N >> n$ of node repairs. For simplicity, we assume that repair is carried out instantaneously. Then at any given time instant $t$, there are $n$ functioning nodes whose contents taken together comprise a regenerating code. At this time instant a data collector could connect to $k$ nodes, download all of their contents and decode to recover underlying message vector $\uu$. 
Thus in all, there are at most $N{n \choose k}$ distinct data collectors which are distinguished based on the particular set of $k$ nodes to which the data collector connects. 

Next, we create a source node $S$ that possesses the $B$ message symbols $\{u_i\}_{i=1}^B$, and draw edges connecting the source to the initial set of $n$ nodes. 
We also draw edges between the $d$ helper nodes that assist a replacement node and the replacement node itself as well as edges connecting each data collector with the corresponding set of $k$ nodes from which the data collector downloads data. All edges are directed in the direction of information flow. We associate a capacity $\beta$ with edges emanating from a helper node to a replacement node and an $\infty$ capacity with all other edges. Each node can only store $\alpha$ symbols over $\fq$. We take this constraint into account using a standard graph-theory construct, in which a node is replaced by $2$ nodes separated by a directed edge (leading towards a data collector) of capacity $\alpha$. We have in this way, arrived at a graph (see Fig.\ref{fig:cutsetbound}) in which there is one source $S$ and at most $N{n \choose k}$ sinks $\{T_i\}$. 

Each sink $T_i$ would like to be able to reconstruct all the $B$ source symbols $\{u_i\}$ from the symbols it receives. This is precisely the multicast setting of network coding. A principal result in network coding tells us that in a multicast setting,  one can transmit messages along the edges of the graph in such a way that each sink $T_i$ is able to reconstruct the source data, provided that the minimum capacity of a cut separating $S$ from $T_i$ is $\ge B$.

A cut separating $S$ from $T_i$ is simply a partition of the nodes of the network into $2$ sets: $A_i$ containing $S$ and $A_i^c$ containing $T_i$. The capacity of the cut is the sum of capacities of the edges leading from a node in $A_i$ to a node in $A_i^c$. A careful examination of the graph will reveal that the minimum capacity $Q$ of a cut separating a sink $T_i$ from source $S$ is given by $
Q = \sum\limits_{i=0}^{k-1} \min \{\alpha, (d-i)\beta\}
$
(see Fig.\ref{fig:cutsetbound} for an example of a cut separating source from sink). This leads to the following upper bound on file size \cite{DimGodWuWaiRam}:
\bea
\label{eq:cutsetbound}
B &\le& \sum\limits_{i=0}^{k-1} \min \{\alpha, (d-i)\beta\}.
\eea
\paragraph{Interpretaion of Bound on File Size:} The above bound could be interpreted as satisfying data collection property by collecting the information on a set of $k$ nodes. First we collect $\min \{\alpha,d\beta \}$ symbols for recovering the information of node one. Since we have information on node one the set of symbols needed to recover the information in another node is at most $\min \{ \alpha,(d-1)\beta \}$. Continuing this way we can collect the information on node $i$ by using the information we already have on node $1$ to $i-1$ with excess information of at most $\min \{ \alpha,(d-(i-1))\beta \}$. The above bound hence could be interpreted as saying that the actual file size is upper bound by the number of symbols with which we can recover it.

Network coding also tells us that when only a finite number of regenerations take place, this bound is achievable and furthermore achievable using linear network coding, i.e., using only linear operations at each node in the network when the size $q$ of the finite field \fq\ is sufficiently large. 
In a subsequent result \cite{Wu}, Wu established using the specific structure of the graph, that even in the case when the number of sinks is infinite, the upper bound in \eqref{eq:cutsetbound} continues to be achievable using linear network coding.

In summary, by drawing upon network coding, we have been able to characterize the maximum file size of a regenerating code given parameters $\{k, d, \alpha, \beta\}$ for the case of functional repair when there is constraint placed on the size $q$ of the finite field $\fq$. Note interestingly, that the upper bound on file size is independent of $n$. Quite possibly, the role played by $n$ is that of determining the smallest value of field size $q$ for which a linear network code can be found having file size $B$ satisfying \eqref{eq:cutsetbound}.
A functional regenerating code having parameters: $(n, \ k, \ d) \ (\alpha, \ \beta, \ B)$ is said to be optimal (a) the file size $B$ achieves the bound in \eqref{eq:cutsetbound} with equality and (b) reducing either $\alpha$ or $\beta$ will cause the bound in \eqref{eq:cutsetbound} to be violated.

\subsection{Storage-Repair Bandwidth Tradeoff}

\begin{wrapfigure}{R}{2.4in}
	\bc
	\includegraphics[width=2.4in]{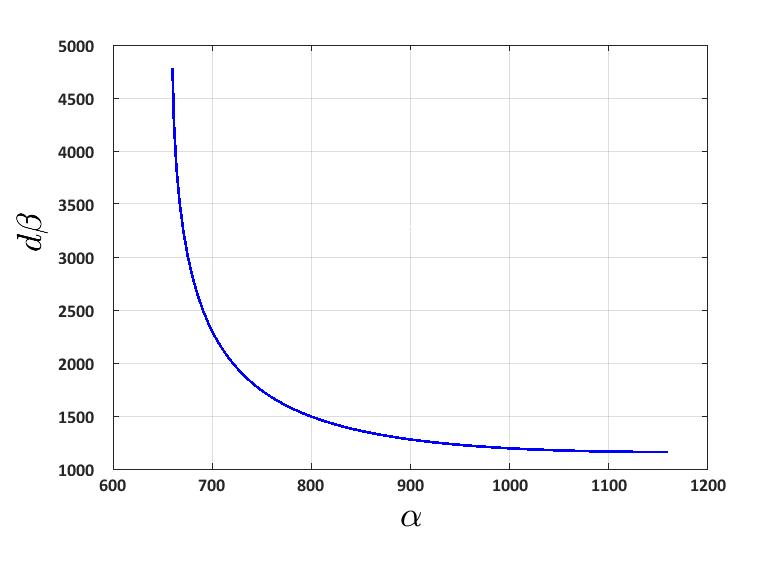}
	\caption{Storage-repair bandwidth tradeoff. Here, $(n = 60, k = 51, d=58, B = 33660)$.\label{fig:srb_tradeoff}}
	\ec
\end{wrapfigure}

We have thus far, specified code parameters $(k,d) (\alpha, \beta)$ and asked what is the largest possible value of file size $B$. If however, we fix parameters $(n,k,d,B)$ and ask instead what is the smallest values of $(\alpha, \beta)$ for which one can hope to achieve \eqref{eq:cutsetbound}, it turns out as might be evident from the form of the summands on the RHS of \eqref{eq:cutsetbound}, that there are several pairs $(\alpha, \beta)$ for which equality holds in \eqref{eq:cutsetbound}.  In other words, there are different flavors of optimality. 

For a given file size $B$, the storage overhead and normalized repair bandwidth are given respectively by $\frac{n\alpha}{B}$ and $\frac{d\beta}{B}$. Thus $\alpha$ reflects the amount of storage overhead while $\beta$ determines the normalized repair bandwidth. For fixed $(n,k,d,B)$ there are several pairs $(\alpha, \beta)$ for which equality holds in \eqref{eq:cutsetbound}. These pairs represent a tradeoff between storage overhead on the one hand and normalized repair bandwidth on the other as can be seen from the example plot in Fig:~\ref{fig:srb_tradeoff}.  Clearly, the smallest value of $\alpha$ for which the equality can hold in \eqref{eq:cutsetbound} is given by $\alpha = \frac{B}{k}$. Given $\alpha = \frac{B}{k}$, the smallest permissible value of $\beta$ is given by $\beta = \frac{\alpha}{d-k+1}$. This represents the minimum storage regeneration point and codes achieving \eqref{eq:cutsetbound} with $\alpha = \frac{B}{k}$ and $\beta = \frac{\alpha}{d-k+1}$ are known as minimum storage regenerating (MSR) codes.  At the other end of the tradeoff, we have the minimum bandwidth regenerating (MBR) code whose associated $(\alpha, \beta)$ values are given by $\beta = \frac{B}{dk-{k \choose 2}}$, $\alpha = d\beta$.
\begin{note}
	Since a regenerating code can tolerate $(n-k)$ erasures by the data collection property, it follows that the minimum Hamming weight $d_{\min}$ of a regenerating code must satisfy $d_{\min} \ge (n-k+1)$.  By the Singleton bound, the largest size $M$ of a code of block length $N$ and minimum distance $d_{\min}$ is given by $M \le Q^{n - d_{\min} + 1} \le Q^k$, where $Q$ is the size of alphabet of the code. Thus $Q = q^{\alpha}$ in the case of regenerating code and it follows therefore that size $M$ of a regenerating code must satisfy $M \le q^{k\alpha}$, or equivalently $q^B \le q^{k\alpha}$ or $B \le k\alpha$. But $B=k\alpha$ in the case of MSR code and it follows that an MSR code is an MDS code over a vector alphabet. Such codes also go by the name MDS array code.
\end{note}

From a practical perspective, exact-repair regenerating codes are easier to implement as the contents of the $n$ nodes in operation do not change with time. Partly for this reason and partly for reasons of tractability, with few exceptions, most constructions of regenerating codes belong to the class of exact-repair regenerating codes. Examples of functional-repair regenerating code include the $d=(k+1)$ construction in \cite{ShaRasKumRam_ia} as well as the construction in \cite{HuChenLee_NCcloud}.

Early constructions of regenerating codes focused on the two extreme points of the storage-repair bandwidth (S-RB) tradeoff, namely the MSR, MBR points. The storage industry places a premium on low storage overhead. This is not too surprising, given the vast amount of data, running into the exa-bytes, stored in today's data centers. In this connection, we note that the maximum rate of an MBR code is given by:
\bean
R_{\text{MBR}} = \frac{B}{n\alpha} = \frac{(dk - {k \choose 2})\beta}{n d\beta} = \frac{dk-{k \choose 2}}{nd},
\eean
which can be shown to be upper bounded by $R_{\text{MBR}} \le \frac{1}{2}$ which is achieved when $k=d=(n-1)$. This makes MSR codes of greater practical interest when minimization of storage overhead  is of primary interest. 

\subsection{MSR Codes} 
An $[n,k]$ MSR code is an $[n,k]$ MDS codes over the vector alphabet $\fq^{\alpha}$ satisfying the additional constraint that a failed node can be repaired by contacting $d$ helper nodes, while downloading $\beta$ symbol over \fq\ from each helper node.  In this case, we have $B=k \alpha$ of message symbls over \fq\ encoded by an MSR code.    Thus, MSR codes are characterized by the parameter set 
\bean
\left\{ (n,k,d),\ (\alpha,\beta),\ B, \ \fq) \right\}, 
\eean
where 
\bit 
\item \fq\ is the underlying finite field, 
\item $n$ is the number of code symbols $\{\underline{c}_i\}_{i=1}^n$ each stored on a distinct node or storage unit and 
\item each code symbol $\underline{c}_i$ is an element of  $\fq^{\alpha}$. 
\eit

As we have already seen, the number $\beta$ of symbols downloaded from each helper node in an MSR code is given by 
\bean
\beta = \frac{\alpha}{d-k+1}. 
\eean

As we have seen, each code symbol \uc{i} is typically stored on a distinct node. Thus the index $i$ of a code symbol is synonymous with the index of the node upon which that code symbol is stored. Throughout this chapter, we will focus on a linear MSR code i.e., the encoding is done by:  $[\uc{1}^T,\hdots,\uc{n}^T]^T = G\underline{m}$ where $G$ is an $(n\alpha \times k\alpha)$ generator matrix over \fq\ and $\underline{m}$ is a $(k\alpha \times 1)$ message vector over \fq\ comprising of $B$ message symbols encoded by the MSR code.

\subsubsection{Linear Repair}  

Throughout this chapter, we will assume linear repair of the failed node.   By linear repair, we mean that the $\beta$ symbols passed on from a helper node $i$ to the replacement of a failed node $j$ are obtained through a linear transformation:
\bean
\uc{i} \rightarrow S_{(i,j)} \uc{i} ,
\eean
where $S_{(i,j)}$ is an $(\beta \times \alpha)$ matrix over \fq. The matrix $S_{(i,j)}$ is called a repair matrix.  While the matrix $S_{(i,j)}$ could potentially be a function of which nodes are participating in the repair of failed node $j$, this chapter is mostly concerned with the case $d=(n-1)$, in which case, all the remaining $(n-1)$ nodes participate as helper nodes. 

As a result, when $d=n-1$, the input to the replacement of failed node $j$ is the set:
\bean
\left\{   S_{(i,j)} \uc{i}\mid    i \in [n], \ i \neq j             \right\} .
\eean 
By linear repair of a node $j$ we also mean: the code symbol $\uc{j} = f_j \left( \left\{   S_{(i,j)} \uc{i}\mid    i \in [n], \ i \neq j  \right\} \right)$ where $f_j:\fq^{d\beta} \rightarrow \fq^{\alpha}$ is a deterministic linear function.
We refer to an MSR code as an {\em optimal access} MSR code if the matrix $S_{(i,j)}$ has rows picked from the standard basis $\{e_1,\hdots,e_{\alpha}\}$ for $\fq^{\alpha}$., i.e., the symbols over \fq\ downloaded for the repair of a failed node from node $i$ are simply a subset of size $\beta$, of the $\alpha$ components of the vector $\underline{c}_i$. This property of repair is also termed as as help-by-transfer repair. 

When $S_{(i,j)}=S_j$ $,\forall i \neq j \in [n]$, we say repair matrices are independent of helper-node index $i$ (or constant repair matrix case).
Now lets consider the MSR code for any $d(\leq n-1)$. Let, $S^D_{(i,j)}$ be $(\beta \times \alpha)$ the repair matrix where $S^D_{(i,j)} \uc{i}$ is downloaded for the repair of the node $j$ when the helper nodes ($d$ nodes from which data is downloaded for the repair of node $j$) belong to the set $D$. 
As a result, the input to the replacement of failed node $j$ is the set:
\bean
\left\{   S^D_{(i,j)} \uc{i}\mid    i \in D \right\},
\eean 
for helper nodes in the set $D \subseteq [n]-\{j\}$ such that $|D|=d$. Here also: the code symbol $\uc{j} = f_{D,j} \left( \left\{   S^D_{(i,j)} \uc{i}\mid    i \in D, \ i \neq j  \right\} \right)$ where $f_{D,j}:\fq^{d\beta} \rightarrow \fq^{\alpha}$ is a deterministic linear function.
When $S^D_{(i,j)}=S_{(i,j)}$, we say that repair matrices are independent of identity of remaining helper nodes. Similar to $d=n-1$ case, the term {\em optimal access} MSR code or repair by help-by-transfer for any $d$ means that the rows of $S^D_{(i,j)}$ are picked from standard basis $\{e_1,\hdots,e_{\alpha}\}$ of $\fq^{\alpha}$.  We drop the superscript $D$ in  $S^D_{(i,j)}$ when $d=n-1$.

An open problem in the literature on regenerating codes is that of determining the smallest value of sub-packetization level $\alpha$ of an optimal-access (equivalently, help-by-transfer) MSR code, given the parameters $\{(n,k,d=(n-1)\}$. 
This question is addressed in \cite{TamWanBru_access_tit}, where a lower bound on $\alpha$ is given for the case of a regenerating code that is MDS and where only the systematic nodes are repaired in help-by-transfer fashion with minimum repair bandwidth.   In the literature these codes are often referred to as optimal access MSR codes with systematic node repair. The authors of  \cite{TamWanBru_access_tit} establish that:
\bean
\alpha \geq r^{\frac{k-1}{r}},
\eean
in the case of an optimal access MSR code with systematic node repair. 

In a slightly different direction, lower bounds are established in \cite{GopTamCal} on the value of $\alpha$ in a general MSR code that does not necessarily possess the help-by-transfer repair property. In \cite{GopTamCal} it is established that:
\bean
k \leq 2 \log_2(\alpha) (\lfloor \log_{\frac{r}{r-1}} (\alpha) \rfloor + 1),
\eean
while more recently, in \cite{HuaparXia} the authors prove that:
\bean
k \leq 2 \log_r(\alpha) (\lfloor \log_{\frac{r}{r-1}} (\alpha) \rfloor + 1). 
\eean

\subsubsection{Desirable Properties of MSR Code}
Desirable properties of MSR codes are high rate, low field size, low sub-packetization, repair property for all nodes (All node repair property), help by transfer repair of failed node and optimal update (optimal update is not defined here. Please see \cite{TamWanBru_access_tit}.). The properties mentioned above such as high rate, low field size, low sub-packetization, repair property for all nodes are clearly desirable. Help-by-Transfer repair is useful because a node may not have computational power on its own and hence it is desirable that repair only requires nodes to access its data and transmit it. In this chapter, we present new results on lower bounds on sub-packetization of MSR codes with help-by-transfer repair property.

\subsubsection{Overview of Existing Constructions of MSR Code}
The existing constructions of MSR codes are summarized in the Table \ref{Const_MSR_Codes}.
\begin{table}[ht!]
	
	\bean
	\scalebox{0.7}{$
		\begin{array}{||l|c|c|c|c|c||}
		\hline \hline
		\text{MSR Code} &  \text{Parameters} & \alpha & \text{Field Size} & \text{All Node} & \text{Optimal} \\
		&&&& \text{Repair}& \text{ Access} \\ \hline
		\text{Low Rate Constructions} \\
		\hline
		\text{\cite{ShaRasKumRam_ia} Shah et.al (using Interference Alignment)} & (n,k,d=n-1\ge 2k-1) & r & 2r & \text{No} & \text{Yes} \\ \hline
		\text{\cite{SuhRam} Suh et.al (using Interference Alignment)} & (n, k, d \ge 2k-1) & s & 2r & \text{Yes} & \text{No} \\
		& (n, k \le 3, d) & & & & \\ \hline
		\text{\cite{RasShaKum_pm} Rashmi et.al (Product matrix Framework)}  & (n \ge 2k-1, k, d) & r & n & \text{Yes} & \text{No} \\  \hline
		\text{High Rate Constructions} \\
		\hline
		\text{\cite{PapDimCad} Papailiopoulos et.al } & (n,k,d=n-1)& r^k & O(r) &\text{No} &  \text{No} \\ \hline 
		\text{\cite{TamWanBru} Tamo et.al (Zig Zag Codes)}  & (n,k,d=n-1)  & r^{k+1} & \leq 4 \text{ when } r \leq 4, & \text{Yes} & \text{Yes} \\
		\text{\cite{WangTamoBruck} Wang et.al} & & &\text{else non-explicit} & & \\
		\hline
		\text{\cite{YeBar_1} Ye et.al (for any $d$)} & (n,k,d) & s^n & sn & \text{Yes} & \text{Yes} \\ \hline
		\text{\cite{YeBar_2} Ye et.al} & & & & & \\ 
		\text{\cite{SasVajKum_arxiv} Sasidharan et.al (Clay Code)} & (n, k, d=n-1)& r^{\left\lceil\frac{n}{r}\right\rceil} & n & \text{Yes} & \text{Yes} \\
		\text{\cite{LiTangTian} Li et.al} & &  & & & \\ \hline
		\text{\cite{VajBalKum} Vajha et.al} & (n,k,d) &  & & &  \\
		& d \in \{k+1, k+2, k+3\} & s^{\left\lceil \frac{n}{s} \right\rceil} & O(n) & \text{Yes} &\text{Yes}\\ \hline
		\hline \hline
		\end{array}	
		$}
	\eean 
	\caption{A list of MSR constructions and the parameters. In the table $r=n-k$, $s=d-k+1$ and when All Node Repair is No, the constructions are systematic MSR. By `non-explicit' field-size, we mean that the order of the size of the field from which coefficients are picked is not given explicitly. } \label{Const_MSR_Codes}
\end{table}
There are several known MSR constructions. In \cite{ShaRasKumRam_ia,SuhRam} authors provide low rate constructions with small sub-packetization and small field size. The product matrix construction presented in \cite{RasShaKum_pm} is for any $2k-2 \le d \le n-1$ and low sub-packetization and low field size. In \cite{PapDimCad}, the authors provide a high-rate MSR construction using Hadamard designs for systematic node repair. In \cite{TamWanBru}, high-rate systematic node repair MSR codes called Zigzag codes were constructed for $d=n-1$ with help-by-transfer repair property. These codes however had large field size and sub-packetization that is exponential in in $k$ but provided low field size for $r \leq 4$. This construction was extended in \cite{WangTamoBruck} to enable the repair of all nodes. The existence of MSR codes for any value of $(n,k,d)$  as $\alpha$ tends to infinity is shown in \cite{CadJafMalRamSuh}.  In \cite{YeBar_1} authors provided high rate construction for any $d$ but with large sub-packetization and small field size with help-by-transfer repair property. Construction of MSR codes with small sub-packetization with small field size with help-by-transfer repair property is presented in \cite{YeBar_2}, \cite{SasVajKum_arxiv}, \cite{LiTangTian}. As we have seen already, in \cite{TamWanBru_access_tit},\cite{GopTamCal},\cite{HuaparXia} lower bounds for sub-packetization($\alpha$) were presented. In \cite{TamWanBru_access_tit} a lower bound $\alpha \ge r^{\frac{k-1}{r}}$ for the special case of an optimal-access MSR code with systematic node repair was provided. This bound was recently improved by Balaji et al. in \cite{BalKum_subpkt} (These lower bounds on sub-packetization are presented in this chapter. A summary of lower bounds presented in this chapter is in Table~\ref{tab:results}.). Results in Table~\ref{tab:results}, proves sub-packetization-optimality of the explicit codes provided in \cite{YeBar_2}, \cite{SasVajKum_arxiv}, \cite{LiTangTian} for $d=n-1$ and the non-explicit code in \cite{RawKoyVis_msr} for $d < n-1$.

Though the literature contains prior optimal access constructions for $d < n-1$, the resultant codes were either non-explicit \cite{RawKoyVis_msr}, or else have large sub-packetization \cite{YeBar_1}, or are of high field size \cite{GopFazVar}. In the present chapter, we also give a brief overview of our optimal-access MSR codes that have optimal sub-packetization for any $(n,k,d=k+1, k+2, k+3)$ and which can be constructed using a field of size $q = O(n)$.

\begin{table}[h!] \caption{A tabular summary of the new lower bounds on sub-packetization-level $\alpha$ contained in the present chapter.   In the table, the number of nodes repaired is a reference to the number of nodes repaired with minimum possible repair bandwidth $d\beta$ and moreover, in help-by-transfer fashion.  An * in the first column indicates that the bound is tight, i.e., that there is a matching construction in the literature.  With respect to the  first entry in the top row, we note that all MSR codes are MDS codes.   For the definition of an MDS code with optimal access repair for any node belonging to a set of $w$ nodes please see Section \ref{sec:MDScodes_ch6}.\label{tab:results}}
	\begin{center}
		\resizebox{\columnwidth}{!}{%
			\begin{tabular}{ | P{4em} | P{1cm}| P{3cm} |  P{4cm}| P{7cm} |  P{3cm}| P{3cm} |} 
				\hline
				\shortstack{MSR \\ or MDS \\  Code ?} & $d$ & \shortstack{No. of Nodes \\ Repaired} & \shortstack{Assumption on \\ Repair Matrices \\ $S_{(i,j)}$ } & \shortstack{Lower Bound on $\alpha$ \\ and Constructions achieving our \\ lower bound on $\alpha$} & \shortstack{Reference \\ in this chapter} & \shortstack{Previous \\ Known Bound} \\ 
				\hline
				\shortstack{MSR* } & $n-1$ & $n$ & none & \shortstack{\ \ \ \\$\alpha \geq  \min \{r^{\lceil \frac{n-1}{r} \rceil},r^{k-1}\}$ \\ Constructions (when $r \nmid (n-1)$): \cite{YeBar_2,SasVajKum_arxiv}} & Theorem \ref{thm:main} & $\alpha \geq r^{\frac{k-1}{r}}$ \\ 
				\hline
				\shortstack{MSR*}  & $n-1$ & $n$ & \shortstack{ independent of \\ helper-node index $i$} & \shortstack{\ \ \ \\ $\alpha \geq  \min \{ r^{\lceil \frac{n}{r} \rceil},r^{k-1}\}$ \\ Constructions: \cite{YeBar_2,SasVajKum_arxiv}} & Corollary \ref{sb_5} & $\alpha \geq r^{\frac{k}{r}}$ \\
				\hline        
				\shortstack{MSR*}  & any $d$ & $n$ & \shortstack{\ \ \ \\ independent \\ of identity of \\ remaining \\ helper nodes } &  \shortstack{ $s=d-k+1$ \\ $\alpha \geq  \min \{ s^{\lceil \frac{n-1}{s} \rceil},s^{k-1} \}$ \\ Construction: \cite{RawKoyVis_msr}} & Corollary \ref{sb_6} & none\\
				\hline       
				MDS* & $n-1$ & $w$ ($\leq n-1$) & none & \shortstack{$\alpha \geq  \left\{ \begin{array}{rl} \min \{ r^{\lceil \frac{w}{r} \rceil},r^{k-1} \}, \  & w > (k-1)  \\
					r^{\lceil \frac{w}{r} \rceil}, \ & w \leq (k-1) . \end{array} \right. $ \\ Construction: \cite{LiTangTian}} & Corollary \ref{cor:node_subset} & \shortstack{for $w=k$ \\ $\alpha \geq r^{\frac{k-1}{r}}$}  \\
				\hline
				MDS & any $d$ & $w$ ($\leq d$) & none &  \shortstack{ \ \ \ \\ $s=d-k+1$ \\ $\alpha \geq  \left\{ \begin{array}{rl} \min \{ s^{\lceil \frac{w}{s} \rceil},s^{k-1}\}, \  & w > (k-1) \\
					s^{\lceil \frac{w}{s}\rceil}, \ & w \leq (k-1) . \end{array} \right. $} & Corollary \ref{sb_4} & none \\ 
				\hline
			\end{tabular}
		}
	\end{center}
\end{table}

\paragraph{Comparison of our Lower Bound on $\alpha$ with Exisitng Code Constructions:}
\ben
\item When $r \nmid n-1$, our bound on $\alpha$ (Theorem \ref{thm:main}) for optimal access MSR code with $d=n-1$ becomes:
\bean
\alpha \geq  \min \{r^{\lceil \frac{n}{r} \rceil},r^{k-1}\} 
\eean
For $k \geq 5$, this reduces to $\alpha \geq  r^{\lceil \frac{n}{r} \rceil}$. 
The latter lower bound on $\alpha$ is achievable by the constructions in \cite{YeBar_2,SasVajKum_arxiv}. Hence our lower bound on $\alpha$ is tight.
Although our bound on $\alpha$ is shown to be tight only for $(n-1)$ not a multiple of $(n-k)$, the lower bound is valid for all parameters (with $d=n-1$) and when $r$ divides $(n-1)$.
\item Our bound on $\alpha$ (Corollary \ref{cor:node_subset} ) for MDS code with optimal access repair (repair with minimum repair bandwidth and help-by-transfer repair) for a failed node when it belongs to a fixed set of $w (\leq n-1)$ nodes (see Section \ref{sec:MDScodes_ch6}) is:
\bean
\alpha \geq  \left\{ \begin{array}{rl} \min \{ r^{\lceil \frac{w}{r} \rceil},r^{k-1} \}, \  & w > (k-1) \\
	r^{\lceil \frac{w}{r} \rceil}, \ & w \leq (k-1) . \end{array} \right. 
\eean
The above bound for $k \geq 5$ becomes $\alpha \geq r^{\lceil \frac{w}{r}\rceil}$.
This lower bound is achieved by the construction given in \cite{LiTangTian}. Hence our lower bound on $\alpha$ for the repair of $w$ nodes is tight.
\item The constructions in \cite{YeBar_2,SasVajKum_arxiv} have repair matrices that are independent of the helper node index $i$ i.e., $S_{(i,j)}=S_j$ and has sub-packetization $\alpha = r^{\lceil \frac{n}{r} \rceil}$ which achieves our lower bound on $\alpha$ (Corollary \ref{sb_5}) under the assumption that $S_{(i,j)}=S_j$ for an optimal access MSR code with $d=n-1$. 
\item Our bound given in Corollary \ref{sb_6} is achieved by construction in \cite{RawKoyVis_msr} when $(d-k+1)$ does not divide $n-1$ under the assumption $S^D_{i,j}=S_{i,j}$.
\een
Hence our lower bounds on $\alpha$ are tight for four cases.

\section{Bounds on Sub-Packetization Level of an Optimal Access MSR Code} \label{sec:LowerBoundMSR_ch6}
In this section, we derive three lower bounds on sub-packetization of optimal access MSR codes for three different cases. First, we derive a lower bound on sub-packetization of optimal access MSR codes for $d=n-1$ case with no assumptions. Second, we derive a more tighter lower bound on sub-packetization of optimal access  MSR codes for $d=n-1$ case with assumptions on repair matrices. Third, we derive a lower bound on sub-packetization of MSR codes for $d<n-1$ case with assumptions on repair matrices. Tightness of each of these bounds is discussed already. The proof of the last two bounds follows as a corollary of the first bound. The first bound is proved by analysing the intersection of row space of repair matrices.
\subsection{Notation}
We adopt the following notation throughout the chapter. 
\ben
\item Given a matrix $A$, we use $<A>$ to refer to the row space of the matrix $A$, 
\item Given a subspace $V$ and a matrix $A$, by $VA$ we will mean the subspace 
$\{ \underline{v} A \mid \underline{v} \in V\}$ obtained through transformation of $V$ by $A$.
\bit
\item Thus for example, $(\bigcap_i <S_i>) A$ will indicate the subspace obtained by transforming the intersection subspace $(\bigcap_i <S_i>)$ through right multiplication by $A$.
\eit 
\een

\subsection{Improved Lower Bound on Sub-packetization of an Optimal Access MSR Code with $d=n-1$} \label{sec:dnGMSR_ch6}
\begin{thm} \label{thm:main} \textbf{(Sub-packetization Bound)}: 
	Let $\mathcal{C}$ be a linear optimal access MSR code having parameter set 
	\bean
	\left\{ (n,k,d),\ (\alpha,\beta),\ B, \ \fq) \right\},
	\eean
	with $d=(n-1)$ and linear repair for all $n$ nodes. Let $r=n-k$. Then we must have:
	\bean
	\alpha \geq  \min \{r^{\lceil \frac{n-1}{r} \rceil},r^{k-1}\} \label{eq:Main_Bound}.
	\eean
\end{thm}

\vspace*{0.3in}

The proof of the theorem will make use of Lemma~\ref{lem:rsi} below.   We begin with some helpful notations.  We will use the indices (two disjoint subsets of $[n]$):
\bean
\{ u_1,u_2,\cdots, u_k \}, \ \{p_1,p_2,\cdots,p_r\}
\eean
to denote the $n$ nodes in the network over which the code symbols are stored (Note that the code symbol \uc{i} is stored in node $i$.).  Let $\ell, \ 2 \leq \ell \leq (k-1)$, be an integer and set
\bean
U & = & \{ u_1,u_2,\cdots, u_{\ell} \}, \\
V & = & \{ u_{\ell+1},u_{\ell+2},\cdots, u_k \}, \\
P & = & \{p_1,p_2,\cdots,p_r\} .
\eean
Note that our choice of $\ell$ ensures that neither $U$ nor $V$ is empty. We refer to the row space of $S_{(i,j)}$ as a repair subspace.

\begin{lem} \label{lem:rsi} \textbf{(Repair Subspace Intersection)}: 
	For the code $\mathcal{C}$ which is a linear optimal access MSR code with $d=(n-1)$ and linear repair for all $n$ nodes. We must have:
	\bea \label{eq:lemma_2}
	\sum_{i=1}^{r} \dim \left( \bigcap_{u \in U} < S_{(p_i,u)} > \right)  \ \leq \ 
	\dim \left( \bigcap_{u \in U - \{u_{\ell}\}} <S_{(p,u)}> \right), 
	\eea
	
	where $U,P,u_{\ell},p_i$ are as defined above in this section and $p$ is an arbitrary node in  $P$.  Furthermore, $\dim \left( \bigcap_{u \in U} <S_{(p,u)} > \right)$ is the same for all $p \in P$. 
\end{lem}	
\begin{proof} \ 
	
	\paragraph{Invariance of $\ell$-fold Intersection of Repair Subspaces Contributed by a Parity Node}  Let us consider the nodes in $U \cup V$ as systematic nodes and nodes in $P$ as parity nodes. Note that the sets $U,V,P$ are pairwise disjoint and are arbitrary subsets of $[n]$, under the size restrictions $2 \leq |U|=\ell \leq k-1,|P|=r$ and $U \cup V \cup P = [n]$. First we prove that $\dim \left(\bigcap_{u \in U} <S_{(p,u)}> \right)$ is the same for all $p \in P$. Note that $< S_{(p,u)} >$ is the row space of the repair matrix carrying repair information from helper (parity) node $p$ to the replacement of the failed node $u$.  Thus we are seeking to prove that the $\ell$-fold intersection of the $\ell$ subspaces $\{ <S_{(p,u)}> \}_{u \in U}$ obtained by varying the failed node $u \in U$ is the same, regardless of the parity node $p \in P$ from which the helper data originates.   
	
	To show this, consider a generator matrix $G$ for the code $\mathcal{C}$ in which the nodes of $P$ are the parity nodes and the nodes in $U \cup V \ = \ \{u_1,u_2,\cdots,u_k\}$ are the systematic nodes.  Then $G$ will take on the form: 	
	\bea
	G = \left[
	\scalemath{1}{
		\begin{array}{cccc}
			I_{\alpha} & 0 & \hdots & 0 \\
			0 & I_{\alpha} & \hdots & 0 \\
			\vdots & \vdots & \vdots & \vdots \\
			0 & 0 & \hdots & I_{\alpha} \\
			A_{p_1,u_1} & A_{p_1,u_2} & \hdots & A_{p_1,u_k} \\
			A_{p_2,u_1} & A_{p_2,u_2} & \hdots & A_{p_2,u_k} \\
			\vdots & \vdots & \vdots & \vdots \\
			A_{p_r,u_1} & A_{p_r,u_2} & \hdots & A_{p_r,u_k} \\
		\end{array} 
	}
	\right]. \ \ \label{Gform_3}
	\eea
	Wolog the generator matrix assumes an ordering of the nodes in which the first $k$ nodes in $G$ (node $i$ in $G$ correspond to rows $[(i-1)\alpha+1,i\alpha]$) correspond respectively to $\{u_i \mid 1 \leq i \leq k\}$ and the remaining $r$ nodes in $G$ to $\{p_1,\cdots,p_r\}$ in the same order. 
	A codeword is formed by $G\underline{m}$ where $\underline{m}$ is the vector of $k \alpha$ message symbols over \fq\ encoded by the MSR code and the index of code symbols in the codeword is according to the nodes in $U,V,P$ i.e., the codeword will be of the form $[\uc{u_1}^T,\hdots,\uc{u_k}^T,\uc{p_1}^T,\hdots,\uc{p_r}^T]^T = G \underline{m}$.
	
	\begin{center}
		\begin{figure}
			\begin{center}
				\includegraphics[width=2.5in]{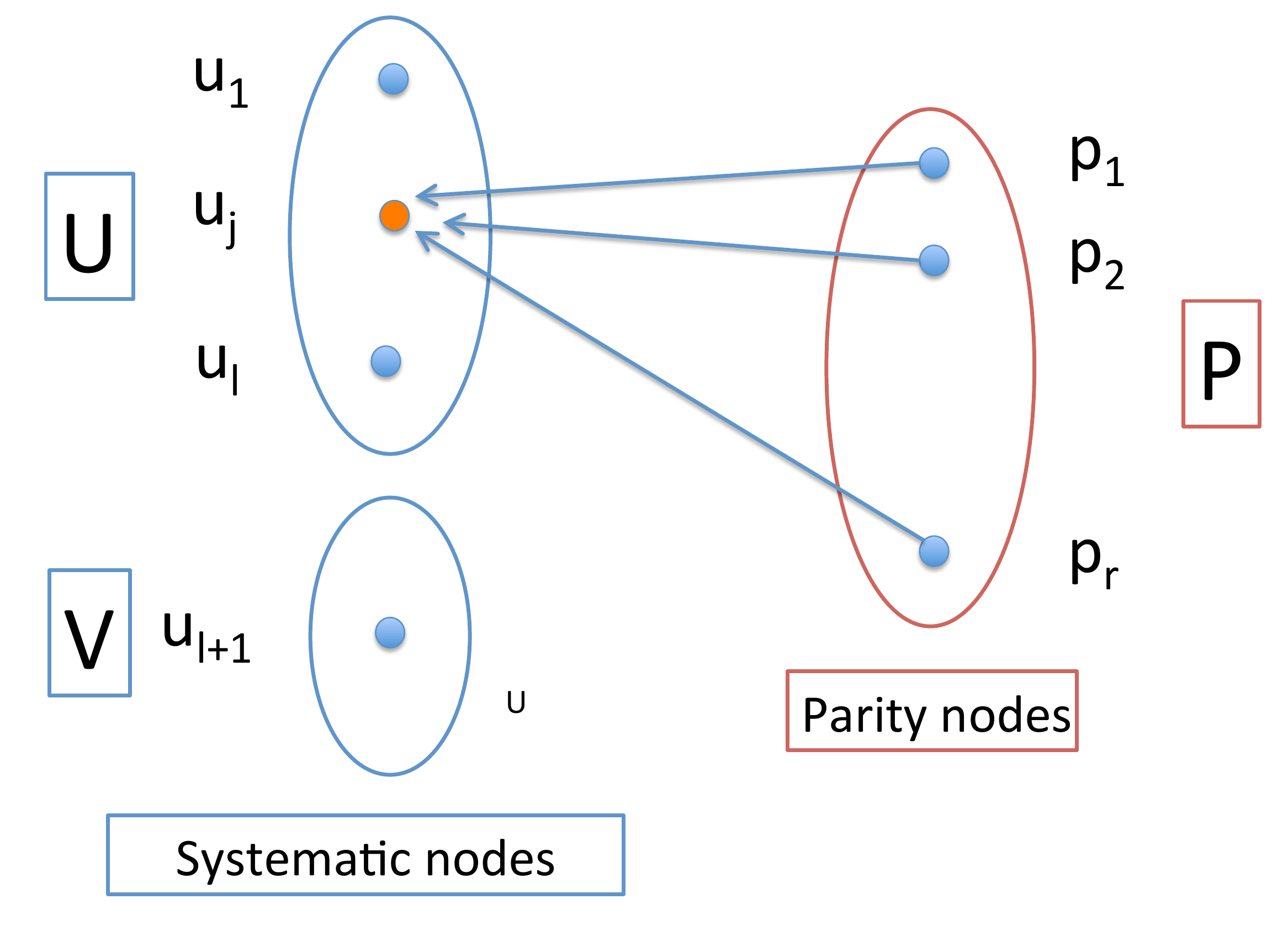}
			\end{center} 
			\caption{The general setting considered here where helper data flows from the parity-nodes $\{p_i\}_{i=1}^r$ forming set $P$ to a failed node $u_j \in U$. \label{fig:setting}}
		\end{figure}
	\end{center}

	By the interference-alignment conditions~\cite{ShaRasKumRam_ia},\cite{TamWanBru_access_tit} applied to the repair of a systematic node $u_j \in U$ (Fig \ref{fig:setting}), we obtain 	(see Lemma \ref{lem:rank_cond_p_node} in the Appendix, for a more complete discussion on interference alignment equations given below): 
	\bea
	<S_{(p_i,u_j)} A_{p_i,u_{\ell+1}}>  & = &  <S_{(p_m,u_j)} A_{p_m,u_{\ell+1}}>, \ \text{ for every pair } p_i, p_m \in P. \label{eq:IA_P_1}
	\eea	 
	
	Equation \eqref{eq:IA_P_1} and Lemma \ref{Row_Spaces} (in Appendix) implies (as $A_{i,j}$ are invertible for all $i,j$) that for every pair $p_i, p_m \in P$ :
	\bea
	\left ( \bigcap_{j=1}^{\ell} <S_{(p_i,u_j)}> \right )  A_{p_i,u_{\ell+1}} = \left ( \bigcap_{j=1}^{\ell} <S_{(p_m,u_j)}> \right)  A_{p_m,u_{\ell+1}}. \label{eq:l_fold}
	\eea
	
	It follows then from the non-singularity of the matrices $A_{i,j}$ and equation \eqref{eq:l_fold},  that $\text{dim}(\bigcap_{u \in U} <S_{(p,u)}>)$ is the same for all $p \in P$.  It remains to prove the main inequality \eqref{eq:lemma_2}.    		
	
	\paragraph{$(\ell-1)$-fold Intersection of Repair Subspaces} 
	We proceed similarly in the case of an $(\ell-1)$-fold intersection, replacing $\ell$ by $\ell-1$ in \eqref{eq:l_fold}.
	We will then obtain: 
	\bea
	\left ( \bigcap_{j=1}^{\ell-1}<S_{(p_i,u_j)}> \right )  A_{p_i,u_\ell} = \left ( \bigcap_{j=1}^{\ell-1} <S_{(p_m,u_j)}> \right )  A_{p_m,u_\ell},\text{ for every pair } p_i, p_m \in P . \label{eq:l_minus_one}
	\eea
	
	\paragraph{Relating $\ell$-fold and $(\ell-1)$-fold intersections}
	Next consider the repair of the node $u_{\ell}$. Then from the full-rank condition of node repair, (see Lemma \ref{lem:rank_cond_p_node} in the Appendix), we must have that 			
	\bea
	\text{rank} \left( \left[
	\scalemath{1}{
		\begin{array}{c}
			S_{(p_1,u_{\ell})} A_{p_1,u_\ell} \\
			S_{(p_2,u_{\ell})} A_{p_2,u_\ell} \\
			S_{(p_3,u_{\ell})} A_{p_3,u_\ell} \\
			\vdots \\
			S_{(p_r,u_{\ell})} A_{p_r,u_\ell} \\
		\end{array} 
	}
	\right]\right) & = & \alpha. \ \ \label{Cond_1}
	\eea
	It follows as a consequence, that 
	\bea
	\bigoplus_{i=1}^r < S_{(p_i,u_{\ell})} A_{p_i,u_\ell} > & = & \mathbb{F}_q^{\alpha},  \label{eq:stratify} 
	\eea 
	and hence, for every $j \in \ [r]$, we must have that 
	\bea
	< S_{(p_j,u_{\ell})} A_{p_j,u_\ell} >  \bigcap \ 		\bigoplus_{i \in [r], i \neq j  }  <S_{(p_i,u_{\ell})} A_{p_i,u_\ell} > & = & \{ \underline{0} \} .  \label{eq:sum_disjoint} 
	\eea 			
	
	It follows from \eqref{eq:l_minus_one} that for any $i_0 \in [r]$ and all $p_m \in P$:
	\bea \nonumber 
	\left ( \bigcap_{u \in U} <S_{(p_m,u)}> \right ) A_{p_m,u_\ell} 	& = & 
	<S_{(p_m,u_{\ell})} A_{p_m,u_\ell}> \bigcap \left ( \left ( \bigcap_{j=1}^{\ell-1} <S_{(p_{m},u_j)}> \right )  A_{p_m,u_\ell} \right )\\ 
	& \subseteq & \left ( \bigcap_{j=1}^{\ell-1} <S_{(p_{i_0},u_j)}> \right )  A_{p_{i_0},u_\ell}.  			\label{eq:l_to_lminus1}
	\eea
	As a consequence of \eqref{eq:sum_disjoint} and \eqref{eq:l_to_lminus1} we can make the stronger assertion:
	\bea
	\bigoplus_{m=1}^{r} \left ( \bigcap_{u \in U} <S_{(p_m,u)}> \right ) A_{p_m,u_\ell}
	& \subseteq & \left (  \bigcap_{j=1}^{\ell-1} <S_{(p_{i_0},u_j)}> \right ) A_{p_{i_0},u_\ell} .
	\eea			
	Since the $A_{i,j}$ are nonsingular, this allows us to conclude that: 
	\bea
	\sum_{p \in P} \dim \left ( \bigcap_{u \in U} <S_{(p,u)}> \right ) \ \leq \ 
	\dim \left ( \bigcap_{u \in U-\{u_{\ell}\}} <S_{(p_{i_0},u)}> \right ). \label{eq:lemma22}
	\eea
	However, since $p_{i_0}$ is an arbitrary node in $P$, this can be rewritten in the form: 
	\bea
	\sum_{i=1}^r \dim \left ( \bigcap_{u \in U} <S_{(p_i,u)}> \right ) \ \leq \ 
	\dim \left ( \bigcap_{u \in U - \{u_{\ell}\}} <S_{(p,u)}> \right ), \label{eq:lemma2_final}
	\eea
	for any $p \in P$, which is precisely the desired equation \eqref{eq:lemma_2}. 
\end{proof} 

\begin{proof} (Proof of Theorem~\ref{thm:main}) \ 
	
	\begin{enumerate}
		\item {\em Invariance of Repair Matrices to Choice of Generator Matrix} We first observe that the repair matrices can be kept constant, even if the generator matrix of the code changes.  This is because the repair matrices only depend upon relationships that hold among code symbols of any codeword in the code and are independent of the particular generator matrix used in encoding.  In particular, the repair matrices are insensitive to the characterization of a particular node as being either a systematic or parity-check node.
		\item {\em Implications for the Dimension of the Repair Subspace} 
		From Lemma~\ref{lem:rsi}, we have that 
		\bea 
		\sum_{i=1}^{r} \dim \left ( \bigcap_{u \in U} <S_{(p_i,u)}> \right )  \ \leq \ 
		\dim \left ( \bigcap_{u \in U - \{u_{\ell}\}} <S_{(p,u)}> \right ), 
		\eea
		and moreover that $\dim \left (\bigcap_{u \in U} <S_{(p,u)}> \right )$ is the same for all $p \in P=\{p_1,...,p_r\}$. 
		It follows that 
		\bea 
		r \times \dim \left (\bigcap_{u \in U} <S_{(p,u)}> \right ) \leq \dim \left ( \bigcap_{u \in U-\{u_{\ell}\}} <S_{(p,u)}> \right ) \label{Cond_3},
		\eea
		i.e., 
		\bea \nonumber 
		\dim \left ( \bigcap_{u \in U} <S_{(p,u)}> \right ) & \leq & \frac{\dim \left (\bigcap_{u \in U-\{u_{\ell}\}} <S_{(p,u)}> \right )}{r}  \\ \nonumber 
		&  \leq & \frac{\dim \left ( \bigcap_{u \in U-\{u_{\ell},u_{\ell-1}\}} <S_{(p,u)}> \right )}{r^2}  \\ \nonumber 
		& \leq & \frac{\dim \left (<S_{(p,u_1)}> \right )}{r^{\ell-1}} \\ 
		& = & \frac{\alpha}{r^{\ell}} \ = \  \frac{\alpha}{r^{|U|}}.
		\label{eq:Cond_4}
		\eea	   
		Lemma \ref{lem:rsi} and its proof holds true for any set $U \subseteq [n]$ of size $2 \leq |U| \le (k-1)$.    As a result, equation \eqref{eq:Cond_4}, also holds for any set $U \subseteq [n]$ of size $2 \leq |U| \le (k-1)$.   
		
		We would like to extend the above inequality to hold even for the case when $U$ is of size $k \leq |U| \leq (n-1)$.  We get around the restriction on $U$ as follows.   It will be convenient in the argument, to assume that $U$ does not contain the $n$th node, i.e., $U \subseteq [n-1]$ and $ n \in P$.   Let us next suppose that  $\alpha < r^{k-1}$ and that $U$ is of size $(k-1)$. We would then have: 
		\bea \label{eq:intersection}
		\dim \left ( \bigcap_{u \in U} <S_{(n,u)}> \right ) & \leq  & \frac{\alpha}{r^{k-1}} \ < \ 1,
		\eea
		which is possible iff
		\bean
		\dim \left(  \bigcap_{u \in U} <S_{(n,u)}> \right) & = & 0.
		\eean
		But this would imply that 
		\bean
		\dim \left(  \bigcap_{u \in F} <S_{(n,u)}> \right) & = & 0.
		\eean
		for any subset $F\subseteq [n-1]$ of nodes of size $|F|$ satisfying $(k-1) \leq |F| \leq (n-1)$.  We are therefore justified in extending the inequality in \eqref{eq:Cond_4} to the case when $U$ is replaced by a subset $F$ whose size now ranges from $2$ to $(n-1)$, i.e., we are justified in writing: 
		\bea \label{eq:size_bd} 
		\dim \left ( \bigcap_{u \in F} <S_{(n,u)}> \right ) &  \leq  & \frac{\alpha}{r^{|F|}} 
		\eea	   
		for any $F \subseteq [n-1]$, of size $2 \leq |F| \leq (n-1)$.  A consequence of the inequality \eqref{eq:size_bd} is that 
		\bean
		\dim \left ( \bigcap_{u \in F} <S_{(n,u)}> \right ) & \geq 1 ,
		\eean
		implies that $|F| \leq \lfloor \log_{r} (\alpha) \rfloor $.    In other words, a given non-zero vector can belong to at most $\lfloor \log_{r} (\alpha) \rfloor$ repair subspaces among the repair subspaces: $<S_{(n,1)}>,\hdots,<S_{(n,n-1)}>$.
		
		\item {\em Counting in a Bipartite Graph} The remainder of the proof then follows the steps outlined in Tamo et. al \cite{TamWanBru_access_tit}.   We form a bipartite graph with $\{e_1,...,e_{\alpha}\}$ (standard basis) as left nodes and $S_{(n,1)},...,S_{(n,n-1)}$ as right nodes as shown in Fig.~\ref{fig:biparite}. We declare that edge $(e_i,S_{(n,j)})$ belongs to the edge set of this bipartite graph iff $(e_i \in <S_{(n,j)}>)$. Now since the MSR code is an optimal access code, the rows of each repair matrix $S_{(n,j)}$ must all be drawn from the set $\{e_1,...,e_{\alpha}\}$. 
		
		\begin{figure}
			\begin{center}
				\includegraphics[width=2in]{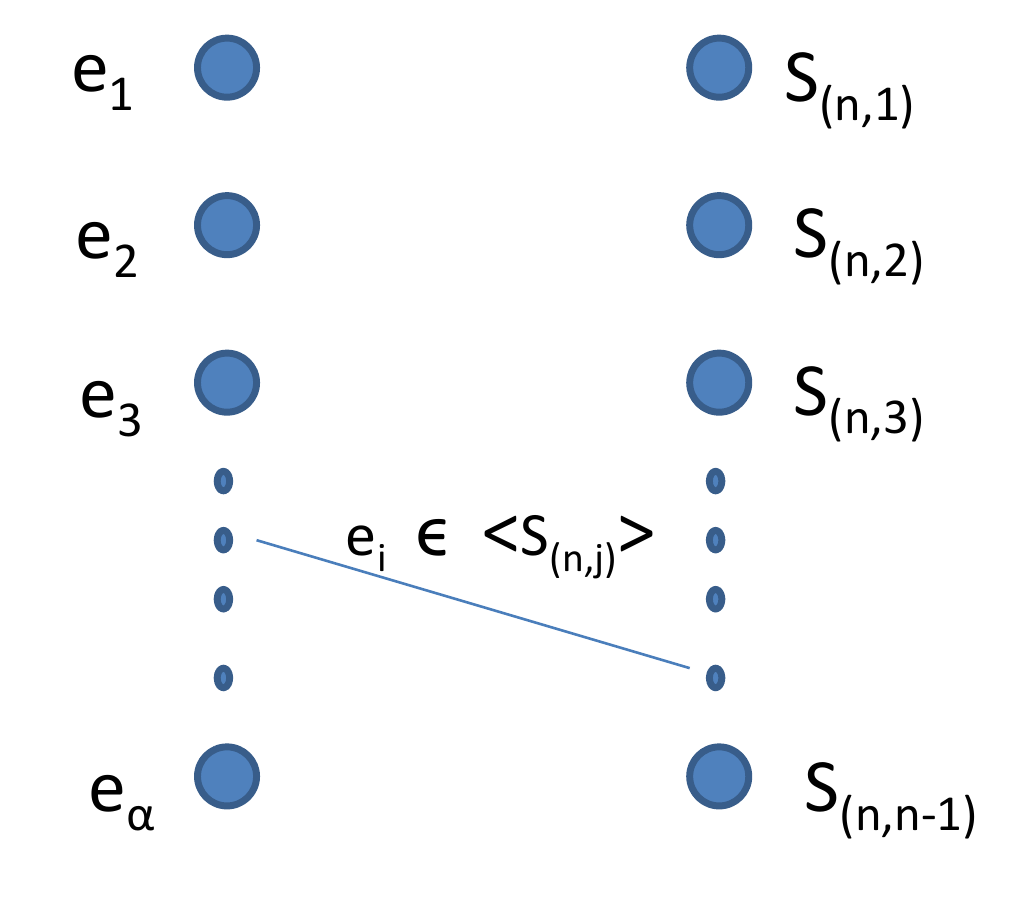}
			\end{center} 
			\caption{The above figure shows the bipartite graph appearing in the counting argument used to provie Theorem~\ref{thm:main}.  Each node on the left corresponds to an element of the standard basis $\{e_1,...,e_{\alpha}\}$.  The nodes to the right are associated to the repair matrices $S_{(n,1)},...,S_{(n,n-1)}$.  \label{fig:biparite} }
		\end{figure}
		
		Counting the number of edges of this bipartite graph in terms of node degrees on the left and the right, we obtain:
		\bean
		\alpha \lfloor \log_r(\alpha) \rfloor \geq (n-1)\frac{\alpha}{r},  \\
		\log_r(\alpha) \geq \lfloor \log_r(\alpha) \rfloor \geq \lceil \frac{(n-1)}{r} \rceil,  \\
		\log_r(\alpha) \geq \lceil \frac{(n-1)}{r} \rceil,   \\
		\alpha \geq r^{\lceil \frac{n-1}{r} \rceil} 	.	  
		\eean
		Thus we have shown that if $\alpha < r^{k-1}$, we must have $\alpha \geq r^{\lceil \frac{n-1}{r} \rceil}$.  It follows that 
		\bean
		\alpha \geq \min \{r^{\lceil \frac{n-1}{r} \rceil},r^{k-1}\}.
		\eean		   
	\end{enumerate}    
\end{proof}	
In the following, we will derive lower bounds on sub-packetization of optimal access MSR code for $d=n-1$ and for arbitrary $d$ under specific assumptions. These bounds are derived based on the proof of Theorem~\ref{thm:main} and hence we state them as corollaries.

\subsection{Sub-packetization Bound for $d=n-1$ and Constant Repair Subspaces}	\label{sec:dnMSR_ch6}
In the following, we will derive a lower bound on sub-packetization of an optimal access MSR code for $d=n-1$ under the assumption that the repair matrix $S_{(i,j)}$ is independent of $i$.

\begin{cor} \label{sb_5} 
	Given a linear optimal access $\{(n,k,d),(\alpha,\beta)\}$ MSR code $\mathcal{C}$ with  $d =n-1$ and linear repair for all $n$ nodes with $S_{(i,j)}=S_j$,$\forall i,j \in [n],i \neq j$ (thus the repair matrix $S_{(i,j)}$ is independent of $i$), we must have:
	\bean
	\alpha \geq  \min \{ r^{\lceil \frac{n}{r} \rceil},r^{k-1}\}.
	\eean
\end{cor}

\begin{proof}
	For a given subset $U \subseteq [n]$ of nodes, of size $2 \leq |U| \leq k-1$, let $P \subseteq [n]$ be a second subset disjoint from $U$ (i.e., $ U \bigcap P = \emptyset$), of size $|P|=r$.    In this setting, the proof of Lemma \ref{lem:rsi} will go through for the pair of subsets $U,P$ and we will obtain that for any $p \in P$:
	\bean
	\dim \left( \bigcap_{u \in U} <S_{u}> \right) = \dim \left( \bigcap_{u \in U} <S_{(p,u)}> \right) \leq \frac{\alpha}{r^{|U|}}
	\eean
	As before, we next extend the validity of the above inequality for any $U \subseteq [n]$ by assuming that $\alpha < r^{k-1}$ and following the same steps as in the proof of Theorem \ref{thm:main}.    Following this, we repeat the bipartite-graph construction and subsequent counting argument as in the proof of Theorem \ref{thm:main} with one important difference.   In the bipartitie graph constructed here, there are $n$ nodes on the right (as opposed to $(n-1)$), with $\{S_{1},\hdots,S_{n}\}$ as nodes in the right. The result then follows.
\end{proof}

\subsection{Sub-packetization Bound for Arbitrary $d$ and Repair Subspaces that are Independent of the Choice of Helper Nodes}	\label{sec:arbdMSR_ch6}
In the following, we will derive a lower bound on sub-packetization of an optimal access MSR code for any $d<n-1$ under the assumption that the repair matrix $S^D_{(i,j)}$ is independent of the choice of the remaining $(d-1)$ helper nodes in $D$.
\begin{cor} \label{sb_6} 
	Let  $\mathcal{C}$ be a linear optimal-access $\{(n,k,d),(\alpha,\beta)\}$ MSR code for some $d$, $k \leq d \leq n-1$, and linear repair for all $n$ nodes.  We assume in addition, that every node $j \in [n]$ can be repaired by contacting a subset $D \subseteq [n]-\{j\}$, $|D|=d$ of helper nodes in such a way that the repair matrix $S^D_{(i,j)}$ is independent of the choice of the remaining $(d-1)$ helper nodes, i.e., $S^D_{(i,j)}=S_{(i,j)}$,$\forall j \in [n]$ and $i \in D$, $\forall D \subseteq [n]-\{j\}, |D|=d$.  Then we must have:
	\bean
	\alpha \geq  \min \{ (d-k+1)^{\lceil \frac{n-1}{d-k+1} \rceil},(d-k+1)^{k-1} \}.
	\eean
	
\end{cor}

\begin{proof}
	Given a set $U, U \subseteq [n-1]$, of nodes of size $|U|$, $2 \leq |U| \leq k-1$, let us form a set $P$ of size $|P|=d-k+1$, such that $P \subseteq [n]$,  with $n \in P$, $ U \bigcap P = \emptyset$.   Let $V$ be any subset of $[n]\setminus \{U \cup P\}$ such that $|U \cup V \cup P|=d+1$.   Next, consider  the punctured code obtained by restricting attention to the node subset $\{U \cup V \cup P\}$.  The proof of Lemma \ref{lem:rsi} applied to the subset $\{U \cup V \cup P\}$ of nodes will then go through and we will obtain:
	\bean
	\dim \left( \bigcap_{u \in U} <S_{(n,u)}> \right) \leq \frac{\alpha}{(d-k+1)^{|U|}}.
	\eean
	We then repeat the process of extending the above inequality for any $U \subseteq [n-1]$ by assuming $\alpha < (d-k+1)^{k-1}$ and following the proof of Theorem \ref{thm:main}. Following this, we construct the bipartite graph as always and repeat the counting argument employed in the proof of Theorem \ref{thm:main} with one difference.   On the right side of the bipartite graph, we now have the $(n-1)$ repair matrices $S_{(n,1)},\hdots,S_{(n,n-1)}$ as the right nodes of the bipartite graph.   This will give us the desired result.   We omit the details.
\end{proof}

\section{Vector MDS Codes with Optimal Access Repair of $w$ Nodes:} \label{sec:MDScodes_ch6}
In this section, we define optimal access MDS codes with repair of a subset of $w$ nodes with each node repaired in help-by-transfer method with optimal repair bandwidth. We then derive two lower bounds on sub-packetization of such codes. One lower bound for $d=n-1$ and another for $d < n-1$. These two lower bounds on sub-packetization are derived again based on the proof of Theorem \ref{thm:main} and hence stated as corollaries. 

\subsection{Optimal-Access MDS Codes} In this section, we derive results which correspond to an $[n,k]$ MDS code over the vector alphabet $\mathbb{F}_q^{\alpha}$, having the property that it can repair the failure of any node in a subset of $w (\leq n-1)$ nodes where repair of each particular node on failure, can be carried out using linear operations, by uniformly downloading $\beta$ symbols from a collection of $d$ helper nodes. Linear repair of any node among the $w$ nodes is defined as in MSR code using repair matrices. It can be shown that even here, the minimum amount $d\beta$ of data download needed for repair of a single failed node, is given by $d\beta=d\frac{\alpha}{(d-k+1)}$ (minimum repair bandwidth). Our objective here as well, is on lower bounds on the sub-packetization level $\alpha$ of an MDS code that can carry out repair of any node in a subset of $w$ nodes, $1 \leq w \leq (n-1)$ where each node is repaired (linear repair) by help-by-transfer with minimum repair bandwidth (Such codes will be referred to as optimal access MDS codes). We prove a lower bound on $\alpha$ of such codes for the case of $d=(n-1)$ in Corollary \ref{cor:node_subset}.  This bound holds for any $w (\leq n-1)$ and is already shown to be tight in the beginning, by comparing with a recent code constructions \cite{LiTangTian}. Also provided, are bounds for the case $d<(n-1)$. The $w=n$ case correspond to the MSR code case which is described before. 
\subsection{Bounds on Sub-Packetization Level of a Vector MDS Code  with Optimal Access Repair of $w$ Nodes}
In this section, we will derive lower bounds on sub-packetization of vector MDS codes which can repair any node in a subset of $w$ nodes by help-transfer method with optimal repair bandwidth both for $d=n-1$ case and any $d$ case.
\subsubsection{Sub-packetization Bound for Optimal Access MDS codes with $d=n-1$} \label{sec:dnMDS_ch6}

In this section, we will derive a lower bound on sub-packetization of optimal accees MDS codes for $d=n-1$.

\begin{cor} \label{cor:node_subset}  
	Let $\mathcal{C}$ be a linear $[n,k]$ MDS code over the vector alphabet $\fq^{\alpha}$ containing a distinguished set $W$ of $|W|=w \leq (n-1)$ nodes.  Each node in $W$ can be repaired, through linear repair, by accessing and downloading, precisely $\beta=\frac{\alpha}{d-k+1}$ symbols over \fq \  from each of the remaining $d=(n-1)$ nodes.
	In other words, the repair of each node in $W$ can be carried out through help-by-transfer with minimum repair bandwidth with only linear operations.   Then we must have 
	\bean
	\alpha \geq  \left\{ \begin{array}{rl} \min \{ r^{\lceil \frac{w}{r} \rceil},r^{k-1} \}, \  & w > (k-1) \\
		r^{\lceil \frac{w}{r} \rceil}, \ & w \leq (k-1) . \end{array} \right. 
	\eean
\end{cor}

\begin{proof}  We remark that even in this setting, it is known that $d\beta=d\frac{\alpha}{r}$ is the minimum repair bandwidth needed to repair the nodes in $W$, hence the nodes in $W$ are those for which the repair is optimal.  To prove the corollary, consider a subset $U \subseteq W$ and $|U| \leq k-1$ and repeat the steps used to prove Lemma \ref{lem:rsi} with this set $U$ and by choosing any $P$ disjoint from $U$ with (wolog we assume $n \notin W$) $n \in P$ and $|P|=r$. Since in the proof of Lemma \ref{lem:rsi}, we only consider equations regarding repair of nodes in $U$, the proof of Lemma \ref{lem:rsi} will go through. We will arrive at the following analogue of \eqref{eq:Cond_4}: 
	\bea
	\dim \left( \bigcap_{u \in U} <S_{(n,u)}> \right) \leq \frac{\alpha}{r^{|U|}}. \label{eq:Degree_ch6}
	\eea
	For the case when $|W| > (k-1)$, we again extend the range of validity of this inequality to the case when $U$ is any subset of $W$, by first assuming that $\alpha < r^{k-1}$ and proceeding as in the proof of Theorem~\ref{thm:main} above.  For the case when $|W| \leq(k-1)$, no such extension is needed.  We then repeat the bipartite-graph-counting argument used in the proof of Theorem \ref{thm:main}, with the difference that the number of nodes on the right equals $w$ with $\{ S_{(n,j)}, j \in W\}$ as nodes in the right.   This will then give us the desired result.
\end{proof}

\subsubsection{Sub-packetization Bound for Optimal Access MDS codes for an Arbitrary Number $d$ of Helper Nodes} \label{sec:arbdMDS_ch6}
In this section, we will derive a lower bound on sub-packetization of optimal accees MDS codes for any $d<n-1$.
\begin{cor} \label{sb_4} 
	Let $\mathcal{C}$ be a linear $[n,k]$ MDS code over the vector alphabet $\fq^{\alpha}$ containing a distinguished set $W$ of $|W|=w \leq d$ nodes.  Each node $j$ in $W$ can be repaired, through linear repair, by accessing and downloading, precisely $\beta=\frac{\alpha}{d-k+1}$ symbols over \fq \  from each of the $d$ helper nodes where the $d$ helper nodes are any set of $d$ nodes apart from the failed node $j$.
	In other words, the repair of each node in $W$ can be carried through help by transfer with minimum repair bandwidth with only linear operations.   Then we must have 
	\bean
	\alpha \geq  \left\{ \begin{array}{rl} \min \{ (d-k+1)^{\lceil \frac{w}{d-k+1} \rceil},(d-k+1)^{k-1}\}, \  & w > (k-1) \\
		(d-k+1)^{\lceil \frac{w}{d-k+1}\rceil}, \ & w \leq (k-1) . \end{array} \right. 
	\eean

\end{cor}

\begin{proof}
	To prove this, we simply restrict our attention to the (punctured) code obtained by selecting a subset of nodes of size $n'=(d+1)$ that includes the subset $W$ for which optimal repair is possible.   Applying the results of Corollary~\ref{cor:node_subset} then gives us the desired result.
\end{proof}

An optimal access MDS code is said to have optimal sub-packetization if the code has $\alpha$ equal to the lower bound we derived.

\section{Structure of a Vector MDS Code with Optimal Access Repair of $w$ Nodes with Optimal Sub-packetization} \label{sec:struct_ch6}
In this section, we deduce the structure of optimal access MDS code with optimal sub-packetization with repair matrices of the form $S^D_{i,j}=S_j$ and show that the structure we deduced is also present in existing constructions.
\subsection{Deducing the Structure of Optimal access MDS Code with Optimal Sub-packetization with Repair matrices of the form $S^D_{i,j}=S_j$}

\begin{thm} \label{thm:struct} \textbf{(Structure of Optimal access MDS Code with Optimal Sub-packetization with Repair Matrices of the form $S^D_{i,j}=S_j$)}: 
	Let $\mathcal{C}$ be a linear $[n,k]$ MDS code over the vector alphabet $\fq^{\alpha}$. Let the nodes $U=\{u_1,...,u_k\} \subseteq [n]$ be systematic and $\{p_1,...,p_r\} \subseteq [n]$ be parity nodes with $r=n-k$. Let $W=\{u_1,...,u_{w}\}$ be a set of $w$ nodes such that any node $j$ in $W$ can be repaired, through linear repair, by accessing and downloading, precisely $\beta=\frac{\alpha}{d-k+1}$ symbols over \fq \  from each of the set of $d$ helper nodes where the $d$ helper nodes are any set of $d$ nodes apart from the failed node $j$. In other words, the repair of each node in $W$ can be carried out through help-by-transfer with minimum repair bandwidth with only linear operations.  We assume in addition, that every node $j \in W$ can be repaired by contacting a subset $D \subseteq [n]-\{j\}$, $|D|=d$ of helper nodes in such a way that the repair matrix $S^D_{(i,j)}$ is dependent only on the node $j$, i.e., $S^D_{(i,j)}=S_{j}$,$\forall i \in D$, $\forall D \subseteq [n]-\{j\}, |D|=d$. Let $d \geq k+1$ and $k \geq 3$.
	Let the generator matrix be of the form,
	\bea
	G = \left[
	\scalemath{1}{
		\begin{array}{cccc}
			I_{\alpha} & 0 & \hdots & 0 \\
			0 & I_{\alpha} & \hdots & 0 \\
			\vdots & \vdots & \vdots & \vdots \\
			0 & 0 & \hdots & I_{\alpha} \\
			A_{p_1,u_1} & A_{p_1,u_2} & \hdots & A_{p_1,u_k} \\
			A_{p_2,u_1} & A_{p_2,u_2} & \hdots & A_{p_2,u_k} \\
			\vdots & \vdots & \vdots & \vdots \\
			A_{p_r,u_1} & A_{p_r,u_2} & \hdots & A_{p_r,u_k} \\
		\end{array} 
	}
	\right]. \ \ \label{Gform_struct}
	\eea
	
	\ben
	\item Case 1: $w=k$ (repair with optimal bandwidth is possible for all systematic nodes) and $(d-k+1)$ divides $k$:\\
	By Corollary \ref{sb_4}, $\alpha \geq (d-k+1)^{\frac{k}{d-k+1}}$. We assume optimal subpacketization i.e., $\alpha = (d-k+1)^{\frac{k}{d-k+1}}$ (achieving the lower bound).
	Under the above conditions, wlog we must have: \\
	For $1 \leq j \leq w$,$1 \leq i \leq r$ and assuming $<S_{u_j}> = <e_1,...,e_{\beta}>$:
	\bea
	A_{p_i,u_j}=\left[
	\scalemath{1}{
		\begin{array}{c}
			v_{i,j,1} \\
			\vdots \\
			v_{i,j,\beta} \\
			M_{i,j} 
		\end{array}
	}
	\right]. \ \ \label{Gform_struct11}
	\eea
	where 
	\ben
	\item If $<S_{u_j}> = <e_{j_1},...,e_{j_\beta}>$ then the row vectors $\{v_{i,j,1},...,v_{i,j,\beta}\}$ will be in rows $\{j_1,...,j_{\beta}\}$ of $A_{p_i,u_j}$ respectively and the rest of $\alpha-\beta$ rows of $A_{p_i,u_j}$ will be termed as $M_{i,j}$,
	\item Disjointness of support of vectors with uniform cardinality of support: $v_{i,j,t}$ is a $1 \times \alpha$ vector such that $|\text{Support}(v_{i,j,t})|=d-k+1$,$\forall 1 \leq t \leq \beta$ and 
	$\text{Support}(v_{i,j,t_1}) \cap \text{Support}(v_{i,j,t_2})=\emptyset$, $\forall 1\leq t_1 \neq t_2 \leq \beta$,
	\item 	Same Support independent of $i$: any given $j,t$, $\text{Support}(v_{i_1,j,t}) = \text{Support}(v_{i_2,j,t})$, $\forall 1 \leq i_1 \neq i_2 \leq r$,
	\item For any given $j,t$, $\{ v^T_{i,j,t} : 1 \leq i \leq r\}$ are a set of $r$ vectors such that they form the columns of generator matrix of an $[r,d-k+1]$ MDS code over $\mathbb{F}_q$,
	\item $M_{i,j}$ is a $(\alpha-\beta) \times \alpha$ matrix such that each distinct row is a distinct standard basis vector from the set $\{e_1,...,e_{\alpha}\}-\text{rows of } S_{u_j}$ after getting scaled by some element from $\fq$.
	\een
	
	\item Case 2: $w=d-k+1$ and $w \leq k$:\\
	By Corollary \ref{sb_4}, $\alpha \geq w$. We assume optimal subpacketization i.e., $\alpha=w$ (achieving the lower bound). Hence $\beta=1$.
	Under the above conditions, wlog we must have: \\
	For $1 \leq j \leq w$,$1 \leq i \leq r$:
	\bea
	A_{p_i,u_j}=\left[
	\scalemath{1}{
		\begin{array}{c}
			P_{i,j,1} \\
			v_{i,j} \\
			P_{i,j,2} 
		\end{array} 
	}
	\right]. \ \ \label{Gform_struct1}
	\eea
	where $P_{i,j,1}$ is a $j-1 \times w$ matrix and $P_{i,j,2}$ is a $w-j \times w$ matrix such that 
	\bea
	P_{i,j,1} &=&[D_{(i,j,1)} \ \ \  0_{j-1 \times w-j+1}],\\
	P_{i,j,2} &=& [0_{w-j \times j} \ \ \  D_{(i,j,2)}],
	\eea 
	where $D_{(i,j,1)}$ is a $j-1 \times j-1$ diagonal matrix and  $D_{(i,j,2)}$ is a $w-j \times w-j$ diagonal matrix 
	and $\{ v^T_{i,j} : 1\leq i \leq r \}$ are a set of $r$ vectors forming the  columns of the generator matrix of an $[r,w]$ MDS code over $\fq$.\\
	For $w+1 \leq j \leq k$,$1 \leq i \leq r$:
	\bea
	A_{p_i,u_j}=
	\scalemath{1}{
		\begin{array}{c}
			P'_{i,j-w}
		\end{array} 
	},
	\ \ \label{Gform_struct2}
	\eea
	where $P'_{i,j-w}$ is a $w \times w$ diagonal matrix such that for a fixed $1 \leq \ell \leq w$, $[P'^{\ell,\ell}_{i,j} : 1\leq i \leq r, 1 \leq j \leq k-w]$ is a $r \times (k-w)$ matrix such that every square submatrix is invertible, where $P'^{\ell,\ell}_{i,j}$ is the $(\ell,\ell)^{th}$ entry of the matrix $P'_{i,j}$.
	\een
\end{thm}
\begin{proof}
	Proof is given in the Appendix  \ref{sec:proofstructthm_ch6} .
\end{proof}

\subsection{A Coupled Layer Interpretation of Theorem \ref{thm:struct}}
In this subsection we show evidence that the structure we deduced in Theorem \ref{thm:struct} is present in existing codes.
We restrict our attention to the case when $w=d-k+1$ and $w \leq k$. From Theorem \ref{thm:struct}, we get the general form of 
$\{A_{p_i,u_j}\}$. Since the generator matrix is systematic, we can directly write down the parity-check matrix by inspection.  In the following, we show that this parity-check matrix can be written such that the construction can be viewed to have a coupled-layer structure, that is present in the Ye-Barg~\cite{YeBar_2} construction (see \cite{SasVajKum_arxiv} for details). We do note however, that we are only considering the case where $w$ nodes are repaired, not all the nodes as in Ye-Barg~\cite{YeBar_2}. 
For simplicity we illustrate the connection via an example.   Let $w=3,r=4,k=4$, so that $\alpha=w=3$. Let $U=[1,k]$ and $\{p_1,...,p_r\}=[k+1,n]$.

From Theorem \ref{thm:struct}, we obtain that:
\bea
A_{1,1}=\left[
\scalemath{0.8}{
	\begin{array}{ccc}
		v_{1,1}(1) & v_{1,1}(2) & v_{1,1}(3) \\
		0 & a^2_{1,1} & 0 \\
		0 & 0 & a^3_{1,1} \\
	\end{array} 
}
\right]. \ \ \label{P9}
\eea
$a^m_{i,j}$ is a non zero scalar in $\fq$. The other matrices can similarly be written down leading to the expression below for the parity-check matrix :
\bea
H = \left[
\scalemath{0.8}{
	\begin{array}{cccc|cccc}
		A_{1,1} & A_{1,2} & A_{1,3} & A_{1,4} & I_3 & & & \\
		A_{2,1} & A_{2,2} & A_{2,3} & A_{2,4} & & I_{3} & &  \\
		A_{3,1} & A_{3,2} & A_{3,3} & A_{3,4} & & & I_3 & \\
		A_{4,1} & A_{4,2} & A_{4,3} & A_{4,4} & & & & I_3 \\
	\end{array} 
}
\right]. \ \ \label{P10}
\eea
Upon substituting for the $A_{i,j}$, we obtain:
\bean
H = \left[
\scalemath{0.6}{
	\begin{array}{ccc|ccc|ccc|ccc|c}
		v_{1,1}(1) & v_{1,1}(2) & v_{1,1}(3)& a^1_{1,2} & 0 & 0 & a^1_{1,3} & 0 & 0 & a^1_{1,4} & 0 & 0 & \\
		0 & a^2_{1,1} & 0 & v_{1,2}(1) & v_{1,2}(2) & v_{1,2}(3) & 0 &  a^2_{1,3} & 0 & 0 & a^2_{1,4} & 0 & \\
		0 & 0 & a^3_{1,1} & 0 & 0 & a^3_{1,2} & v_{1,3}(1) & v_{1,3}(2) & v_{1,3}(3) & 0 & 0 & a^3_{1,4} & \\
		\cline{1-12}
		v_{2,1}(1) & v_{2,1}(2) & v_{2,1}(3)& a^1_{2,2} & 0 & 0 & a^1_{2,3} & 0 & 0 & a^1_{2,4} & 0 & 0 & \\
		0 & a^2_{2,1} & 0 & v_{2,2}(1) & v_{2,2}(2) & v_{2,2}(3) & 0 &  a^2_{2,3} & 0 & 0 & a^2_{2,4} & 0 & \\
		0 & 0 & a^3_{2,1} & 0 & 0 & a^3_{2,2} & v_{2,3}(1) & v_{2,3}(2) & v_{2,3}(3) & 0 & 0 & a^3_{2,4} & I_{12}\\
		\cline{1-12}
		v_{3,1}(1) & v_{3,1}(2) & v_{3,1}(3)& a^1_{3,2} & 0 & 0 & a^1_{3,3} & 0 & 0 & a^1_{3,4} & 0 & 0 & \\
		0 & a^2_{3,1} & 0 & v_{3,2}(1) & v_{3,2}(2) & v_{3,2}(3) & 0 &  a^2_{3,3} & 0 & 0 & a^2_{3,4} & 0 & \\
		0 & 0 & a^3_{3,1} & 0 & 0 & a^3_{3,2} & v_{3,3}(1) & v_{3,3}(2) & v_{3,3}(3) & 0 & 0 & a^3_{3,4} & \\
		\cline{1-12}
		v_{4,1}(1) & v_{4,1}(2) & v_{4,1}(3)& a^1_{4,2} & 0 & 0 & a^1_{4,3} & 0 & 0 & a^1_{4,4} & 0 & 0 & \\
		0 & a^2_{4,1} & 0 & v_{4,2}(1) & v_{4,2}(2) & v_{4,2}(3) & 0 &  a^2_{4,3} & 0 & 0 & a^2_{4,4} & 0 & \\
		0 & 0 & a^3_{4,1} & 0 & 0 & a^3_{4,2} & v_{4,3}(1) & v_{4,3}(2) & v_{4,3}(3) & 0 & 0 & a^3_{4,4} & \\
	\end{array} 
}
\right],
\eean
which in turn can be written in the form:
\bean
H=\left[
\scalemath{0.6}{
	\begin{array}{ccccc|ccccc|ccccc}
		v_{1,1}(1) & a^1_{1,2} & a^1_{1,3} & a^1_{1,4} & & v_{1,1}(2) & & & & & v_{1,1}(3) & & & & \\
		v_{2,1}(1) & a^1_{2,2} & a^1_{2,3} & a^1_{2,4} & I_4 & v_{2,1}(2) & & & & & v_{2,1}(3) & & & &\\
		v_{3,1}(1) & a^1_{3,2} & a^1_{3,3} & a^1_{3,4} & & v_{3,1}(2) & & & & & v_{3,1}(3) & & & &\\
		v_{4,1}(1) & a^1_{4,2} & a^1_{4,3} & a^1_{4,4} & & v_{4,1}(2) & & & & & v_{4,1}(3) & & & &\\
		\hline
		& v_{1,2}(1) & & & &  a^2_{1,1} & v_{1,2}(2) & a^2_{1,3} & a^2_{1,4} &  & & v_{1,2}(3) & & & \\
		& v_{2,2}(1) & & & &  a^2_{2,1} & v_{2,2}(2) & a^2_{2,3} & a^2_{2,4} & I_4 & & v_{2,2}(3) & & & \\
		& v_{3,2}(1) & & & & a^2_{3,1} & v_{3,2}(2) & a^2_{3,3} & a^2_{3,4} & & & v_{3,2}(3) & & &\\
		& v_{4,2}(1) & & &	& a^2_{4,1} & v_{4,2}(2) & a^2_{4,3} & a^2_{4,4} & & & v_{4,2}(3) & & &  \\
		\hline
		& & v_{1,3}(1) & & & & & v_{1,3}(2) & & & a^3_{1,1} & a^3_{1,2} & v_{1,3}(3) & a^3_{1,4} & \\
		& & v_{2,3}(1)  & & & & & v_{2,3}(2) & & & a^3_{2,1} & a^3_{2,2}  & v_{2,3}(3) & a^3_{2,4} & I_4 \\
		& & v_{3,3}(1) & & & & & v_{3,3}(2) & & & a^3_{3,1} & a^3_{3,2} & v_{3,3}(3)  & a^3_{3,4} & \\
		& & v_{4,3}(1) & & & & & v_{4,3}(2) & & & a^3_{4,1} & a^3_{4,2} &  v_{4,3}(3) & a^3_{4,4} & \\
	\end{array} 
}
\right].
\eean
The coupled nature of the construction is now apparent: the columns in the matrix can be divided into three sections, each can be viewed as corresponding to a different plane.  While for the most part, each parity-check equations runs over elements of a single plane, there are exceptions and this corresponds to the coupling across planes.

\section{Construction of an MSR code for $d<(n-1)$} \label{sec:MSRconstr_ch6} 
The following is a collaborative work with a person named Myna Vajha. In this section we briefly describe a construction of optimal access MSR code with optimal sub-packetization with $d \in \{k+1,k+2,k+3\}$ with field size of $O(n)$. Here we only we briefly give an intuitive outline. For complete description of the construction please refer to \cite{VajBalKum}. Based on the structure theorem for vector MDS codes in the last section we can formulate a general form of parity check matrix of optimal access MSR code with optimal sub-packetization but with help-by-transfer for all nodes. The code symbols will be divided into several planes of code symbols with parity check equations given such that each parity check equation involves code symbols within a plane and some symbols out of plane. The main property to prove is MDS property of the code. Since the code is defined in terms of planes by recursive argument we can write a general form of matrices we want to be non-singular for the MDS property. An example submatrix to prove non-singular is given by :
\bean\scalebox{0.8}{$
	M_{i} = \left[\begin{array}{c|c|c}
	M_{i-1} & 0 & D\\ \hline
	0 & M_{i-1} & D \Lambda
	\end{array}\right]
	\begin{array}{l}
	\MyLBrace{2ex}{$z=0$} \\
	\MyLBrace{2ex}{$z=1$} \\
	\end{array}$},
\eean
for some matrices $M_j,\Lambda,D$ where $z$ is the variable indicating plane index. We prove that these matrices are non-singular by formulating the entries of $M_i$ as variables and finding the degree of a root of the determinant polynomial i.e., the polynomial $\det(M_i)$ and by repeating this we determine all the factors which involves a particular variable. For more details please refer to \cite{VajBalKum}.

\section{Summary and Contributions}
In this chapter, we derived lower bounds on sub-packetization of MSR codes and vector MDS codes with optimal access repair of a subset of nodes for any $d$. Some of our lower bounds were general and some of them had assumptions on repair matrices. We compared our lower bounds with existing constructions and showed the tightness of our lower bounds on sub-packetization. We then deduced the structure of vector MDS codes with optimal access repair of a subset of nodes for any $d$ with optimal sub-packetization assuming a certain form of repair matrices (this form of repair matrices that we assumed is present in a lot of existing constructions) and showed evidence that this structure is present in existing codes. We then briefly gave an intuitive outline of a new construction of an optimal access MSR code with optimal sub-packetization.

\begin{appendices}
	\chapter{Row Spaces}

	\begin{lem} \label{Row_Spaces}
		Let $A,B$ be nonsingular $(\alpha \times \alpha)$ matrices and $\{P_1,P_2,Q_1,Q_2\}$ be matrices of size $(m \times \alpha)$.  Then if 
		\bean
		<P_1A> = <P_2B>, & & <Q_1A> = <Q_2B> ,
		\eean
		we have 
		\bean
		\left(<P_1> \bigcap <Q_1> \right) A  & = & 
		\left(<P_2> \bigcap <Q_2> \right) B. 
		\eean
		
		\label{lem:row space} 
	\end{lem}
	
	\begin{proof} 
		For the case when $A$ is nonsingular, we have that: 
		\bean
		\left(<P_1> \bigcap <Q_1> \right) A  & = & <P_1A>  \ \bigcap \ <Q_1 A>  .
		\eean
		The result then follows from noting that: 
		\bean
		\left(<P_1> \bigcap <Q_1> \right) A & = & <P_1A> \bigcap <Q_1A> \\
		\ = \ <P_2B> \bigcap <Q_2B> & = & \left(<P_2> \bigcap <Q_2> \right) B. 
		\eean
	\end{proof}
	\chapter{Condition for Repair of a systematic Node} 
	
	\begin{lem} \label{lem:rank_cond_p_node}
		Let the linear $\{ (n,k,d=n-1), (\alpha,\beta), B, \mathbb{F}_q \}$ MSR code $\mathcal{C}$ be encoded in systematic form, where nodes $\{u_1, \cdots, u_k\}$ are the systematic nodes and nodes $\{p_1, \cdots, p_r\}$ are the parity nodes. Let the corresponding generator matrix $G$ be given by: 
		\bea
		G = \left[
		\scalemath{1}{
			\begin{array}{cccc}
				I_{\alpha} & 0 & \hdots & 0 \\
				0 & I_{\alpha} & \hdots & 0 \\
				\vdots & \vdots & \vdots & \vdots \\
				0 & 0 & \hdots & I_{\alpha} \\
				A_{p_1,u_1} & A_{p_1,u_2} & \hdots & A_{p_1,u_k} \\
				A_{p_2,u_1} & A_{p_2,u_2} & \hdots & A_{p_2,u_k} \\
				\vdots & \vdots & \vdots & \vdots \\
				A_{p_r,u_1} & A_{p_r,u_2} & \hdots & A_{p_r,u_k} \\
			\end{array} 
		}
		\right], \ \ \label{G_syst}
		\eea
		Let $\left\{ S_{(i,u_1)} \mid i \in [n]-\{u_1\}\right\}$ be the repair matrices	associated to the repair of the systematic node $u_1$. Then we must have that for $j \in [k]$, $j \neq 1$:
		\bea
		<S_{(u_j,u_1)}> = <S_{(p_1,u_1)} A_{p_1,u_j}>= \hdots = <S_{(p_r,u_1)} A_{p_r,u_j}> \label{eq:IA} 
		\eea
		and
		\bea
		rank \left( \left[ \begin{array}{c}  S_{(p_1,u_1)} A_{p_1,u_1}  \\ \vdots \\ S_{(p_r,u_1)} A_{p_r,u_1} \end{array} \right] \right) = \alpha \label{eq:FR} .
		\eea
	\end{lem}
	\begin{proof}
		Let $[m_1,...,m_k]$ be the $k$ message symbols (each message symbol $m_{j}$ being a $(1 \times \alpha)$ row vector) encoded by the MSR code i.e., the resultant codeword is given by $G[m_1,...,m_k]^T$. 
		For the repair of systematic node $u_1$, the data collected by the replacement node is given by:
		\bean 
		\left[ \begin{array}{c} S_{(u_2,u_1)} \uc{u_2}^T  \\ \vdots \\ S_{(u_k,u_1)} \uc{u_k}^T \\ S_{(p_1,u_1)} \uc{p_1}^T \\ \vdots \\ S_{(p_r,u_1)} \uc{p_r}^T \end{array} \right] 
		& = & 
		\left[ \begin{array}{c} S_{(u_2,u_1)} m^T_{2}  \\ \vdots \\ S_{(u_k,u_1)} m^T_{k} \\ S_{(p_1,u_1)} \left( \sum_{j=1}^k A_{p_1,u_j}m^T_j  \right) \\ \vdots \\ S_{(p_r,u_1)} \left( \sum_{j=1}^k A_{p_r,u_j}m^T_j\right) \end{array} \right] .
		\eean
		Let 
		\bean
		T_{u_1} & = \left[ \begin{array}{cccccc}T_{(u_2,u_1)}, & \cdots & T_{(u_k,u_1)}, & T_{(p_1,u_1)}, & \cdots & T_{(p_r,u_1)} \end{array} \right] 
		\eean
		be the $(\alpha \times (n-1)\beta)$ matrix used to derive the contents of the replacement node $u_1$, with each $T_{(i,j)}$ being an $(\alpha \times \beta)$ submatrix.  Then we must have that 
		\bean
		\left[ \begin{array}{cccccc}T_{(u_2,u_1)}, & \cdots & T_{(u_k,u_1)}, & T_{(p_1,u_1)}, & \cdots & T_{(p_r,u_1)} \end{array} \right] 	\left[ \begin{array}{c} S_{(u_2,u_1)} m^T_{2},  \\ \vdots \\ S_{(u_k,u_1)} m^T_{k} \\ S_{(p_1,u_1)} \left( \sum_{j=1}^k A_{p_1,u_j}m^T_j \right) \\ \vdots \\ S_{(p_r,u_1)} \left( \sum_{j=1}^k A_{p_r,u_j}m^T_j \right) \end{array} \right]
		& = & m^T_1 .
		\eean
		Since this must hold for all data vectors $m^T_j$, we can equate the matrices that premultiply  $m^T_j$ on the left, on both sides.  If we carry this out for $j=1$, we will obtain that: 
		\bean
		\sum_{i=1}^r T_{(p_i,u_1)}S_{(p_i,u_1)}A_{p_i,u_1} & = & I ,
		\eean
		which implies:
		\bean
		\left[ \begin{array}{ccc} T_{(p_1,u_1)} & \cdots & T_{(p_r,u_1)} \end{array} \right] \left[ \begin{array}{c}  S_{(p_1,u_1)} A_{p_1,u_1}  \\ \vdots \\ S_{(p_r,u_1)} A_{p_r,u_1} \end{array} \right]
		& = & I , 
		\eean
		which in turn, forces:
		\bea
		\text{rank} \left( \left[ \begin{array}{c}  S_{(p_1,u_1)} A_{p_1,u_1}  \\ \vdots \\ S_{(p_r,u_1)} A_{p_r,u_1} \end{array} \right] \right) = \alpha ,
		\eea
		and
		\bea
		\text{rank} \left( \left[ \begin{array}{ccc} T_{(p_1,u_1)} & \cdots & T_{(p_r,u_1)} \end{array} \right] \right) = \alpha . \label{Temp:FR}
		\eea
		The above equation proves \eqref{eq:FR}.
		For the case $j \neq 1$:
		\bean
		T_{(u_j,u_1)}S_{(u_j,u_1)} + \sum_{i=1}^r T_{(p_i,u_1)}S_{(p_i,u_1)}A_{p_i,u_j} & = & 0, 
		\eean
		and
		\bea
		\left[ \begin{array}{cccc}T_{(u_j,u_1)}, & T_{(p_1,u_1)}, & \cdots & T_{(p_r,u_1)} \end{array} \right] \left[ \begin{array}{c} S_{(u_j,u_1)} \\ S_{(p_1,u_1)} A_{p_1,u_j}  \\ \vdots \\ S_{(p_r,u_1)} A_{p_r,u_j}, \end{array} \right]
		& = & 0. \label{eq:IA_1}
		\eea
		Equations \eqref{eq:IA_1} and \eqref{Temp:FR} imply:
		\bea
		\text{rank} \left( \left[ \begin{array}{c} S_{(u_j,u_1)} \\ S_{(p_1,u_1)} A_{p_1,u_j}  \\ \vdots \\ S_{(p_r,u_1)} A_{p_r,u_j}, \end{array} \right] \right) \leq \beta. 
		\eea
		Given that $\text{rank}(S_{(p_i,u_1)} A_{p_i,u_j})=\beta$ for all $i \in [r]$ and $\text{rank}(S_{(u_j,u_1)})= \beta$, this implies
		\bea
		<S_{(u_j,u_1)}> = <S_{(p_1,u_1)} A_{p_1,u_j}>= \hdots = <S_{(p_r,u_1)} A_{p_r,u_j}> \label{IA:1}. 
		\eea
		The above equation proves \eqref{eq:IA}.
		
	\end{proof}	
	
	\chapter{Proof of Theorem \ref{thm:struct}} \label{sec:proofstructthm_ch6}
	\begin{proof}
		\ben
		\item Case $1$: $w=k$ and $(d-k+1)$ divides $k$:\\
		Let us recall the counting argument on bipartite graph implicit in the proof of Corollary \ref{sb_4} (which is based on the counting argument on bipartite graph in the proof of Corollary~\ref{cor:node_subset}). Let $e_1,...,e_{\alpha}$ be left nodes and $S_{u_1},...,S_{u_k}$ be right nodes of the bipartite graph. Form an edge between $e_i$ and $S_{u_j}$ iff $e_i \in <S_{u_j}>$. Now count the edges on both sides. The degree of $e_i$ is atmost $\log_{(d-k+1)}(\alpha)$ which is proved implicitly in the proof of Corollary \ref{sb_4} (which is clear from the extension of inequality \eqref{eq:Degree_ch6} to more than $k-1$ nodes by assuming $\alpha < (d-k+1)^{(k-1)}$ and the counting argument in the proof of Corollary \ref{cor:node_subset} by applying them to the punctured code with $n'=d+1$ nodes containing $W$.). The degree of $S_{u_j}$ is exactly $\beta$. Hence counting and equating the number of edges from $\{e_1,...,e_{\alpha}\}$ and $\{ S_{u_1},...S_{u_k} \}$, for $\alpha < (d-k+1)^{(k-1)}$:
		\bea
		\alpha \times \log_{d-k+1}(\alpha) \geq k \beta, \\
		\alpha \geq (d-k+1)^{\frac{k}{(d-k+1)}}. \label{P1}
		\eea
		By above inequality \eqref{P1}, since $\alpha = (d-k+1)^{\frac{k}{(d-k+1)}}$ and $\alpha < (d-k+1)^{(k-1)}$ is true as $k \geq 3$ and $d \geq k+1$, we have that degree of $e_i$ is exactly equal to $L=\log_{(d-k+1)}(\alpha) = \frac{k}{(d-k+1)}$. Let $N_j$ be the nodes to which node $S_{u_j}$ is connected. Let  $L_i$ be the nodes to which node $e_i$ is connected. If nodes to which $e_i$ is connected is $\{S_{i_1},...,S_{i_L}\}$ then we define $L_i=\{i_1,...,i_L\}$.
		
		Since, $\forall T \subseteq \{u_1,...,u_k\}$, $\text{dim}(\cap_{b \in T} <S_b>) \leq \frac{\alpha}{(d-k+1)^{|T|}}$ (follows from the proof of Corollary \ref{sb_4} and Corollary \ref{cor:node_subset} by extending dimension inequality \eqref{eq:Degree_ch6} for more than $k-1$ nodes by applying it to the punctured code with $n'=d+1$ nodes containing $W$ as  $\alpha < (d-k+1)^{(k-1)}$), we have that for $1 \leq t \leq \alpha $: $\text{dim}(\cap_{b \in L_t} <S_b>) \leq \frac{\alpha}{(d-k+1)^{|L_t|}} = 1$. Hence $\cap_{b \in L_t} <S_b> = \{e_t\}$. 
		
		We must have for $1 \leq j \neq b \leq k$ and $1\leq i \neq m \leq r$, by applying interference alignment conditions (i.e., by applying equation \eqref{eq:IA} in the Appendix for the punctured code with $n'=d+1$ nodes containing $W \cup \{p_i,p_m\}$ for the repair of node $u_b$):
		\bea
		<S_{u_b} A_{p_i,u_{j}}> & =& <S_{u_b} A_{p_m,u_{j}}> = <S_{u_b}>. \label{P6}
		\eea
		Equation \eqref{P6} implies that for $1 \leq t \leq \alpha$,$1 \leq j \leq k$ such that $u_j \notin L_t$ and $1 \leq i \neq m \leq r$:
		\bea
		\cap_{b \in L_t} <S_{b} A_{p_i,u_{j}}> & = & \cap_{b \in L_t} <S_{b} A_{p_m,u_{j}}> = \cap_{b \in L_t} <S_{b}>, \notag \\
		(\cap_{b \in L_t} <S_{b}>) A_{p_i,u_{j}} & =& (\cap_{b \in L_t} <S_{b}>) A_{p_m,u_{j}} = \cap_{b \in L_t} <S_{b}>, \label{P3} \\ 
		<e_t> A_{p_i,u_{j}} &=& <e_t> A_{p_m,u_{j}} = <e_t>. \label{P2}
		\eea
		Let us fix $j$ such that  $1\leq j \leq k$. Throughout the rest of the proof, we use this $j$ and proof is applicable for any $j$ such that $1 \leq j \leq k$. For any given $1 \leq i \neq m \leq r$, equation \eqref{P2} is true for all $e_t \notin N_j$.   Since $|N_j| = \beta$ (wlog let $N_j=\{e_1,...e_{\beta}\}$), this proves the $M_{i,j}$ part of the matrix $A_{p_i,u_j}$ in equation \eqref{Gform_struct11}.\\
		For $1\leq t \leq \beta$ :
		\bea
		\text{dim}(\cap_{b \in L_t-\{u_j\}} <S_b>) \leq \frac{\alpha}{(d-k+1)^{|L_t|-1}} = d-k+1. \label{P7}
		\eea
		Hence for $1\leq t \leq \beta$, $1 \leq i \neq m \leq r$, from equation $\eqref{P3}$ with $L_t$ replaced by $L_t-\{u_j\}$ and from equation \eqref{P7}, we have :
		\bean
		(\cap_{b \in L_t-\{u_j\}} <S_{b}>) A_{p_i,u_{j}} = (\cap_{b \in L_t-\{u_j\}} <S_{b}>) A_{p_m,u_{j}} \\
		= \cap_{b \in L_t-\{u_j\}} <S_{b}> \subseteq <e_{t},e_{s_{1,j,t}}...,e_{s_{d-k,j,t}}>,
		\eean
		for some $\{s_{1,j,t},...s_{d-k,j,t}\} \subseteq \{1,...,\alpha\}-\{t\}$.
		Hence for $1\leq t \leq \beta$, $v_{i,j,t} = e_t A_{p_i,u_{j}} \in <e_{t},e_{s_{1,j,t}}...,e_{s_{d-k,j,t}}>$ independent of $i$. Hence this proves that $|\text{Support}(v_{i,j,t})|\leq d-k+1$,$\forall 1 \leq t \leq \beta$ independent of $i$. 
		
		Let $D'$ be a set such that $D' \subseteq \{p_1,...,p_r\}$ and $|D'| = d-k+1$. By applying full rank condition of repair (i.e., by applying equation \eqref{eq:FR} in the Appendix for the punctured code with $n'=d+1$ nodes containing $W \cup D'$ for the repair of node $u_j$):
		\bea
		\bigoplus_{p_i \in D'} < S_{u_j} A_{p_i,u_j} >  = \bigoplus_{\{i: p_i \in D'\}} <\{v_{i,j,t}: e_t \in N_j\}> & = & \mathbb{F}_q^{\alpha}. \label{P4}
		\eea 
		Equation \eqref{P4} implies:
		\bean
		\cup_{\{i:p_i \in D',1\leq t \leq \beta\}} \text{Support}(v_{i,j,t}) = \{1,...,\alpha\}.
		\eean
		but 
		\bean
		&\cup_{\{i:p_i \in D',1\leq t \leq \beta\}} \text{Support}(v_{i,j,t})  \subseteq  \cup_{1 \leq t \leq \beta} \{t,s_{1,j,t},...,s_{d-k,j,t}\},\\
		&\beta (d-k+1) = \alpha,  \\
		&\scalemath{0.9}{ \alpha = |\cup_{\{i:p_i \in D',1\leq t \leq \beta\}} \text{Support}(v_{i,j,t})| \leq \sum_{1 \leq t \leq \beta} |\{t,s_{1,j,t},...,s_{d-k,j,t}\}| \leq (d-k+1)\beta=\alpha }.
		\eean
		Hence for $1 \leq t_1 \neq t_2 \leq \beta$, $\{t_1,s_{1,j,t_1},...,s_{d-k,j,t_1}\} \cap \{t_2,s_{1,j,t_2},...,s_{d-k,j,t_2}\} = \emptyset $ and independent of $i$, $\text{Support}(v_{i,j,t})=\{t,s_{1,j,t},...,s_{d-k,j,t}\}$. This proves that (Since $D'$ is arbitrary) $|\text{Support}(v_{i,j,t})|=d-k+1$,$\forall 1 \leq t \leq \beta$ and $\text{Support}(v_{i,j,t_1}) \cap \text{Support}(v_{i,j,t_2})=\emptyset$, $\forall 1\leq t_1 \neq t_2 \leq \beta$. We have also proved that for fixed $t$, support of $v_{i,j,t}$ is same independent of $i$. By equation \eqref{P4}:
		\bea
		\bigoplus_{\{ i: p_i \in D' \}} <\{v_{i,j,t}: e_t \in N_j\}> & = & \mathbb{F}_q^{\alpha}, \notag \\
		\text{Hence for a fixed $t$:} & & \notag \\
		\text{dim} \left(\bigoplus_{\{ i: p_i \in D' \}} <\{v_{i,j,t}\}> \right) &=& |D'|=d-k+1. \label{P5}
		\eea
		For a fixed $t$, from equation \eqref{P5}, since $D'$ is an arbitrary subset of $\{p_1,...,p_r\}$ of cardinality $d-k+1$, any subset of $\{ v^T_{i,j,t} : 1 \leq i \leq r\}$ of cardinality $d-k+1$ must be of rank $d-k+1$ and we have already seen that for a fixed $t$, $\text{Support}(v_{i,j,t})$ is same for all $1 \leq i \leq r$ and is of cardinality $d-k+1$. Hence $\{ v^T_{i,j,t} : 1 \leq i \leq r\}$ form a set of $r$ vectors with support of all $r$ vectors equal and of cardinality $d-k+1$ such that they form the columns of generator matrix of an $[r,d-k+1]$ MDS code over $\mathbb{F}_q$. This completes the proof of case $1$.
		\item Case $2$: $w=d-k+1$ and $w \leq k$:\\
		We first note that wlog for $1 \leq i \leq w$, $S_{u_i}=e_i$. This is because as $\text{dim}(<S_i> \cap <S_j>) \leq \frac{\alpha}{w^2}<1$ and hence $\text{dim}(<S_i> \cap <S_j>)=0$ (follows from the proof of Corollary \ref{sb_4} and Corollary \ref{cor:node_subset} by applying the dimension inequality \eqref{eq:Degree_ch6} to the punctured code with $n'=d+1$ nodes containing $W$).\\
		We next look at interference alignment conditions. We must have for $1 \leq j \leq k$, $1 \leq b \leq w$, $j \neq b$ and $1 \leq i \neq m \leq r$ (by applying equation \eqref{eq:IA} in the Appendix for the punctured code with $n'=d+1$ nodes containing $W \cup \{p_i,p_m\}$ for the repair of node $u_b$):
		\bea
		<S_{u_b} A_{p_i,u_{j}}> = <S_{u_b} A_{p_m,u_{j}}> = <S_{u_b}> = <e_{b}>.
		\eea
		The above interference alignment proves the structure of $P_{i,j,1}$ and $P_{i,j,2}$ for $1 \leq j \leq w$. It also proves that $P'_{i,j-w}$ is a diagonal matrix for $w+1 \leq j \leq k$.
		
		Let $D'$ be a set such that $D' \subseteq \{p_1,...,p_r\}$ and $|D'| = w$. By applying full rank conditions (i.e., applying equation \eqref{eq:FR} in the Appendix for the punctured code with $n'=d+1$ nodes containing $W \cup D'$ for the repair of node $u_j$): for $1 \leq j \leq w$:
		\bea
		\bigoplus_{p_i \in D'} < S_{u_j} A_{p_i,u_j} >  = \bigoplus_{\{i: p_i \in D'\}} <v_{i,j}> & = & \mathbb{F}_q^{\alpha}.
		\eea 
		Hence any subset of $w$ vectors from $\{v_{i,j} : 1 \leq i \leq r\}$ must form a basis of $\fq^{\alpha}$ which implies $\{v^T_{i,j} : 1 \leq i \leq r\}$ form the columns of generator matrix of an $[r,w]$ MDS code.\\
		We finally look at MDS property of code $\mathcal{C}$ to conclude the proof.
		For MDS property of $\mathcal{C}$, we must have:\\
		For $1 \leq m \leq \text{min}(r,k)$ any matrices of the following form must be full rank.
		\bea
		\left[
		\scalemath{1}{
			\begin{array}{cccc}
				A_{p_{i_1},u_{j_1}} & A_{p_{i_1},u_{j_2}} & \hdots & A_{p_{i_1},u_{j_m}} \\
				A_{p_{i_2},u_{j_1}} & A_{p_{i_2},u_{j_2}} & \hdots & A_{p_{i_2},u_{j_m}} \\
				\vdots & \vdots & \vdots & \vdots \\
				A_{p_{i_m},u_{j_1}} & A_{p_{i_m},u_{j_2}} & \hdots & A_{p_{i_m},u_{j_m}}
			\end{array} 
		}
		\right]. \ \ \label{Gform_structT}
		\eea
		Applying the above condition for $\{j_1,...,j_m\} \subseteq [w+1,k]$ and observing that determinant of above matrix is the product of determinant of the matrices $[P'^{\ell,\ell}_{i_{a},j_b-w} : 1\leq a \leq m, 1 \leq b \leq m]$ for $1\leq \ell \leq w$, (as $P'_{i,j}$ is a diagonal matrix) we have the condition that any $m \times m$ sub matrix of the matrix $[P'^{\ell,\ell}_{i,j} : 1\leq i \leq r, 1 \leq j \leq k-w]$ must be invertible.
		\een
	\end{proof}
\end{appendices}

\chapter{Partial Maximal and Maximal Recoverable Codes} \label{ch:MR}
In this chapter, we discuss a sub-class of Locally Recoverable (LR) codes (defined in Chapter \ref{ch:LRSingleErasure}) called Maximal Recoverable (MR) codes. MR codes of block length $n$ and dimension $k$ are a class of sub-codes of a given code $\mathcal{C}_{sup}$ with block length $n$ and dimension $>k$ such that an MR code can correct all possible erasure patterns which is not precluded by the code $\mathcal{C}_{sup}$. We will refer to the parity check matrix of $\mathcal{C}_{sup}$ as $H_0$ in this chapter. If  $\mathcal{C}_{sup} = \fq^{n}$ with $q \geq n$, then MR codes coincide with MDS codes. When $\mathcal{C}_{sup}$ is an LR code then MR codes are a subclass of LR codes. In this case, MR codes are LR counterpart of MDS codes. Hence as with MDS codes, it is an important problem to construct MR codes with low field size. In this chapter, we will discuss MR codes with low field size. We also consider a relaxation of MR codes called Partial Maximal Recoverable (PMR) codes.  Below is a summary of our results in this chapter.

\paragraph{Organization of the Chapter} The chapter is organized as follows.  We begin by formally defining MR codes in Section \ref{sec:MRdef_ch7} and PMR codes in Section~\ref{sec:PMR_ch7}. Section \ref{sec:PMR_ch7} and Section \ref{sec:lower_rate_PMR_ch7} presents results on PMR codes. Section \ref{sec:MRConst_ch7} presents three constructions of MR codes with low field size. Finally Section \ref{sec:Sum_ch7}, presents a summary of this chapter.

\paragraph{Contributions}  Contributions of the thesis on the topic of PMR and MR codes include:
\ben
\item Introduction of a relaxation of MR codes called PMR codes. A general form of parity check matrix of PMR codes is given in Section \ref{sec:characerizing_PMR_ch7}. This general structure is also applicable to MR codes. In Section \ref{sec:simple_PMR_construction_ch7}, we present a simple, high-rate, low field size construction of PMR codes. Also provided is an approach for a general construction of PMR codes (Section \ref{sec:lower_rate_PMR_ch7}) for parameters with slightly lower rate.
\item Also contained in the chapter are three constructions of MR codes with improved parameters, primarily field size. 
\item The first construction of MR codes is for the special case of $r=2$ (Section \ref{sec:coset_construction_ch7}). The code is obtained by puncturing codes constructed in \cite{TamBar_LRC}. The code constructed has field size better than the existing constructions in the rate regime $0 < \text{rate} \leq 0.233$ although the construction gives codes for all rates in the rate regime $0 < \text{rate} < 0.667$.
\item We next give two explicit constructions of MR codes (Section \ref{sec:r12MR_ch7} and Section \ref{sec:rdelta2MR_ch7}) with block length $n$ with field size of $O(n)$ for the case when dimension $k=\frac{nr}{r+1}-2$.
\een

\section{Maximal Recoverable Codes} \label{sec:MRdef_ch7}

An $[n,k]$ MDS code can recover from any pattern of $(n-k)$ erasures.  Maximal Recoverable (MR) codes are codes that operate under some pre-specified linearity constraints and which can recover from any pattern of $(n-k)$ erasures that is not precluded by the pre-specified linearity constraints imposed. In the context of locality, the pre-specified constraints are of course, the constraints imposed for satisfying the definition of an LR code. 

\begin{defn}
	Let $H_0$ be a $(\rho \times n)$ matrix over a finite field \fq\ with $\text{rank}(H_0) = \rho$ whose row space has $m=q^{\rho}-1$ nonzero vectors with respective support sets $A_i \subseteq [n], i=1,2, \cdots, \ m$.  We view $H_0$ as the matrix that imposes constraints necessary to satisfy the definition of an LR code.   Let us define a subset $S \subset [n]$ to be a $k$-core with respect to $H_0$ if $|S|=k$ and $| A_i \cap S^c | \geq 1, \forall i \in [m]$.  Then with respect to $H_0$, an MR code is an $[n,k]$ code ${\cal C}$ over \fq\ possessing a $(k \times n)$ generator matrix $G$ with $k \leq n-\rho$ satisfying the property that $H_0G^T =[0]$ and for any $k$-core $S$,  
	\bea \label{eq:mrc} 
	\text{rank} \left( G \mid_S \right) & = & k. 
	\eea
\end{defn}
\begin{note}
	Let $H = \left [ \frac{H_0}{H_1} \right] $ denote the parity-check matrix of the MR code, where $H_1$ is an $((n-k-\rho)\times n)$ matrix representing the additional parity-checks imposed in such a way that the code  with parity check matrix $H$ satisfies the requirements of an MR code. It could happen that the elements of $H_0$ belong to a small base field $\mathbb{B}$ and over that field it is not possible to find a matrix $H_1$ which will result in an MR code.  It turns out that in such instances, one can always choose the elements of $H_1$ to lie in a suitable extension field \fq\ of $\mathbb{B}$, resulting in an MR code over \fq. 
	
\end{note}
\begin{note}
	The condition in \eqref{eq:mrc}  imposed on the $k$-core subsets $S$ is equivalent to the following condition: 
	Let $B \subseteq [n]$ be such that $\mid B^c \cap A_i \mid \geq 1$,   $\forall i \in [m]$. Then $G|_B$ is a generator matrix of an $[n=|B|,k]$ MDS code.  This follows since any $k$ columns of $G |_B$ are required to be linearly independent. 
\end{note}

\subsection{General Construction with Exponential Field Size} 

Saying that $S$ is a $k$-core is equivalent to saying that  $S$ is an information set since the $k$ underlying message symbols can be uniquely recovered from the $k$ code symbols $\{c_i \mid i \in S\}$. From the perspective of the parity check matrix $H$, $S$ is a $k$-core if only if $\text{rk} \left(H \mid_{S^c}\right) = (n-k)$.  This suggests a construction technique.  Setting $H = \left[ \begin{array}{c} H_0 \\ H_1 \end{array} \right]$ as earlier, we regard the symbols in the $((n-k-\rho) \times n)$ matrix $H_1$ as variables.   We need to select $H_1$ such that any $(n-k) \times (n-k)$ sub-matrix of $H$ corresponding to columns indexed by the complement $S^c$ of a $k$-core, has nonzero determinant.  Let $P(H_1)$ be the polynomial in the symbols of $H_1$ obtained by taking the product of these determinants. Note that the definition of a $k$-core ensures that each of these determinants are non-zero polynomials.  The product polynomial is a polynomial in the entries (variables) of the matrix $H_1$ and each variable appears with degree at most ${n-1 \choose n-k-1}$.  By the Combinatorial Nullstellensatz \cite{Nog}, it follows that there is a field of size $q > {n-1 \choose n-k-1}$ such that this product of determinants can be made nonzero.    Thus a MR code always exists of field size $q > {n-1 \choose n-k-1}$. The interest is of course, in explicit constructions of MR codes having low field size $q$.   It is also possible to use linearized polynomials to construct MR codes, but while this results in an explicit construction, the field size is still in general,  of exponential size. 

\subsection{Partial MDS Codes} 

In the literature, the focus motivated by practical considerations, is on the following subclass of MR code, also sometimes termed as Partial MDS (P-MDS) codes.
\begin{defn}
	An $(r,\delta,s)$ MR code or partial MDS code over a finite field \fq\ is defined as an $[n=m(r+\delta),k=mr-s]$ code over $\mathbb{F}_q$ in which the $n$ code symbols can be arranged as an array of $(m \times (r+\delta) \ )$ code symbols in such a way that each row in the array forms a $[r+\delta,r,\delta+1]$ MDS code and upon puncturing any $\delta$ code symbols from each row of the array, the resulting code becomes an $[mr,mr-s]$ MDS code.
\end{defn}

A tabular listing of some constructions of $(r,\delta,s)$ MR codes or partial-MDS codes appears in Table \ref{tab:pmds}. 
\begin{table}[h!] 
	\caption{Constructions for partial MDS or MR codes.} 
	\label{tab:pmds} 
	\begin{center}
		\scalebox{0.9}{
			\begin{tabular}{ | m{4.5em} | m{2cm}| m{14cm} |} 
				\hline \hline 
				\multicolumn{1}{|c|}{Reference} &  \centering Parameters of MR Code &  \multicolumn{1}{|c|}{Field Size} \\
				\hline \hline 
				\multicolumn{3}{|c|}{General $r,\delta,s$} \\ \hline 
				\hline
				\cite{Cal_Koy} & $(r,\delta,s)$ &  $(q')^{mr}$ where $q'$ is a prime power $\geq r+\delta$. \\
				\hline
				\cite{Gab_Yak_Bla_Sie} & $(r,\delta,s)$ &  $\geq \max((q')^{\delta+s}m^{s-1},(q')^{s(\delta+s)})$ with $q'$ a prime power $\geq r+\delta$. \\
				\hline \hline 
				\multicolumn{3}{|c|}{$\delta=1$} \\   \hline 
				\hline 
				\cite{Bla_Haf_Het} & $(r,1,s)$ & $O(2^n)$ \\
				\hline 
				\cite{GopHuaJenYek} &  $(r,1,s)$ & $O(m^{\lceil (s-1)(1-\frac{1}{2^r}) \rceil})$ or $\geq n^{\frac{m+s}{2}}$ for $m+s$ even and  $\geq 2n^{\frac{m+s-1}{2}}$ for $m+s$ odd, when $r+1$ and $m$ are powers of $2$. \\
				\hline 
				\cite{Hu_Yek} & $(r,1,s)$ &  $\geq (q')^{\lfloor (1-\frac{1}{m})s \rfloor +m-1}$ ($q'$ is prime power $\geq n $) and for some special case, the field size of their construction is $\geq (q')^{\lfloor (1-\frac{1}{m})s \rfloor +m-2}$. For  $m=2$, $4 | s$, $\geq (q')^{\frac{s}{2}}$ where $q'\geq n$ is a power of $2$.\\
				\hline
				\cite{Gab_Yak_Bla_Sie} & $(r,1,s)$ &  $\geq {2^{\ell}}^{(1+(s-1)\lceil \log_{2^{\ell}}(m) \rceil)}$ where $\ell = \lceil \frac{s+1}{2} \rceil \lceil \log_2(r+\delta) \rceil$. \\
				\hline \hline 
				\multicolumn{3}{|c|}{$s=1$}  \\ \hline 
				\hline 
				\cite{Bla_Haf_Het} & $(r,\delta,1)$ & $O(\max(m,r+\delta))$ \\ 
				\hline 
				\cite{Che_Shu_Yu_Sun} &  $(r,\delta,1)$ & $O(r+\delta)$ \\
				\hline \hline 
				\multicolumn{3}{|c|}{$s=2$}  \\ \hline 
				\hline 
				\cite{Blaum_1} and & $(r,1,2)$ & $O(n)$ \\
				Theorem \ref{thm:2oneS_ch7} & & \\
				\hline 
				\cite{Blaum_3} &  $(r,\delta,2)$ & $\geq m((\delta+1)(r-1)+1)$ $\approx$ $\delta \times n$ \\
				\hline 
				Theorem \ref{thm:Mod_bluam_ch7}  & $(r,\delta,2)$ & $O(n)$ \\
				\hline \hline 
				\multicolumn{3}{|c|}{$s=3$}  \\ \hline 
				\hline 
				\cite{GopHuaJenYek} &  $(r,1,3)$ & $O(k^{\frac{3}{2}})$ \\
				\hline 
				\cite{Gab_Yak_Bla_Sie} &  $(r,\delta,3)$ & if $m < (r+\delta)^3$ then $O((r+\delta)^{3(
					\delta+3)})$ otherwise $O((r+\delta)^{\delta+3} m^{1.5})$ \\
				\hline \hline 
				\multicolumn{3}{|c|}{$s=4$}  \\ \hline  \hline 
				\hline 
				\cite{GopHuaJenYek} &  $(r,1,4)$ & $O(k^{\frac{7}{3}})$ \\
				\hline       \hline                  
			\end{tabular}
		}
	\end{center}
\end{table}
In \cite{LalLok}, the authors characterize the weight enumerators and higher support weights of an $(r,1,s)$ MR code. Throughout this chapter, a local parity check refers to a code word in dual code of Hamming weight $\leq r+1$.

\section{Partial Maximal Recoverability} \label{sec:PMR_ch7} 

We have already seen in Chapter \ref{ch:LRSingleErasure} that the following upper bound on the minimum distance of a code under a weaker notion called information-symbol locality was derived in \cite{GopHuaSimYek}:
\bea
d_{\min} & \leq & (n-k+1) -\left( \left \lceil \frac{k}{r} \right \rceil -1 \right). \label{eq:dmin} 
\eea
The same bound also applies to codes with all-symbol locality and is often (but not always) tight, see \cite{PraLalKum} for instance. We will be constructing codes achieving this bound \eqref{eq:dmin} along with further constraints on erasure correction capabilities. In this chapter, we will refer to an $[n,k]$ LR code with AS (All Symbol) locality (defined in Chapter \ref{ch:LRSingleErasure}) over a finite field \fq\ with locality parameter $r$ and minimum distance $d$ as an $(n,k,r,d)$ LR code over \fq.

Given that the construction of MR codes having small field size is challenging, we seek here to construct codes that satisfy a weaker condition which we will refer to in this chapter as the partial maximal recoverable (PMR) condition. Let ${\cal C}$ be an $(n,k,r,d)$ LR code with $d$ satisfying the bound in \eqref{eq:dmin} with equality. Let $\{E_i: 1 \leq i \leq m \}$ denote the support of distinct local parity checks such that $\cup_{i=1}^{m}E_i=[n]$ (Recall that the term local parity check refers to a codeword of weight $\leq r+1$ in the dual code). Wolog we asssume $E_i \setminus \left( \bigcup_{j \in [m], \  j \neq i} E_j \right) \neq \emptyset$ as otherwise we can remove $E_i$ and relabel the sets again until it satisfies this condition. In the context of PMR codes, an admissible puncturing pattern $\{e_1,e_2,\cdots, e_m\}$ is one which satisfy the condition:
\bean
e_i & \in & E_i \setminus \left( \bigcup_{j \in [m], \  j \neq i} E_j \right).
\eean
\begin{defn}
	An $(n,k,r,d)$ PMR code over a finite field \fq\ is then defined simply as an $(n,k,r,d)$ LR code ${\cal C}$ over \fq\ whose minimum distance $d$ satisfies the bound in \eqref{eq:dmin} with equality and which becomes an MDS code upon puncturing the code $\calc$ in co-ordinates corresponding to one admissible puncturing pattern $\{e_1,e_2,\cdots, e_m\}$ i.e., the code $\mathcal{C}|_{[n] \setminus \{e_1,e_2,\cdots, e_m\} }$ is an MDS code for one admissible puncturing pattern $\{e_1,e_2,\cdots, e_m\}$. 
\end{defn}
The parity-check matrix of a PMR code is characterized below. We assume w.l.o.g. in the section below, that $(e_1,e_2,\cdots,e_m)=(1,2,\cdots,m)$ (can be ensured by symbol re-ordering) is an admissible puncturing pattern. 

\subsection{Characterizing $H$ for a PMR Code} \label{sec:characerizing_PMR_ch7}

\begin{thm} \label{thm:canonical_form} 
	Let ${\cal C}$ be an $(n,k,r,d)$ PMR code over a finite field \fq\ as defined above which becomes MDS on puncturing at the co-ordinates corresponding to the admissible puncturing pattern $e=\{e_1,...,e_m\}=\{1,...,m\}$. Let $n=k_0+m,k=k_0-\Delta$ for some natural numbers $k_0,\Delta$. Then the parity check matrix $H$ of \code (upto a paermutation of columns) can be written in the form: 
	\bea
	H & = & \left[ \begin{array}{c|c} 
		I_m & \underbrace{F}_{(m \times k_0)}   \\ \hline 
		[0]& \underbrace{\hmds}_{(\Delta \times k_0)} \end{array} \right], \label{eq:Hmatrix_ch7}
	\eea
	where \hmds is a parity-check matrix of a $[k_0,k_0-\Delta]$ MDS code and $F$ is of the form: 
	\bean
	F & = & 
	\left[ \begin{array}{c} 
		\ux{1}^t  \\
		\ux{2}^t  \\
		\ddots  \\
		\ux{m}^t  \end{array} \right]
	\eean
	in which each \ux{i} is a vector of Hamming weight at most $r$.
\end{thm}

\begin{proof}
	Clearly, $H$ (upto a permutation of columns) can be written in the form (where the first $m$ rows are codewords corresponding to the support sets $\{E_i:1\leq i \leq m \}$)
	\bean
	H & = & \left[ \begin{array}{c|c} 
		I_m & \underbrace{F}_{(m \times k_0)}   \\ \hline 
		H_1 & \underbrace{H_2}_{(\Delta \times k_0)} \end{array} \right],
	\eean
	which can be transformed, upon row reduction to the form:
	\bean
	H & = & \left[ \begin{array}{c|c} 
		I_m & \underbrace{F}_{(m \times k_0)}   \\ \hline 
		[0] & \underbrace{H_3}_{(\Delta \times k_0)} \end{array} \right].
	\eean
	For the code \calc\ to be a PMR code it is necessary that upon puncturing the first $m$ coordinates (corresponding to columns of the identity matrix $I_m$ in the upper left), the code become an MDS code.  But since the dual code of a punctured code is the shortened code in the same coordinates, it follows that $H_3$ must be a parity-check matrix of an MDS code.  
\end{proof}
\subsection{A Simple Parity-Splitting Construction for a PMR Code when $\Delta \leq (r-1)$} \label{sec:simple_PMR_construction_ch7}
Let,
\bean
n & = & m(r+1) , \hspace*{0.3in} k_0 \ = \ mr , \\ 
k & = & k_0-\Delta \ = \  n-(m+\Delta)
\eean
for some natural numbers $m,r,\Delta$. Assume that $\Delta \leq (r-1)$.  Let $H_g$ be a $(\Delta+1 \times k_0)$ parity-check matrix of a $[k_0,k_0-(\Delta+1)]$ MDS code. Let $\uxx^t$  be the  last row of $H_g$ and $H_1$ be $H_g$ with the last row deleted, i.e.,
\bean
H_g & = & \left[ \begin{array}{c} H_1 \\ \uxx^t \end{array} \right]. 
\eean
In the construction, we will require that $H_1$ also be a parity-check matrix of an MDS code and set
\bea
\hmds=H_1. \label{const_highrate_pmr1}
\eea
For example, this is the case when $H_g$ is either a Cauchy or a Vandermonde matrix.   Let $\{y_{i}^t\}_{i=1}^m$ be the $m$ contiguous component $(1 \times r)$ vectors of $\uxx^t$ defined through  
\bea
\uxx^t \  & = & \left( \underline{y}_{1}^t \ \underline{y}_{2}^t \cdots, \underline{y}_{m}^t \right). \label{const_highrate_pmr2}
\eea
Let $F$ be given by 
\bea
F & = & 
\left[ \begin{array}{cccc} 
	\underline{y}_1^t & & & \\
	& \underline{y}_2^t & & \\
	& & \ddots & \\
	& & & \underline{y}_{m}^t  \end{array} \right]. \label{const_highrate_pmr3}
\eea

\begin{thm} [Parity-Splitting Construction] \label{thm:simple_PMR} 
	The $[n=m(r+1),k=mr-\Delta]$ code \code having parity-check matrix $H$ given by 
	\bean
	H & = & \left[ \begin{array}{c|c} 
		I_m & \underbrace{F}_{(m \times k_0)}   \\ \hline 
		[0] & \underbrace{\hmds}_{(\Delta \times k_0)} \end{array} \right],
	\eean
	with $k_0=mr$ and $\hmds,F$ as defined by equations \eqref{const_highrate_pmr1},\eqref{const_highrate_pmr3} and $\Delta \leq (r-1)$, is a $(n,k,r,d_{\min})$ PMR code and hence with minimum distance $d_{\min}$ achieving the bound 
	\bean
	d_{\min} & = & (n-k+1) \ - \ \left( \left \lceil\frac{k}{r} \right \rceil -1 \right) \\
	& = & \Delta + 2 .
	\eean
\end{thm}

\begin{proof}
	Since on puncturing the first $m$ co-ordinates the code becomes MDS and the rows corresponding to $[I_m | F]$ gives the all symbol locality property to the code, it is enough to show that any $(\Delta+1)$ columns of $H$ are linearly independent. From the structure of $H$ and the properties of the matrix \hmds, it is not hard to see that it suffices to show that any $(\Delta+1)$ columns of 
	\bean
	H_a & = & \left[ \begin{array}{c} 
		F  \\ \hline 
		\hmds  \end{array} \right],
	\eean
	are linearly independent.  But the rowspace of $F$ contains the vector $\uxx^t$, hence it suffices to show that any $(\Delta+1)$ columns of 
	\bean
	H_b & = & \left[ \begin{array}{c} 
		\hmds \\ \hline 
		\uxx^t  \end{array} \right] \ = \ H_g
	\eean
	are linearly independent, but this is clearly the case, since $H_g$ is the parity-check matrix of an MDS code having redundancy $(\Delta+1)$. 
\end{proof}
In the above construction $\Delta$ represents the number of ``global'' parity checks imposed on top of $m$ ``local'' parity checks where the $m$ local parity checks refers to the $m$ codewords in the dual code which are in the rows of $[I_m | F]$ corresponding to $m$ support sets $\{ E_i : 1 \leq i \leq m \}$. 
\begin{note} 
	The construction gives rise to $(n=m(r+1),k=mr-\Delta,r,d=\Delta+2)$ LR code which is also a PMR code over a field size of at most $n$ (small field size) and also, high rate:
	\bean
	R & = & 1 -\frac{\Delta+m}{m(r+1)} \ \geq \ 1 - \frac{r-1+m}{m(r+1)}.
	\eean 
\end{note}

\section{A General Approach to PMR Construction} \label{sec:lower_rate_PMR_ch7}

We attempt to handle the general case 
\bean
\Delta & = & ar+b ,
\eean 
in this section and outline one approach.  At this time, we are only able to provide constructions for selected parameters with $\Delta=2r-2$ and field size that is cubic in the block length of the code and hold out hope that this construction can be generalized. 
Let,
\bean
n & = & m(r+1) , \hspace*{0.3in} k_0 \ = \ mr , \\ 
k & = & k_0-\Delta \ = \  n-(m+\Delta). \\
\Delta & = & ar+b
\eean
for some natural numbers $m,r,a,b$ with $b < r$.
Let $\mathcal{C}$ be an $(n,k,r,d)$ LR code with parity check matrix $H$ as given in equation \eqref{eq:Hmatrix_ch7} in Theorem \ref{thm:canonical_form}  and $H_{MDS}$ chosen to be a Vandermonde matrix and with Hamming weight of all the rows of $F$ being exactly $r$ (i.e., Hamming Weight of $\ux{i}^t = r$) and intersection of support of any two rows of $F$ being an empty set (i.e., $\text{Support}(\ux{i}^t) \cap \text{Support}(\ux{j}^t) = \emptyset$, $\forall i \neq j$). In the following we will derive conditions so that the code $\mathcal{C}$ becomes a PMR code.
For  $\mathcal{C}$ to be a PMR code, the desired minimum distance code  $d$ can be shown to equal in this case, 
\bean
d \ := \ \dmin & = & (n-k+1) - \left( \lceil \frac{k}{r} \rceil -1 \right) \\
& = & (m+\Delta+1) - \left( \lceil \frac{mr-\Delta}{r} \rceil -1 \right) \\
& = & \Delta + 2 + a .
\eean
It follows that even the code on the right having parity-check matrix
\bean
H_{\text{pun}} & = &  \left[ \begin{array}{c} 
	F  \\ \hline 
	\hmds  \end{array} \right],
\eean
must have the same value of \dmin and therefore, the sub matrix formed by any $(d-1)$ columns of $H_{pun}$ must have full rank.  Let $A$ be the column indices of a subset of $(d-1)$ columns of $H_{pun}$. Let $A$ have non-empty intersection with the support of  some $s$ rows $\{\underline{r}_1,...,\underline{r}_s\}$ of $F$. Let the intersection of $A$ with the support of row $\underline{r}_i$ be $A_i$ of size $\mid A_i \mid \ = \ \ell_i$. Let wolog row $\underline{r}_i$ refer to the $i$th row of $F$, $\forall 1 \leq i \leq s$. The corresponding sub matrix formed by columns of $H_{pun}$ indexed by $A$ will then take on the form:

\bea
\left[ \begin{array}{c} 
	\begin{array}{cccccccccc} 
		& a_1(\theta_{1i}) & & & & & & & &  \\
		& & & & a_2(\theta_{2i})  & & & & &  \\
		&  & & & & & \ddots  & & &  \\
		& & & & & & & & a_s(\theta_{si})&  \\ \hline 
		\cdots & 1 & \cdots & \cdots & 1 & \cdots & \cdots & \cdots & 1 & \cdots \\
		\cdots & \theta_{1i} & \cdots & \cdots & \theta_{2i} & \cdots & \cdots & \cdots & \theta_{si} & \cdots \\
		\cdots & \theta_{1i}^2 & \cdots & \cdots & \theta_{2i}^2 & \cdots & \cdots & \cdots & \theta_{si}^2 & \cdots \\
		\cdots & \cdots & \cdots & \cdots & \cdots & \cdots & \cdots & \cdots & \cdots & \cdots \\
		\cdots & \theta_{1i}^{\Delta-1} & \cdots & \cdots & \theta_{2i}^{\Delta-1}  & \cdots & \cdots & \cdots & \theta_{si}^{\Delta-1} & \cdots \\
	\end{array} 
\end{array} \right] \label{matrix_pmr}
\eea

where $a_i(x)$ is the polynomial whose evaluations provide the components of the $i$th row of $F$, $\forall 1 \leq i \leq m$ as indicated in equation \eqref{matrix_pmr}.
Since we want this matrix to have full rank $(d-1)$ it must be that the left null space of the matrix must be of dimension $(\Delta+s)-(\Delta+a+1)\ = \ s-(a+1)$. 
Computing the dimension of this null space is equivalent (by polynomial interpolation formula) to computing the number of solutions $\{s_i : 1 \leq i \leq s \}$ to the equation: 
\bea
\sum_{i=1}^s s_i \sum_{j=1}^{\ell_i}a_i (\theta_{ij}) \prod_{(k,l) \neq (i,j) } \frac{(x-\theta_{kl})}{(\theta_{ij}-\theta_{kl})} & = & f(x) \label{Nullspace_pmr}
\eea
where $f(x)$ is generic notation for a polynomial of degree $\leq (\Delta-1)$.  Let us define 
\bean
E_i(x) & = & \sum_{j=1}^{\ell_i}a_i (\theta_{ij}) \prod_{(k,l) \neq (i,j) } \frac{(x-\theta_{kl})}{(\theta_{ij}-\theta_{kl})} ,
\eean
and note that each $E_i(x)$ will in general, have degree $(\Delta+a)$. 
Consider the matrix $E$ whose rows correspond to the coefficients of $E_i(x)$.
\bean
E & = & \text{Coefficients of} \left(\left[ \begin{array}{c}
	E_1(x) \\
	\vdots \\
	E_{a+1}(x) \\ \hline 
	E_{a+2}(x) \\
	\vdots \\
	E_s(x) \end{array} \right] \right).
\eean
$E$ is a $s \times (\Delta+a)$ matrix and we want to calculate the number of solutions $\underline{y}$ such that $\underline{y}E$ has first $a+1$ compoenents equal to $0$ to calculate number of solutions to equation \eqref{Nullspace_pmr}.
It follows that the first $(a+1)$ columns of $E$ must have full rank so that the dimension of solution space of equation \eqref{Nullspace_pmr} and hence the dimension of left null space of the matrix in equation \eqref{matrix_pmr} is exactly $s-(a+1)$.

\subsection{Restriction to the Case $a=1$, i.e., $r \leq \Delta\leq 2r-1$}

We now assume that $a=1$ so that $(a+1)=2$ and we need the first $2$ columns of $E$ to have rank $=2$.  We consider the $(2 \times 2)$ sub matrix made up of the first two rows and first two columns of $E$.  The determinant of this $(2 \times 2)$ upper-left matrix of $E$ is given by
\bean
-\det \left[ \begin{array}{cc} 
	\sum_{j=1}^{\ell_1} \frac{a_1(\theta_{1j})}{P_{1j}} & \sum_{j=1}^{\ell_1} \frac{a_1(\theta_{1j})\left( \sum_{(k,l) \neq (1,j)} \theta_{kl} \right)}{P_{1j}} \\
	\sum_{j=1}^{\ell_2} \frac{a_2(\theta_{2j})}{P_{2j}} & \sum_{j=1}^{\ell_2} \frac{a_2(\theta_{2j})\left( \sum_{(k,l) \neq (2,j)} \theta_{kl} \right)}{P_{2j}} \\
\end{array} \right] \\
=   \det  \left[ \begin{array}{cc} 
	\sum_{j=1}^{\ell_1} \frac{a_1(\theta_{1j})}{P_{1j}} & \sum_{j=1}^{\ell_1} \frac{a_1(\theta_{1j}) \theta_{1j} }{P_{1j}} \\
	\sum_{j=1}^{\ell_2} \frac{a_2(\theta_{2j})}{P_{2j}} & \sum_{j=1}^{\ell_2} \frac{a_2(\theta_{2j}) \theta_{2j}}{P_{2j}} \\
\end{array} \right] 
\eean
where 
\bean
P_{ij} & = & \prod_{(k,l) \neq (i,j)} (\theta_{ij}-\theta_{kl})
\eean

This is equal to 
\bean
\sum_{j=1}^{\ell_1} \sum_{t=1}^{\ell_2} \frac{a_1(\theta_{1j})a_2(\theta_{2t})}{P_{1j}P_{2t}}(\theta_{1j}-\theta_{2t}).
\eean
Let $\Delta=2r-1$ and $a_1(\theta_{1j}) =   \theta_{1j}$, $a_2(\theta_{2t}) =  \theta_{2t}$, $\theta_{ij} = \xi + h_{ij}, \ h_{ij} \in \mathbb{F}_q$ and 
$\xi \  \in \  \mathbb{F}_{q^3} \setminus \mathbb{F}_q$ for some prime power $q$ with $q \geq n$. 
Then this becomes:
\bean
\sum_{j=1}^{\ell_1} \sum_{t=1}^{\ell_2} \frac{ \left(\xi^2 + \xi (h_{1j}+h_{2t}) +h_{1j}h_{2t} \right)} {P_{1j}(\{h_{ij}\})P_{2t}(\{h_{ij}\})}(h_{1j}-h_{2t}) \\
=  A\xi^2 + B\xi + C 
\eean
with $A,B,C \in \mathbb{F}_q$ which will be nonzero if the minimum polynomial of $\xi$ over $\mathbb{F}_q$ has degree $=3$, unless all the coefficients are equal to zero.

\paragraph{Numerical Evidence} Computer verification was carried out for the $\Delta=5,r=3$ case for $n=12$ over $F_{(2^4)^3}$ ($q=2^4$) and $n=36$ over $F_{(2^6)^3}$ ($q=2^6$) with $h_{ij}=\alpha^{(i-1)} \beta{(ij)}$ where $\alpha$ is the primitive element of $F_{2^4}$ and $F_{2^6}$ respectively for the two cases and $\beta(ij)$ is fifth and seventh root of unity respectively for the two cases (the choice of fifth and seventh roots of unity varies for each $i,j$). For both cases, it was found that the elements $A,B,C$ never simultaneously vanished for all instances.

\section{Maximal Recoverable Codes} \label{sec:MRConst_ch7} 
In this section, we give three constructions of MR codes with low field size. The first construction is a construction of $(r=2,\delta=1,s)$ MR code and the rest of the two constructions correspond to MR codes with $s=2$. The last two constructions has field size of $O(n)$. 
\subsection{A Coset-Based Construction with Locality $r=2$} \label{sec:coset_construction_ch7} 
We now give a construction of $(2,1,s)$ MR code. Since this construction is based on Construction $1$ in \cite{TamBar_LRC} of all-symbol LR codes, we briefly review the construction $1$ here. This construction 1 in \cite{TamBar_LRC} is already described in Chapter \ref{ch:LRSingleErasure}.

\begin{const} \cite{TamBar_LRC} \label{const:Tamo_Barg_ch7}
	Let $n=m(r+1)$ for some natural numbers $m,r$ and $q$ be a power of a prime such that $n \leq (q-1)$, for example, $q$ could equal $(n+1)$ and $(r+1) \  | \ (q-1)$. Let $\alpha$ be a primitive element of $\mathbb{F}_q$ and $\beta = \alpha^{\frac{q-1}{r+1}}$. Note that $\beta$ is an element of order $(r+1)$. Let $H = \{1,\beta,\beta^2,\cdots, \beta^r\}$ be the group of $(r+1)$th roots of identity element of the field $\mathbb{F}_q$. Let 
	\bean
	A_i & = & \alpha^{i-1} H , \ \ 1 \leq i \leq m.
	\eean
	Note that $\{A_i\}_{i=1}^m$ are pairwise disjoint and $| \cup_{i=1}^mA_i | = n$ and $A_i$ are cosets of the group $H$. Let $k=ar+b$. Note that the monomial $x^{r+1}$ evaluates to a constant on elements of any set $A_i$ where the constant depends only on $i$. Let the set of $k$ message symbols be $A=\{a_{ij}: 0 \leq j \leq a-1, 0 \leq i \leq r-1 \} \cup \{a_{ij}: j = a, 0 \leq i \leq b-1 \} \subseteq \fq$. Let
	\bean
	f_A(x) & = & \sum_{j=0}^{a-1} \sum_{i=0}^{r-1} a_{ij}x^{j(r+1)+i} \ + \   \sum_{j=a} \sum_{i=0}^{b-1} a_{ij}x^{j(r+1)+i} ,
	\eean
	where the second term is vacuous for $b=0$, i.e., is not present when $r \mid k$.  Let $ \bigcup_{i=1}^m A_i = \{u_1,..,u_n\}$.
	Consider the code ${\cal C}$ of block length $n$ and dimension $k$ defined by the following encoding function $E: \fq^k \rightarrow \fq^n$:
	\bean
	E(A) =[c_1,...,c_i,...,c_n] = [f_A(u_1),...,f_A(u_i),...,f_A(u_n)],
	\eean
	where $[c_1,...,c_i,...,c_n]$ indicate the code symbols forming the codeword $E(A)$. Hence the message symbols $A$ are encoded by evaluating the polynomial $f_A(x)$ at the field elements in $\bigcup_{i=1}^m A_i$. The code ${\cal C}$ is an $(n,k,r,d)$ LR code over \fq\ with $d$ satisfying \eqref{eq:dmin} with equality where the conditions of LR code are satisfied by a set of $m$ local parity checks (i.e., $m$ codewords in the dual code of weight $\leq r+1$) with support sets covering [n]. The support of the $i^{th}$ local parity check are the indices of those code symbols of the codeword $E(A)$ obtained by evaluating $f_A(x)$ at elements of $A_i$. 
\end{const}

Note that the exponents $e$ in the monomial terms forming each polynomial $f_A(x)$ satisfy 
$e \text{ mod } (r+1)  \neq  r $. 
It is this property that makes the code $\calc$ an LR code.

Our construction of $(2,1,s)$ MR code here is based on Construction \ref{const:Tamo_Barg_ch7} corresponding to the code ${\cal C}$ with parameters given by $n=q-1, r=2,k=2D+1$ so that $a=D$ and $b=1$. Thus all $m$ local parity checks have support of cardinality $3$.  Let us denote the algebraic closure of $\mathbb{F}_q$ by $\mathbb{F}$. 

\begin{thm} \label{thm:neat5_ch7}
	Given positive integers $N, D$ with $\frac{2D}{N} < \frac{2}{3}$ and $3 \  | \  N$,
	Let ${\cal C}$ be the code over \fq\ described in Construction \ref{const:Tamo_Barg_ch7} with parameters $n=q-1, r=2, k=2D+1$ $(k=ar+b \text{ with } r=2,a=D,b=1)$ for $q$, a prime power with $3 \ | \ (q-1)$ such that:
	\bea
	q & > & 2 \left( \Sigma_{j=2}^{2D} \lfloor{j g(j)} \rfloor {(\frac{N}{3}-1) \choose j} 3^{j} \right) + N-2  , \label{eq:fieldsize_ch7}
	\eea
	where 
	\bean
	g(j) & = & \left\{ \begin{array}{cc} 1 & \text{for $j$ even and $2(D-1) \geq j \geq 4$} \\
		\frac{1}{2} & \text{otherwise,} \end{array} \right.
	\eean
	Then there exists an $(r=2,\delta=1,s=N-k-\frac{N}{3})$  MR code $\mathcal{C}_{MR}$ over $\mathbb{F}_q$ with block length $N$ and dimension $k=2D+1$ that is obtained from ${\cal C}$ by puncturing the code ${\cal C}$ at the co-ordinates $S$ corresponding to code symbols obtained by evaluating $f_A(x)$ at a carefully selected set of $c \ = \ \frac{q-1}{3} - \frac{N}{3}$ cosets (evaluating positions) given by $\{A_{i_1}, A_{i_2}, \cdots, A_{i_c}\}$ i.e., $\mathcal{C}_{MR} = \mathcal{C}|_{[n]\setminus S}$ where $\{c_j: j \in S\} = \{f_A(u): u \in \cup_{j=1}^c A_{i_j}\}$. 
\end{thm}

\begin{proof}
	For proof please refer to the Appendix \ref{sec:MR_Proofs_ch7}. 
\end{proof}
\begin{note}
	Although Theorem \ref{thm:neat5_ch7}, states that $(r=2,\delta=1,s=N-k-\frac{N}{3})$  MR code $\mathcal{C}_{MR}$ over $\mathbb{F}_q$ exists, one can construct such codes by following the procedure given in the proof of Theorem \ref{thm:neat5_ch7}. The proof of Theorem \ref{thm:neat5_ch7} is essentially a greedy algorithm which exploits the structure of the code ${\cal C}$. Hence can be interpreted as a construction of MR code.
\end{note}
\paragraph{Analysis of Field size:}
\begin{thm} \label{thm:improvedfieldsize_ch7}
	Theorem \ref{thm:neat5_ch7} showing the existence of $(r=2,\delta=1,s=N-k-\frac{N}{3})$ MR code $\mathcal{C}_{MR}$ over $\mathbb{F}_q$ with block length $N$ and dimension $k=2D+1$ obtained by puncturing ${\cal C}$ is true with field size $q$ such that:
	\bea
	q & \leq & \left\{ \begin{array}{cc} constant \times (m-1)4^{m-1} & \text{for $0 < y \leq 1$ } \\
		constant \times m^{2}  \left(e_1(y) \right)^m & \text{for $1 < y < 2$} \\
		N+1 & \text{for $y=2$}, \end{array} \right. \label{eq:fldupper_ch7}
	\eea
	where:
	\bean
	k=2D+1 &=& (2-y)m+1, \\
	m &=&\frac{N}{3}, \\
	e_1(y) & = &  \max_{0 <  \gamma \leq 2-y} \frac{3^{\gamma}}{(\gamma)^{\gamma} \times (1-\gamma)^{(1-\gamma)}}.
	\eean
\end{thm}
\begin{proof}
	Theorem \ref{thm:neat5_ch7} is clearly true with field size $q$ such that:
	\bean
	q \leq Q = constant \times \left ( \Sigma_{j=2}^{2D} \lfloor{j g(j)} \rfloor {(\frac{N}{3}-1) \choose j} 3^{j} + N-2 \right ).
	\eean
	Note that  $s=\frac{2N}{3}-k=ym-1$. Hence (the value of $constant$ in each equation in the following could be different):
	\bea
	Q & \leq & constant \times \left( \Sigma_{j=2}^{(2-y)m} j  {(m-1) \choose j} 3^{j} + 3m-2 \right). \label{eq:fieldsize1}
	\eea
	For $0 < y \leq 1$, from equation \eqref{eq:fieldsize1}:
	\bea
	Q & \leq & constant \times \Sigma_{j=0}^{(2-y)m} j {(m-1) \choose j} 3^{j} \leq constant \times (m-1)4^{m-1}. \label{eq:fieldsize2}
	\eea
	For  $1 < y < 2$, from equation \eqref{eq:fieldsize1} (in the following we assume that maximum over a null set is $1$): 
	\bea
	Q & \leq & constant \times \Sigma_{j=0}^{(2-y)m} j {(m-1) \choose j} 3^{j}, \notag \\
	Q & \leq & constant \times m^2 \max_{\frac{1}{m} \leq \gamma \leq \min(2-y,1-\frac{2}{m})} {(m-1) \choose \gamma m} 3^{\gamma m}, \notag \notag \\
	Q & \leq & constant \times m^2 \max_{\frac{1}{m} \leq \gamma \leq \min(2-y,1-\frac{2}{m})} \frac{(m-1)!}{(\gamma m)! \times (m(1-\gamma)-1)!} 3^{\gamma m}. \label{eq:fieldsize3}
	\eea
	
	Using the inequality (\cite{Rob}) $(2 \pi)^{0.5} m^{m+\frac{1}{2}} e^{-m} \leq m! \leq (2 \pi)^{0.5} e^{1/12} m^{m+\frac{1}{2}} e^{-m}$ and using equation \eqref{eq:fieldsize3}:
	\bean
	Q & \leq & constant \times m^2 \max_{\frac{1}{m} \leq  \gamma \leq \min(2-y,1-\frac{2}{m})} \frac{(m-1)^{m-1+\frac{1}{2}}}{(\gamma m)^{\gamma m + \frac{1}{2}} \times (m(1-\gamma)-1)^{(m(1-\gamma)-1)+\frac{1}{2}}} 3^{\gamma m}, \\
	Q & \leq & constant \times m^2 \max_{\frac{1}{m} \leq  \gamma \leq \min(2-y,1-\frac{2}{m})} \frac{(m)^{m-1+\frac{1}{2}}}{(\gamma m)^{\gamma m + \frac{1}{2}} \times (m(1-\gamma))^{(m(1-\gamma)-1)+\frac{1}{2}}} 3^{\gamma m}, \\
	Q & \leq & constant \times m^{2} \max_{\frac{1}{m} \leq  \gamma \leq \min(2-y,1-\frac{2}{m})} \left(\frac{3^{\gamma}}{(\gamma)^{\gamma} \times (1-\gamma)^{(1-\gamma)}} \right)^m, \\
	Q & \leq & constant \times m^{2}  \left( \max_{0 <  \gamma \leq 2-y} \frac{3^{\gamma}}{(\gamma)^{\gamma} \times (1-\gamma)^{(1-\gamma)}} \right)^m, \\
	Q & \leq & constant \times m^{2}  \left(e_1(y) \right)^m.
	\eean
	For $y=2$, equation \eqref{eq:fldupper_ch7} directly follows from the fact that for $k=1$, the code $\mathcal{C}$ is a repetition code and hence an MR code. Note that repetition code can be constructed over binary field but since here we constructing by puncturing $\mathcal{C}$ we give a field size of $N+1$.
\end{proof}

\begin{table}[h!] 
	\caption{Comparison of field size of the code $\mathcal{C}_{MR}$ given in Theorem \ref{thm:improvedfieldsize_ch7} with existing constructions.} 
	\label{tab:pmds_r2} 
	\begin{center}
		\scalebox{0.9}{
			\begin{tabular}{ | m{7em} | m{10cm}|} 
				\hline \hline 
				\multicolumn{2}{|c|}{$(r=2,\delta=1,s)$ MR code with blcok length $N=3m$, dimension $k=(2-y)m+1$ } \\
				\hline \hline
				\multicolumn{1}{|c|}{Reference} &  \multicolumn{1}{|c|}{Field Size $q$} \\
				\hline \hline 
				\multicolumn{2}{|c|}{ $1.3 \leq y < 2$, $0 < \text{rate} \leq 0.233$} \\ 
				\hline \hline
				This thesis Theorem \ref{thm:improvedfieldsize_ch7} &  $q \leq constant \times (e_1(y))^m m^{2}$ \\
				\hline
				\cite{GopHuaJenYek,Hu_Yek,Gab_Yak_Bla_Sie} & $q \geq m^{\frac{3my}{4}-2}$  \\
				\hline
				\cite{Cal_Koy} & $q \geq 4^m$ \\
				\hline
				\cite{Bla_Haf_Het} & $q = 8^m$ \\
				\hline
				\cite{Chen} & $q \geq {3m-1 \choose (2-y)m} \geq constant \times (e_2(y))^m m^{-\frac{1}{2}}$, \\
				& $e_2(y) = \frac{27}{(2-y)^{(2-y)} (1+y)^{(1+y)}}$ \\
				& $e_1(y) < e_2(y)$ \\
				\hline
			\end{tabular}
		}
	\end{center}
\end{table}
The following conclusions are true as $m$ increases for a fixed $y$.
Table \ref{tab:pmds_r2} shows a comparison of field size of the code $\mathcal{C}_{MR}$ given in Theorem \ref{thm:improvedfieldsize_ch7} with existing constructions for $1.3 \leq y < 2$ i.e., $0 < \text{rate} \leq 0.233$. From the Table \ref{tab:pmds_r2}, it can be seen that the code $\mathcal{C}_{MR}$ given in Theorem \ref{thm:improvedfieldsize_ch7} has the smallest field size compared to all the existing constructions for $1.3 \leq y < 2$ as it can be verfied that $e_1(y) < 4$ for $1.3 \leq y < 2$.
For the range $0 < y < 1.3$ (rate regime of $0.233 < \text{rate} < 0.667$), from \eqref{eq:fldupper_ch7}, the code $\mathcal{C}_{MR}$ has field size of at most $constant \times (m-1) 4^{m-1}$. Hence our construction is not that far off from the smallest known field size for $(r,1,s)$ MR codes in the rate range $0.233 < \text{rate} < 0.667$ also as from Table \ref{tab:pmds_r2}, it can be seen that the best known field size is $4^m$ for $0 < y < 1.3$.  Note that although the field size of $\mathcal{C}_{MR}$ is exponential, there are no constructions in literature in rate regime $0 < \text{rate} < \frac{r}{r+1}$ for any $r$, which as block length tends to infinity (for a fixed rate) gives a polynomial (in block length) field size MR code. 

\subsection{Explicit $(r,1,2)$ MR codes with field Size of $O(n)$} \label{sec:r12MR_ch7}

In this section we give an explicit construction of $(r,1,2)$ MR codes using the structure of $H$ matrix given in Theorem~\ref{thm:canonical_form} with $H_{MDS}$ chosen as Vandermonde matrix. We choose co-efficients corresponding to local parity checks from cosets of a subfield which will ensure appropriate erasure correcting capability required for MR codes.

\begin{thm} \label{thm:2oneS_ch7}  
	Let $m,r$ be positive integers and $n=m(r+1)$ and $k_0=mr$,$k=mr-2$. Let,
	\bea
	H & = & \left[ \begin{array}{c|c} 
		I_m & \underbrace{F}_{(m \times k_0)}   \\ \hline 
		[0]& \underbrace{\hmds}_{(2 \times k_0)} \end{array} \right], \label{pcmatrix_ch7}
	\eea
	\bean
	\hmds & = & 
	\left[ \begin{array}{ccc} 
		1 & \cdots & 1 \\
		\theta_1 & \cdots & \theta_{k_0}
	\end{array} \right],
	\eean
	\bean
	F & = & 
	\left[ \begin{array}{cccc} 
		\underline{x}_1^t & & & \\
		& \underline{x}_2^t & & \\
		& & \ddots & \\
		& & & \underline{x}_{m}^t  \end{array} \right],
	\eean
	\bean
	\underline{x}_i^t & = & (\theta_{(i-1)r+1}^2,\ \theta_{(i-1)r+2}^2, \cdots ,\  \theta_{(i-1)r+r}^2),
	\eean
	\bean
	\theta_{(i-1)r + j} = \alpha^{i-1} \beta^j,
	\eean
	where $\alpha$ is a primitive element of a finite field $\mathbb{F}_q$ and $\beta = \alpha^{\frac{q-1}{2^{\ell}-1}}$ is a $(2^{\ell} - 1)^{th}$ root of unity and $\ell$ is the smallest integer such that $2^{\ell} - 1 \geq r+1$ with $q$ satisfying equation \eqref{fld_2}. The code $\mathcal{C}$ with parity check matrix $H$ (described in \eqref{pcmatrix_ch7}) is an $(r,1,2)$ MR code over \fq\  with block length $n=m(r+1)$ and dimension $k=mr-2$ with $q$ satisfying:
	\bea
	q = 2^{\ell \rho } > \frac{(2^{\ell} - 1) (k+2)}{r} + 1, \label{fld_2}
	\eea
	and $\rho \geq 1$ is the smallest integer satisfying the above inequality (where $(2^{\ell} - 1)^{th}$ roots of unity form a subfield of $\mathbb{F}_q$).
	
\end{thm}
\begin{proof}
	Since we need to show that we can correct any erasure pattern with two random erasures at any positions with rest of the erasures such that there is atmost one erasure in the support of each local parity check (a local parity check refers to  one of the codewords (or rows) in the first $m$ rows of $H$), It is enough to check that any $(2+m) \times (2+j) $ submatrix of $H$ with $2+j$ columns having non zero intersection with support of exactly $j$ local parity checks ($j$ rows among the first $m$ rows of $H$), is full rank. Since a single erasure in the support of a local parity check can be corrected, its enough to consider $j \leq s=2 $.
	We choose the $(2^{\ell} - 1)^{th}$ roots of unity $\{ 1,\beta,\cdots,\beta^{2^{\ell} - 2} \}$ from $\mathbb{F}_q$ to be a subfield of $\mathbb{F}_q$. This can be done as $q = 2^{\ell \rho}$ with $\rho \geq 1$ (Just take the field $F_{2^{\ell}}$ and take $\beta$ to be its primitive element and go to an finite extension to get $\mathbb{F}_q$ satisfying the inequality $q = 2^{\ell \rho } > \frac{(2^{\ell} - 1) (k+2)}{r} + 1$).\\
	
	Note that all $\theta_{((i-1)r + j )} = \alpha^{(i-1)} \beta^j$ are distinct. If not $\alpha^{(i_1-1)} \beta^{j_1} = \alpha^{(i_2-1)} \beta^{j_2} $ but this can't be true as elements in both sides of the equation belong to different cosets of $(2^{\ell} - 1)^{th}$ roots of unity due to the condition $\frac{k+2}{r} < \frac{q-1}{2^{\ell} - 1}$.
	
	Based on the discussion above its enough to check non singularity of following matrices:\\
	$1$ local parity check with $2+1=3$ erasures in it:
	\bean
	H_1 & = & 
	\left[ \begin{array}{ccc} 
		\theta_{(i-1)r+j_1}^2 & \theta_{(i-1)r+j_2}^2 & \theta_{(i-1)r+j_3}^2\\
		1 & 1 & 1 \\
		\theta_{(i-1)r+j_1} & \theta_{(i-1)r+j_2} & \theta_{(i-1)r+j_3}
		
	\end{array} \right]
	\eean
	\bean
	H_2 & = & 
	\left[ \begin{array}{ccc} 
		1 & \theta_{(i-1)r+j_1}^2 & \theta_{(i-1)r+j_2}^2 \\
		0 & 1 & 1 \\
		0 & \theta_{(i-1)r+j_1} & \theta_{(i-1)r+j_2}
		
	\end{array} \right]
	\eean
	$H_1$ is a Vandermonde matrix (with rows permuted) and hence non singular and its clear the $H_2$ is also non singular for any local code $i$ and any erasure locations $j_1,j_2,j_3 \in [r]$.
	
	The case of $2$ local parity checks with $3$ erasure in the support of first local parity check and $1$ erasure in the support of second local parity check boils down to the above $2$ cases as the $1$ erasure in the support of second local parity check can be corrected by locality.
	
	Submatrix of $H$ corresponding to $2$ local parity checks with each having $2$ erasures in it:
	\bean
	H_3 & = & 
	\left[ \begin{array}{cccc} 
		0 & 0 & \theta_{(i_2-1)r+k_1}^2 & \theta_{(i_2-1)r+k_2}^2 \\
		\theta_{(i_1-1)r+j_1}^2 & \theta_{(i_1-1)r+j_2}^2 & 0 & 0 \\
		1 & 1 & 1 & 1 \\
		\theta_{(i_1-1)r+j_1} & \theta_{(i_1-1)r+j_2} & \theta_{(i_2-1)r+k_1} & \theta_{(i_2-1)r+k_2}
		
	\end{array} \right]
	\eean
	\bean
	H_4 & = & 
	\left[ \begin{array}{cccc} 
		0 & 0 & \theta_{(i_2-1)r+k_1}^2 & \theta_{(i_2-1)r+k_2}^2 \\
		1 & \theta_{(i_1-1)r+j_1}^2 & 0 & 0 \\
		0 & 1 & 1 & 1 \\
		0 & \theta_{(i_1-1)r+j_1} & \theta_{(i_2-1)r+k_1} & \theta_{(i_2-1)r+k_2}
		
	\end{array} \right]
	\eean
	\bean
	H_5 & = & 
	\left[ \begin{array}{cccc} 
		1 & 0 & 0 & \theta_{(i_2-1)r+k_1}^2 \\
		0 & 1 & \theta_{(i_1-1)r+j_1}^2 & 0 \\
		0 & 0 & 1 & 1 \\
		0 & 0 & \theta_{(i_1-1)r+j_1} & \theta_{(i_2-1)r+k_1}
		
	\end{array} \right]
	\eean
	
	The matrix $H_5$ is clearly non singular.
	
	\bean
	H_3 & = & 
	\left[ \begin{array}{cccc} 
		0 & 0 & \theta_{(i_2-1)r+k_1}^2 & \theta_{(i_2-1)r+k_2}^2 \\
		\theta_{(i_1-1)r+j_1}^2 & \theta_{(i_1-1)r+j_2}^2 & 0 & 0 \\
		1 & 1 & 1 & 1 \\
		\theta_{(i_1-1)r+j_1} & \theta_{(i_1-1)r+j_2} & \theta_{(i_2-1)r+k_1} & \theta_{(i_2-1)r+k_2}
	\end{array} \right]
	\eean
	
	we do permuation of rows or columns and row and column operations (as these do not change the singularity or non singularity of the matrix) on $H_3$:\\
	we substitute $\theta_{(i-1)r+j} = \alpha^{(i-1)} \beta^j$ in $H_3$:
	wlog we assume $j_2 >j_1$ and $k_2 > k_1$ and $i_2 > i_1$
	
	\bean
	H_3 & = & 
	\left[ \begin{array}{cccc} 
		0 & 0 & \alpha^{2(i_2-1)} \beta^{2k_1} & \alpha^{2(i_2-1)} \beta^{2k_2} \\
		\alpha^{2(i_1-1)} \beta^{2j_1} & \alpha^{2(i_1-1)} \beta^{2j_2} & 0 & 0 \\
		1 & 1 & 1 & 1 \\
		\alpha^{(i_1-1)} \beta^{j_1} & \alpha^{(i_1-1)} \beta^{j_2} & \alpha^{(i_2-1)} \beta^{k_1} & \alpha^{(i_2-1)} \beta^{k_2}
	\end{array} \right]
	\eean
	multiply the non zero scalars $\alpha^{-2(i_1-1)},\alpha^{-2(i_2-1)},\alpha^{-(i_1-1)}$ to rows $2,1,4$ respectively as this does not affect the non-singularity of the matrix. The matrix $H_3$ now becomes
	\bean
	\left[ \begin{array}{cccc} 
		0 & 0 & \beta^{2k_1} & \beta^{2k_2} \\
		\beta^{2j_1} & \beta^{2j_2} & 0 & 0 \\
		1 & 1 & 1 & 1 \\
		\beta^{j_1} & \beta^{j_2} & \alpha^{(i_2-i_1)} \beta^{k_1} & \alpha^{(i_2-i_1)} \beta^{k_2}
	\end{array} \right]
	\eean
	multiply row 1 with $- \alpha^{(i_2-i_1)} \beta^{-k_1}$ and add it row $4$ as this won't change the determinant. The matrix now becomes:
	\bean
	\left[ \begin{array}{cccc} 
		0 & 0 & \beta^{2k_1} & \beta^{2k_2} \\
		\beta^{2j_1} & \beta^{2j_2} & 0 & 0 \\
		1 & 1 & 1 & 1 \\
		\beta^{j_1} & \beta^{j_2} & 0 & \alpha^{(i_2-i_1)} (\beta^{k_2}-\beta^{2k_2-k_1})
	\end{array} \right]
	\eean
	the determinant of the above matrix is:
	\bean
	\alpha^{(i_2-i_1)} \beta^{k_2}(1-\beta^{k_2-k_1}) \beta^{2k_1} \beta^{2j_1} (1 - \beta^{2(j_2 - j_1)}) - \{ \beta^{\psi} & or &  0\} 
	\eean
	
	for some $\psi$ as $(2^{\ell} - 1)^{th}$ roots of unity form a subfield.\\
	$(1-\beta^{k_2-k_1}) \neq 0 $, as $0 < k_2-k_1 < 2^{\ell} - 1$ because $k_i < 2^{\ell} - 1$ \\
	$(1 - \beta^{2(j_2 - j_1)}) \neq 0 $ as $ 0 < (j_2 - j_1) < 2^{\ell} - 1$ which implies $ 0 < 2(j_2 - j_1) < 2(2^{\ell} - 1)$. Hence for $(1 - \beta^{2(j_2 - j_1)}) = 0 $, $2(j_2-j_1) = 2^{\ell} - 1$ which can't be true as $2^{\ell} - 1$ is odd.
	
	Hence for $H_3$ to have zero determinant :
	\bean
	\alpha^{(i_2-i_1)} \beta^{k_2}(1-\beta^{k_2-k_1}) \beta^{2k_1} \beta^{2j_1} (1 - \beta^{2(j_2 - j_1)}) - \{ \beta^{\psi} & or & 0 \} = 0 \\
	\alpha^{(i_2-i_1)}  = \{ \beta^{\psi_1} & or & 0\}
	\eean
	for some $0 \leq \psi_1 < 2^{\ell} - 1$ which is not true as $0 < i_2 - i_1 < \frac{q-1}{2^{\ell}- 1}$ because $0 < i_j < \frac{q-1}{2^{\ell} - 1}$ (due to condition $\frac{k+2}{r} < \frac{q-1}{2^{\ell} - 1}$ given in the theorem). In otherwords both sides of the equation belong to different cosets which can't be true.
	Hence $H_3$ is non - singular.
	
	The non singularity of $H_4$ can be checked similarly.	
\end{proof}

\subsection{Construction of $(r,\delta,2)$ MR code with field size of $O(n)$ based on \cite{Blaum_3}} \label{sec:rdelta2MR_ch7} 

In \cite{Blaum_3}, the authors provide a construction for an $(r,\delta,2)$ MR code (the code is referred to as a partial MDS code in \cite{Blaum_3}).  We present a modification of this construction here.  The modification essentially amounts to a different choice of finite-field elements in the construction of the parity check matrix given in \cite{Blaum_3} for the partial MDS code.  The modified parity-check matrix is provided below.
\begin{thm} \label{thm:Mod_bluam_ch7}
	Let $m,r,\delta$ be positive integers. Let $r'=r+\delta-1$. Let,
	
	\bea
	H  & =& 
	\begin{pmatrix}
		H_0 & 0 & \cdots & 0 \\
		0 & H_0 & \cdots & 0 \\
		\vdots & \vdots & \vdots & \vdots \\
		0 & 0 & \cdots & H_0 \\
		H_1 & H_2 & \cdots & H_m
	\end{pmatrix}, \label{Hmatrix_ch7}
	\eea
	where 
	\bea H_j  & = & 
	\begin{pmatrix}
		1 & \beta^{\delta} & \beta^{2\delta} & \cdots & \beta^{(r')\delta} \\
		\alpha^{j-1} & \alpha^{j-1} \beta^{-1} & \alpha^{j-1} \beta^{-2} & \cdots & \alpha^{j-1}\beta^{-(r')} \\
	\end{pmatrix}, \label{Hjmatrix_ch7}
	\eea
	and 
	\bea H_0  & = & 
	\begin{pmatrix}
		1 & 1 & 1 & \cdots & 1 \\
		1 & \beta^{1} & \beta^{2} & \cdots & \beta^{r'} \\
		1 & \beta^{2} & \beta^{4} & \cdots & \beta^{2r'} \\
		\vdots & \vdots & \vdots & \vdots & \vdots \\
		1 & \beta^{\delta-1} & \beta^{2(\delta-1)} & \cdots & \beta^{r'(\delta-1)} 
	\end{pmatrix}. \label{H0matrix_ch7}
	\eea
	where, $\alpha$ is a primitive element of a finite field $\mathbb{F}_q$ and $\beta = \alpha^{\frac{q-1}{\psi}}$ is a $\psi$th root of unity for any $\psi \geq r'+1$ and hence $\psi$ divides $q-1$.
	The code $\mathcal{C}$ with parity check matrix $H$ is a $(r,\delta,2)$ MR code over \fq\ with block length $n=m(r'+1)=m(r+\delta)$ and dimension $k=mr-2$ with $q$ satisfying $q-1 \geq \psi m$ $\approx$ block length.
\end{thm}
\begin{proof}
	Using the derivation of closed-form expression for the determinant given in Lemma $3$ in \cite{Blaum_3}, it can be seen that the code $\mathcal{C}$ is an $(r,\delta,2)$ MR code.
\end{proof}
Note that the field size $q$ is dependent on $\delta$ only through block length and the code $\mathcal{C}$ is an MR code for any $q$ a prime power such that $\psi | (q-1)$ and $q-1 \geq \psi m$ $\approx$ block length  whereas the construction in \cite{Blaum_3} has field size $\geq m((\delta+1)(r-1)+1)$ $\approx$ $\delta \times$ block length. The construction given in \cite{Cal_Koy} has field size of $\geq (r+\delta)^{mr}$ and the construction given in \cite{Gab_Yak_Bla_Sie} has field size of $\geq \max((r+\delta)^{\delta+2}m,(r+\delta)^{2(2+\delta)}) \geq (r+\delta)^{\delta+1}n$.

\section{Summary and Contributions} \label{sec:Sum_ch7}
In this chapter, we introduced a class of codes called PMR codes and presented a general structure of parity check matrix of such codes and used it to give explicit low field size, high rate construction of such codes. We also presented a possible generalization of our construction of PMR codes. We then presented three constructions of MR codes with low field size. The first construction is a construction of $(2,1,s)$ MR code which has low field size compared to existing constructions in literature and the rest of the two constructions correspond to MR codes with $s=2$. The last two constructions had field size of $O(n)$.

\begin{appendices}
	\chapter{Proof of Theorem \ref{thm:neat5_ch7}} \label{sec:MR_Proofs_ch7}
	
	\begin{proof}
		The code $\mathcal{C}$ has (optimum) maximum possible minimum distance with all symbol locality $r=2$ \cite{TamBar_LRC} for the given $n=q-1,k=2D+1,r=2$. Hence it can be seen that puncturing $\mathcal{C}$ at co-ordinates corresponding to any number of cosets in $\{A_i: 1 \leq i \leq m\}$ (puncturing local codes) without changing $k$ will maintain the optimum minimum distance.
		
		Let $\mathbb{F}$ be the algebraic closure of $\mathbb{F}_q$. Throughout the proof whenever we say a pattern $e$ or just $e$, it refers to an admissible puncturing pattern for an $[N,k]$ code (the code will be clear from the context) with all symbol locality $r$. Throughout the discussion any $[N,k]$  code referred to are polynomial evaluation codes and we assume that the set of all finite field elements (evaluation positions) say $q_1,...,q_{N}$ at which we evaluate a message polynomial to form the $[N,k]$ code to be ordered. If $i \in e$, it is also used to refer to $q_i$. Hence $e$ is used to indicate both indices of the co-ordinates as well as corresponding finite field elements at which we evaluate.
		
		\paragraph{Maximal Recoverability:}
		Let $\ell=\frac{N}{3}$.\\
		We denote an encoding (message) polynomial of $\mathcal{C}$ by $f(x)$ and let $deg(f) = 3D$.
		Let $H$ denote the cyclic group of cube roots of unity in $\mathbb{F}_q$.
		Let $\alpha$ be a primitive element in $\mathbb{F}_q$.
		If $\{ X_1,...X_{3D} \} \subset \mathbb{F}$ are the roots of $f(x)$ then it must satisfy (because co-efficients of monomials $x^i$ in $f(x)$ such that $i \text{ mod } (r+1) = r \text{ mod } (r+1)=2 \text{ mod } 3$ are $0$):
		\bean
		\sigma_1(X_1,...,X_{3D}) = 0 \\
		\sigma_4(X_1,...,X_{3D}) = 0 \\
		\vdots \\
		\sigma_{1+3(D-1)}(X_1,...,X_{3D}) = 0
		\eean
		where $\sigma_i$ refers to the $i$th elementary symmetric function.
		Lets denote the above set of conditions based on elementary symmetric functions on $X_1,...,X_{3D}$ by $R(D)$.
		
		If we have a MR code with block length $N$ and dimension $k=2D+1$ based on the theorem we are proving now (assuming the theorem is true ) and let $H_1,...H_{\ell}$ be the cosets of cube roots of unity chosen for the evaluation positions for forming the codeword of the $[N,k]$ MR code and if we puncture this $[N,k]$ code by a pattern $e$ then for the resulting $[N-\ell,k]$ (assuming $k$ does not change after puncturing) code to be MDS we need $d_{min} = N-\ell-k+1 = N-2D-\ell$. Based on the degree of $f(x)$ (degree of $f(x)$ is $\leq 3D$), we know that $d_{min} \geq N-\ell-deg(f) \geq  N-\ell-3D$. Hence out of $deg(f)$ roots of $f(x)$, we want atmost $2D$ distinct roots to lie in $\cup_{i=1}^{\ell} (H_{i} \setminus \{e_{i}\})$ for any admissible $e$. In other words, its enough if we choose $\ell$ cosets of cube roots of unity $H_1,...H_{\ell}$ such that for any $\{ X_1,...,X_{3D} \} \subset \mathbb{F}$ which satisfies the condition $R(D)$, atmost only $2D$ distinct elements out of them will lie in $\cup_{i=1}^{\ell} (H_{i} \setminus \{e_{i}\})$ for any $e$. 
		
		Note that there is another set of conditions like $R(D)$, if $deg(f)$ is of the form $3D-2$. If $\{X_1,...,X_{3D-2}\} \subset \mathbb{F}$ are the roots of $f$ then:
		\bean
		\sigma_2(X_1,...,X_{3D-2}) = 0 \\
		\sigma_5(X_1,...,X_{3D-2}) = 0 \\
		\vdots \\
		\sigma_{3(D-1)-1}(X_1,...,X_{3D-2}) = 0
		\eean
		where $\sigma_i$ refers to the $i$th elementary symmetric function. We call the above set of conditions $Q(D)$.
		Similar to above, we need to make sure that for any $\{ X_1,...,X_{3D-2} \} \subset \mathbb{F}$ which satisfies the condition $Q(D)$, atmost only $2D$ distinct elements out of them will lie in $\cup_{i=1}^{\ell} (H_{i} \setminus \{e_{i}\})$ for any $e$.
		
		We repeat the following procedure parallely for the set of conditions represented by $Q(D)$ as well. For brevity we only explain choosing cosets based on satisfying the above set of conditions based on $R(D)$. At the end of it, we account for these extra set of conditions represented by $Q(D)$ by multiplying by a factor of $2$ in the resulting constraint on field size. If $deg(f) \in \{3j-2,3j\}$ then we multiply by appropriate power of $x$ to bring the $deg(f)$ to the form $3D-2,3D$ and make sure atmost $2D$ roots of $f$ lie inside $\cup_{i=1}^{\ell} (H_{i} \setminus \{e_{i}\})$ for any $e$.
		
		Note that this condition will also ensure that the dimension of a $N-\ell$ length punctured code obtained by puncturing the $[N,k]$ code by a pattern $e$ is $k$ for any $e$. If not there are 2 distinct non zero message polynomials $f_1(x),f_2(x)$ which after evaluating at $\ell$ cosets (evaluation positions) of the $[N,k]$ code yields the same codeword after puncturing by a pattern $e$ to $N-\ell$ length. This means $f_1-f_2$ is another non zero message or evaluation polynomial with $N-\ell$ zeros in the chosen $\ell$ cosets but by the condition of choosing cosets mentioned previously (as roots of $f_1-f_2$ satisfies $R(D)$ or $Q(D)$) there can be atmost $2D$ distinct zeros in the $N-\ell$ evaluation positions obtained after puncturing by $e$. This is a contradiction as $N-\ell=\frac{2N}{3}>2D$ (by the condition $\frac{2D}{N}<\frac{2}{3}$ given in the theorem). Hence if we choose $\ell$ cosets such that for any pattern $e$ and any $2D$ distinct elements $X_1,..,X_{2D}$ from the $\ell$ cosets after puncturing by $e$, none of $X_{2D+1},..,X_{3D}$ from $\mathbb{F}$ such that $X_1,..X_{3D}$ satisfies $R(D)$   which are distinct from $X_1,...,X_{2D}$ lie in the chosen $\ell$ cosets after puncturing by $e$ then we are done. There is an equivalent set of conditions based on $Q(D)$. We do the following procedures parallely for $Q(D)$ as well but we skip the description.
		
		\begin{prop}
			Let $S$ be a set of $3A$ elements from $\mathbb{F}$ satisfying $R(A)$ and $S$ contains $\alpha^i H$ for some $ i$ then $S-\alpha^i H$ (set difference of the sets $S$ and $\alpha^i H$) satisfies $R(A-1)$.
		\end{prop}
		
		\begin{proof}
			Since $S$ satisfies $R(A)$, this implies $\sigma_{1+3(i-1)}(S) = 0$ for $ i=1,..,A$. 
			\bean
			\sigma_{1+3(i-1)}(S) = \sigma_{3}(\alpha^i H) \times \sigma_{1+3(i-1)-3}(S- \alpha^i H) +
			\sigma_{2}(\alpha^i H) \times \sigma_{1+3(i-1)-2}(S- \alpha^i H) + \\ \sigma_{1}(\alpha^i H) \times \sigma_{1+3(i-1)-1}(S- \alpha^i H) + \sigma_{1+3(i-1)}(S- \alpha^i H).
			\eean
			\bean
			\sigma_{3}(\alpha^i H) = a, \sigma_{2}(\alpha^i H) = 0, \sigma_{1}(\alpha^i H) = 0,\ \ \text{ for some } \ \ a \neq 0.
			\eean
			Hence,
			\bea
			\sigma_{1+3(i-1)}(S) = & a \sigma_{1+3(i-1)-3}(S- \alpha^i H) +
			& \sigma_{1+3(i-1)}(S- \alpha^i H). \label{eq_red_mr}
			\eea
			For $ i=A$, \ \  $\sigma_{1+3(A-1)}(S- \alpha^i H) = 0$ as $S- \alpha^i H$ has only $3(A-1)$ elements.\\
			Hence, 
			\bean
			\sigma_{1+3(A-1)}(S) = a \sigma_{1+3(A-1)-3}(S- \alpha^i H).
			\eean
			Hence,
			\bean
			\sigma_{1+3(A-1)}(S) = 0 \implies \sigma_{1+3(A-1)-3}(S- \alpha^i H) = 0.
			\eean
			For $i=A-1$,
			\bean
			\sigma_{1+3(A-2)}(S) = & a \sigma_{1+3(A-2)-3}(S- \alpha^i H) +
			&    \sigma_{1+3(A-2)}(S- \alpha^i H).
			\eean
			Since, $\sigma_{1+3(A-2)}(S) = 0$ and $\sigma_{1+3(A-2)}(S- \alpha^i H) = 0$, this implies that $\sigma_{1+3(A-2)-3}(S- \alpha^i H) = 0$.
			
			By induction, if we assume, $\sigma_{1+3(i-1)}(S- \alpha^i H) = 0$ then by equation \eqref{eq_red_mr} since $\sigma_{1+3(i-1)}(S)=0$, we have : $\sigma_{1+3(i-1)-3}(S- \alpha^i H) = 0$ ($i=A$ is the starting condition of the induction which we already proved).\\
			Hence $S- \alpha^i H$ satisfies $R(A-1)$.
		\end{proof}

		$\it{Claim}$:\\
		Its enough to choose $\ell$ cosets  such that for any $(A \leq D)$ and any $X_1,...,X_{2A}$(contained in the chosen $\ell$ cosets) which are distinct and contains atmost one element from each coset, none of the $X_{2A+1},..,X_{3A}$ from $\mathbb{F}$ such that $X_1,...,X_{3A}$ satisfies $R(A)$, which are distinct from $X_1,...,X_{2A}$ lies in the chosen $\ell$ cosets after puncturing by $e$ for any $e$ disjoint from $X_1,...,X_{2A}$.
		
		$\it{Proof}$:\\
		This is because if $X_1,...,X_{3D}$ satisfying $R(D)$ contains at least 2 element from some coset $\alpha^i H$ for some $i$, since the polynomial $f_1(x)=(x-X_1)...(x-X_{3D})$ restricted to any coset is a degree $1$ polynomial, the third element from coset is also a root of $f_1$. Hence the entire coset is contained in $X_1,...,X_{3D}$ and by similar reasoning $X_1,...,X_{3D}$ can be written as $X_1,...,X_{3(D-j)} \cup \alpha^{i_1} H  \cup ... \alpha^{i_j} H$ for some $i_1,...,i_j$ where $X_1,...,X_{3(D-j)}$ contains at most one element from
		each coset and satisfies $R(D-j)$ by proposition $1$.
		
		Now by the property of the chosen cosets, we have that for any distinct $X_1,...,X_{2(D-j)}$ from the chosen $\ell$ cosets containing atmost one element from each coset, any of $X_{2(D-j)+1},.$ $..,X_{3(D-j)}$ which are distinct from $X_1,...,X_{2(D-j)}$ such that $X_1,...,X_{3(D-j)}$ satisfies $R(D-j)$ will not lie inside the chosen cosets after puncturing by $e$ for any $e$ such that $e \cap \{ X_1,...,X_{2(D-j)} \} = \emptyset$. Wlog this implies the chosen cosets after puncturing by any $e$ can contain atmost only (writing only distinct elements) $X_1,...,X_{2(D-j)} \cup \alpha^{i_1} H -\{e_{i_1}\} \cup ... \alpha^{i_j}H-\{e_{i_j}\}  $ among the $3D$ elements $\{X_1,..,X_{3D}\}$. Hence there can be atmost $2(D-j)+3j-j = 2D$ distinct roots out of $3D$ roots inside the chosen cosets after puncturing by any $e$. Hence we are done.
		
		From here we term a set of $\ell$ cosets satisfying the above claim, to be satisfying $R_1(\ell)$.\\
		We are going put another set of conditions $R_2(\ell)$ on a set of $\ell$ cosets. The necessity of this condition will be clear in the proof.\\
		$R_2(\ell) :$\\
		A given set of $\ell$ cosets, is said to satisfy condition $R_2(\ell)$ if, \\
		For any $1 \leq A \leq D$ and any $X_1,...,X_{2A}$(contained in the given $\ell$ cosets) which are distinct and contains atmost one element from each of the $\ell$ cosets, the matrix $P(A)$ given by
		
		
		\bean
		P(A)=\begin{pmatrix}
			1 & 0 & \cdots & \cdots & 0 \\
			\sigma_{3}(S) & \sigma_{2}(S) & \cdots & \cdots & \sigma_{3+1-A}(S) \\
			\vdots  & \vdots & \vdots  & \ddots & \vdots  \\
			\sigma_{3(i-1)}(S) & \sigma_{3(i-1)-1}(S) & \cdots & \cdots & \sigma_{3(i-1)+1-A}(S) \\
			\vdots  & \vdots & \vdots  & \ddots & \vdots  \\
			0 & \cdots & \sigma_{2A}(S) & \sigma_{2A-1}(S) & \sigma_{2(A-1)}(S) \\
		\end{pmatrix}
		\eean
		is non-singular, where $S=\{X_1,...,X_{2A}\}$.
		
		Furthermore, for any $3 \leq A \leq D$ and any $X_1,...,X_{2A-1}$ (contained in the given $\ell$ cosets) which are distinct and contains at most one element from each of $\ell$ cosets, the matrix $P_1(A)$ given by
		
		\bean
		P_1(A)= \begin{pmatrix}
			1 & 0 & \cdots & \cdots & 0 \\
			\sigma_{3}(S) & \sigma_{2}(S) & \cdots & \cdots & \sigma_{3+1-A}(S) \\
			\vdots  & \vdots & \vdots  & \ddots & \vdots  \\
			\sigma_{3(i-1)}(S) & \sigma_{3(i-1)-1}(S) & \cdots & \cdots & \sigma_{3(i-1)+1-A}(S) \\
			\vdots  & \vdots & \vdots  & \ddots & \vdots  \\
			0 & \cdots & 0 & \sigma_{2A-1}(S) & \sigma_{2(A-1)}(S) \\
		\end{pmatrix} 
		\eean
		
		is non-singular, where $S=\{X_1,...,X_{2A-1}\}$.
		
		From here on we proceed to find a set of $\ell$ cosets satisfying $R_1(\ell)$ and $R_2(\ell)$. We proceed by choosing $1$ new coset at each step inductively until we choose the required set of $\ell$ cosets.
		
		At each step we select and add one coset to our list and throw away a collection of cosets from the cosets not chosen. The cosets are thrown away such that it is straight forward to pick a coset from the cosets not thrown away so that the collection of chosen cosets satisfy $R_1(\ell)$ and $R_2(\ell)$. Let the cosets chosen upto $i^{th}$ step (including step $i$) be $G(i)$ and the cosets thrown away upto $i^{th}$ step be $T(i)$ and let the total collection of cosets in the field $\mathbb{F}_q$ be $W=\{H,\alpha H,...\}$.
		\ben
		\item Step 1: The first coset is chosen to be any coset. Hence $G(1)$ consists of just the coset chosen. We don't throw away any cosets at this step. Hence $T(1)$ is empty. $G(1)$ satisfies $R_1(1)$ and $R_2(1)$ trivially.
		
		\item Step 2: The second coset is also chosen to be any coset from $W-G(1)$. Hence $G(2)$ consists of the $2$ chosen cosets.
		\ben
		\item $R_1(2)$:\\
		For $A=1$, and for any $2A=2$ distinct elements $X_1,X_2$, one from each coset in $G(2)$, any $X_3$ such that $\sigma_1(X_1,X_2,X_3) = 0$ cannot be distinct from $X_1,X_2$ and lie in any of the cosets in $G(2)$. If it does, wlog let $X_1$ and $X_3$ lie in same coset which is in $G(2)$ then $X_2 = -(X_1 + X_3)$ but every coset is a coset of cube roots of unity. Hence $X + X_1 + X_3 = 0$ where X is the third element from the same coset as $X_1,X_3$. Hence $X = -(X_1 + X_3)$ which implies $X = X_2$ but X is in the same coset as $X_1,X_3$ and $X_2$ is in the other coset in $G(2)$. Hence a contradiction.
		For $A \geq 2$, $2A \geq 4$, we need to pick 4 distinct elements, from distinct cosets but there are only 2 cosets in $G(2)$. Hence $R_1(2)$ is satisfied.
		\item $R_2(2)$:\\
		For $A=1$, $P_1(1)=[1],P_2(1)=[1]$, hence non-singular.\\
		For $A \geq 2$, we need to pick $2A \geq 4$ and $2A-1 \geq 3$ distinct elements from distinct cosets but there are only 2 cosets. Hence $R_2(2)$ is satisfied.
		\item $T(2)$:\\
		For every two distinct elements $X_1,X_2$ chosen one from each of the 2 cosets in $G(2)$, find the third element $X_3$ such that $\sigma_1(X_1,X_2,X_3) = 0$ and throw away the coset in $W-(G(2))$ which contains it. Since $G(2)$ satisfies $R_1(2)$, $X_3$ will either not lie in any coset in $G(2)$ or won't be distinct from $X_1,X_2$. In the first case, we throw the coset and in the latter case, we don't do anything. There are 3x3=9 possible summations $X_1+X_2$ but if $X_1+X_2+X_3 = 0$ then $\theta(X_1+X_2+X_3) = 0$ and $\theta X_3$ is in the same coset as $X_3$ for any cube root of unity $\theta$. Hence solutions for 9 possible summations $X_3 = -(X_1+X_2)$ lie in atmost 3 cosets and we throw away these 3 cosets. 
		\een
		\item Step $i$ (Assumng that a new coset is chosen at step $i$ and $G(i)$ satifying $R_1(i),R_2(i)$ is formed, we are now showing the cosets we are throwing away into $T(i)$ at step $i$): Let $i \geq 2D$ and assume we have $G(i)$ satifying $R_1(i),R_2(i)$.
		\paragraph{$T(i)$:}
		\ben
		\item For every $A \leq D$,
		Choose $2A$ cosets (say $H_1,...,H_{2A}$) from the cosets in $G(i)$, and choose $X_1,$$...$$,X_{2A}$ one from each of these $2A$ cosets i.e., say $X_i \in H_i$, now find the set of all  $X_{2A+1},$$...$$,X_{3A}$ from $\mathbb{F}$ such that $\sigma_1(X_1,...,X_{3A}) = 0,..., \sigma_{1+3(A-1)}(X_1,...,X_{3A}) = 0$ and throw away all the cosets in which $X_{2A+1},$$...$ $,X_{3A}$ lies. Since $G(i)$ satisfies $R_1(i)$, the elements in $X_{2A+1},...,X_{3A}$ will either be not distinct from $X_1,...,X_{2A}$ or elements in $X_{2A+1},...,X_{3A}$ which are distinct from $X_1,...,X_{2A}$ will lie in cosets outside $G(i)$. In the first case we do not do anything and in the latter case, we throw away any coset containing any of $X_{2A+1},...,X_{3A}$. \\
		To find the number of solutions $X_{2A+1},...,X_{3A}$ such that \\
		$\sigma_1(X_1,...,X_{3A}) = 0,..., \sigma_{1+3(A-1)}(X_1,...,X_{3A}) = 0$, we solve for $X_{2A+1},...,X_{3A}$ given $X_{1},...,X_{2A}$.
		It can be seen that to satisfy $\sigma_1(X_1,...,X_{3A}) = 0$ ,..., $\sigma_{1+3(A-1)}(X_1,...,X_{3A}) = 0$, $\sigma_1(X_{2A+1},...,X_{3A}),...,$$\sigma_{A}(X_{2A+1},...,X_{3A})$ has to satsify a linear equation of the form 
		\bean
		P(A)[\sigma_1(X_{2A+1},...,X_{3A}),...,\sigma_{A}(X_{2A+1},...,X_{3A})]^T = \\
		-[\sigma_1(X_{1},...,X_{2A}),...,\sigma_{1+3(A-1)}(X_{1},...,X_{2A})]^T.
		\eean
		Since $G(i)$ satifies $R_2(i)$, $P(A)$ is non singular and there is a unique solution, for
		$[\sigma_1(X_{2A+1},...,X_{3A}),...,\sigma_{A}(X_{2A+1},...,X_{3A})]$ which implies a unique solution for $X_{2A+1},$... $,X_{3A}$ because $\sigma_1(X_{2A+1},...,X_{3A}),...,\sigma_{A}(X_{2A+1},...,X_{3A})$ are the co-efficients of the polynomial with roots exactly equal to $X_{2A+1},...,X_{3A}$. Hence for a given distinct $X_1,...,X_{2A}$, from distinct cosets, there is a unique solution for $X_{2A+1},...,X_{3A}$ such that $\sigma_1(X_1,...,X_{3A}) = 0,$ $...,\sigma_{1+3(A-1)}(X_1,...,X_{3A}) = 0$. Hence its enough to throw away $j \leq A$ cosets which contain $X_{2A+1},...,X_{3A}$ (unique solution).
		
		The above procedure is done for every choice of $2A$ cosets from cosets in $G(i)$ and every choice of $X_1,...,X_{2A}$ from the chosen $2A$ cosets such that one element is chosen from each coset.
		Hence the total number of cosets thrown away are atmost ${ |G(i)| \choose 2A}  3^{2A}  A$ but if $X_1,...,X_{2A}$, $X_{2A+1},...,X_{3A}$ put together satisfies $R(A)$ then $\theta X_1,$ ... $,\theta X_{2A}$, $\theta X_{2A+1},$ ... $,\theta X_{3A}$ (which doesn't change the cosets in which  $X_{2A+1},..,X_{3A}$ lie for any cube root of unity $\theta$) also satisfies $R(A)$ and this choice is unique as seen before. Hence out of $3^{2A}$ choices for $X_1,...,X_{2A}$ from a given chosen $2A$ cosets, its enough to throw away cosets for $\frac{3^{2A}}{3}$ choices of $X_1,...,X_{2A}$.
		Hence the total number of cosets thrown away are atmost ${|G(i)| \choose 2A}  3^{2A-1}  A$.
		\item For every $ 3 \leq A \leq D$,
		Choose $2A-1$ cosets (say $H_1,...,H_{2A-1}$) from the cosets in $G(i)$, and choose $X_1,...,X_{2A-1}$ one from each of these $2A-1$ cosets i.e., say $X_i \in H_i$, now find the set of all $X_{2A}$ from $\mathbb{F}$ such that $P(A)$ is singular. This $X_{2A}$ cannot be in any coset in $G(i)$ which does not contain $X_1,...,X_{2A-1}$ as $G(i)$ satisfies $R_2(i)$. If $X_{2A}$ lies in the coset which contains any of $X_1,...,X_{2A-1}$, then we do not do anything. If $X_{2A}$ lies outside $G(i)$, we throw away the coset.
		To find the number of solutions of $X_{2A}$ for a given $X_1,...,X_{2A-1}$ such that $P(A)$ is singular,
		let $S_1 = \{X_1,...,X_{2A-1} \}$ and $S= \{ X_1,...,X_{2A-1},X_{2A} \}$.
		\bean
		P(A)&=& \begin{pmatrix}
			1 & 0 & \cdots & \cdots & 0 \\
			\sigma_{3}(S) & \sigma_{2}(S) & \cdots & \cdots & \sigma_{3+1-A}(S) \\
			\vdots  & \vdots & \vdots  & \ddots & \vdots  \\
			\sigma_{3(i-1)}(S) & \sigma_{3(i-1)-1}(S) & \cdots & \cdots & \sigma_{3(i-1)+1-A}(S) \\
			\vdots  & \vdots & \vdots  & \ddots & \vdots  \\
			0 & \cdots & \sigma_{2A}(S) & \sigma_{2A-1}(S) & \sigma_{2(A-1)}(S) \\
		\end{pmatrix} 
		\eean
		\bean
		&=&\scalemath{0.4}{\begin{pmatrix}
				1 & 0 & \cdots & \cdots & 0 \\
				\sigma_{3}(S_1)+\sigma_{2}(S_1)X_{2A} & \sigma_{2}(S_1)+\sigma_{1}(S_1)X_{2A} & \cdots & \cdots & \sigma_{3+1-A}(S_1)+\sigma_{3+1-A-1}(S_1)X_{2A} \\
				\vdots  & \vdots & \vdots  & \ddots & \vdots  \\
				\sigma_{3(i-1)}(S_1)+\sigma_{3(i-1)-1}(S_1)X_{2A} & \sigma_{3(i-1)-1}(S_1)+ \sigma_{3(i-1)-1-1}(S_1) X_{2A} & \cdots & \cdots & \sigma_{3(i-1)+1-A}(S_1)+\sigma_{3(i-1)+1-A-1}(S_1) X_{2A} \\
				\vdots  & \vdots & \vdots  & \ddots & \vdots  \\
				0 & \cdots & \sigma_{2A-1}(S_1) X_{2A} & \sigma_{2A-1}(S_1)+\sigma_{2A-1-1}(S_1) X_{2A} & \sigma_{2(A-1)}(S_1)+\sigma_{2(A-1)-1}(S_1) X_{2A} \\
			\end{pmatrix}.}
		\eean
		The determinant of above matrix $P(A)$ can be seen as a polynomial in $X_{2A}$ and its degree is atmost $A-1$. The constant term of this polynomial is the determinant of following matrix:
		\bean
		\begin{pmatrix}
			1 & 0 & \cdots & \cdots & 0 \\
			\sigma_{3}(S_1) & \sigma_{2}(S_1) & \cdots & \cdots & \sigma_{3+1-A}(S_1) \\
			\vdots  & \vdots & \vdots  & \ddots & \vdots  \\
			\sigma_{3(i-1)}(S_1) & \sigma_{3(i-1)-1}(S_1) & \cdots & \cdots & \sigma_{3(i-1)+1-A}(S_1) \\
			\vdots  & \vdots & \vdots  & \ddots & \vdots  \\
			0 & \cdots & 0 & \sigma_{2A-1}(S_1) & \sigma_{2(A-1)}(S_1) \\
		\end{pmatrix}
		\eean
		The above matrix is equal to $P_1(A)$ and is non-singular for $A \geq 3$ since $G(i)$ satisfies $R_2(i)$.  Hence the $det(P(A))$ as a polynomial in $X_{2A}$ is a non-zero polynomial (has a non zero constant term), and since its degree is atmost $A-1$, it can have atmost $A-1$ solutions for $X_{2A}$. Hence its enough to throw away these $j \leq A-1$ cosets containing these $A-1$ solutions.\\
		The above procedure is done for every choice of $2A-1$ cosets from cosets in $G(i)$ and every choice of $X_1,...,X_{2A-1}$ from the chosen $2A-1$ cosets such that one element is chosen from each coset.
		
		The number of cosets thrown away are atmost: $ { |G(i)| \choose 2A-1} 3^{2A-1} (A-1)$. It can be seen that $det(P(A))$ is a homogenous polynomial in $X_1,$ $...$ $,X_{2A}$ and hence as before if for $X_1,$...$,X_{2A-1}$ $,X_{2A}$, $det(P(A)) = 0$ then for $\theta(X_1,...,X_{2A-1},X_{2A})$ also $det(P(A)) = 0$. Hence its enough to throw away atmost: ${|G(i)| \choose 2A-1} 3^{2A-2} (A-1)$ cosets.
		
		For $A=1$, $P(A) = [1]$ which is trivially non- singular and we don't do anything. For $A=2$, choose 3 cosets from $G(i)$ and choose distinct $X_1,X_2,X_3$ one from each of these distinct cosets, now find the set of all $X_{4}$ such that $P(A)$ is singular. This $X_{4}$ can't be in any coset in $G(i)$ which does not contain $X_1,...,X_{3}$ as $G(i)$ satisfies $R_2(i)$. If $X_{4}$ lies in the coset which contains any of $X_1,...,X_{3}$, then we don't do anything. If $X_{4}$ lies outside $G(i)$, we throw the coset.
		To find the number of solutions of $X_{4}$ for a given $X_1,...,X_{3}$ such that $P(A)$ is singular,
		\bean
		det(P(A)) = \sigma_2(X_1,X_2,X_3,X_4) = \sigma_2(X_1,X_2,X_3) + X_4 \sigma_1(X_1,X_2,X_3) 
		\eean
		Given the chosen $X_1,X_2,X_3$, the above expression for $det(P(A))$ can be seen as a linear expression in $X_4$. if $\sigma_2(X_1,X_2,X_3)=0,\sigma_1(X_1,X_2,X_3) = 0$ then $(X-X_1)(X-X_2)(X-X_3) = X^3 - X_1 X_2 X_3 = X^3 - \gamma$. Here $X_1,X_2,X_3$ constitutes the solution set for $X^3 = \gamma$ but $X_1 H$  also constitutes 3 solutions for the equation $X^3 = \gamma$ but there can be atmost 3 solutions for the equation $X^3 = \gamma$. Hence $X_1 H = \{X_1,X_2,X_3\}$ which implies they all belong to same coset which is a contradiction. Hence either $\sigma_2(X_1,X_2,X_3) \neq 0$ or $\sigma_1(X_1,X_2,X_3) \neq 0$ which implies $det(P(A))$ is a non zero polynomial in $X_4$ with degree atmost $1$. Hence we can find the solution and throw away the coset containing it.
		
		The number of cosets thrown are atmost: ${|G(i)| \choose 3} 3^{3} $ but by similar argument as before we can see that the number of cosets thrown are atmost: ${ |G(i)| \choose 3} 3^{2} $.
		\item For every $3 \leq A \leq D$,
		Choose $2A-2$ cosets (say $H_1,...,H_{2A-2}$) from the cosets in $G(i)$, and choose $X_1,...,X_{2A-2}$ one from each of these $2A-2$ cosets i.e., say $X_i \in H_i$, now find the set of all $X_{2A-1}$ from $\mathbb{F}$ such that $P_1(A)$ is singular. This $X_{2A-1}$ cannot be in any coset in $G(i)$ which does not contain $X_1,...,X_{2A-2}$ as $G(i)$ satisfies $R_2(i)$. If $X_{2A-1}$ lies in the coset which contains any of $X_1,...,X_{2A-2}$, then we don't do anything. If $X_{2A-1}$ lies outside $G(i)$, we throw away the coset.
		
		To find the number of solutions of $X_{2A-1}$ for a given $X_1,...,X_{2A-2}$ such that $P_1(A)$ singular,\\
		let $S_1 = \{ X_1,...,X_{2A-2} \}$ and $S= \{ X_1,...,X_{2A-2},X_{2A-1} \}$.
		\bean
		P_1(A)=	\begin{pmatrix}
			1 & 0 & \cdots & \cdots & 0 \\
			\sigma_{3}(S) & \sigma_{2}(S) & \cdots & \cdots & \sigma_{3+1-A}(S) \\
			\vdots  & \vdots & \vdots  & \ddots & \vdots  \\
			\sigma_{3(i-1)}(S) & \sigma_{3(i-1)-1}(S) & \cdots & \cdots & \sigma_{3(i-1)+1-A}(S) \\
			\vdots  & \vdots & \vdots  & \ddots & \vdots  \\
			0 & \cdots & 0 & \sigma_{2A-1}(S) & \sigma_{2(A-1)}(S) \\
		\end{pmatrix}
		\eean
		\bean
		=\scalemath{0.5}{
			\begin{pmatrix}
				1 & 0 & \cdots & \cdots & 0 \\
				\sigma_{3}(S_1)+\sigma_{2}(S_1)X_{2A-1} & \sigma_{2}(S_1)+\sigma_{1}(S_1)X_{2A-1} & \cdots & \cdots & \sigma_{3+1-A}(S_1)+\sigma_{3+1-A-1}(S_1)X_{2A-1} \\
				\vdots  & \vdots & \vdots  & \ddots & \vdots  \\
				\sigma_{3(i-1)}(S_1)+\sigma_{3(i-1)-1}(S_1)X_{2A-1} & \sigma_{3(i-1)-1}(S_1)+ \sigma_{3(i-1)-1-1}(S_1) X_{2A-1} & \cdots & \cdots & \sigma_{3(i-1)+1-A}(S_1)+\sigma_{3(i-1)+1-A-1}(S_1) X_{2A-1} \\
				\vdots  & \vdots & \vdots  & \ddots & \vdots  \\
				0 & \cdots & 0 & \sigma_{2A-1-1}(S_1) X_{2A-1} & \sigma_{2(A-1)}(S_1)+\sigma_{2(A-1)-1}(S_1) X_{2A-1} \\
			\end{pmatrix}.}
		\eean
		The determinant of above matrix $P_1(A)$ can be seen as a polynomial in $X_{2A-1}$ and its degree is atmost $A-1$. The constant term of this polynomial is the determinant of following matrix:
		\bea
		\begin{pmatrix}
			1 & 0 & \cdots & \cdots & 0 \\
			\sigma_{3}(S_1) & \sigma_{2}(S_1) & \cdots & \cdots & \sigma_{3+1-A}(S_1) \\
			\vdots  & \vdots & \vdots  & \ddots & \vdots  \\
			\sigma_{3(i-1)}(S_1) & \sigma_{3(i-1)-1}(S_1) & \cdots & \cdots & \sigma_{3(i-1)+1-A}(S_1) \\
			\vdots  & \vdots & \vdots  & \ddots & \vdots  \\
			0 & \cdots & 0 & 0 & \sigma_{2(A-1)}(S_1) \\
		\end{pmatrix} \label{m_mr}
		\eea
		Now $\sigma_{2(A-1)}(S_1) \neq 0$ (because this is just the product of $2(A-1)$ non zero elements) , since the determinant of the matrix:
		\bean
		\begin{pmatrix}
			1 & 0 & \cdots & \cdots & 0 \\
			\sigma_{3}(S_1) & \sigma_{2}(S_1) & \cdots & \cdots & \sigma_{3+1-(A-1)}(S_1) \\
			\vdots  & \vdots & \vdots  & \ddots & \vdots  \\
			\sigma_{3(i-1)}(S_1) & \sigma_{3(i-1)-1}(S_1) & \cdots & \cdots & \sigma_{3(i-1)+1-(A-1)}(S_1) \\
			\vdots  & \vdots & \vdots  & \ddots & \vdots  \\
			0 & \cdots & \sigma_{2(A-1)}(S_1) & \sigma_{2(A-1)-1}(S_1) & \sigma_{2(A-2)}(S_1) \\
		\end{pmatrix}
		\eean
		is non zero as $G(i)$ satisfies $R_2(i)$ and this matrix is equal to $P(A-1)$, we have that the determinant of the matrix mentioned in equation \eqref{m_mr} corresponding to the constant term of the polynomial $det(P_1(A))$ is also non zero. Hence the $det(P_1(A))$ as a polynomial in $X_{2A-1}$ is a non-zero polynomial (has a non zero constant term), and since its degree is atmost $A-1$, it can have atmost $A-1$ solutions for $X_{2A-1}$. Hence its enough to throw away these $j \leq A-1$ cosets containing these $A-1$ solutions.\\
		The above procedure is done for every choice of $2A-2$ cosets from cosets in $G(i)$ and every choice of $X_1,...,X_{2A-2}$ from the chosen $2A-2$ cosets such that one element is chosen from each coset.
		
		The number of cosets thrown away are atmost: $ {|G(i)| \choose 2A-2} 3^{2A-2} (A-1)$. It can be seen that $det(P_1(A))$ is a homogenous polynomial in $X_1,...,X_{2A-1}$ and hence as before if for $X_1,...,X_{2A-2},X_{2A-1}$, $det(P_1(A)) = 0$ then for $\theta(X_1,...,X_{2A-2},X_{2A-1})$ also $det(P_1(A)) = 0$. Hence its enough to throw away atmost: $ { |G(i)| \choose 2A-2} 3^{2A-3} (A-1)$ cosets.
		\een
		\item Step $i+1$ (only showing how to choose a new coset such that $G(i+1)$ satisfies $R_1(i+1),R_2(i+1)$ and the cosets thrown away at this step follows the same procedure indicated in previous point step $i$): Following the previous step, we want to select one more coset to form $G(i+1)$ such that it satisfies $R_1(i+1),R_2(i+1)$ :\\
		Choose any coset (say) $H_1$ from the collection $W- (T(i) \cup G(i))$. Hence $G(i+1) = G(i) \cup \{ H_1 \}$. It can be easily shown that $G(i+1)$ satisfies $R_1(i+1),R_2(i+1)$ using the properties of $T(i)$ and $G(i)$. 
		\ben 
		\item $R_1(i+1)$:\\
		For any $A \leq D$, for any distinct $X_1,...,X_{2A}$ from distinct cosets (say wlog $C_1,...,C_{2A}$ and $X_i \in C_i$, $\forall i \in [2A]$) from $G(i+1)$, find the set of all $X_{2A+1},...,X_{3A}$ such that $X_1,...,X_{3A}$ satisfies $R(A)$. The following argument is for each such $X_1,...,X_{3A}$. If some $X_{2A+j},j > 0$ distinct from $X_1,...,X_{2A}$ lies in a coset in $G(i+1)$, and $H_1 \neq C_i$ for any $ i=1,...,2A$, then this gives a contradiction because $G(i)$ satisfies $R_1(i)$  (gives contradiction if $X_{2A+j}$ is in a coset distinct from $H_1$) and since $H_1 \in W- (T(i) \cup G(i))$, gives contradiction if $X_{2A+j} \in H_1$ as previous step implies $H_1 \in T(i)$. Hence its enough to consider the case where $H_1 = C_i$ for some $i=1,...,2A$. Wlog let $H_1 = C_{2A}$.\\
		If some  $X_{2A+j},j>0$ distinct from $X_1,...,X_{2A}$ lies in a coset in $G(i+1)$ which is different from $C_1,...,C_{2A}$, then $X_1$ ,..., $X_{2A-1}$, $X_{2A+j}$ ($2A$ elements), $X_{2A}$, $X_{2A+1}$,..., $X_{2A+j-1}$, $X_{2A+j+1}$ ,..., $X_{3A}$ form a solution such that $X_1$ ,..., $X_{2A-1}$, $X_{2A+j}$, $X_{2A}$, $X_{2A+1}$ ,..., $X_{2A+j-1}$,$X_{2A+j+1}$ ,..., $X_{3A}$ satisfies $R(A)$. But by previous step we are throwing away all cosets containing $X_{2A}$, $X_{2A+1}$, ..., $X_{2A+j-1}$, $X_{2A+j+1}$, ..., $X_{3A}$ as $X_1$, ..., $X_{2A-1}$, $X_{2A+j}$ are distinct $2A$ elements from distinct cosets in $G(i)$  which implies $H_1 \in T(i)$ (as $X_{2A} \in H_1=C_{2A}$) which is a contradiction to the fact that $H_1 \in W- (T(i) \cup G(i))$.\\
		If some $X_{2A+j},j>0$ distinct from $X_1,...,X_{2A}$ lies in a coset from $C_1,...,C_{2A}$, if $X_{2A+j} \in C_i$, $i \neq 2A$, then there exists $X_{2A+j_1} \in \{ X_{2A+1},...,X_{3A} \} - \{ X_{2A+j}\}$ and $X_{2A+j_1} \neq X_i$ such that $X_{2A+j_1} \in C_i$ (this is because $f_1(x)= (x-X_1)...(x-X_{3A})$ is a degree 1 polynomial when restricted to a coset. Hence if there are 2 distinct zeros of $f_1$ in a coset then $f_1$ restricted to the coset must be 0. Hence the entire coset is subset of roots of $f_1$). Hence $C_i \subset \{X_1,...,X_{3A}\}$. By proposition 1, $\{X_1,...,X_{3A}\}-C_i$ satisfies $R(A-1)$ and $\{X_1,...,X_{2A-1}\}-\{X_i\}$ is a set of distinct $2(A-1)$ elements from distinct cosets in $G(i)$ and hence the cosets in which $\{X_{2A},X_{2A+1},...,X_{3A}\} - \{X_{2A+j},X_{2A+j_1}\}$ lies will be thrown away in previous step and $X_{2A} \in H_1$. Hence $H_1 \in T(i)$ which is a contradiction to the fact $H_1 \in W-(T(i) \cup G(i))$.\\
		If some $X_{2A+j}$ ($j>0$) distinct from $X_1,...,X_{2A}$ lies in a coset from $C_1,...,C_{2A}$, and if $X_{2A+j} \in C_{2A}$, then there exists $X_{2A+j_1} \in \{ X_{2A+1},...,X_{3A} \} - \{ X_{2A+j}\}$ and $X_{2A+j_1} \neq X_{2A}$ such that $X_{2A+j_1} \in C_{2A}$ (for the same as explained before). Hence $C_{2A} \subset \{X_1,...,X_{3A}\}$. By proposition 1, $\{X_1,...,X_{3A}\}-C_{2A}$ satisfies $R(A-1)$. $\{X_1,...,X_{2A-2}\}$ is a set of distinct $2(A-1)$ elements from distinct cosets in $G(i)$ and $X_{2A-1}$ is in a coset other than $C_1,...,C_{2A-2}$ in $G(i)$ and $\{X_1,...,X_{3A}\} - C_{2A}$ satisfies $R(A-1)$ which is a contradiction to the fact $G(i)$ satisfies $R_1(i)$.
		\item $R_2(i+1)$:\\
		For any $A \leq D$, for any distinct $X_1,...,X_{2A-1}$ from distinct cosets (say wlog $C_1,...,C_{2A-1}$ i.e., $X_i \in C_i$,$\forall i \in [2A-1]$) from $G(i+1)$, find the set of all $X_{2A}$ such that $X_1,...,X_{2A}$ makes $P(A)$ singular. The following argument is for each such $X_1,...,X_{2A-1},X_{2A}$. If $X_{2A}$ lie in a coset in $G(i+1)$ and in a coset in $C_1,...,C_{2A-1}$, we don't care as this doesn't violate $R_2(i+1)$. If $X_{2A}$ lie in a coset in $G(i+1)$ and in a coset distinct from in $C_1,...,C_{2A-1}$, it gives a contradiction as $G(i)$ satisfies $R_2(i)$ (gives contradiction if none of $X_1,...,X_{2A}$ lie in $H_1$ as this means we have 2A distinct elements from distinct cosets from $G(i)$ making $P(A)$ singular) and since $H_1 \in W-(T(i) \cup G(i))$, gives contradiction if one of $X_1,...,X_{2A}$ lie in $H_1$ as this would imply $H_1 \in T(i)$.\\
		For any $A \leq D$, for any distinct $X_1,...,X_{2A-2}$ from distinct cosets (say wlog $C_1,...,C_{2A-2}$ i.e., $X_i \in C_i$,$\forall i \in [2A-2]$) from $G(i+1)$, find the set of all $X_{2A-1}$ such that $X_1,...,X_{2A-1}$ makes $P_1(A)$ singular. The following argument is for each such $X_1,...,X_{2A-2},$ $X_{2A-1}$. If $X_{2A-1}$ lie in a coset in $G(i+1)$ and in a coset in $C_1,...,C_{2A-2}$, we don't care as this doesn't violate $R_2(i+1)$. If $X_{2A-1}$ lie in a coset in $G(i+1)$ and in a coset distinct from in $C_1,...,C_{2A-2}$, it gives a contradiction as $G(i)$ satisfies $R_2(i)$ (gives contradiction if none of $X_1,...,X_{2A-1}$ lie in $H_1$ as this means we have $2A-1$ distinct elements from distinct cosets from $G(i)$ making $P_1(A)$ singular) and since $H_1 \in W-(T(i) \cup G(i))$, gives contradiction if one of $X_1,...,X_{2A-1}$ lie in $H_1$ as this would imply $H_1 \in T(i)$.
		\een
		\item  For $i<2D$: The argument for throwing cosets is similar to the above arguments (step $i$) except that we skip the parts where it becomes vacuous. The procedure for selecting new coset to form $G(i)$ and showing that it satisfies $R_1(i)$ and $R_2(i)$ can be done in a similar manner.
		\een
		We repeat the procedure described above in steps $i$ and step $i+1$ until we pick $\ell$ cosets. Note that the set of cosets thrown away at $i^{th}$ step contains the set of cosets thrown away at $(i-1)^{th}$ step.
		
		Hence the total number of cosets thrown until $i^{th}$ step is (from step $i$ described above):
		\bean
		|T(i)| \leq \Sigma_{j=1}^{D} { |G(i)| \choose 2j}  3^{2j-1}  j + \Sigma_{j=3}^{D} { |G(i)| \choose 2j-1} 3^{2j-2} (j-1) + \\
		{|G(i)| \choose 3} 3^{2} + \Sigma_{j=3}^{D} { |G(i)| \choose 2j-2} 3^{2j-3} (j-1)
		\eean
		we can pick $(i+1)^{th}$ coset to form $G(i+1)$ as long as $|T(i)| + |G(i)| < |W|$. Now the procedure for throwing away cosets for satisfying parallel conditions based on $Q(D)$ as mentioned in the beginning is similar to above and hence we throw away atmost $2 |T(i)|$ cosets and satisfy all the necessary conditions. Hence 	we can pick $(i+1)^{th}$ coset to form $G(i+1)$ as long as $2|T(i)| + |G(i)| < |W|$. Hence we can pick $\ell=\frac{N}{3}$ cosets (evaluating positions) to form maximally recoverable code of block length $N$ and dimension $k$ as long as $2|T(\ell-1)| + |G(\ell-1)| < |W|$. $|W| = \frac{q-1}{3}$. Hence we can form $[N,k=2D+1]$ maximally recoverable code as long as:
		\bean
		2  \Sigma_{j=1}^{D} { |G(\ell-1)| \choose 2j}  3^{2j-1}  j + 2 \Sigma_{j=3}^{D} { |G(\ell-1)| \choose 2j-1} 3^{2j-2} (j-1) + \\
		2 { |G(\ell-1)| \choose 3} 3^{2} + 2 \Sigma_{j=3}^{D} { |G(\ell-1)| \choose 2j-2} 3^{2j-3} (j-1) + |G(\ell-1)| < \frac{q-1}{3}
		\eean
		Using $|G(i)| = i$, it can be seen that the above inequality is implied by:
		\bean
		2\left( \Sigma_{j=2}^{2D} \lfloor{j g(j)} \rfloor {\ell-1 \choose j} 3^{j-1} \right) + (\ell-1)  < \frac{q-1}{3}  
		\eean
		where,
		$g(j) = 1$ for $2(D-1) \geq j \geq 4$ and $j$ even and $g(j) = \frac{1}{2}$ otherwise.
		Hence:
		\bean
		2\left( \Sigma_{j=2}^{2D} \lfloor{j g(j)} \rfloor {(\frac{N}{3}-1)\choose j} 3^{j} \right) + N-2  < q  
		\eean
		The cosets $\{A_{i_1}, A_{i_2}, \cdots, A_{i_c}\}$ given in the Theorem statement are actually the set of cosets in $W-G(\ell)$.
	\end{proof}
\end{appendices}

 \bibliographystyle{IEEEtran}
\bibliography{master}

\begin{thebibliography}{100}
\providecommand{\url}[1]{#1}
\csname url@samestyle\endcsname
\providecommand{\newblock}{\relax}
\providecommand{\bibinfo}[2]{#2}
\providecommand{\BIBentrySTDinterwordspacing}{\spaceskip=0pt\relax}
\providecommand{\BIBentryALTinterwordstretchfactor}{4}
\providecommand{\BIBentryALTinterwordspacing}{\spaceskip=\fontdimen2\font plus
\BIBentryALTinterwordstretchfactor\fontdimen3\font minus
  \fontdimen4\font\relax}
\providecommand{\BIBforeignlanguage}[2]{{%
\expandafter\ifx\csname l@#1\endcsname\relax
\typeout{** WARNING: IEEEtran.bst: No hyphenation pattern has been}%
\typeout{** loaded for the language `#1'. Using the pattern for}%
\typeout{** the default language instead.}%
\else
\language=\csname l@#1\endcsname
\fi
#2}}
\providecommand{\BIBdecl}{\relax}
\BIBdecl

\bibitem{Blaum_3}
M.~Blaum, J.~S. Plank, M.~Schwartz, and E.~Yaakobi, ``{Construction of Partial
  MDS and Sector-Disk Codes With Two Global Parity Symbols},'' \emph{IEEE
  Trans. Inf. Theory}, vol.~62, no.~5, pp. 2673--2681, 2016.

\bibitem{RashmiShahGuKuang}
K.~V. Rashmi, N.~B. Shah, D.~Gu, H.~Kuang, D.~Borthakur, and K.~Ramchandran,
  ``{A Solution to the Network Challenges of Data Recovery in Erasure-coded
  Distributed Storage Systems: {A} Study on the Facebook Warehouse Cluster},''
  in \emph{Proc. 5th {USENIX} Workshop on Hot Topics in Storage and File
  Systems, San Jose, CA, USA, 2013}, 2013.

\bibitem{SonYue}
\BIBentryALTinterwordspacing
W.~Song and C.~Yuen, ``Locally repairable codes with functional repair and
  multiple erasure tolerance,'' \emph{CoRR}, vol. abs/1507.02796, 2015.
  [Online]. Available: \url{http://arxiv.org/abs/1507.02796}
\BIBentrySTDinterwordspacing

\bibitem{TamBarFro}
I.~Tamo, A.~Barg, and A.~Frolov, ``{Bounds on the Parameters of Locally
  Recoverable Codes},'' \emph{IEEE Trans. Inf. Theory}, vol.~62, no.~6, pp.
  3070--3083, 2016.

\bibitem{SathiaAstPap_Xorbas}
M.~Sathiamoorthy, M.~Asteris, D.~S. Papailiopoulos, A.~G. Dimakis, R.~Vadali,
  S.~Chen, and D.~Borthakur, ``{XORing Elephants: Novel Erasure Codes for Big
  Data},'' \emph{{PVLDB}}, vol.~6, no.~5, pp. 325--336, 2013.

\bibitem{DimGodWuWaiRam}
A.~Dimakis, P.~Godfrey, Y.~Wu, M.~Wainwright, and K.~Ramchandran, ``{Network
  coding for distributed storage systems},'' \emph{IEEE Trans. Inf. Theory},
  vol.~56, no.~9, pp. 4539--4551, Sep. 2010.

\bibitem{Wu}
Y.~Wu, ``{Existence and Construction of Capacity-Achieving Network Codes for
  Distributed Storage},'' \emph{{IEEE} Journal on Selected Areas in
  Communications}, vol.~28, no.~2, pp. 277--288, 2010.

\bibitem{GopHuaSimYek}
P.~Gopalan, C.~Huang, H.~Simitci, and S.~Yekhanin, ``{On the Locality of
  Codeword Symbols},'' \emph{IEEE Trans. Inf. Theory}, vol.~58, no.~11, pp.
  6925--6934, 2012.

\bibitem{GuruWootRSRepair17}
V.~Guruswami and M.~Wootters, ``{Repairing Reed-Solomon Codes},'' \emph{{IEEE}
  Trans. Inf. Theory}, vol.~63, no.~9, pp. 5684--5698, 2017.

\bibitem{LubyPadRic}
\BIBentryALTinterwordspacing
M.~G. Luby, R.~Padovani, T.~J. Richardson, L.~Minder, and P.~Aggarwal, ``Liquid
  cloud storage,'' \emph{CoRR}, vol. abs/1705.07983, 2017. [Online]. Available:
  \url{http://arxiv.org/abs/1705.07983}
\BIBentrySTDinterwordspacing

\bibitem{ShaRasKumRam_rbt}
N.~Shah, K.~Rashmi, P.~{Vijay Kumar}, and K.~Ramchandran, ``{Distributed
  Storage Codes With Repair-by-Transfer and Nonachievability of Interior Points
  on the Storage-Bandwidth Tradeoff},'' \emph{IEEE Trans. Inf. Theory},
  vol.~58, no.~3, pp. 1837--1852, Mar. 2012.

\bibitem{Tia}
C.~Tian, ``{Characterizing the Rate Region of the {(4, 3, 3)} Exact-Repair
  Regenerating Codes},'' \emph{{IEEE} Journal on Selected Areas in
  Communications}, vol.~32, no.~5, pp. 967--975, 2014.

\bibitem{SasSenKum_isit}
B.~Sasidharan, K.~Senthoor, and P.~Kumar, ``{An Improved Outer Bound on the
  Storage Repair-Bandwidth Tradeoff of Exact-Repair Regenerating Codes},'' in
  \emph{Proc. {IEEE} International Symposium on Information Theory, Honolulu,
  HI, USA, 2014}, 2014, pp. 2430--2434.

\bibitem{SasPraKriVajSenKum}
B.~Sasidharan, N.~Prakash, M.~N. Krishnan, M.~Vajha, K.~Senthoor, and P.~V.
  Kumar, ``{Outer bounds on the storage-repair bandwidth trade-off of
  exact-repair regenerating codes},'' \emph{International Journal of Inf. and
  Coding Theory}, vol.~3, no.~4, pp. 255--298, 2016.

\bibitem{PawarRouayRam}
S.~Pawar, S.~E. Rouayheb, and K.~Ramchandran, ``{Securing Dynamic Distributed
  Storage Systems Against Eavesdropping and Adversarial Attacks},'' \emph{IEEE
  Trans. on Inf. Theory}, vol.~57, no.~10, pp. 6734--6753, 2011.

\bibitem{RouRamFR10}
S.~Y.~E. Rouayheb and K.~Ramchandran, ``{Fractional Repetition Codes for Repair
  in Distributed Storage Systems},'' \emph{CoRR}, vol. abs/1010.2551, 2010.

\bibitem{HuXuWanZhaLi}
Y.~Hu, Y.~Xu, X.~Wang, C.~Zhan, and P.~Li, ``{Cooperative Recovery of
  Distributed Storage Systems from Multiple Losses with Network Coding},''
  \emph{{IEEE} Journal on Selected Areas in Communications}, vol.~28, no.~2,
  pp. 268--276, 2010.

\bibitem{KerScoStr}
A.~M. Kermarrec, N.~L. Scouarnec, and G.~Straub, ``{Repairing Multiple Failures
  with Coordinated and Adaptive Regenerating Codes},'' in \emph{Proc.
  International Symposium on Networking Coding, Beijing, China, 2011}, 2011,
  pp. 1--6.

\bibitem{ShuHu}
K.~W. Shum and Y.~Hu, ``{Cooperative Regenerating Codes},'' \emph{{IEEE} Trans.
  Inf. Theory}, vol.~59, no.~11, pp. 7229--7258, 2013.

\bibitem{RasShaRam_Piggyback}
K.~V. Rashmi, N.~B. Shah, and K.~Ramchandran, ``{A Piggybacking Design
  Framework for Read-and Download-Efficient Distributed Storage Codes},''
  \emph{{IEEE} Trans. Inf. Theory}, vol.~63, no.~9, pp. 5802--5820, 2017.

\bibitem{RawTamGur_epsilonMSR}
A.~S. Rawat, I.~Tamo, V.~Guruswami, and K.~Efremenko, ``{$\epsilon$}-{MSR}
  codes with small sub-packetization,'' in \emph{Proc. {IEEE} International
  Symposium on Information Theory, Aachen, Germany, 2017}, 2017, pp.
  2043--2047.

\bibitem{HanMon}
J.~Han and L.~A. Lastras-Montano, ``{Reliable Memories with Subline
  Accesses},'' in \emph{Proc. IEEE International Symposium on Information
  Theory, Nice, France, 2007}, June 2007, pp. 2531--2535.

\bibitem{HuaChenLi}
C.~Huang, M.~Chen, and J.~Li, ``{Pyramid codes: Flexible schemes to trade space
  for access efficiency in reliable data storage systems},'' in \emph{{Proc.
  6th IEEE Int. Symposium on Network Computing and Applications, Cambridge,
  Massachusetts, USA, 2007}}, 2007, pp. 79--86.

\bibitem{OggDat}
F.~Oggier and A.~Datta, ``{Self-repairing homomorphic codes for distributed
  storage systems},'' in \emph{Proc. IEEE INFOCOM, Shanghai, China,2011}, April
  2011, pp. 1215--1223.

\bibitem{PapDim}
D.~Papailiopoulos and A.~Dimakis, ``{Locally repairable codes},'' in
  \emph{Proc. IEEE International Symposium on Information Theory, Cambridge,
  MA, USA, 2012}, July 2012, pp. 2771--2775.

\bibitem{PraKamLalKum}
N.~Prakash, G.~M. Kamath, V.~Lalitha, and P.~V. Kumar, ``{Optimal linear codes
  with a local-error-correction property},'' in \emph{Proc. IEEE International
  Symposium on Information Theory Proceedings, Cambridge, MA, USA, 2012}, 2012,
  pp. 2776--2780.

\bibitem{SonDauYueLi}
W.~Song, S.~H. Dau, C.~Yuen, and T.~J. Li, ``Optimal locally repairable linear
  codes,'' \emph{IEEE Journal on Selected Areas in Communications}, vol.~32,
  no.~5, pp. 1019--1036, 2014.

\bibitem{PraLalKum}
N.~Prakash, V.~Lalitha, and P.~V. Kumar, ``Codes with locality for two
  erasures,'' in \emph{Proc. IEEE International Symposium on Information
  Theory, Honolulu, HI, USA, 2014}, 2014, pp. 1962--1966.

\bibitem{BalPraKum}
\BIBentryALTinterwordspacing
S.~B. Balaji, K.~P. Prasanth, and P.~V. Kumar, ``Binary codes with locality for
  multiple erasures having short block length,'' \emph{CoRR}, 2016. [Online].
  Available: \url{http://arxiv.org/abs/1601.07122/}
\BIBentrySTDinterwordspacing

\bibitem{SongCaiYue_L}
W.~Song, K.~Cai, C.~Yuen, K.~Cai, and G.~Han, ``On sequential locally
  repairable codes,'' \emph{IEEE Trans. Inf. Theory}, vol.~PP, no.~99, pp.
  1--1, 2017.

\bibitem{BalKinKum_ISIT}
S.~B. Balaji, G.~R. Kini, and P.~V. Kumar, ``A tight rate bound and a matching
  construction for locally recoverable codes with sequential recovery from any
  number of multiple erasures,'' in \emph{2017 IEEE International Symposium on
  Information Theory (ISIT)}, June 2017, pp. 1778--1782.

\bibitem{BalKinKum_NCC}
------, ``{A Rate-Optimal Construction of Codes with Sequential Recovery with
  Low Block Length},'' in \emph{Proc. National Conference on Communications,
  Hyderabad, India, 2018}, 2018.

\bibitem{WanZha}
A.~Wang and Z.~Zhang, ``{Repair Locality With Multiple Erasure Tolerance},''
  \emph{IEEE Trans. Inf. Theory}, vol.~60, no.~11, pp. 6979--6987, 2014.

\bibitem{WanZhaLiu}
A.~Wang, Z.~Zhang, and M.~Liu, ``Achieving arbitrary locality and availability
  in binary codes,'' in \emph{Proc. IEEE International Symposium on Information
  Theory, Hong Kong, 2015}, 2015, pp. 1866--1870.

\bibitem{BalKum}
S.~B. Balaji and P.~V. Kumar, ``Bounds on the rate and minimum distance of
  codes with availability,'' in \emph{2017 IEEE International Symposium on
  Information Theory (ISIT)}, June 2017, pp. 3155--3159.

\bibitem{RawMazVis}
A.~S. Rawat, A.~Mazumdar, and S.~Vishwanath, ``{Cooperative Local Repair in
  Distributed Storage},'' \emph{CoRR}, vol. abs/1409.3900, 2014.

\bibitem{KamPraLalKum}
G.~Kamath, N.~Prakash, V.~Lalitha, and P.~Kumar, ``Codes with local
  regeneration,'' in \emph{Information Theory and Applications Workshop (ITA),
  2013}, Feb 2013, pp. 1--5.

\bibitem{ShanPapDimCaiMDSRepair14}
K.~Shanmugam, D.~S. Papailiopoulos, A.~G. Dimakis, and G.~Caire, ``A repair
  framework for scalar {MDS} codes,'' \emph{{IEEE} Journal on Selected Areas in
  Communications}, vol.~32, no.~5, pp. 998--1007, 2014.

\bibitem{HuaSimXu_Azure}
C.~Huang, H.~Simitci, Y.~Xu, A.~Ogus, B.~Calder, P.~Gopalan, J.~Li, and
  S.~Yekhanin, ``Erasure coding in windows azure storage,'' in \emph{Proc.
  {USENIX} Annual Technical Conference, Boston, MA, USA, 2012}, 2012, pp.
  15--26.

\bibitem{ForYek}
M.~Forbes and S.~Yekhanin, ``{On the Locality of Codeword Symbols in Non-linear
  Codes},'' \emph{Discrete Math.}, vol. 324, pp. 78--84, Jun. 2014.

\bibitem{SonYue_Prdr}
\BIBentryALTinterwordspacing
W.~Song and C.~Yuen, ``Binary locally repairable codes - sequential repair for
  multiple erasures,'' \emph{CoRR}, vol. abs/1511.06034, 2015. [Online].
  Available: \url{http://arxiv.org/abs/1511.06034}
\BIBentrySTDinterwordspacing

\bibitem{BalPraKum4Era}
\BIBentryALTinterwordspacing
S.~B. Balaji, K.~P. Prasanth, and P.~V. Kumar, ``Binary codes with locality for
  four erasures,'' \emph{CoRR}, vol. abs/1607.02817, 2016. [Online]. Available:
  \url{http://arxiv.org/abs/1607.02817}
\BIBentrySTDinterwordspacing

\bibitem{HuaYaaUchSie}
P.~Huang, E.~Yaakobi, H.~Uchikawa, and P.~H. Siegel, ``Binary linear locally
  repairable codes,'' \emph{IEEE Trans. Inf. Theory}, vol.~62, no.~11, pp.
  6268--6283, Nov 2016.

\bibitem{WanZhaLin1}
A.~Wang, Z.~Zhang, and D.~Lin, ``Two classes of (r, t)-locally repairable
  codes,'' in \emph{2016 IEEE International Symposium on Information Theory
  (ISIT)}, July 2016, pp. 445--449.

\bibitem{SilZeh}
\BIBentryALTinterwordspacing
N.~Silberstein and A.~Zeh, ``Anticode-based locally repairable codes with high
  availability,'' \emph{Designs, Codes and Cryptography}, vol.~86, no.~2, pp.
  419--445, Feb 2018. [Online]. Available:
  \url{https://doi.org/10.1007/s10623-017-0358-0}
\BIBentrySTDinterwordspacing

\bibitem{TamBar_LRC}
I.~Tamo and A.~Barg, ``A family of optimal locally recoverable codes,''
  \emph{IEEE Trans. Inf. Theory}, vol.~60, no.~8, pp. 4661--4676, Aug 2014.

\bibitem{BarTamVla}
A.~Barg, I.~Tamo, and S.~Vlăduţ, ``{Locally Recoverable Codes on Algebraic
  Curves},'' \emph{IEEE Trans. Inf. Theory}, vol.~63, no.~8, pp. 4928--4939,
  2017.

\bibitem{KruFro}
S.~Kruglik and A.~Frolov, ``Bounds and constructions of codes with all-symbol
  locality and availability,'' in \emph{2017 IEEE International Symposium on
  Information Theory (ISIT)}, June 2017, pp. 1023--1027.

\bibitem{BhaTha}
S.~Bhadane and A.~Thangaraj, ``Unequal locality and recovery for locally
  recoverable codes with availability,'' in \emph{2017 Twenty-third National
  Conference on Communications (NCC)}, March 2017, pp. 1--6.

\bibitem{KadCal}
S.~Kadhe and A.~R. Calderbank, ``Rate optimal binary linear locally repairable
  codes with small availability,'' \emph{CoRR}, vol. abs/1701.02456, 2017.

\bibitem{TamWanBru_access_tit}
I.~Tamo, Z.~Wang, and J.~Bruck, ``{Access Versus Bandwidth in Codes for
  Storage},'' \emph{{IEEE} Trans. Inf. Theory}, vol.~60, no.~4, pp. 2028--2037,
  2014.

\bibitem{GopTamCal}
S.~Goparaju, I.~Tamo, and A.~R. Calderbank, ``{An Improved Sub-Packetization
  Bound for Minimum Storage Regenerating Codes},'' \emph{{IEEE} Trans. on Inf.
  Theory}, vol.~60, no.~5, pp. 2770--2779, 2014.

\bibitem{HuaparXia}
\BIBentryALTinterwordspacing
K.~Huang, U.~Parampalli, and M.~Xian, ``Improved upper bounds on
  systematic-length for linear minimum storage regenerating codes,''
  \emph{CoRR}, vol. abs/1610.08026, 2016. [Online]. Available:
  \url{http://arxiv.org/abs/1610.08026}
\BIBentrySTDinterwordspacing

\bibitem{microsoft}
``Microsoft research blog: A better way to store data,''
  \url{https://www.microsoft.com/en-us/research/blog/better-way-store-data/}.

\bibitem{lrc_plugin}
``Locally repairable erasure code plugin,''
  \url{http://docs.ceph.com/docs/master/rados/operations/erasure-code-lrc/}.

\bibitem{RasShaGu_Hitchhiker}
K.~V. Rashmi, N.~B. Shah, D.~Gu, H.~Kuang, D.~Borthakur, and K.~Ramchandran,
  ``A "hitchhiker's" guide to fast and efficient data reconstruction in
  erasure-coded data centers,'' in \emph{Proc. {ACM} {SIGCOMM} Conference,
  Chicago, IL, USA, 2014}, 2014, pp. 331--342.

\bibitem{KraGliJen_HashTag}
K.~Kralevska, D.~Gligoroski, R.~E. Jensen, and H.~{\O}verby, ``Hashtag erasure
  codes: From theory to practice,'' \emph{IEEE Transactions on Big Data}, 2017.

\bibitem{KraGliOVe}
K.~Kralevska, D.~Gligoroski, and H.~{\O}verby, ``General sub-packetized
  access-optimal regenerating codes,'' \emph{{IEEE} Communications Letters},
  vol.~20, no.~7, pp. 1281--1284, 2016.

\bibitem{HuChenLee_NCcloud}
Y.~Hu, H.~C.~H. Chen, P.~P.~C. Lee, and Y.~Tang, ``{NCCloud: applying network
  coding for the storage repair in a cloud-of-clouds},'' in \emph{Proc. 10th
  {USENIX} conference on File and Storage Technologies, San Jose, CA, USA,
  2012}, 2012, p.~21.

\bibitem{KriPraLal_Pentagon}
M.~N. Krishnan, N.~Prakash, V.~Lalitha, B.~Sasidharan, P.~V. Kumar,
  S.~Narayanamurthy, R.~Kumar, and S.~Nandi, ``{Evaluation of Codes with
  Inherent Double Replication for Hadoop},'' in \emph{Proc. 6th {USENIX}
  Workshop on Hot Topics in Storage and File Systems, Philadelphia, PA, USA,
  2014.}, 2014.

\bibitem{RasNakWan_PMRBT}
K.~V. Rashmi, P.~Nakkiran, J.~Wang, N.~B. Shah, and K.~Ramchandran, ``{Having
  Your Cake and Eating It Too: Jointly Optimal Erasure Codes for I/O, Storage,
  and Network-bandwidth},'' in \emph{Proc. 13th {USENIX} Conference on File and
  Storage Technologies, Santa Clara, CA, USA, 2015}, 2015, pp. 81--94.

\bibitem{JuaBlaMat_Butterfly}
L.~Pamies{-}Juarez, F.~Blagojevic, R.~Mateescu, C.~Guyot, E.~E. Gad, and
  Z.~Bandic, ``Opening the chrysalis: On the real repair performance of {MSR}
  codes,'' in \emph{Proc. 14th {USENIX} Conference on File and Storage
  Technologies, Santa Clara, CA, USA, 2016}, 2016, pp. 81--94.

\bibitem{GadMatBla}
E.~E. Gad, R.~Mateescu, F.~Blagojevic, C.~Guyot, and Z.~Bandic,
  ``Repair-optimal {MDS} array codes over {GF(2)},'' in \emph{Proc. {IEEE}
  International Symposium on Information Theory, Istanbul, Turkey, 2013}, 2013,
  pp. 887--891.

\bibitem{VajRamPur_Clay}
M.~Vajha, V.~Ramkumar, B.~Puranik, G.~R. Kini, E.~Lobo, B.~Sasidharan, P.~V.
  Kumar, A.~Barg, M.~Ye, S.~Narayanamurthy, S.~Hussain, and S.~Nandi, ``Clay
  codes: Moulding {MDS} codes to yield an {MSR} code,'' in \emph{Proc. 16th
  {USENIX} Conference on File and Storage Technologies, Oakland, CA, USA,
  2018}, 2018, pp. 139--154.

\bibitem{YeBar_2}
M.~Ye and A.~Barg, ``{Explicit Constructions of Optimal-Access {MDS} Codes With
  Nearly Optimal Sub-Packetization},'' \emph{{IEEE} Trans. Inf. Theory},
  vol.~63, no.~10, pp. 6307--6317, 2017.

\bibitem{SasVajKum_arxiv}
B.~Sasidharan, M.~Vajha, and P.~V. Kumar, ``{An Explicit, Coupled-Layer
  Construction of a High-Rate {MSR} Code with Low Sub-Packetization Level,
  Small Field Size and All-Node Repair},'' \emph{CoRR}, vol. abs/1607.07335,
  2016.

\bibitem{Sihem_Mesnager}
J.~Liu, S.~Mesnager, and L.~Chen, ``{New Constructions of Optimal Locally
  Recoverable Codes via Good Polynomials},'' \emph{IEEE Trans. Inf. Theory},
  vol.~64, no.~2, pp. 889--899, 2018.

\bibitem{KolBarTamYad}
O.~Kolosov, A.~Barg, I.~Tamo, and G.~Yadgar, ``{Optimal LRC codes for all
  lenghts $n <= q$},'' \emph{CoRR}, vol. abs/1802.00157, 2018.

\bibitem{LinLimCha}
L.~Jin, L.~Ma, and C.~Xing, ``Construction of optimal locally repairable codes
  via automorphism groups of rational function fields,'' \emph{CoRR}, vol.
  abs/1710.09638, 2017.

\bibitem{SilRawKoyVis}
N.~Silberstein, A.~S. Rawat, O.~O. Koyluoglu, and S.~Vishwanath, ``{Optimal
  locally repairable codes via rank-metric codes},'' in \emph{Proc. IEEE
  International Symposium on Information Theory ,Istanbul, Turkey, 2013}, July
  2013, pp. 1819--1823.

\bibitem{CadMaz}
V.~R. Cadambe and A.~Mazumdar, ``{Bounds on the Size of Locally Recoverable
  Codes},'' \emph{IEEE Trans. Inf. Theory}, vol.~61, no.~11, pp. 5787--5794,
  Nov 2015.

\bibitem{WanZhaLin}
A.~Wang, Z.~Zhang, and D.~Lin, ``{Bounds and constructions for linear locally
  repairable codes over binary fields},'' in \emph{Proc. IEEE International
  Symposium on Information Theory, Aachen, Germany, 2017}, 2017, pp.
  2033--2037.

\bibitem{Wei}
V.~K. Wei, ``{Generalized Hamming weights for linear codes},'' \emph{IEEE
  Trans. Inf. Theory}, vol.~37, no.~5, pp. 1412--1418, 1991.

\bibitem{HelKloLevYtr}
T.~Helleseth, T.~Klove, V.~Levenshtein, and O.~Ytrehus, ``{Bounds on the
  minimum support weights},'' \emph{IEEE Trans. Inf. Theory}, vol.~41, no.~2,
  pp. 432--440, 1995.

\bibitem{NamSo}
M.~Y. Nam and H.~Y. Song, ``{Binary Locally Repairable Codes With Minimum
  Distance at Least Six Based on Partial $t$ -Spreads},'' \emph{IEEE
  Communications Letters}, vol.~21, no.~8, pp. 1683--1686, Aug 2017.

\bibitem{HaoXiaChe}
J.~Hao, S.~T. Xia, and B.~Chen, ``Some results on optimal locally repairable
  codes,'' in \emph{Proc. IEEE International Symposium on Information Theory,
  Barcelona, Spain, 2016}, 2016, pp. 440--444.

\bibitem{MaGe}
J.~Ma and G.~Ge, ``{Optimal binary linear locally repairable codes with
  disjoint repair groups},'' \emph{CoRR}, vol. abs/1711.07138, 2017.

\bibitem{ShaKhaArd}
M.~Shahabinejad, M.~Khabbazian, and M.~Ardakani, ``{A Class of Binary Locally
  Repairable Codes},'' \emph{IEEE Transactions on Communications}, vol.~64,
  no.~8, pp. 3182--3193, 2016.

\bibitem{HaoXiaChe1}
J.~Hao, S.~T. Xia, and B.~Chen, ``On optimal ternary locally repairable
  codes,'' in \emph{Proc. IEEE International Symposium on Information Theory,
  Aachen, Germany, 2017}, 2017, pp. 171--175.

\bibitem{HaoXia}
J.~Hao and S.~Xia, ``{Bounds and Constructions of Locally Repairable Codes:
  Parity-check Matrix Approach},'' \emph{CoRR}, vol. abs/1601.05595, 2016.

\bibitem{LiMaXin1}
X.~Li, L.~Ma, and C.~Xing, ``Optimal locally repairable codes via elliptic
  curves,'' \emph{CoRR}, vol. abs/1712.03744, 2017.

\bibitem{GopCal}
S.~Goparaju and A.~R. Calderbank, ``Binary cyclic codes that are locally
  repairable,'' in \emph{Proc. {IEEE} International Symposium on Information
  Theory, Honolulu, HI, USA, 2014}, 2014, pp. 676--680.

\bibitem{TambarGopCal}
I.~Tamo, A.~Barg, S.~Goparaju, and A.~R. Calderbank, ``Cyclic {LRC} codes,
  binary {LRC} codes, and upper bounds on the distance of cyclic codes,''
  \emph{CoRR}, vol. abs/1603.08878, 2016.

\bibitem{ZehYak}
A.~Zeh and E.~Yaakobi, ``Optimal linear and cyclic locally repairable codes
  over small fields,'' in \emph{Proc. IEEE Information Theory Workshop,
  Jerusalem, Israel, 2015}, 2015, pp. 1--5.

\bibitem{KimNo}
C.~Kim and J.~S. No, ``{New Constructions of Binary and Ternary Locally
  Repairable Codes Using Cyclic Codes},'' \emph{IEEE Communications Letters},
  vol.~22, no.~2, pp. 228--231, 2018.

\bibitem{LuoXinYua}
Y.~Luo, C.~Xing, and C.~Yuan, ``Optimal locally repairable codes of distance 3
  and 4 via cyclic codes,'' \emph{CoRR}, vol. abs/1801.03623, 2018.

\bibitem{KriPurKumTamBarISIT17}
M.~N. Krishnan, B.~Puranik, P.~V. Kumar, I.~Tamo, and A.~Barg, ``A study on the
  impact of locality in the decoding of binary cyclic codes,'' in \emph{Proc.
  IEEE International Symposium on Information Theory, Aachen, Germany, 2017},
  2017, pp. 176--180.

\bibitem{VardyBeery94}
A.~Vardy and Y.~Be'ery, ``Maximum-likelihood soft decision decoding of {BCH}
  codes,'' \emph{{IEEE} Trans. Inf. Theory}, vol.~40, no.~2, pp. 546--554,
  1994.

\bibitem{HuaYakUchSie}
P.~Huang, E.~Yaakobi, H.~Uchikawa, and P.~H. Siegel, ``Cyclic linear binary
  locally repairable codes,'' in \emph{Proc. IEEE Information Theory Workshop,
  Jerusalem, Israel, 2015}, 2015, pp. 1--5.

\bibitem{ErdosGallai}
P.~Erdos and T.~Gallai, ``Graphs with prescribed degrees of vertices
  (hungarian), mat. lapok,'' pp. 264--274, 1960.

\bibitem{BurKriLitMil}
D.~Burshtein, M.~Krivelevich, S.~Litsyn, and G.~Miller, ``Upper bounds on the
  rate of ldpc codes,'' \emph{IEEE Trans. Inf. Theory}, vol.~48, no.~9, pp.
  2437--2449, 2002.

\bibitem{BenLit}
Y.~Ben-Haim and S.~Litsyn, ``Upper bounds on the rate of ldpc codes as a
  function of minimum distance,'' \emph{IEEE Trans. Inf. Theory}, vol.~52,
  no.~5, pp. 2092--2100, 2006.

\bibitem{IceSam}
E.~Iceland and A.~Samorodnitsk, ``On coset leader graphs of ldpc codes,''
  \emph{IEEE Trans. Inf. Theory}, vol.~61, no.~8, pp. 4158--4163, 2015.

\bibitem{Fro}
A.~Frolov, ``An upper bound on the minimum distance of ldpc codes over gf(q),''
  \emph{{IEEE} Int. Symp. on Inf. Theory (ISIT)}, pp. 2885--2888, 2015.

\bibitem{Viz}
V.~G. Vizing, ``On an estimate of the chromatic class of a p-graph (russian),''
  \emph{Diskret. Analiz}, no.~3, pp. 25--30, 1964.

\bibitem{dieRich}
\BIBentryALTinterwordspacing
R.~Diestel, \emph{Graph Theory}.\hskip 1em plus 0.5em minus 0.4em\relax
  Springer, 2000. [Online]. Available:
  \url{http://www.esi2.us.es/~mbilbao/pdffiles/DiestelGT.pdf}
\BIBentrySTDinterwordspacing

\bibitem{Lubotzky1988}
\BIBentryALTinterwordspacing
A.~Lubotzky, R.~Phillips, and P.~Sarnak, ``Ramanujan graphs,''
  \emph{Combinatorica}, vol.~8, no.~3, pp. 261--277, Sep 1988. [Online].
  Available: \url{https://doi.org/10.1007/BF02126799}
\BIBentrySTDinterwordspacing

\bibitem{X_Dahan}
\BIBentryALTinterwordspacing
X.~Dahan, ``Regular graphs of large girth and arbitrary degree,'' 2011.
  [Online]. Available: \url{https://arxiv.org/abs/1110.5259}
\BIBentrySTDinterwordspacing

\bibitem{Mor}
\BIBentryALTinterwordspacing
M.~Morgenstern, ``Existence and explicit constructions of q + 1 regular
  ramanujan graphs for every prime power q,'' \emph{J. Comb. Theory Ser. B},
  vol.~62, no.~1, pp. 44--62, Sep. 1994. [Online]. Available:
  \url{http://dx.doi.org/10.1006/jctb.1994.1054}
\BIBentrySTDinterwordspacing

\bibitem{DavSarVal}
G.~Davidoff, P.~Sarnak, and A.~Valette, ``Elementary number theory, group
  theory, and ramanujan graphs,'' \emph{London Mathematical Society Student
  Texts}, vol.~55, 2003.

\bibitem{LazUstWol}
F.~Lazebnik, V.~A. Ustimenko, and A.~J. Woldar, ``A new series of dense graphs
  of high girth,'' \emph{Bulletin of the American Mathematical Society},
  vol.~32, no.~1, pp. 73--79, 1995.

\bibitem{DynCageSur}
G.~Exoo and R.~Jajcay, ``{Dynamic Cage Survey},'' \emph{Electronic Journal
  Combinatorics, Dynamic Survey, DS16}, 2013.

\bibitem{BelGueMenSal}
\BIBentryALTinterwordspacing
E.~Bellini, E.~Guerrini, A.~Meneghetti, and M.~Sala, ``On the griesmer bound
  for nonlinear codes,'' \emph{CoRR}, vol. abs/1502.07379, 2015. [Online].
  Available: \url{http://arxiv.org/abs/1502.07379}
\BIBentrySTDinterwordspacing

\bibitem{Far}
P.~G. Farrell, ``Linear binary anticodes,'' \emph{Electronics Letters}, vol.~6,
  no.~13, pp. 419--421, June 1970.

\bibitem{BalKumAvail}
\BIBentryALTinterwordspacing
S.~B. Balaji and P.~V. Kumar, ``Bounds on codes with locality and
  availability,'' \emph{CoRR}, vol. abs/1611.00159, 2016. [Online]. Available:
  \url{http://arxiv.org/abs/1611.00159}
\BIBentrySTDinterwordspacing

\bibitem{Nog}
N.~Alon, ``Combinatorial nullstellensatz,'' \emph{COMBINATORICS, PROBABILITY
  AND COMPUTING}, vol.~8, pp. 7--29, 1999.

\bibitem{ColDin}
C.~J. Colbourn and J.~H. Dinitz, \emph{Handbook of Combinatorial Designs,
  Second Edition (Discrete Mathematics and Its Applications)}.\hskip 1em plus
  0.5em minus 0.4em\relax Chapman and Hall/CRC, 2006.

\bibitem{Smith}
K.J.C.Smith, ``On the p-rank of the incidence matrix of points and hyperplanes
  in a finite projective geometry,'' \emph{Journal of Combinatorial Theory},
  vol.~7, pp. 122--129, 09 1969.

\bibitem{AssKey}
E.~F. Assmus and J.~D. Key, \emph{Designs and their Codes (Cambridge Tracts in
  Mathematics)}.\hskip 1em plus 0.5em minus 0.4em\relax Cambridge University
  Press, 1994.

\bibitem{Doyen}
\BIBentryALTinterwordspacing
H.~X. V.~M. Doyen, Jean, ``Ranks of incidence matrices of steiner triple
  systems.'' \emph{Mathematische Zeitschrift}, vol. 163, pp. 251--260, 1978.
  [Online]. Available: \url{http://eudml.org/doc/172756}
\BIBentrySTDinterwordspacing

\bibitem{VajBalKum}
M.~Vajha, S.~B. Balaji, and P.~V. Kumar, ``{Explicit MSR Codes with Optimal
  Access, Optimal Sub-Packetization and Small Field Size for $d = k + 1, k + 2,
  k + 3$},'' \emph{CoRR (Accepted at ISIT 2018)}, vol. abs/1804.00598, {2018}.

\bibitem{ShaRasKumRam_ia}
N.~B. Shah, K.~V. Rashmi, P.~V. Kumar, and K.~Ramchandran, ``{Interference
  Alignment in Regenerating Codes for Distributed Storage: Necessity and Code
  Constructions},'' \emph{IEEE Trans. Inf. Theory}, vol.~58, no.~4, pp.
  2134--2158, Apr. 2012.

\bibitem{SuhRam}
C.~Suh and K.~Ramchandran, ``{Exact-repair MDS code construction using
  interference alignment},'' \emph{IEEE Trans. Inf. Theory}, vol.~57, no.~3,
  pp. 1425--1442, Mar. 2011.

\bibitem{RasShaKum_pm}
K.~V. Rashmi, N.~B. Shah, and P.~V. Kumar, ``{Optimal Exact-Regenerating Codes
  for Distributed Storage at the MSR and MBR Points via a Product-Matrix
  Construction},'' \emph{IEEE Trans. Inf. Theory}, vol.~57, no.~8, pp.
  5227--5239, Aug. 2011.

\bibitem{PapDimCad}
D.~Papailiopoulos, A.~Dimakis, and V.~Cadambe, ``{Repair Optimal Erasure Codes
  through Hadamard Designs},'' \emph{IEEE Trans. Inf. Theory}, vol.~59, no.~5,
  pp. 3021--3037, 2013.

\bibitem{TamWanBru}
I.~Tamo, Z.~Wang, and J.~Bruck, ``Zigzag codes: {MDS} array codes with optimal
  rebuilding,'' \emph{IEEE Trans. Inf. Theory}, vol.~59, no.~3, pp. 1597--1616,
  2013.

\bibitem{WangTamoBruck}
Z.~Wang, I.~Tamo, and J.~Bruck, ``On codes for optimal rebuilding access,'' in
  \emph{Proc. 49th Annual Allerton Conference on Communication, Control, and
  Computing 2011}, Sept 2011, pp. 1374--1381.

\bibitem{YeBar_1}
M.~Ye and A.~Barg, ``{Explicit Constructions of High-Rate MDS Array Codes With
  Optimal Repair Bandwidth},'' \emph{{IEEE} Trans. Inf. Theory}, vol.~63,
  no.~4, pp. 2001--2014, 2017.

\bibitem{LiTangTian}
J.~Li, X.~Tang, and C.~Tian, ``{A generic transformation for optimal repair
  bandwidth and rebuilding access in MDS codes},'' in \emph{Proc. IEEE
  International Symposium on Information Theory, Aachen, Germany, 2017}, June
  2017, pp. 1623--1627.

\bibitem{CadJafMalRamSuh}
V.~Cadambe, S.~A. Jafar, H.~Maleki, K.~Ramchandran, and C.~Suh, ``{Asymptotic
  Interference Alignment for Optimal Repair of MDS Codes in Distributed
  Storage},'' \emph{IEEE Trans. Inf. Theory}, vol.~59, no.~5, pp. 2974--2987,
  2013.

\bibitem{BalKum_subpkt}
S.~B. Balaji and P.~V. Kumar, ``A tight lower bound on the sub-packetization
  level of optimal-access {MSR} and {MDS} codes,'' \emph{CoRR, (Accepted at
  ISIT 2018)}, vol. abs/1710.05876, 2017.

\bibitem{RawKoyVis_msr}
A.~S. Rawat, O.~O. Koyluoglu, and S.~Vishwanath, ``{Progress on high-rate {MSR}
  codes: Enabling arbitrary number of helper nodes},'' in \emph{Proc.
  Information Theory and Applications Workshop, La Jolla, CA, USA, 2016}, 2016,
  pp. 1--6.

\bibitem{GopFazVar}
S.~Goparaju, A.~Fazeli, and A.~Vardy, ``{Minimum Storage Regenerating Codes for
  All Parameters},'' \emph{{IEEE} Trans. Inf. Theory}, vol.~63, no.~10, pp.
  6318--6328, 2017.

\bibitem{Cal_Koy}
G.~Calis and O.~O. Koyluoglu, ``{A General Construction for PMDS Codes},''
  \emph{IEEE Communications Letters}, vol.~21, no.~3, pp. 452--455, 2017.

\bibitem{Gab_Yak_Bla_Sie}
R.~Gabrys, E.~Yaakobi, M.~Blaum, and P.~H. Siegel, ``{Constructions of partial
  MDS codes over small fields},'' in \emph{Proc. IEEE International Symposium
  on Information Theory, Aachen, Germany, 2017}, 2017, pp. 1--5.

\bibitem{Bla_Haf_Het}
M.~Blaum, J.~L. Hafner, and S.~Hetzler, ``{Partial-MDS Codes and Their
  Application to RAID Type of Architectures},'' \emph{IEEE Trans. Inf. Theory},
  vol.~59, no.~7, pp. 4510--4519, 2013.

\bibitem{GopHuaJenYek}
P.~Gopalan, C.~Huang, B.~Jenkins, and S.~Yekhanin, ``{Explicit Maximally
  Recoverable Codes With Locality},'' \emph{IEEE Trans. Inf. Theory}, vol.~60,
  no.~9, pp. 5245--5256, 2014.

\bibitem{Hu_Yek}
G.~Hu and S.~Yekhanin, ``{New constructions of SD and MR codes over small
  finite fields},'' in \emph{Proc. IEEE International Symposium on Information
  Theory, Barcelona, Spain, 2016}, 2016, pp. 1591--1595.

\bibitem{Che_Shu_Yu_Sun}
J.~Chen, K.~W. Shum, Q.~Yu, and C.~W. Sung, ``{Sector-disk codes and partial
  MDS codes with up to three global parities},'' in \emph{Proc. IEEE
  International Symposium on Information Theory, Hong Kong, 2015}, 2015, pp.
  1876--1880.

\bibitem{Blaum_1}
M.~Blaum, ``Construction of {PMDS} and {SD} codes extending {RAID} 5,''
  \emph{CoRR}, vol. abs/1305.0032, 2013.

\bibitem{LalLok}
V.~Lalitha and S.~V. Lokam, ``Weight enumerators and higher support weights of
  maximally recoverable codes,'' in \emph{Proc. 53rd Annual Allerton Conference
  on Communication, Control, and Computing, Monticello, IL, USA, 2015}, 2015,
  pp. 835--842.

\bibitem{Rob}
H.~Robbins, ``A remark on stirling's formula,'' \emph{The American Mathematical
  Monthly}, vol.~62, no.~1, pp. 26--29, 1955.

\bibitem{Chen}
M.~Chen, C.~Huang, and J.~Li, ``On the maximally recoverable property for
  multi-protection group codes,'' in \emph{2007 IEEE International Symposium on
  Information Theory}, June 2007, pp. 486--490.

\end{thebibliography}

\end{document}